\documentclass[runningheads,orivec]{llncs}
\newcommand{\macrospath}{.}

\usepackage{ifthen}

\newboolean{talk}
\setboolean{talk}{false}
\newboolean{paper}
\setboolean{paper}{false}

\newboolean{IEEEstyle}
\setboolean{IEEEstyle}{false}
\newboolean{lipicsstyle}
\setboolean{lipicsstyle}{false}
\newboolean{eptcsstyle}
\setboolean{eptcsstyle}{false}
\newboolean{entcsstyle}
\setboolean{entcsstyle}{false}
\newboolean{sigplanstyle}
\setboolean{sigplanstyle}{false}
\newboolean{easychairstyle}
\setboolean{easychairstyle}{false}

\newboolean{lncsstyle}
\setboolean{lncsstyle}{false}

\newboolean{needstheorems}
\setboolean{needstheorems}{false}
\newboolean{withimages}
\setboolean{withimages}{false}
\newboolean{withproofs}
\setboolean{withproofs}{true}

\newboolean{french}
\setboolean{french}{false}

\setboolean{lncsstyle}{true}
\usepackage{amssymb}
\usepackage{amsmath}
\usepackage{stmaryrd}
\usepackage{proof}
\usepackage{multicol}
\usepackage{mathtools}
\usepackage{yfonts}
\usepackage{appendix}
\usepackage{compactspacing}
\usepackage{xcolor}
\usepackage{fancybox}
\usepackage[normalem]{ulem}
\usepackage{nameref}
\usepackage{cmll}
\usepackage{appendix}
\usepackage{rotating}

\ifthenelse{\boolean{easychairstyle}}{
\setboolean{needstheorems}{true}}{}

\renewcommand{\oc}{{\tt !}}

\newcommand\dom[1]{{\tt dom}({#1})}

\renewcommand\l{\lambda}
\newcommand\eqdef{\eqqcolon}

\newcommand\normal{{\sf{normal}}\xspace}
\newcommand\normalpr[1]{\normal({#1})}
\newcommand{\shortnormal}{\sf{n}}
\newcommand{\cbnnormaltonormal}{\shortnormal^{\mult{\shortnormal}}}
\newcommand\normalcbv{{\normal_\cbvsym}\xspace}
\newcommand\normalcbvpr[1]{\normalcbv({#1})}

\newcommand\col{ \!:\! }

\newcommand\Deribbase[5]{{#3}\ {\pmb\vdash}_{#2}^{#1} {#4} \hastype  {#5}}

\makeatletter

\newcommand{\Deribase}[1]{%
  \def\DeribW[##1]{\Deribbase{##1}{#1}}%
  \def\DeribWO{\Deribbase{}{#1}}%
  \@ifnextchar[\DeribW\DeribWO%
  }

  \newcommand{\Deri}{%
  \def\DeriW_##1{\Deribase{##1}}%
  \def\DeriWO{\Deribase{}}%
  \@ifnextchar_\DeriW\DeriWO%
  }
\makeatother

\newcommand\single[1]{\mult{#1}}

\newcommand\MSigma[2]{[#1]_{{#2}}}

\newcommand{\steps}{b}
\newcommand{\msteps}{m}

\newcommand{\esteps}{e}
\newcommand{\estepstwo}{\esteps'}
\newcommand{\estepsthree}{\esteps''}

\newcommand{\mstepstwo}{\msteps'}
\newcommand{\mstepsthree}{\msteps''}

\newcommand{\result}{r}

\newcommand{\mtype}{M}
\newcommand{\mtypetwo}{N}
\newcommand{\mtypethree}{O}
\newcommand{\mtypefour}{P}

\newcommand{\M}{\mtype}

\newcommand{\type}{L}
\newcommand{\typetwo}{{\type'}}

\newcommand{\I}{I}
\newcommand{\J}{J}
\newcommand{\K}{K}
\newcommand{\iI}{{i \in \I}}
\newcommand{\jJ}{{j \in \J}}
\newcommand{\kK}{{k \in \K}}

\makeatletter
\newcommand{\exder}{%
  \def\exderW[##1]{\triangleright_{##1}\ }%
  \def\exderWO{\triangleright\ }%
  \@ifnextchar[\exderW\exderWO%
  }
\makeatother

\newcommand{\tight}{{\tt tight}}

\renewcommand\single[1]{[#1]}

\newcommand{\fun}{{\tt fun}\xspace}

\newcommand{\app}{{\tt app}\xspace}
\newcommand{\appsteps}{\app_\steps}
\makeatletter
\newcommand{\appresult}{%
  \def\appresultW<##1>{\app_\result^{##1}}%
  \def\appresultWO{\app_\result}%
  \@ifnextchar<\appresultW\appresultWO%
  }
\makeatother
\newcommand{\esrule}{{\tt ES}\xspace}

\newcommand{\tarrow}[2]{#1 \multimap #2}
\newcommand{\ty}[2]{\tarrow{#1}{#2}}
\newcommand{\mult}[1]{[ #1 ] }

\newcommand{\tyjp}[4]{{#3} \vdash^{#1} #2 \hastype #4}

\newcommand{\tderiv}{\Phi}
\newcommand{\tderivtwo}{\Psi}
\newcommand{\tderivthree}{\Theta}
\newcommand{\tderivfour}{\Xi}

\newcommand{\typctx}{\Gamma}
\newcommand{\typctxtwo}{\Pi}
\newcommand{\typctxthree}{\Delta}
\newcommand{\typctxfour}{\Sigma}

\newcommand{\precise}{tight\xspace}

\newcommand{\Precise}{Tight\xspace}

\newcommand{\emptymset}{\zero}

\newcommand{\ndsym}{\mathtt{nd}}

\newcommand{\hole}[1]{\langle #1 \rangle}
\newcommand{\chole}{\langle \, \rangle}

\newcommand{\ES}{\text{ES}\xspace}
\newcommand{\gc}{\text{gc}\xspace}

\newcommand{\ESgc}{\ES_{\gc}\xspace}
\newcommand{\appgc}{\app_{\gc}\xspace}
\newcommand{\esgc}{\esrule_{\gc}\xspace}

\newcommand{\rtond}{\rootRew{{\tt nd}}}
\newcommand{\tond}{\Rew{\ndsym}}

\newcommand{\sm}{\mathbin{{\setminus}\mspace{-5mu}{\setminus}}}

\ifthenelse{\boolean{talk}}{}{\usepackage[utf8]{inputenc}}
\ifthenelse{\boolean{french}}{\usepackage[french]{babel}}{
  \ifthenelse{\boolean{lipicsstyle}}{}{\usepackage[english]{babel}}}
\usepackage{amssymb}
\usepackage{amsmath}
\usepackage{graphicx}
\usepackage{cmll}
\usepackage{stmaryrd}
\usepackage{multirow}
\usepackage{fancybox}
\usepackage{xspace}

\ifthenelse{\boolean{IEEEstyle}}{
\setboolean{needstheorems}{true}}{}

\ifthenelse{\boolean{eptcsstyle}}{
\setboolean{needstheorems}{true}}{}

\ifthenelse{\boolean{sigplanstyle}}{
\setboolean{needstheorems}{true}}{}

\ifthenelse{\boolean{needstheorems}}{

\usepackage{amsthm}

    \newtheorem{theorem}{Theorem}[section]
    \newtheorem{lemma}[theorem]{Lemma}
    \newtheorem{proposition}[theorem]{Proposition}
    \newtheorem{definition}{Definition}
    \newtheorem{remark}{Remark}
}{}

\newcommand{\myproof}[1]{
\ifthenelse{\boolean{withproofs}}{#1}{}
}

\newcommand{\ltermslsc}{\Lambda_{{\tt lsc}}}
\newcommand{\la}[1]{\lambda #1.}

\newcommand{\tm}{t}
\newcommand{\tmtwo}{s}
\newcommand{\tmthree}{u}
\newcommand{\tmfour}{r}
\newcommand{\tmfive}{q}

\newcommand{\var}{x}
\newcommand{\vartwo}{y}
\newcommand{\varthree}{z}
\newcommand{\varfour}{w}

\newcommand{\rootRew}[1]{\mapsto_{#1}}

\makeatletter

\newcommand{\Rewbase}{%
  \def\RewbaseW[##1]##2##3{\ {\xrightarrow{##1}}{}_{##2}^{##3}\xspace }%
  \def\RewbaseWO##1##2{\ {\xrightarrow{}}{}_{##1}^{##2}\xspace }%
  \@ifnextchar[\RewbaseW\RewbaseWO%
  }

\newcommand{\Rew}[1]{%
  \def\RewW[##1]{\Rewbase[##1]{{#1}}{}}%
  \def\RewWO{\Rewbase{{#1}}{}}%
  \@ifnextchar[\RewW\RewWO%
  }

\newcommand{\Rewn}[2][*]{%
  \def\RewnW[##1]{\Rewbase[##1]{#2}{#1}}%
  \def\RewnWO{\Rewbase{#2}{#1}}%
  \@ifnextchar[\RewnW\RewnWO%
  }

\makeatother

\renewcommand{\to}{\Rew{}}

\newcommand{\lRew}[1]{\; \mbox{}_{#1}{\leftarrow}\ }

\newcommand{\gcsym}{{\tt gc}}

\newcommand{\esym}{{\mathtt e}}

\newcommand{\msym}{{\mathtt m}}

\newcommand{\mcbn}{\msym_\cbnsym}
\newcommand{\ecbn}{\esym_\cbnsym}
\newcommand{\mcbv}{\msym_\cbvsym}
\newcommand{\ecbv}{\esym_\cbvsym}
\newcommand{\mcbneed}{\msym_\cbneedsym}
\newcommand{\ecbneed}{\esym_\cbneedsym}

\newcommand{\cbvsym}{\textup{cbv}}
\newcommand{\cbnsym}{\textup{cbn}}
\newcommand{\cbneedsym}{\textup{need}}

\newcommand{\cbv}{\text{CbV}\xspace}
\newcommand{\cbn}{\text{CbN}\xspace}
\newcommand{\cbneed}{\text{CbNeed}\xspace}
\newcommand{\cbvup}{\cbvsym\!}
\newcommand{\cbnup}{\cbnsym\!}
\newcommand{\cbneedup}{\cbneedsym\!}

\newcommand{\val}{v}

\newcommand{\rtocbv}{\rootRew{\cbvsym}}

\newcommand{\cbvctx}{V}
\newcommand{\cbvctxtwo}{\cbvctx'}
\newcommand{\cbvctxthree}{\cbvctx''}

\newcommand{\cbvctxp}[1]{\cbvctx\ctxholep{#1}}

\newcommand{\cbneedctx}{E}
\newcommand{\cbneedctxtwo}{\cbneedctx'}

\newcommand{\cbneedctxp}[1]{\cbneedctx\ctxholep{#1}}

\newcommand{\ctxholep}[1]{\langle #1\rangle}
\newcommand{\ctxhole}{\ctxholep{\cdot}}

\newcommand{\nbvctxtwo}[1]{\nbvctxtwo{#1}}

\newcommand{\sctx}{S}
\newcommand{\sctxtwo}{{\sctx'}}

\newcommand{\sctxp}[1]{\sctx\ctxholep{#1}}
\newcommand{\sctxptwo}[1]{\sctxtwo\ctxholep{#1}}

\newcommand{\sctxtwop}[1]{\sctxtwo\ctxholep{#1}}

\newcommand{\cbnctx}{C}
\newcommand{\cbnctxtwo}{D}

\newcommand{\cbnctxp}[1]{\cbnctx\ctxholep{#1}}

\newcommand{\cbnctxtwop}[1]{\cbnctxtwo\ctxholep{#1}}

\newcommand{\defeq}{\coloneqq}
\newcommand{\grameq}{\Coloneqq}

\newcommand{\esub}[2]{[#1{\leftarrow}#2]}

\newcommand{\letexp}[3]{{\tt let}\ #1=#2\ {\tt in}\ #3}

\newcommand{\rtogc}{\rootRew{\gcsym}}

\newcommand{\cwc}[1]{\langle  \! \langle    #1 \rangle  \! \rangle}

\newcommand{\llbrace}{\{ \kern -0.27em \vert}
\newcommand{\rrbrace}{\vert \kern -0.27em \}}

\renewcommand{\l}{\lambda}

\newcommand{\ie}{{\em i.e.}\xspace}

\newcommand{\ih}{{\emph{i.h.}}\xspace}
\newcommand{\fv}[1]{{\tt fv}(#1)}

\newcommand{\ignore}[1]{}

\newcommand{\colspace}{@{\hspace{.5cm}}}
\newcommand{\myinput}[1]{\ifthenelse{\boolean{withimages}}{\input{#1}}{}}

\newcommand{\reflemma}[1]{Lemma~\ref{l:#1}}
\newcommand{\reflemmap}[2]{Lemma~\ref{l:#1}.\ref{p:#1-#2}}
\newcommand{\reflemmaeq}[1]{{L.\ref{l:#1}}}

\newcommand{\refthm}[1]{Theorem~\ref{thm:#1}}

\newcommand{\refprop}[1]{Proposition~\ref{prop:#1}}

\newcommand{\refsect}[1]{Sect.~\ref{sect:#1}}

\ifthenelse{\boolean{easychairstyle}}{}{
}
\newcommand{\reffig}[1]{Fig.~\ref{fig:#1}}
\newcommand{\refcoro}[1]{Corollary~\ref{coro:#1}}

\newcommand{\refrem}[1]{Remark~\ref{rem:#1}}

\newcommand{\refex}[1]{Example~\ref{ex:#1}}
\newcommand{\refexample}[1]{\refex{#1}}

\newcommand{\form}{A}
\newcommand{\formtwo}{B}

\newcommand{\ax}{\mathsf{ax}\xspace}
\newcommand{\many}{\mathsf{many}\xspace}

\newcommand{\lolli}{\multimap}

\newcommand{\set}[1]{\{#1\}}

\newcommand{\nat}{\mathbb{N}}

\newcommand{\rtom}{\mapsto_\msym}
\newcommand{\rtoe}{\mapsto_\esym}
\newcommand{\tom}{\Rew{\msym}}
\newcommand{\toe}{\Rew{\esym}}

\newcommand{\tomcbn}{\Rew{\mcbn}}
\newcommand{\toecbn}{\Rew{\ecbn}}
\newcommand{\tomcbv}{\Rew{\mcbv}}
\newcommand{\toecbv}{\Rew{\ecbv}}
\newcommand{\tomcbneed}{\Rew{\msym_\cbneedsym}}
\newcommand{\toecbneed}{\Rew{\esym_\cbneedsym}}

\newcommand{\rtoecbn}{\mapsto_{\esym_\cbnsym}}

\newcommand{\rtoecbv}{\mapsto_{\ecbv}}
\newcommand{\rtomcbneed}{\mapsto_{\msym_\cbneedsym}}
\newcommand{\rtoecbneed}{\mapsto_{\esym_\cbneedsym}}

\newcommand{\wctx}{W}
\newcommand{\wctxp}[1]{\wctx\ctxholep{#1}}

\newcommand{\size}[1]{|#1|}
\newcommand{\sizep}[2]{|#1|_{#2}}
\newcommand{\sizem}[1]{\sizep{#1}{\msym}}
\newcommand{\sizee}[1]{\sizep{#1}{\esym}}

\newcommand{\deriv}{\ensuremath{d}}

\newcommand{\derivtwo}{\ensuremath{d'}}

\newcommand\mplus{\uplus}
\newcommand\bigmplus{\biguplus}

\newcommand{\sem}[1]{[\![#1]\!]}

\newcommand\emptytype{\mathbf{0}}

\newcommand\tocbn{\Rew{\cbnsym}}
\newcommand\tocbnn{\Rewn{\cbnsym}}

\newcommand{\tocbv}{\Rew{\cbvsym}}
\newcommand{\tocbvn}{\Rewn{\cbvsym}}

\newcommand{\tocbneed}{\Rew{\cbneedsym}}
\newcommand{\tocbneedn}{\Rewn{\cbneedsym}}

\newcommand{\Rule}{\mathsf{r}}
\newcommand{\torule}{\Rew{\Rule}}

\newcommand{\ctxplus}{\mplus}
\newcommand{\hastype}{\,{:}\,}

\newcommand{\zero}{\emptytype}

\usepackage{microtype}
\usepackage{ebproof}
\hypersetup{hidelinks}

\begin{document}
\title{Types by Need (Extended Version)}
%
%

\author{Beniamino Accattoli\inst{1} \and Giulio Guerrieri\inst{2} \and Maico Leberle\inst{1}}
\authorrunning{Accattoli, Guerrieri, and Leberle}
%
\institute{Inria \& LIX, \'Ecole Polytechnique, UMR 7161, Palaiseau, France
\email{\{\href{mailto:beniamino.accattoli@inria.fr}{beniamino.accattoli},\href{mailto:maico.leberle@inria.fr}{maico.leberle}\}@inria.fr}
\and
University di Bath, Departement of Computer Science, Bath, United Kingdom
\email{\href{mailto:g.guerrieri@bath.ac.uk}{g.guerrieri@bath.ac.uk}} 
}

\maketitle              
\begin{abstract}
A cornerstone of the theory of $\lambda$-calculus is that intersection types characterise termination properties. They are a flexible tool that can be adapted to various notions of termination, and that also induces adequate denotational models. 

Since the seminal work of de Carvalho in 2007, it is  known that multi types (\ie non-idempotent intersection types) refine intersection types with quantitative information and a strong connection to linear logic. Typically, type derivations provide bounds for evaluation lengths, and minimal type derivations provide exact bounds. 

De Carvalho studied call-by-name evaluation, and Kesner used his system to show the termination equivalence of call-by-need and call-by-name. De Carvalho's system, however, cannot provide exact bounds on call-by-need evaluation lengths.

In this paper we develop a new multi type system for call-by-need. Our system produces exact bounds and induces a denotational model of call-by-need, providing the first tight quantitative semantics of call-by-need.

\end{abstract}

\section{Introduction}
Duplications and erasures have always been considered as key phenomena in the $\lambda$-calculus---the $\lambda I$-calculus, where erasures are forbidden, is an example of this. The advent of linear logic \cite{DBLP:journals/tcs/Girard87} gave them a new, prominent logical status. Forbidding erasure and duplication enables single-use resources, i.e. linearity, but limits expressivity, as every computation terminates in linear time. Their controlled reintroduction via the non-linear modality $\oc$ recovers the full expressive power of cut-elimination and allows a fine analysis of resource consumption. Duplication and erasure are therefore the key ingredients for logical expressivity, and---via Curry-Howard---for the expressivity of the $\lambda$-calculus. 
They are also essential to understand evaluation strategies.

In a $\lambda$-term there can be many $\beta$-redexes, that is, places where $\beta$-reduction can be applied. In this sense, the $\lambda$-calculus is non-deterministic. Non-determinism does not affect the result of evaluation, if any, but it affects whether evaluation terminates, and in how many steps. There are two natural deterministic evaluation strategies, \emph{call-by-name} (shortened to \cbn) and \emph{call-by-value} (\cbv), which have dual behaviour with respect to duplication and erasure. 

\paragraph{Call-by-Name = Silly Duplication + Wise Erasure.} \cbn \emph{never} evaluates arguments of $\beta$-redexes before the redexes themselves. 
As a consequence, it never evaluates in subterms that will be erased. This is wise, and makes \cbn a \emph{normalising strategy}, that is, a strategy that reaches a result whenever one exists\footnote{If a term $\tm$ admits both converging and diverging evaluation sequences then the diverging sequences occur in erasable subterms of $\tm$, which is why \cbn avoids~them.}. 
A second consequence is that if the argument 
of the redex is duplicated then it may be evaluated more than once. This is silly, as it repeats~work~already~done. 

\paragraph{Call-by-Value = Wise Duplication + Silly Erasure.} \cbv, on the other hand, \emph{always} evaluates arguments of $\beta$-redexes before the redexes themselves. Consequently, arguments are not re-evaluated---this is wise with respect to duplication---but they are also evaluated when they are going to be erased. For instance, on $\tm \defeq (\la\var\la\vartwo\vartwo) \Omega$, where $\Omega$ is the famous looping $\lambda$-term, \cbv evaluation diverges (it keeps evaluating $\Omega$) while \cbn converges in one  $\beta$-step (simply erasing $\Omega$). This \cbv treatment of erasure is clearly as silly as the duplicated work of \cbn.

\paragraph{Call-by-Need = Wise Duplication + Wise Erasure.} It is natural to try to combine the advantages of both \cbn and \cbv. The strategy that is wise with respect to both duplications and erasures is usually called \emph{call-by-need} (\cbneed), it was introduced by Wadsworth \cite{Wad:SemPra:71}, and dates back to the '70s. 
Despite being at the core of  Haskell, one of the most-used functional programming languages, and---in its strong variant---being at work in the kernel of Coq as designed by Barras \cite{barras-phd}, the theory of \cbneed is much less developed than that of \cbn or \cbv. 

One of the reasons for this is that it cannot be defined inside the $\lambda$-calculus without some hacking. Manageable presentations of \cbneed indeed require first-class sharing and  micro-step operational semantics where variable occurrences are replaced one at a time (when needed), and not all at once as in the $\lambda$-calculus. Another reason is the less natural logical interpretation.

\paragraph{Linear Logic, Names, Values, and Needs.} \cbn and \cbv have neat interpretations in linear logic. 
They correspond to two different representations of 
intuitionistic logic in linear logic, based on two different representations of implication\footnote{The \cbn translation maps $\form \Rightarrow \formtwo$ to $(\oc\form^\cbn) \lolli \formtwo^\cbn$, while the \cbv maps it to $\oc\form^\cbv \lolli \oc\formtwo^\cbv$, or equivalently to $\oc(\form^\cbv \lolli \formtwo^\cbv)$.}.

The  logical interpretation of \cbneed---studied by Maraist et al. in \cite{DBLP:journals/tcs/MaraistOTW99}---is less neat than those of \cbn and \cbv. Within linear logic, \cbneed is usually understood as corresponding to the \cbv representation where erasures are generalised to all terms, not only those under the scope of a $\oc$ modality. 
So, it is seen as a sort of \emph{affine} \cbv. Such an interpretation however is unusual, because it does not match exactly with cut-elimination in linear logic, as for \cbn~and~\cbv.

\paragraph{Call-by-Need, Abstractly.} The main theorem of the theory of \cbneed is that it is termination equivalent to \cbn, that is, on a fixed term, \cbneed evaluation terminates if and only if \cbn evaluation terminates, and, moreover, they essentially produce the same result (up to some technical details that are irrelevant here). This is due to the fact that both strategies avoid silly divergent sequences such as that of $(\la\var\la\vartwo\vartwo) \Omega$. Termination equivalence is an abstract theorem stating that \cbneed erases as wisely as \cbn. Curiously, in the literature there are no abstract theorems reflecting the dual fact that \cbneed duplicates as wisely as \cbv---we provide one, as a side contribution of this paper.

\paragraph{Call-by-Need and Denotational Semantics.} \cbneed is then usually considered as a \cbv optimisation of \cbn. In particular, every denotational model of \cbn is also a model of \cbneed, and adequacy---that is the fact that the denotation of $\tm$ is not degenerated if and only if $\tm$ terminates---transfers from \cbn to \cbneed. 

Denotational semantics is invariant by evaluation, and so is insensitive to evaluation lengths by definition. 
It then seems that denotational semantics cannot distinguish between \cbn and \cbneed. 
The aim of this paper is, somewhat counter-intuitively, to separate \cbn and \cbneed semantically. 
We develop a type system whose type judgements induce a model---this is typical of \emph{intersection} type systems---and whose type derivations provide exact bounds for \cbneed evaluation---this is usually obtained via \emph{non-idempotent} intersection types. 
Unsurprisingly, the design of the type system requires a delicate mix of erasure and duplication and builds on the linear logic understanding~of~\cbn~and~\cbv.

\paragraph{Multi Types.} Our typing framework is given by \emph{multi types}, which is an alternative name for \emph{non-idempotent intersection types}\footnote{The new terminology is due to the fact that a non-idempotent intersection $A \wedge A \wedge B \wedge C$ can be seen as a multi-set $\mult{A,A,B,C}$.}. Multi types characterise termination properties exactly as intersection types, having moreover the advantages that they are closely related to (the relational semantics of) linear logic, their type derivations provide quantitative information about evaluation lengths, and the proof techniques are simpler---no need for the reducibility method.

The seminal work of de Carvalho \cite{DBLP:journals/mscs/Carvalho18} (appeared in 2007 but unpublished until  2018) showed how to use multi types to obtain exact bounds on evaluation lengths in \cbn. 
Ehrhard adapted multi types to \cbv \cite{DBLP:conf/csl/Ehrhard12}, and very recently Accattoli and Guerrieri adapted de Carvalho's study of exact bounds to Ehrhard's system and \cbv evaluation \cite{aplas18}. 
Kesner used de Carvalho's \cbn multi types to obtain a simple proof that \cbneed is termination equivalent with respect to \cbn \cite{DBLP:conf/fossacs/Kesner16} (first proved with other techniques by Maraist, Odersky, and Wadler \cite{DBLP:journals/jfp/MaraistOW98} and Ariola and Felleisen \cite{DBLP:journals/jfp/AriolaF97} in the nineties), and then Kesner and coauthors continued exploring the theory of \cbneed via \cbn multi types \cite{DBLP:journals/pacmpl/BalabonskiBBK17,KesnerVR18,DBLP:conf/ppdp/BarenbaumBM18}. 

Kesner's use of \cbn multi types to study \cbneed is \emph{qualitative}, as it deals with termination and not with exact bounds. 
For a \emph{quantitative} study of \cbneed, de Carvalho's \cbn system cannot really be informative: \cbn multi types provide bounds for \cbneed 
which cannot be exact because they already provide exact bounds for \cbn, which generally takes more steps than \cbneed.

\paragraph{Multi Types by Need.} In this paper we provide the first multi type system characterising \cbneed termination and whose minimal type derivations provide \emph{exact} bounds for \cbneed evaluation lengths. The design of the type system is delicate, as we explain in \refsect{cbneed}. One of the key points is that, in contrast to Ehrhard's system for \cbv \cite{DBLP:conf/csl/Ehrhard12}, multi types for \cbneed cannot be directly extracted by the relational semantics of linear logic, given that \cbneed does not have a clean representation in it. A by-product of our work is a new denotational semantics of \cbneed, 
the first one to precisely reflect its quantitative properties.

Beyond the result itself, the paper tries to stress how the key ingredients of our type system are taken from those for \cbn and \cbv and combined together. To this aim, we first present multi types for \cbn and \cbv, and only then we proceed to build the \cbneed system and prove its properties. 

Along the way, we also prove the missing fundamental property of \cbneed, that is, that it duplicates as efficiently as \cbv. The result is obtained by dualising Kesner's approach \cite{DBLP:conf/fossacs/Kesner16}, showing that the \cbv multi type system is correct also with respect to \cbneed evaluation, that is, its bounds are also valid with respect to \cbneed evaluation lengths. Careful: the \cbv system is correct but of course not complete with respect to \cbneed, because \cbneed may normalise when \cbv diverges. The proof of the result is straightforward, because of our presentations of (\cbn,) \cbv and \cbneed. We adopt a liberal, non-deterministic formulation of \cbv, and  assuming that garbage collection is always postponed. These two ingredients turn \cbneed into a fragment of \cbv, obtaining the new fundamental result as a corollary of correctness of \cbv multi types for \cbv evaluation.

\paragraph{Technical Development.} The paper is extremely uniform, technically speaking. The three evaluations are presented as strategies of Accattoli and Kesner's Linear Substitution Calculus (shortened to LSC) \cite{DBLP:conf/rta/Accattoli12,DBLP:conf/popl/AccattoliBKL14}, a calculus with a simple but expressive form of explicit sharing. The LSC is strongly related to linear logic \cite{DBLP:conf/ictac/Accattoli18}, and provides a neat and manageable presentation of \cbneed, introduced by Accattoli, Barenbaum, and Mazza in \cite{DBLP:conf/icfp/AccattoliBM14}, and further developed by various authors in \cite{DBLP:conf/wollic/AccattoliC14,DBLP:conf/fossacs/Kesner16,DBLP:journals/pacmpl/BalabonskiBBK17,DBLP:conf/ppdp/AccattoliB17,DBLP:conf/aplas/AccattoliB17,KesnerVR18,DBLP:conf/ppdp/BarenbaumBM18}. Our type systems count evaluation steps by annotating typing rules in the \emph{exact} same way, and the proofs of correctness and completeness all follow the \emph{exact} same structure. While the results for \cbn are very minor variations with respect to those in the literature \cite{DBLP:journals/mscs/Carvalho18,DBLP:journals/pacmpl/AccattoliGK18}, those for \cbv are the first ones with respect to a presentation of \cbv with sharing.

As it is standard for \cbneed, we restrict our study to closed terms and weak evaluation (that is, out of abstractions). The main consequence of this fact is that normal forms are particularly simple (sometimes called \emph{answers} in the literature). Compared with other recent works dealing with exact bounds such as Accattoli, Graham-Lengrand, and Kesner's \cite{DBLP:journals/pacmpl/AccattoliGK18} and Accattoli and Guerrieri's \cite{aplas18} the main difference is that the size of normal forms is not taken into account by type derivations. This is because of the simple notions of normal forms in the closed and weak case, and not because the type systems are not accurate.

\paragraph{Related work about \cbneed.} Call-by-need was introduced by Wadsworth \cite{Wad:SemPra:71} in the 
'70s. In the 
'90s, it was first reformulated as operational semantics by Launchbury \cite{DBLP:conf/popl/Launchbury93}, Maraist, Odersky, and Wadler \cite{DBLP:journals/jfp/MaraistOW98}, and Ariola and Felleisen \cite{DBLP:journals/jfp/AriolaF97}, and then implemented by Sestoft \cite{DBLP:journals/jfp/Sestoft97} and further studied by Kutzner and Schmidt{-}Schau{\ss} \cite{DBLP:conf/icfp/KutznerS98}. 
More recent papers are Garcia, Lumsdaine, and
               Sabry's \cite{DBLP:conf/popl/GarciaLS09}, Ariola, Herbelin, and Saurin's \cite{DBLP:conf/tlca/AriolaHS11}, Chang and Felleisen's \cite{DBLP:conf/esop/ChangF12}, Danvy and Zerny's \cite{DBLP:conf/ppdp/DanvyZ13}, Downen et al.'s \cite{DBLP:conf/ppdp/DownenMAV14}, P{\'{e}}drot and
               Saurin's \cite{DBLP:conf/esop/PedrotS16}, and Balabonski~et~al.'s \cite{DBLP:journals/pacmpl/BalabonskiBBK17}.

\paragraph{Related work about Multi Types.}               Intersection types are a standard tool to study $\lambda$-calculi---see Coppo and Dezani \cite{DBLP:journals/aml/CoppoD78,DBLP:journals/ndjfl/CoppoD80}, Pottinger \cite{Pottinger80}, and Krivine \cite{Kri}. Non-idempotent intersection types, \ie multi types, were first considered by Gardner \cite{DBLP:conf/tacs/Gardner94}, and then by Kfoury \cite{DBLP:journals/logcom/Kfoury00}, Neergaard and Mairson \cite{DBLP:conf/icfp/NeergaardM04}, and de Carvalho \cite{DBLP:journals/mscs/Carvalho18}---a survey is Bucciarelli, Kesner, and Ventura's~\cite{BKV17}. 

Many recent works rely on multi types or relational semantics to study properties of programs and proofs. Beyond the cited ones, Diaz-Caro, Manzonetto, and Pagani's \cite{DBLP:conf/lfcs/Diaz-CaroMP13}, Carraro and Guerrieri's \cite{DBLP:conf/fossacs/CarraroG14}, Ehrhard and Guerrieri's \cite{DBLP:conf/ppdp/EhrhardG16}, and Guerrieri's \cite{Guerrieri18} deal with \cbv, while Bernadet and Lengrand's \cite{Bernadet-Lengrand2013}, de Carvalho, Pagani, and Tortora de Falco's \cite{DBLP:journals/tcs/CarvalhoPF11} provide exact bounds. Further related work is by Bucciarelli, Ehrhard, and Manzonetto 
\cite{DBLP:journals/apal/BucciarelliEM12}, de Carvalho and Tortora de Falco \cite{DBLP:journals/iandc/CarvalhoF16}, Tsukada and Ong \cite{DBLP:conf/lics/TsukadaO16}, Kesner and Vial \cite{DBLP:conf/rta/KesnerV17}, Piccolo, Paolini and Ronchi Della Rocca \cite{DBLP:journals/mscs/PaoliniPR17}, Ong \cite{DBLP:conf/lics/Ong17}, Mazza, Pellissier, and Vial \cite{DBLP:journals/pacmpl/MazzaPV18}, Bucciarelli, Kesner and Ronchi Della Rocca \cite{DBLP:journals/lmcs/BucciarelliKR18}---this list is not exhaustive.

\medskip
This is the long version (with all proofs) of a paper accepted to ESOP 2019.

\section{Closed $\lambda$-Calculi}
\label{sect:closed}

In this section we define the \cbn, \cbv, and \cbneed evaluation strategies. We present them in the context of the Accattoli and Kesner's \emph{linear substitution calculus} (LSC) \cite{DBLP:conf/rta/Accattoli12,DBLP:conf/popl/AccattoliBKL14}. We mainly follow the uniform presentation of these strategies given by Accattoli, Barenbaum, and Mazza \cite{DBLP:conf/icfp/AccattoliBM14}. The only difference is that we adopt a non-deterministic presentation of \cbv, subsuming both the left-to-right and the right-to-left strategies in \cite{DBLP:conf/icfp/AccattoliBM14}, that makes our results slightly more general. 
Such a  non-determinism is harmless: not only \cbv evaluation is confluent, it even has the diamond property, so that all evaluations 
have the same length.

\paragraph{Terms and Contexts.} The set of terms $\ltermslsc$ of the LSC is given by the
following grammar, where $\tm\esub{\var}{\tmtwo}$ is called an \emph{explicit substitution} (shortened to ES), that is 
a more compact notation for $\letexp\var\tmtwo\tm$: 
\begin{align*}
    \textsc{LSC Terms} \quad \tm,\tmtwo & \grameq \var \mid \val \mid \tm\tmtwo\mid \tm\esub{\var}{\tmtwo} & & &
    \textsc{LSC Values} \quad \val & \grameq  \la\var\tm
  \end{align*} 
The set $\fv{\tm}$ of \emph{free} variables of a term $\tm$  is defined as expected, in particular,
$\fv{\tm\esub{\var}{\tmtwo}} :=
(\fv{\tm} \setminus \set{\var}) \cup \fv{\tmtwo}$. 
A term $\tm$ is \emph{closed} if $\fv{\tm} = \emptyset$, \emph{open} otherwise.
As usual, terms are identified up to $\alpha$-equivalence.

Contexts are terms with exactly one occurrence of the \emph{hole} $\ctxhole$, an additional constant. 
We 
shall use many different contexts. 
The most general ones 
are \emph{weak contexts} $\wctx$ (i.e. not under abstractions).
The (evaluation) contexts $\cbnctx$, $\cbvctx$ and $\cbneedctx$---used to define \cbn, \cbv and \cbneed evaluation strategies, respectively---are special cases of weak contexts (in fact, \cbv contexts coincide with weak contexts, the consequences of that are discussed on p.~\pageref{prop:syntactic-characterization-closed-normal}).
To define evaluation strategies, \emph{substitution contexts} (\ie lists of explicit substitutions) also play a role.
\begin{align*}
\textsc{Weak contexts} & & \wctx & \grameq \ctxhole \mid \wctx \tm \mid \wctx \esub
\var\tm \mid \tm \wctx \mid \tm \esub\var\wctx\\
\textsc{Substitution contexts} & &  \sctx & \grameq \ctxhole \mid \sctx\esub{\var}{\tm} \\
\textsc{\cbn contexts} & &  \cbnctx   & \grameq  \ctxhole  \mid \cbnctx \tm \mid \cbnctx\esub{\var}{\tm} \\
\textsc{\cbv contexts} & &  \cbvctx   & \grameq  \wctx
\\
\textsc{\cbneed contexts} & &  \cbneedctx   & \grameq  \ctxhole  \mid \cbneedctx \tm \mid \cbneedctx \esub{\var}\tm \mid \cbneedctx\cwc{\var}\esub{\var}{\cbneedctxtwo}
\end{align*} 

We write 
$\wctxp{\tm}$ for the term obtained by replacing the hole $\ctxhole$ in context 
$\wctx$ by the term $\tm$. 
This \emph{plugging} operation, as usual with contexts, can
capture variables---for instance $(\ctxhole \tm)\esub\var\tmtwo)\ctxholep\var = (\var\tm)\esub\var\tmtwo$. 
We write $\wctx\cwc{\tm}$ when we want to stress that the context $\wctx$ does not capture the free variables of $\tm$. 

\paragraph{Micro-step semantics.}
The rewriting rules decompose the usual small-step semantics for $\lambda$-calculi, by substituting one variable occurrence at the time, and only
when such an occurrence is in evaluation position. 
We emphasise this fact saying that we adopt a \emph{micro-step semantics}. We now give the definitions, examples of evaluation sequences follow right next.

Formally, a micro-step semantics is defined by first giving its \emph{root-steps} and then taking the closure of root-steps under suitable contexts.
\begin{align*}
\textsc{Multiplicative root-step} &&
\sctxp{\la \var \tm} \tmtwo &\rtom \sctxp{\tm \esub{\var}{\tmtwo}}
\\
\textsc{Exponential \cbn root-step} && 
\cbnctx\cwc{\var}\esub{\var}{\tm} &\rtoecbn  \cbnctx\cwc{\tm}\esub{\var}{\tm} 
\\
\textsc{Exponential \cbv root-step} && 
\cbvctx\cwc{\var}\esub{\var}{\sctxp{\val}} & \rtoecbv \sctxp{\cbvctx\cwc{\val}\esub{\var}{\val}}
\\
\textsc{Exponential \cbneed root-step} && 
\cbneedctx\cwc{\var}\esub{\var}{\sctxp\val} &\rtoecbneed \sctxp{\cbneedctx\cwc{\val}\esub{\var}{\val}}
\end{align*}
\noindent where, in the root-step $\rtom$ (resp. $\rtoecbv$; $\rtoecbneed$), if $\sctx \defeq \esub{\vartwo_1}{\tmtwo_1} \dots \esub{\vartwo_n}{\tmtwo_n}$ for some $n \in \nat$, then 
$\fv{\tmtwo}$ (resp.~$\fv{\cbvctx\cwc{\var}}$; $\fv{\cbneedctx\cwc{\var}}$) and $\{\vartwo_1, \dots, \vartwo_n\}$ are disjoint.
This condition can always be fulfilled by~$\alpha$-equivalence.

The \emph{evaluation strategies} $\tocbn$ for \cbn, $\tocbv$ for \cbv, and $\tocbneed$ for \cbneed, are defined as the closure of root-steps under \cbn, \cbv and \cbneed evaluation contexts, respectively (so, all evaluation strategies do not reduce under abstractions, since all such contexts are weak):
$$\begin{array}{c|c|c}
\cbn & \cbv & \cbneed \\
\begin{aligned}
\tomcbn &\!\defeq  \cbnctxp{\rtom} \\
\toecbn &\!\defeq  \cbnctxp{\rtoecbn} \\
\tocbn &\!\defeq  \cbnctxp{\rtom \!\cup \rtoecbn}
\end{aligned} 
&  
\begin{aligned}
\tomcbv &\!\defeq \cbvctxp{\rtom} \\
\toecbv &\!\defeq \cbvctxp{\rtoecbv} \\
\tocbv &\!\defeq \cbvctxp{\rtom \!\cup \rtoecbv}
\end{aligned} 
& 
\begin{aligned}
\tomcbneed &\!\defeq \cbneedctxp{\rtom} \\
\toecbneed &\!\defeq \cbneedctxp{\rtoecbneed} \\
\tocbneed &\!\defeq \cbneedctxp{\rtom \!\cup \rtoecbneed}
\end{aligned}
\end{array}$$
where the notation $\to \defeq \wctxp\mapsto$ means that, given a root-step $\mapsto$, the evaluation $\to$ is defined as follows: $\tm \to \tmtwo$ if and only if there are terms $\tm'$ and $\tmtwo'$ and a context $\wctx$ such that $\tm = \wctxp{\tm'}$ and $\tmtwo = \wctxp{\tmtwo'}$ and $\tm' \mapsto \tmtwo'$.

Note that evaluations $\tocbn$, $\tocbv$ and $\tocbneed$ can equivalently be defined as $\tomcbn \cup \toecbn$, $\tomcbn \cup \toecbv$ and $\tomcbneed \cup \toecbneed$, respectively.

Given an evaluation sequence $\deriv \colon \tm \tocbnn \tmtwo$ we note with $\size\deriv$ the length of $\deriv$, and with $\sizem\deriv$ and $\sizee\deriv$ the number of multiplicative and exponential steps in $\deriv$, respectively---and similarly for $\tocbv$ and $\tocbneed$.

\paragraph{Erasing Steps.} The reader may be surprised by our evaluation strategies, as none of them includes erasing steps, despite the absolute relevance of erasures pointed out in the introduction. There are no contradictions: in the LSC---in contrast to the $\lambda$-calculus---erasing steps can always be postponed, and so they are often simply omitted. This is actually close to programming language practice, as the garbage collector acts asynchronously with respect to the evaluation flow. For the sake of clarity let us spell out the erasing rules---they shall nonetheless be ignored in the rest of the paper. In \cbn and \cbneed every term is erasable, so the root erasing step takes the following form
\begin{align*}
\tm \esub\var\tmtwo \rtogc \tm \qquad &\textup{if $\var\notin\fv\tm$}
\intertext{and it is then closed by 
weak evaluation contexts. \newline
\indent In \cbv only values are erasable; so, the root erasing step in \cbv is:}
\tm \esub\var{\sctxp\val} \rtogc \sctxp\tm \qquad &\textup{if $\var\notin\fv\tm$}
\end{align*}
and it is then closed by 
weak evaluation contexts.

\begin{example}
\label{ex:evaluation}
A good example to observe the differences between \cbn, \cbv, and \cbneed is given by the term $\tm \defeq ((\la\var\la\vartwo\var\var) (I I)) (I I)$ where $I \defeq \la \varthree \varthree$ is the identity combinator. In \cbn, it evaluates with 5 multiplicative steps and 5 exponential steps, as follows:

\[{\small
	\begin{alignedat}{2}
  \tm 
  & \tomcbn  (\la\vartwo\var\var) \esub\var{II} (II) &
  & \tomcbn  (\var\var) \esub\vartwo{II}\esub\var{II} \\
  & \toecbn  ((II)\var) \esub\vartwo{II}\esub\var{II} & 
  & \tomcbn  (\varthree\esub\varthree I \var) \esub\vartwo{II}\esub\var{II} \\
  & \toecbn  (I\esub\varthree I \var) \esub\vartwo{II}\esub\var{II} &
  & \tomcbn  \varfour \esub\varfour\var \esub\varthree I  \esub\vartwo{II}\esub\var{II}  \\
  & \toecbn  \var \esub\varfour\var \esub\varthree I  \esub\vartwo{II}\esub\var{II}  &
  & \toecbn (II) \esub\varfour\var \esub\varthree I \esub\vartwo{II}\esub\var{II}  \\
  & \tomcbn  \var' \esub{\var'} I \esub\varfour\var \esub\varthree I \esub\vartwo{II}\esub\var{II} &
  & \toecbn  I \esub{\var'} I \esub\varfour\var \esub\varthree I \esub\vartwo{II}\esub\var{II}
  \end{alignedat}
}\]

In \cbv, $\tm$ evaluates with 5 multiplicative steps and 5 exponential steps, for instance from right to left, as follows:

\[{\small\begin{alignedat}{2}
  \tm 
  & \tomcbv  (\la\var\la\vartwo\var\var) (II) (\varthree \esub\varthree I) &
  & \toecbv (\la\var\la\vartwo\var\var) (II) (I \esub\varthree I) \\
  & \tomcbv (\la\var\la\vartwo\var\var) (\varfour \esub\varfour I) (I \esub\varthree I) &
  & \toecbv (\la\var\la\vartwo\var\var) (I \esub\varfour I) (I \esub\varthree I) \\
  & \tomcbv (\la\vartwo\var\var) \esub\var{I \esub\varfour I} (I \esub\varthree I)  &
  & \tomcbv (\var\var) \esub\vartwo {I \esub\varthree I} \esub\var{I \esub\varfour I}\\
  & \toecbv (\var I) \esub{\vartwo}{I \esub{\varthree}{I}} \esub{\var}{I} \esub{\varfour}{I} &
  & \toecbv (II) \esub{\vartwo} {I \esub\varthree I} \esub\var{I} \esub{\varfour}{I}\\
  & \tomcbv \var'\esub{\var'}{I} \esub{\vartwo}{I \esub{\varthree}{I}}  \esub{\var}{I} \esub{\varfour}{I} &
  & \toecbv I\esub{\var'}I \esub{\vartwo}{I \esub{\varthree}{I}} \esub\var{I} \esub{\varfour}{I}
  \end{alignedat}
}\]

Note that the fact that \cbn and \cbv take the same number of steps is by chance, as they reduce different redexes: \cbn never reduce the unneeded redex $II$ associated to $y$, but it reduces twice the needed $II$ redex associated to $x$, while \cbv reduces both, but each one only once.

In \cbneed, $\tm$ evaluates in 4 multiplicative steps and 4 exponential steps.

\[{\small\begin{alignedat}{2}
  \tm 
  & \tomcbneed   (\la\vartwo\var\var) \esub\var{II} (II) &
  & \tomcbneed  (\var\var) \esub\vartwo{II}\esub\var{II} \\
  & \tomcbneed   (\var\var) \esub\vartwo{II}\esub\var{\varthree\esub\varthree I}  &
  & \toecbneed   (\var\var) \esub\vartwo{II}\esub\var{I\esub\varthree I} \\
  & \toecbneed   (I\var) \esub{\vartwo}{II} \esub{\var}{I} \esub{\varthree}{I} &
  & \tomcbneed   (\varfour \esub\varfour \var) \esub{\vartwo}{II} \esub{\var}{I} \esub{\varthree}{I}\\
  & \toecbneed  \varfour\esub\varfour I \esub{\vartwo}{II} \esub{\var}{I} \esub{\varthree}{I} &
  & \toecbneed  I\esub\varfour I \esub{\vartwo}{II}\esub{\var}{I} \esub{\varthree}{I}
  \end{alignedat}
}\]

\end{example}

\paragraph{\cbv Diamond Property.} \cbv contexts coincide with  weak ones.
 As a consequence, our presentation of \cbv is non-deterministic, as for instance one can have
$$ \var\esub \var I (\vartwo \esub\vartwo I) \lRew{\mcbv} (II) (\vartwo \esub\vartwo I) \toecbv (II) (I \esub\vartwo I)$$
but it is easily seen that diagrams can be closed in exactly one step (if the two reducts are different). For instance,
$$ \var\esub \var I (\vartwo \esub\vartwo I) \toecbv \var\esub \var I (I \esub\vartwo I) \lRew{\mcbv} (II) (I \esub\vartwo I)$$
Moreover, the kind of steps is preserved, as the example illustrates. This is an instance of the strong form of confluence called \emph{diamond property}. A consequence is that either all evaluation sequences normalise or all diverge, and if they normalise they have all the same length and the same number of steps of each kind. Roughly, the diamond property is a form of relaxed determinism. In particular, it makes sense to talk about the number of multiplicative\,/\,exponential steps to normal form, independently of the evaluation sequence. The proof of the property is an omitted routine check of diagrams.

\paragraph{Normal Forms.} We use two predicates to characterise normal forms, one for both \cbn and \cbneed normal forms, for which ES can contain whatever term, and one for \cbv normal forms, where ES can only contain normal terms: 
\begin{align*}
\infer{\normalpr{\la\var\tm}}{\phantom{\normalpr{\la{\var}{\tm}}}}
&&
\infer{ \normalpr{ \tm \esub\var\tmtwo } }{ \normalpr\tm}  
&&&&
\infer{\normalcbvpr{\la\var\tm}}{\phantom{\normalpr{\la{\var}{\tm}}}}
&&
\infer{ \normalcbvpr{ \tm \esub\var\tmtwo } }{ \normalcbvpr\tm \quad \normalcbvpr{\tmtwo} } 
\end{align*}

\begin{toappendix}
\begin{proposition}[Syntactic characterisation of closed normal forms]
	\label{prop:syntactic-characterization-closed-normal}
	$\!$Let $\tm$ be a closed term.
	\begin{enumerate}
		\item \emph{\cbn and \cbneed}: For $\Rule \in \{\cbnsym, \cbneedsym \}$, $\tm$ is $\Rule$-normal if and only if $\normalpr{\tm}$.
		\item \emph{\cbv}: $\tm$ is $\cbvsym$-normal if and only if $\normalcbvpr{\tm}$.
	\end{enumerate} 
\end{proposition}
\end{toappendix}

The simple structure of normal forms is the main point where the restriction to closed calculi plays a role in this paper.

From the syntactic characterization of normal forms (\refprop{syntactic-characterization-closed-normal}) it follows immediately that among closed terms, \cbn and \cbneed normal forms coincide, while \cbv normal forms are a subset of them.
Such a subset is proper since the closed term $I \esub{\var}{\delta\delta}$ (where $I \defeq \la{\varthree}\varthree$ and $\delta \defeq \la{\vartwo}{\vartwo\vartwo}$) is \cbn normal but not \cbv normal (and it cannot normalise in \cbv).

\section{Preliminaries About Multi Types}
\label{sect:prelim}
In this section we define basic notions about multi types, type contexts, and (type) judgements that are shared by the three typing systems of the paper.

\paragraph{Multi-Sets.} The type systems are based on two layers of types, defined in a mutually recursive way, \emph{linear types} $\type$ and finite \emph{multi-sets} $\mtype$ of linear types. The intuition is that a linear type $\type$ corresponds to a single use of a term, and
that an argument $\tm$ is typed with a multi-set $\mtype$ of $n$ linear types if it is going to end up (at most)
$n$ times in evaluation position, with respect to the strategy associated with the type system. 
The three systems differ on the definition of linear types, that is therefore not specified here, while all adopt the same notion of finite multi-set $\mtype$ of linear types (named \emph{multi type}), that we now introduce:
  \begin{align*}
		\textsc{Multi types}
    && \mtype, \mtypetwo  & \grameq  \mult{\type_i}_{i \in J} \ \ (J \mbox{ a finite set}) 
  \end{align*}
where $\MSigma{\ldots}{}$ denotes the multi-set constructor. The empty multi-set $[\,]$ (the multi type obtained for $J = \emptyset$) is called \emph{empty (multi) type} and denoted by the special symbol $\emptymset$. An example of multi-set is $\MSigma{\type, \type, \typetwo}{}$, that contains two occurrences of $\type$ and one occurrence of $\typetwo$. Multi-set union is noted $\mplus$.

\paragraph{Type Contexts.} A \emph{type context} $\typctx$ is a map
  from variables to 
  multi types
  such that only finitely many variables are not mapped to $\emptytype$.
  The \emph{domain} of $\typctx$ is the set $\dom{\typctx} \defeq \set{x \mid \typctx(\var) \neq \emptytype}$.
  The type context $\typctx$ is \emph{empty} if $\dom{\typctx} = \emptyset$.  
  
  \emph{Multi-set union} $\mplus$ is extended to type contexts point-wise,
  \ie\  $\typctx \mplus \typctxtwo$ 
  maps  each variable $\var$ to $\typctx(\var) \mplus \typctxtwo(\var)$.
  This notion is extended to several contexts as expected, so that
  $\bigmplus_{i \in J} \typctx_i$ denotes a finite union of contexts---when $J = \emptyset$ the notation is to be understood as the empty context.
  We write $\typctx; \var \col \mtype$ for $\typctx\mplus (\var \mapsto \mtype)$
  only if  $\var \notin \dom{\typctx}$. 
  More generally,  we write $\typctx; \typctxtwo$ if the intersection between
  the domains of $\typctx$ and $\typctxtwo$ is empty.
  
  The \emph{restricted} context $\typctx$ with respect to the variable $\var$,
  written $\typctx \sm \var$ is defined by 
  $(\typctx \sm \var)(\var) \defeq \emptytype$ and
   $(\typctx \sm \var)(\vartwo) \defeq \typctx(\vartwo)$ if $\vartwo \neq \var$.
 
  \paragraph{Judgements.} \emph{Type judgements} are of the form
$\Deri[(\msteps, \esteps)] {\typctx}{\tm}{\type}$ or $\Deri[(\msteps, \esteps)] {\typctx}{\tm}{\mtype}$,
where the \emph{indices} $\msteps$ and $\esteps$ are natural numbers 
whose intended meaning is that $\tm$ evaluates to normal form in $\msteps$ multiplicative steps and $\esteps$ exponential steps, with respect to the evaluation strategy associated with the type system. 

To make clear in which type systems the judgement is derived, we write
$\tderiv \exder[\cbnup] \Deri[(\msteps, \esteps)] {\typctx}{\tm}{\type}$
if $\tderiv$ is a derivation in the \cbn system ending
in the judgement
$\Deri[(\msteps, \esteps)] { \typctx}{\tm}{ \type}$, and similarly for \cbv and \cbneed.

\section{Types by Name}
\label{sect:cbn}
In this section we introduce the \cbn multi type system, together with intuitions about multi types. We also prove that derivations provide exact bounds on \cbn evaluation sequences, and define the induced denotational model.

\begin{figure}[t]
  \centering
  \ovalbox{
    $\begin{array}{cccc}
      \infer[\ax]{\Deri[(0, 1)] {\var \col \single \type} \var \type}{} 
      &
      \infer[\normal]{
        \Deri[(0,0)] {} { \la\var\tm } \normal
      }{}
      \\\\
      \infer[\!\fun]{
        \Deri[(\msteps, \esteps)]{\typctx \sm \var} {\la\var\tm} {\ty{\typctx(\var)}{\type}}
      }{\Deri[(\msteps, \esteps)] {\typctx } \tm \type} 
      &
      
      \infer[\!\many]{
        \Deri[(\sum_{i \in J\!} \msteps_i, \sum_{i \in J\!} \esteps_i )] {
          \bigmplus_{i \in J}\typctxtwo_i  } \tm {\MSigma {\type_i} {i \in J}
        }
      }{
        (\Deri[(\msteps_i, \esteps_i)] {\typctxtwo_i} \tm {\type_i})_{i \in J}
      }
        \\\\

      \infer[\!\app]{
        \Deri[(\msteps + \msteps' + 1, \esteps + \esteps')] {\typctx \ctxplus \typctxtwo}
             {\tm \tmtwo} \type
      }{
        \Deri[(\msteps, \esteps)] \typctx \tm {\ty{\mtype}{\type}}
        \quad
        \Deri[(\msteps', \esteps')] {\typctxtwo} \tmtwo {\mtype}
      }
      &
      
      \infer[\!\esrule]{
        \Deri[(\msteps + \msteps', \esteps + \esteps')] {\typctx \ctxplus \typctxtwo}
             {\tm \esub\var\tmtwo} \type
      }{
        \Deri[(\msteps, \esteps)] {\typctx, \var \col \mtype} \tm {\type}
        \quad
        \Deri[(\msteps', \esteps')] {\typctxtwo} \tmtwo {\mtype}
      } 
      
    \end{array}$

  }
  \caption{Type system for \cbn evaluation}
  \label{fig:type-system-cbn}
\end{figure}

\paragraph{\cbn Types.} The system is essentially a reformulation of de Carvalho's system $\tt R$ \cite{DBLP:journals/mscs/Carvalho18}, itself being a type-based presentation of the relational model of the \cbn $\lambda$-calculus induced by relational model of linear logic via the \cbn translation of $\lambda$-calculus into linear logic. Definitions:
\begin{itemize}
  \item \cbn \emph{linear types} are given by the following grammar:
  \begin{align*}
    \textsc{\cbn linear types} && \type, \typetwo  & \grameq   \normal \mid  \ty{\mtype}{\type}
  \end{align*}
 Multi(-sets) types are defined as in \refsect{prelim}, relatively to \cbn linear types. 
 Note the linear constant $\normal$ (used to type abstractions, which are normal terms): it plays a crucial role in our quantitative analysis of \cbn evaluation.
  
  \item The \cbn \emph{typing rules} are in \reffig{type-system-cbn}. 
  
  \item The \emph{$\many$ rule}: it has as many premises as the elements in the (possibly empty) set of indices $J$. 
  When $J = \emptyset$, the rule has no premises, and it types $\tm$ with the empty multi type $\zero$. The $\many$ rule is needed to derive the right premises of the rules $\app$ and $\esrule$, that have a multi type $\mtype$ on their right-hand side. Essentially, it corresponds to the promotion rule of linear logic, that, in the \cbn representation of the $\lambda$-calculus, is indeed used for typing the right subterm of applications and the content of explicit substitutions.

\item The \emph{size} of a  derivation $\tderiv \exder[\cbnup] \Deri[(\msteps, \esteps)] {\typctx}{\tm}{\type}$ is the sum $\msteps + \esteps$ of the indices.
A quick look to the typing rules shows that indices on typing judgements are not 
needed, as $\msteps$ can be recovered as the number of $\app$ rules, and $\esteps$ as the number of $\ax$ rules. It is however handy to note them~explicitly. 
\end{itemize}

\paragraph{Subtleties and easy facts.} Let us overview some facts about our presentation of the type system.

  \begin{enumerate} 
\item \emph{Introduction and destruction of multi-sets}: multi-set are introduced on the right by the $\many$ rule and on the left by $\ax$. Moreover, on the left they are summed by $\app$ and $\esrule$.

\item \emph{Vacuous abstractions}:  we rely on the convention that the abstraction rule $\fun$ can always abstract a variable $\var$ not explicitly occurring in the context.
Indeed, if $\var \notin \dom{\typctx}$, then
$\typctx \sm \var$ is equal to $\typctx$ since $\typctx(\var) = \emptytype$.

\item \emph{Relevance}: 
    No weakening is allowed in axioms. An easy induction on type derivations shows that 
\begin{lemma}[Type contexts and variable occurrences for \cbn]
\label{l:name-typctx-varocc}Let $\tderiv \exder[\cbnup] \Deri[(\msteps, \esteps)] \typctx \tm \type$ be a  derivation. If $\var \not \in \fv\tm$ then $\var \notin \dom\typctx$.
\end{lemma}

\reflemma{name-typctx-varocc} implies 
that  derivations of closed terms have empty type context.
Note that there can be free variables of $\tm$ not 
in $\dom{\typctx}$: 
the ones only occurring in subterms not touched by the evaluation strategy.    
  \end{enumerate}

\paragraph{Key Ingredients.} Two key points of the \cbn system that play a role in the design of the \cbneed one in \refsect{cbneed} are:
\begin{enumerate}  
  \item \emph{Erasable terms and $\zero$}: the empty multi type $\zero$ is the type of erasable terms. Indeed, abstractions that erase their argument---whose paradigmatic example is $\la\var\vartwo$---can only be typed with $\ty\zero\type$, because of $\reflemma{name-typctx-varocc}$. Note that in \cbn every term---even diverging ones---can be typed with $\zero$ by rule $\many$ (taking 0 premises), because, correctly, in \cbn every term can be erased.
  \item \emph{Adequacy and linear types}: all \cbn typing rules but 
 $\many$ assign linear types. 
 And $\many$ is used only as right premise of the rules $\app$ and $\esrule$, to derive $\mtype$. It is with respect to linear types, in fact, that the adequacy of the system is going to be proved: a term is \cbn normalising if and only if it is typable with a linear type, given by \refthm{name-correctness} and \refthm{name-completeness} below.
\end{enumerate}

\paragraph{Tight derivations.} A 
term 
may have 
several  derivations, indexed by different pairs $(\msteps, \esteps)$.
They always provide upper bounds on \cbn evaluation lengths. The interesting aspect of our type systems, however, is that there is a simple description of a class of  derivations that provide \emph{exact} bounds for these quantities, as we shall show. Their definition relies on the $\normal$ type constant.

\begin{definition}[\Precise derivations for \cbn]
	\label{def:tightderiv}\strut
A  derivation $\tderiv \exder[\cbnup\!] \Deri[(\msteps, \esteps)]{\typctx}{\tm\!}{\!\type}$
  is \emph{\precise} 
  if $\type = \normal$ and $\typctx$ is empty. 
\end{definition}

\begin{example}
\label{ex:cbn-type-deriv}
Let us return to the term $\tm \defeq ((\la{\var}{\la{\vartwo}{\var \var}})(II))(II)$ used in \refex{evaluation} for explaining the difference in reduction lengths among the different strategies. We now give a derivation for it in the $\cbn$ type system. 

First, let us shorten  $\normal$ to $\shortnormal$. Then, we define $\tderiv$ as the following derivation for the subterm $\la{\var}{\la{\vartwo}{\var \var}}$ of $\tm$:
$$
\infer
	[\fun]
	{\tyjp{(1,2)}{\la{\var}{\la{\vartwo}{\var \var}}}{}{\ty{\mult{\shortnormal, \ty{\mult{\shortnormal}}{\shortnormal}}}{(\ty{\zero}{\shortnormal}})}}
	{\infer
		[\fun]
		{\tyjp{(1,2)}{\la{\vartwo}{\var \var}}{\var : \mult{\shortnormal, \ty{\mult{\shortnormal}}{\shortnormal}}}{\ty{\zero}{\shortnormal}}}
		{\infer
			[\app]
			{\tyjp{(1,2)}{\var \var}{\var \hastype \mult{\shortnormal, \ty{\mult{\shortnormal}}{\shortnormal}}}{\shortnormal}}
			{\infer
				[\ax]
				{\tyjp{(0,1)}{\var}{\var \hastype \mult{\ty{\mult{\shortnormal}}{\shortnormal}}}{\ty{\mult{\shortnormal}}{\shortnormal}}}
				{}
			\quad
			\infer
				[\many]
				{\tyjp{(0,1)}{\var}{\var \hastype \mult{\shortnormal}}{\mult{\shortnormal}}}
				{\infer
					[\ax]
					{\tyjp{(0,1)}{\var}{\var \hastype \mult{\shortnormal}}{\shortnormal}}
					{}
				}
			}
		}
	}
$$
Now, we need two derivations for $I I$, one of type $\shortnormal$, given by $\tderivtwo$ as follows
$$
\infer
	[\app]
	{\tyjp{(1,1)}{II}{}{\shortnormal}}
	{\infer
		[\fun]
		{\tyjp{(0,1)}{\la{\varthree}{\varthree}}{}{\ty{\mult{\shortnormal}}{\shortnormal}}}
		{\infer
			[\ax]
			{\tyjp{(0,1)}{\varthree}{\varthree : \mult{\shortnormal}}{\shortnormal}}
			{}
		}
	\quad
	\infer
		[\many]
		{\tyjp{(0,0)}{\la{\varfour}{\varfour}}{}{\mult{\shortnormal}}}
		{\infer
			[\normal]
			{\tyjp{(0,0)}{\la{\varfour}{\varfour}}{}{\shortnormal}}
			{}
		}
	}
$$
and one of type $\ty{\mult{\shortnormal}}{\shortnormal}$, 
given by $\tderivfour$ as follows

$$
\scalebox{0.945}{
\infer
	[\app]
	{\tyjp{(1,2)}{II}{}{\ty{\mult{\shortnormal}}{\shortnormal}}}
	{\infer
		[\fun]
		{\tyjp
			{(0,1)}
			{\la{\varthree}{\varthree}}
			{}
			{\ty{\mult{\ty{\mult{\shortnormal}}{\shortnormal}}}{(\ty{\mult{\shortnormal}}{\shortnormal})}}}
		{\infer
			[\ax]
			{\tyjp{(0,1)}{\varthree}{\varthree : \mult{\ty{\mult{\shortnormal}}{\shortnormal}}}{\ty{\mult{\shortnormal}}{\shortnormal}}}
			{}
		}
	\quad
	\infer
		[\many]
		{\tyjp{(0,1)}{\la{\varfour}{\varfour}}{}{\mult{\ty{\mult{\shortnormal}}{\shortnormal}}}}
		{\infer
			[\fun]
			{\tyjp{(0,1)}{\la{\varfour}{\varfour}}{}{\ty{\mult{\shortnormal}}{\shortnormal}}}
			{\infer
				[\ax]
				{\tyjp{(0,1)}{\varfour}{\varfour : \mult{\shortnormal}}{\shortnormal}}
				{}
			}
		}
	}
}
$$

\noindent Finally, we put $\tderiv$, $\tderivtwo$ and $\tderivfour$ together 
in the following derivation $\tderivthree$ for $\tm = (\tmtwo (II)) (II)$, where $\tmtwo \defeq \la{\var}\la{\vartwo}\var\var$ and $\cbnnormaltonormal \defeq \ty{\mult{\shortnormal}}{\shortnormal}$
\begin{center}
\vspace{-0.15\baselineskip}
\scalebox{0.945}{
\begin{prooftree}[separation=1em, label separation=0.3em, rule margin=0.5ex]
	\hypo{}
	\ellipsis{$\tderiv$}{\Deri [(1,2)] {}{\tmtwo}{\ty{\mult{\shortnormal, \cbnnormaltonormal}}{(\ty{\zero}{\shortnormal})}}}
	\hypo{}
	\ellipsis{$\tderivtwo$}{\tyjp{(1,1)}{II}{}{\shortnormal}}
	\hypo{}
	\ellipsis{$\tderivfour$}{\tyjp{(1,2)}{II}{}{\cbnnormaltonormal}}
	\infer2[$\many$]{\tyjp{(2,3)}{II}{}{\mult{\shortnormal, \cbnnormaltonormal}}}
	\infer2[$\app$]{\tyjp{(4,5)}{\tmtwo(II)}{}{\ty{\zero}{\shortnormal}}}
	\infer0[$\many$]{\tyjp{(0,0)}{II}{}{\zero}}
	\infer2[$\app$]{\tyjp{(5,5)}{(\tmtwo(II))(II)}{}{\shortnormal}}
\end{prooftree}
}
\end{center}
Note that that $\tderivthree$ is a tight derivation and the indices $(5,5)$ correspond exactly to the number of 
$\mcbn$-steps and $\ecbn$-steps, respectively, from $\tm$ to its $\cbnsym$-normal form, as shown in \refex{evaluation}. 
\refthm{name-correctness} below shows that this is not by chance: tight derivations are minimal and provide exact bounds to evaluation lengths.
\end{example}

The next two subsections prove the two halves of the properties of the \cbn type system, namely correctness and completeness. 

\subsection{\cbn Correctness}\label{s:tightcorrectness}
Correctness is the fact that every typable term is \cbn normalising. In our setting it comes with additional quantitative information: the indices $\msteps$ and $\esteps$ of a  derivation $\tderiv\exder[\cbnup] \Deri[(\msteps, \esteps)]{\typctx}{\tm}\type$ provide bounds for the length of the \cbn evaluation of $\tm$, that are exact when the derivation is tight.

The proof technique is standard. Moreover, the correctness theorems for \cbv and \cbneed in the next sections follow \emph{exactly} the same structure. The proof relies on a quantitative subject reduction property showing that $\msteps$ decreases by \emph{exactly one} at each $\mcbn$-step, and similarly for $\esteps$ and $\ecbn$-steps. In turn, subject reduction relies on a linear substitution lemma. Last, correctness for \emph{tight} derivations requires a further property of normal forms.

Let us point out that correctness is stated with respect to closed terms only, but the auxiliary results have to deal with open terms, since they are proved by inductions (over predicates defined by induction) over the structure of terms.\medskip

\paragraph{Linear Substitution.} The linear substitution lemma states that substituting over a variable occurrence as in the exponential rule consumes exactly one linear type and decreases of one the exponential index $\esteps$. 

\begin{toappendix}
\begin{lemma}[\cbn linear substitution]
\label{l:name-linear-substitution}
If $\tderiv \exder[\cbnup]\Deri[(\msteps, \esteps)]{\typctx, \var \hastype \mtype}{\cbnctx\cwc{\var}}{\type}$ 
then there is a splitting $ \mtype = \mult\typetwo \mplus \mtypetwo$ such that for every derivation $\tderivtwo \exder[\cbnup]\Deri[(\mstepstwo, \estepstwo)]{\typctxtwo}{\tm}{\typetwo}$ there is a derivation 
$\tderivthree \exder[\cbnup] \Deri[(\msteps+\mstepstwo, \esteps + \estepstwo-1)]{\typctx \mplus \typctxtwo, \var \hastype \mtypetwo}{\cbnctx\cwc{\tm}}{\type}$.
\end{lemma}
\end{toappendix}
   
The proof is by induction over \cbn evaluation contexts.

\paragraph{Quantitative Subject Reduction.} A key point of multi types is that the size of type derivations shrinks after every evaluation step, which is what allows to bound evaluation lengths. Remarkably, the size (defined as the sum of the indices) shrinks by exactly 1 at every evaluation step.
   
\begin{toappendix}
\begin{proposition}[Quantitative subject reduction for \cbn]
  \label{prop:name-subject-reduction}
  Let $\tderiv\exder[\cbnup] \Deri[(\msteps, \esteps)]\typctx{\tm}\type$ be a  derivation. 
  \begin{enumerate}
    \item \emph{Multiplicative}: if $\tm\tomcbn\tmtwo$ then $\msteps\geq 1$ and
  there exists a  derivation 
  $\tderivtwo\exder[\cbnup] \Deri[(\msteps-1, \esteps)]\typctx{\tmtwo}\type$.
  
  \item \emph{Exponential}: if $\tm\toecbn\tmtwo$ then $\esteps\geq 1$ and
  there exists a  derivation
  $\tderivtwo\exder[\cbnup] \Deri[(\msteps, \esteps-1)]\typctx{\tmtwo}\type$. 
  \end{enumerate}
\end{proposition}  
\end{toappendix}

The proof is by induction on $\tm\tomcbn\tmtwo$ and $\tm\toecbn\tmtwo$, using the linear substitution lemma for the root exponential step.

\paragraph{Tightness and Normal Forms.} Since the indices are always non-negative, quantitative subject reduction (\refprop{name-subject-reduction}) implies that they bound evaluation lengths. 
The bound is not necessarily exact, as derivations of normal forms can have strictly positive indices. 
If they are tight, however, they are  indexed by $(0,0)$, as we now show. The proof of this fact (by induction on the predicate $\normal$) 
requires a slightly different statement, for the induction to go through.

\begin{toappendix}
\begin{proposition}[\normal typing of normal forms for \cbn]
  \label{prop:name-normal-forms-forall}
Let $\tm$ be such that $\normalpr\tm$, and
$\tderiv\exder[\cbnup] \Deri[(\msteps, \esteps)]{\typctx}{\tm}\normal$ be a derivation. 
Then $\typctx$ is empty, and so $\tderiv$ is tight, and $\msteps = \esteps = 0$.
\end{proposition}
\end{toappendix}


\paragraph{The Tight Correctness Theorem.} The theorem is then proved by a straightforward induction on the evaluation length relying on quantitative subject reduction (\refprop{name-subject-reduction}) for the inductive case, and the properties of tight typings for normal forms (\refprop{name-normal-forms-forall}) for the base case.

\begin{toappendix}
\begin{theorem}[\cbn tight correctness]
  \label{thm:name-correctness} 
  Let $\tm$ be a closed term.
  If $\tderiv \exder[\cbnup]  \Deri[(\msteps, \esteps)] {}{\tm}{\type}$ then there is $\tmtwo$ such that $\deriv \colon \tm
  \tocbnn \tmtwo$, $\normalpr\tmtwo$, $\sizem\deriv \leq \msteps$, $\sizee\deriv \leq \esteps$.  Moreover, if $\tderiv$ is tight then $\sizem\deriv = \msteps$ and $\sizee\deriv = \esteps$.
\end{theorem}
\end{toappendix}

Note that \refthm{name-correctness} implicitly states that tight  derivations have \emph{minimal} size among  derivations.

\subsection{\cbn Completeness}
\label{s:tightcompleteness}
Completeness is the fact that every normalising term has a (tight) type derivation. As for correctness, the completeness theorem is always obtained via three intermediate steps, dual to those for correctness. 

\paragraph{Normal Forms.} The first step is to prove (by induction on the predicate $\normal$) that every normal form is typable, and 
is actually typable with a tight derivation.

\begin{toappendix}
\begin{proposition}[Normal forms are tightly typable for \cbn]
\label{prop:name-normal-forms-exist} 
Let $\tm$ be such that $\normalpr\tm$. Then 
there is \precise derivation $\tderiv\exder[\cbnup] \Deri[(0, 0)]{}{\tm}\normal$.
\end{proposition}
\end{toappendix}


\paragraph{Linear Removal.} In order to prove subject expansion, we have to first show that typability can also be pulled back along substitutions, via a linear removal lemma dual to the linear substitution lemma. 

\begin{toappendix}
\begin{lemma}[Linear removal for \cbn]
\label{l:name-linear-removal}
Let  $\tderiv \exder[\cbnup] \Deri[(\msteps, \esteps)]{\typctx;\var \hastype\mtype}{\cbnctx\cwc{\tmtwo}}{\type}$,
where $\var \notin \fv\tmtwo$. Then  there exist
\begin{itemize}
\item a linear type  $\typetwo$,
\item a derivation $\tderiv_\tmtwo  \exder[\cbnup] \Deri[(\msteps_\tmtwo, \esteps_\tmtwo)]{\typctx_\tmtwo}{\tmtwo}{\typetwo}$, and
\item a derivation $\tderiv_{\cbnctx\cwc{\var}} \exder[\cbnup] 
 \Deri[(\mstepstwo, \estepstwo)]{  \typctx', \var \hastype \mtype\mplus\mult{\typetwo} }{\cbnctx\cwc{\var}}{\type}$
\end{itemize}
such that 
\begin{itemize}
\item \emph{Type contexts}: $\typctx = \typctx_\tmtwo \uplus \typctx'$.

\item \emph{Indices}: $(\msteps, \esteps) = (\mstepstwo +  \msteps_\tmtwo, \estepstwo + \esteps_\tmtwo- 1)$.
\end{itemize}
\end{lemma}
\end{toappendix}

\paragraph{Quantitative Subject Expansion.} This property is the dual of subject reduction.

\begin{toappendix}
\begin{proposition}[Quantitative subject expansion for \cbn]
  \label{prop:name-subject-expansion}
  Let $\tderiv \exder[\cbnup] \Deri[(\msteps, \esteps)]\typctx\tmtwo\type$ be a derivation. 
  \begin{enumerate}
    \item \emph{Multiplicative}: if $\tm\tomcbn\tmtwo$ then 
  there is a derivation 
  $\tderivtwo\exder[\cbnup] \Deri[(\msteps+1, \esteps)]\typctx{\tm}\type$.
  
  \item \emph{Exponential}: if $\tm\toecbn \! \tmtwo$ then 
  there is a derivation 
  $\tderivtwo\exder[\cbnup] \Deri[(\msteps, \esteps+1)]\typctx{\tm}\type$. 
  \end{enumerate}
\end{proposition}
\end{toappendix}

The proof is by induction on $\tm\tomcbn\tmtwo$ and $\tm\toecbn\tmtwo$, using the linear removal lemma for the root exponential step.

\paragraph{The Tight Completeness Theorem.} The theorem is proved by a straightforward induction on the evaluation length relying on quantitative subject expansion (\refprop{name-subject-expansion}) in the inductive case, and the existence of tight typings for normal forms (\refprop{name-normal-forms-exist}) in the base case.

\begin{toappendix}
\begin{theorem}[\cbn tight completeness]
	\label{thm:name-completeness}
	Let $\tm$ be a closed term.
	If $\deriv \colon \!\tm \!\tocbnn \tmtwo$ and $\normalpr\tmtwo$ then there is a \precise derivation
	$\tderiv \exder[\cbnup\!\!\!] \Deri[(\sizem\deriv, \sizee\deriv)] {}\tm \normal$.    
\end{theorem}
\end{toappendix}

\paragraph{Back to Erasing Steps.} Our system can be easily adapted to measure also garbage collection steps (the \cbn erasing rule is just before \refexample{evaluation}, page \pageref{ex:evaluation}). First, a new, third index $g$ on judgements is necessary. Second, one needs to distinguish the erasing and non-erasing cases of the the $\app$ and $\esrule$ rules, discriminated by the $\zero$ type. For instance, the $\esrule$ rules are (the $\app$ rules are similar):
\[{\small
\begin{aligned}
\infer
	[\!\esgc]
	{\tyjp{(\msteps,\esteps,g+1)}{\tm \esub\var\tmtwo}{\typctx} {\type}}
	{\tyjp{(\msteps,\esteps,g)} {\tm} {\typctx} {\type}
	\quad
	\typctx(\var) = \emptytype}
&&
\infer
	[\!\esrule]
	{\tyjp{(\msteps + \msteps', \esteps + \esteps',g+g')}{\tm \esub{\var}{\tmtwo}}{\typctx \mplus \typctxtwo}{\type}}
	{\tyjp{(\msteps,\esteps,g)}{\tm}{\typctx, \var \hastype \mtype}{\type}
	\quad
	\tyjp{(\mstepstwo,\estepstwo,g')}{\tmtwo}{\typctxtwo}{\mtype}
	\quad
	\mtype \neq \emptytype}
\end{aligned}
}\]
\noindent The index $g$ bounds to the number of erasing steps. In the closed case, however, the bound cannot be, in general, exact. Variables typed with $\zero$ by $\typctx$ do not exactly match variables not appearing in the typed term (that is the condition triggering the erasing step), because a variable typed with $\zero$ may appear in the body of abstractions typed with the \normal rule, as such bodies are not typed.

It is reasonable to assume that exact bounds for erasing steps can only by provided by a type system characterising strong evaluation, whose typing rules have to inspect abstraction bodies. These erasing typing rules are nonetheless going to play a role in the design of the \cbneed system in \refsect{cbneed}.

\subsection{\cbn Model}
The idea to build the denotational model from the multi type system is that the interpretation (or 
semantics) of a term is simply the set of its type assignments, \ie the set of its derivable types together with their type contexts. 
More precisely, let $\tm$ be a term and $\var_1, \dots, \var_n$ (with $n \geq 0$) be pairwise distinct variables.
If $\fv{\tm} \subseteq \{\var_1, \dots, \var_n\}$, we say that the list $\vec{\var} = (\var_1, \dots, \var_n)$ is \emph{suitable for} $\tm$.
If $\vec{\var} = (\var_1, \dots, \var_n)$ is suitable for $\tm$, the (\emph{relational}) \emph{semantics} 
\emph{of} $\tm$ \emph{for} $\vec{\var}$ is
\begin{equation*}
    \sem{\tm}_{\vec{\var}}^\cbn \defeq \{((\mtype_1,\dots, \mtype_n),\type) \mid \exists 
    \, 
    \tderiv \exder[\cbnup] \Deri[(\msteps, \esteps)]{\var_1 \colon\! \mtype_1, \dots, \var_n \colon\! \mtype_n}\tm\type
    \} \,.
\end{equation*}
Subject reduction (\refprop{name-subject-reduction}) and expansion (\refprop{name-subject-expansion}) guarantee that the semantics $\sem{\tm}_{\vec{\var}}^\cbn$ of $\tm$ (for \emph{any} term $\tm$, possibly open) is \emph{invariant} by \cbn evaluation. Correctness (\refthm{name-correctness}) and completeness (\refthm{name-completeness}) guarantee that, given a \emph{closed} term $\tm$,  its interpretation $\sem{\tm}_{\vec{\var}}^\cbn$ is non-empty if and only if $\tm$ is \cbn normalisable, that is, they imply that relational semantics is \emph{adequate}. 

In fact, adequacy also holds with respect to open terms. The issue in that case is that the characterisation of tight derivations is more involved, see Accattoli, Graham-Lengrand and Kesner's \cite{DBLP:journals/pacmpl/AccattoliGK18}. Said differently, weaker correctness and completeness theorems without exact bounds also hold in the open case. The same is true for the \cbv and \cbneed systems of the next sections. 

\section{Types by Value}
\label{sect:cbv}

Here we introduce Ehrhard's \cbv multi type system \cite{DBLP:conf/csl/Ehrhard12} adapted to our presentation of \cbv in the LSC, and prove its properties. The system is similar, and yet in many aspects dual, to the \cbn one, in particular the grammar of types is different.
Linear types for \cbv are defined by:
\begin{align*} 
\textsc{\cbv linear types} && \type, \typetwo  & \grameq   
 \ty{\mtype}{\mtypetwo} 
\end{align*}
Multi(-sets) types are defined as in \refsect{prelim}, relatively to \cbv linear types. Note that linear types now have a multi type both as source and as target, and that the $\normal$ constant is absent---in \cbv, its role is played by $\zero$. 

The typing rules are in \reffig{type-system-cbv}. 
It is a type-based presentation of the relational model of the \cbv $\lambda$-calculus induced by relational model of linear logic via the \cbv translation of $\lambda$-calculus into linear logic.
Some remarks:
\begin{itemize}
	\item \emph{Right-hand types}: all rules but $\fun$ assign a multi type to the term on the right-hand side, and not a linear type as in \cbn.
	
	\item \emph{Abstractions and $\many$}: the $\many$ rule has a restricted form with respect to the \cbn one, it can only be applied to abstractions, that in turn are the only terms that can be typed with a linear type.
	
	\item \emph{Indices}: note as the indices are however incremented (on $\ax$ and $\app$) and summed (in $\many$ and $\esrule$) exactly as in the \cbn system.
\end{itemize}
\begin{figure}[t]
	\centering
	\ovalbox{
		$\begin{array}{c\colspace \colspace ccc}
		\infer[\ax]{\Deri[(0, 1)] {\var \col \mtype} \var \mtype}{} 
		&
\infer[\app]{
			\Deri[(\msteps + \msteps' + 1, \esteps + \esteps')] {\typctx \mplus \typctxtwo}
			{\tm \tmtwo} \mtypetwo
		}{
			\Deri[(\msteps, \esteps)] \typctx \tm {\single{\ty{\mtype}{\mtypetwo}}}
			\quad
			\Deri[(\msteps', \esteps')] {\typctxtwo} \tmtwo {\mtype}
		}
				\\\\
		\infer[\fun]{
			\Deri[(\msteps, \esteps)]{\typctx} {\la\var\tm} {\ty{\mtypetwo}{\mtype}}
		}{\Deri[(\msteps, \esteps)] {\typctx, \var \col \mtypetwo} \tm \mtype} 
		&
		
		\infer[\many]{
			\Deri[(\sum_{i \in J\!} \msteps_i, \sum_{i \in J\!} \esteps_i )] {
				\bigmplus_{i \in J}\typctxtwo_i  } {\la\var\tm} {\MSigma {\type_i} {i \in J}
			}
		}{
			(\Deri[(\msteps_i, \esteps_i)] {\typctxtwo_i} {\la\var\tm} {\type_i})_{i \in J}
		}
		\\\\

		&
		
		\infer[\esrule]{
			\Deri[(\msteps + \msteps', \esteps + \esteps')] {\typctx \mplus \typctxtwo}
			{\tm \esub\var\tmtwo} \mtype
		}{
			\Deri[(\msteps, \esteps)] {\typctx, \var \col \mtypetwo} \tm {\mtype}
			\quad
			\Deri[(\msteps', \esteps')] {\typctxtwo} \tmtwo {\mtypetwo}
		} 
		
		\end{array}$
		
	}
	\caption{Type system for \cbv evaluation.}
	\label{fig:type-system-cbv}
\end{figure}

\paragraph{Intuitions: the Empty Type $\emptymset$.} The empty multi-set type $\emptymset$ plays a special role in \cbv. As in \cbn, it is the type of terms that can be erased, but, in contrast to \cbn, not every term is erasable in \cbv.

In the \cbn multi type system every term, even a diverging one, is typable with $\emptymset$. On the one hand, this is correct, because in \cbn every term can be erased, and erased terms can also be divergent, because they are never evaluated. On the other hand, adequacy is formulated with respect to non-empty types: a term terminates if and only if it is typable with a non-empty type.

In \cbv, instead, terms have to be evaluated before being erased; and, of course, their evaluation has to terminate. 
Thus, terminating terms and erasable terms coincide. Since the multi type system is meant to characterise terminating terms, in \cbv a term is typable
if and only if it is typable with $\emptymset$, as we shall prove in this section. 
Then the empty type is not a degenerate type excluded for adequacy from the interesting types of a term, as in \cbn, it rather is \emph{the} type, characterising (adequate) typability altogether. And this is also the reason for the absence of the constant $\normal$---one way to see it is that in \cbv $\normal = \zero$.

Note that, in particular, in a type judgement $\typctx \vdash \tm \hastype \mtype$ the type context $\typctx$ may give the empty type to a variable $\var$ occurring in $\tm$, as for instance in the axiom $\var \hastype \emptymset \vdash \var \hastype \emptymset$---this may seem very strange to people familiar with \cbn multi types. We hope that instead, according to the provided intuition that $\emptymset$ is the type of termination, it would rather seem natural.

\begin{definition}[\Precise derivation for \cbv]\label{def:tightderiv-cbv}
	A derivation $\tderiv \exder[\cbvup] \Deri[(\msteps, \esteps)]{\typctx}{\tm}{\mtype}$
	is \emph{\precise} 
	if $\mtype = \emptymset$ and $\typctx$ is empty. 
\end{definition}


\begin{example}
\label{ex:value-type-deriv}
Let's consider again the term $\tm \defeq ((\la{\var}{\la{\vartwo}{\var \var}})(II))(II)$ of \refex{evaluation} (where $I \defeq \la{\varthree}\varthree$), for which a $\cbn$ tight derivation was given in \refex{cbn-type-deriv}, and let us type it in the $\cbv$ system with a tight derivation.

We define the following derivation $\tderiv_{1}$ for the subterm $\tmtwo \defeq \la{\var}{\la{\vartwo}{\var \var}}$ of $\tm$

	$$
	\infer
		[\many]
		{\tyjp{(1,2)}{\tmtwo}{}{\mult{\ty{\mult{\ty{\zero}{\zero}}}{\mult{\ty{\zero}{\zero}}}}}}
		{\infer
			[\fun]
			{\tyjp{(1,2)}{\tmtwo}{}{\ty{\mult{\ty{\zero}{\zero}}}{\mult{\ty{\zero}{\zero}}}}}
			{\infer
				[\many]
				{\tyjp{(1,2)}{\la{\vartwo}{\var \var}}{\var : \mult{\ty{\zero}{\zero}}}{\mult{\ty{\zero}{\zero}}}}
				{\infer
					[\fun]
					{\tyjp{(1,2)}{\la{\vartwo}{\var \var}}{\var : \mult{\ty{\zero}{\zero}}}{\ty{\zero}{\zero}}}
					{\infer
						[\app]
						{\tyjp{(1,2)}{\var \var}{\var : \mult{\ty{\zero}{\zero}}}{\zero}}
						{\infer
							[\ax]
							{\tyjp{(0,1)}{\var}{\var : \mult{\ty{\zero}{\zero}}}{\mult{\ty{\zero}{\zero}}}}
							{}
						\quad
						\infer
							[\ax]
							{\tyjp{(0,1)}{\var}{\var : \zero}{\zero}}
							{}
						}
					}
				}
			}
		}
	$$
	Note that $\mult{\ty{\zero}{\zero}} \mplus \zero = \mult{\ty{\zero}{\zero}}$, which explains the shape of the type context in the conclusion of the $\app$ rule. 
Next, we define the derivation $\tderiv_{2}$ as follows 
	$$
	\infer
		[\app]
		{\tyjp{(1,2)}{II}{}{\mult{\ty{\zero}{\zero}}}}
		{\infer
			[\many]
			{\tyjp{(0,1)}{\la{\varthree}{\varthree}}{}{\mult{\ty{\mult{\ty{\zero}{\zero}}}{\mult{\ty{\zero}{\zero}}}}}}
			{\infer
				[\fun]
				{\tyjp{(0,1)}{\la{\varthree}{\varthree}}{}{\ty{\mult{\ty{\zero}{\zero}}}{\mult{\ty{\zero}{\zero}}}}}
				{\infer
					[\ax]
					{\tyjp{(0,1)}{\varthree}{\varthree : \mult{\ty{\zero}{\zero}}}{\mult{\ty{\zero}{\zero}}}}
					{}
				}
			}
		\quad
		\infer
			[\many]
			{\tyjp{(0,1)}{\la{\varfour}{\varfour}}{}{\mult{\ty{\zero}{\zero}}}}
			{\infer
				[\fun]
				{\tyjp{(0,1)}{\la{\varfour}{\varfour}}{}{\ty{\zero}{\zero}}}
				{\infer
					[\ax]
					{\tyjp{(0,1)}{\varfour}{\varfour : \zero}{\zero}}
					{}
				}
			}
		}
	$$
and the derivation $\tderiv_{3}$ as follows 
	$$
	\infer
		[\app]
		{\tyjp{(1,1)}{II}{}{\zero}}
		{\infer
			[\many]
			{\tyjp{(0,1)}{\la{\var'}{\var'}}{}{\mult{\ty{\zero}{\zero}}}}
			{\infer
				[\fun]
				{\tyjp{(0,1)}{\la{\var'}{\var'}}{}{\ty{\zero}{\zero}}}
				{\infer
					[\ax]
					{\tyjp{(0,1)}{\var'}{\var' : \zero}{\zero}}
					{}
				}
			}
		\quad
		\infer
			[\many]
			{\tyjp{(0,0)}{I}{}{\zero}}
			{}
		}
	$$
	
Finally, we put $\tderiv_{1}$, $\tderiv_{2}$ and $\tderiv_{3}$ together 
in the following derivation $\tderiv$~for~$\tm$
\begin{prooftree*}
	\hypo{}
	\ellipsis{$\tderiv_{1}$}{\tyjp{(1,2)}{\tmtwo}{}{\mult{\ty{\mult{\ty{\zero}{\zero}}}{\mult{\ty{\zero}{\zero}}}}}}
	\hypo{}
	\ellipsis{$\tderiv_{2}$}{\tyjp{(1,2)}{II}{}{\mult{\ty{\zero}{\zero}}}}
	\infer2[$\app$]{\tyjp{(3,4)}{(\la{\var}{\la{\vartwo}{\var \var}})(II)}{}{\mult{\ty{\zero}{\zero}}}}
	\hypo{}
	\ellipsis{$\tderiv_{3}$}{\tyjp{(1,1)}{II}{}{\zero}}
	\infer2[$\app$]{\tyjp{(5,5)}{((\la{\var}{\la{\vartwo}{\var \var}})(II))(II)}{}{\zero}}
\end{prooftree*}
Note that the indices $(5,5)$ correspond exactly to the number of 
$\mcbv$-steps and $\ecbv$-steps, respectively, from $\tm$ to its \cbv normal form, as shown in \refex{evaluation}, and that $\tderiv$ is a tight derivation. 
Forthcoming \refthm{value-correctness} shows that $\cbv$ tight derivations are minimal and provide exact bounds to evaluation lengths in $\cbv$.
\end{example}

\emph{Correctness} (\ie typability implies normalisability) and \emph{completeness} (\ie normalisability implies typability) of the \cbv type system with respect to \cbv evaluation (together with quantitative information about evaluation lengths) follow exactly the same pattern of the \cbn case, \emph{mutatis mutandis}. 

\subsection{\cbv Correctness}

\begin{toappendix}
\begin{lemma}[\cbv linear substitution]
	\label{l:value-linear-substitution-bis}
	Let $\tderiv\exder[\cbvup] \Deri[(\msteps, \esteps)]{\typctx, \var \col \mtype}{\cbvctx\cwc{\var}}{\mtypetwo}$.
	Then there exists a splitting $\mtype = \mtypethree \mplus \mtypefour$ such that, for every derivation $\tderivtwo\exder[\cbvup] \Deri[(\mstepstwo,\estepstwo)]{\typctxtwo}{\val}{\mtypethree}$, there is a derivation $\tderiv'\exder[\cbvup]\Deri[(\msteps+\mstepstwo, \esteps+\estepstwo-1)]{\typctx \mplus \typctxtwo, \var \col \mtypefour}{\cbvctx\cwc{\val}}{\mtypetwo}$.
\end{lemma}	
\end{toappendix}

\begin{toappendix}
\begin{proposition}[Quantitative subject reduction for \cbv]
	\label{prop:value-subject-reduction}
	Let $\tderiv\exder[\cbvup] \Deri[(\msteps, \esteps)]\typctx{\tm}\mtype$ be a derivation. 
	\begin{enumerate}
		\item \emph{Multiplicative}: if $\tm\tomcbv\tm'$ then $\msteps\geq 1$ and
		there exists a derivation 
		$\tderiv'\exder[\cbvup] \Deri[(\msteps-1, \esteps)]\typctx{\tm'}\mtype$.
		
		\item \emph{Exponential}: if $\tm\toecbv\tm'$ then $\esteps\geq 1$ and
		there exists a derivation 
		$\tderiv'\exder[\cbvup] \Deri[(\msteps, \esteps-1)]\typctx{\tm'}\mtype$. 
	\end{enumerate}
\end{proposition}  
\end{toappendix}

\begin{toappendix}
\begin{proposition}[Tight typings for normal forms for \cbv]
	\label{prop:value-normal-forms-forall}
	Let 
	$\tderiv\exder[\cbvup] \Deri[(\msteps, \esteps)]\typctx{\tm}\emptymset$ be a derivation, with 
	$\normalcbvpr{\tm}$.
	Then $\typctx$ is empty, and so $\tderiv$ is tight, and $\msteps = \esteps = 0$.
\end{proposition}
\end{toappendix}

\begin{toappendix}
\begin{theorem}[\cbv tight correctness]
	\label{thm:value-correctness} 
  Let $\tm$ be a closed term. 
  If $\tderiv \exder[\cbvup] \Deri[(\msteps, \esteps)] {\typctx}{\tm}{\mtype}$ then there is $\tmtwo$ such that $\deriv \colon \tm
	\tocbvn \tmtwo$, $\normalcbvpr\tmtwo$, $\sizem\deriv \leq \msteps$, $\sizee\deriv \leq \esteps$.  
	Moreover, if $\tderiv$ is tight then $\sizem\deriv = \msteps$ and $\sizee\deriv = \esteps$.
\end{theorem}
\end{toappendix}

\subsection{\cbv Completeness}
\begin{toappendix}
	\begin{proposition}[Normal forms are tightly typable for \cbv]
		\label{prop:value-normal-forms-exist} 
		Let $\tm$ be such that $\normalcbvpr\tm$. Then 
		there exists a \precise derivation $\tderiv\exder[\cbvup] \Deri[(0, 0)]{}{\tm}\emptymset$.
	\end{proposition}
\end{toappendix}

\begin{toappendix}
\begin{lemma}[Linear removal for \cbv]
\label{l:value-linear-anti-substitution}
	Let $\tderiv \exder[\cbvup] \Deri[(\msteps, \esteps)]{\typctx, \var \col \mtype}{\cbvctx\cwc{\val}}{\mtypetwo}$ where $\var \notin \fv{\val}$.
	Then, there exist
	\begin{itemize}
		\item a multi type $\mtype'$ and two type contexts $\typctx'$ and $\typctxtwo$, 
		\item a derivation $\tderivtwo \exder[\cbvup] \Deri[(\mstepsthree, \estepsthree)] {\var \col \mtype \mplus \mtype', \typctxtwo}{\cbvctx\cwc{\var}}{\mtypetwo}$, and 
		\item a derivation $\tderiv' \exder[\cbvup] \Deri[(\mstepstwo,\estepstwo)]{\typctx'}{\val}{\mtype'}$
	\end{itemize}
		such that 
		\begin{itemize}
			\item\emph{Type contexts:} $\typctx = \typctx' \mplus \typctxtwo$,
			\item\emph{Indices:} $(\msteps, \esteps) = (\mstepstwo + \mstepsthree, \estepstwo + \estepsthree - 1)$.
		\end{itemize} 
\end{lemma}
\end{toappendix}

\begin{toappendix}
	\begin{proposition}[Quantitative subject expansion for \cbv]
		\label{prop:value-subject-expansion}
		Let $\tderiv' \exder[\cbvup] \Deri[(\msteps, \esteps)]\typctx{\tm'}\mtype$ be a derivation. 
		\begin{enumerate}
			\item \emph{Multiplicative}: if $\tm\tomcbv\tm'$ then 
			there is a derivation 
			$\tderiv\exder[\cbvup] \Deri[(\msteps+1, \esteps)]\typctx{\tm}\mtype$.
			
			\item \emph{Exponential}: if $\tm\toecbv\tm'$ then 
			there is a derivation 
			$\tderiv\exder[\cbvup] \Deri[(\msteps, \esteps+1)]\typctx{\tm}\mtype$. 
		\end{enumerate}
	\end{proposition}
\end{toappendix}

\begin{toappendix}
	\begin{theorem}[\cbv tight completeness]
		\label{thm:value-completeness}
		    Let $\tm$ be a closed term.
		    If $\deriv \colon\! \tm \tocbvn\! \tmtwo$ with $\normalcbvpr\tmtwo$,
		then there is a \precise derivation
		$\tderiv \exder[\cbvup] \Deri[(\sizem\deriv, \sizee\deriv)] {}\tm \emptymset$.    
	\end{theorem}
\end{toappendix}

\paragraph{\cbv Model.} The interpretation of terms with respect to the \cbv system is defined as follows (where $\vec{\var} = (\var_1, \dots, \var_n)$ is a list of variables suitable for $\tm$):
\begin{equation*}
    \sem{\tm}_{\vec{\var}}^\cbv \defeq \{((\mtype_1,\dots, \mtype_n),\mtypetwo) \mid \exists 
    \, 
    \tderiv \exder[\cbvup] \Deri[(\msteps, \esteps)]{\var_1 \colon\! \mtype_1, \dots, \var_n \colon\! \mtype_n}\tm\mtypetwo
    \} \,.
\end{equation*}
Note that rule $\fun$ assigns a linear type but the interpretation considers only multi types. The \emph{invariance} 
and the \emph{adequacy} of $\sem{\tm}_{\vec{\var}}^\cbv$ with respect to \cbv evaluation are obtained exactly as for the \cbn case. 

\section{Types by Need}
\label{sect:cbneed}

\paragraph{\cbneed as a Blend of \cbn and \cbv.} The multi type system for \cbneed is obtained by carefully blending ingredients from the \cbn and \cbv ones:
\begin{itemize}
  \item \emph{Wise erasures from \cbn}: in \cbn wise erasures are induced by the fact that the empty multi type $\zero$ (the type of erasable terms) and the linear type $\normal$ (the type of normalisable terms) are distinct and every term is typable with $\zero$ by using the $\many$ rule with 0 premises. Adequacy is then formulated with respect to (non-empty) linear types.
  \item \emph{Wise duplications from \cbv}: in \cbv wise duplications are due to two aspects. First, only abstractions can be collected in multi-sets by rule $\many$. This fact accounts for the evaluation of arguments to normal form---that is, abstractions---before being substituted. Second, terms are typed with multi types instead of linear types. Roughly, this second fact allows the first one to actually work because the argument is reduced once for a whole multi set of types, and not once for each element of the multi set, as in \cbn.
\end{itemize}
It seems then that a type system for \cbneed can easily be obtained by basically adopting the \cbv system plus 
\begin{itemize}
\item separating $\emptymset$ and $\normal$, that is, adding $\normal$ to the system;
\item modifying the $\many$ rule by distinguishing two cases: with 0 premises it can assign $\zero$ to whatever term---as in \cbn---otherwise it is forced to work on abstractions, as in \cbv;
\item restricting adequacy to non-empty types.
\end{itemize}
Therefore, the grammar of linear types is:
\begin{align*}
    \textsc{\cbneed linear types} && \type, \typetwo  & \grameq   \normal \mid  \ty{\mtype}{\mtypetwo} 
\end{align*}
Multi(-sets) types are defined as in \refsect{prelim}, relatively to \cbneed linear types. The rules of this \emph{na\"ive system} for \cbneed are in \reffig{type-system-cbneed-natural}. 
\begin{figure}[t]
	\centering
	\ovalbox{
		$\begin{array}{c\colspace \colspace ccc}
		\infer[\ax]{\Deri[(0, 1)] {\var \col \mtype} \var \mtype}{} 
		&
\infer[\app]{
			\Deri[(\msteps + \msteps' + 1, \esteps + \esteps')] {\typctx \mplus \typctxtwo}
			{\tm \tmtwo} \mtype
		}{
			\Deri[(\msteps, \esteps)] \typctx \tm {\single{\ty{\mtypetwo}{\mtype}}}
			\quad
			\Deri[(\msteps', \esteps')] {\typctxtwo} \tmtwo {\mtypetwo}
		}
				\\\\
		\infer[\many_0]{
			\Deri[(0, 0)] {} {\tm} {\zero}
		}{			
		}
		&
		\infer[\many_{>0}]{
			\Deri[(\sum_{i \in J\!} \msteps_i, \sum_{i \in J\!} \esteps_i )] {
				\bigmplus_{i \in J}\typctxtwo_i  } {\la\var\tm} {\MSigma {\type_i} {i \in J}
			}
		}{
			(\Deri[(\msteps_i, \esteps_i)] {\typctxtwo_i} {\la\var\tm} {\type_i})_{i \in J}
			\quad
			J \neq \emptyset
		}
		\\\\
		\infer[\fun]{
			\Deri[(\msteps, \esteps)]{\typctx} {\la\var\tm} {\ty{\mtypetwo}{\mtype}}
		}{\Deri[(\msteps, \esteps)] {\var \col \mtypetwo ; \typctx} \tm \mtype} 
&	
		\infer[\esrule]{
			\Deri[(\msteps + \msteps', \esteps + \esteps')] {\typctx \mplus \typctxtwo}
			{\tm \esub\var\tmtwo} \mtype
		}{
			\Deri[(\msteps, \esteps)] {\var \col \mtypetwo ; \typctx} \tm {\mtype}
			\quad
			\Deri[(\msteps', \esteps')] {\typctxtwo} \tmtwo {\mtypetwo}
		} 
		\\\\
		\infer
	[\normal]
	{\tyjp{(0,0)} { \la\var\tm } {} {\normal}}
	{}	

		\end{array}$
		
	}
	\caption{Na\"ive type system for \cbneed evaluation.}
	\label{fig:type-system-cbneed-natural}
\end{figure}

\paragraph{Issue with the Na\"ive System.} 
Unfortunately, the na\"ive system does not work: tight derivations---defined as expected: empty type context and the term typed with $\mult{\normal}$---do not provide exact bounds. 
The problem is that the na\"ive blend of ingredients allows derivations of $\zero$ with strictly positive indices $\msteps$ and $\esteps$. 
Instead, derivations of $\zero$ should always have 0 in both indices---as is the case when they are derived with a $\many_0$ rule with 0 premises---because they correspond to terms to be erased, that are not evaluated in \cbneed. For any term $\tm$, indeed, one can for instance derive the following derivation $\tderiv$:
	$$
	\infer
		[\app]
		{\tyjp{(1,0)}{(\la{\var}{\var}) \tm}{}{\zero}}
		{\infer
			[\many_{>0}]
			{\tyjp{(0,0)}{\la{\var}{\var}}{}{\mult{\ty{\zero}{\zero}}}}
			{\infer
				[\fun]
				{\tyjp{(0,0)}{\la{\var}{\var}}{}{\ty{\zero}{\zero}}}
				{\infer
					[\many_0]
					{\tyjp{(0,0)}{\var}{}{\zero}}
					{}
				}
			}
		\quad
		\infer
			[\many_0]
			{\tyjp{(0,0)}{\tm}{}{\zero}}
			{}}
	$$
Note that introducing $\tyjp{(0,1)}{\var}{}{\zero}$ with rule $\ax$ rather than via $\many_0$ (the typing context $\var \hastype \zero$ is equivalent to the empty type context) would give a derivation with final judgement $\tyjp{(1,1)}{(\la{\var}{\var}) \tm}{}{\zero}$---thus, the system messes up both indices.

Such bad derivations of $\zero$ are not a problem \emph{per se}, because in \cbneed one expects correctness and completeness to hold only for derivations of non-empty multi types. 
However, they do mess up also derivations of non-empty multi types because they can still appear \emph{inside} tight derivations, as sub-derivations of sub-terms to be erased; consider for instance:
\begin{prooftree*}
	\infer0[$\normal$]{\tyjp{(0,0)}{I}{}{\normal}}
	\infer1[$\many_{>0}$]{\tyjp{(0,0)}{I}{}{\mult{\normal}}}
	\infer1[$\fun$]{\tyjp{(0,0)}{\la{\vartwo}I}{}{\ty{\zero}\mult{\normal}}}
	\infer1[$\many_{>0}$]{\tyjp{(0,0)}{\la{\vartwo}I}{}{\mult{\ty{\zero}\mult{\normal}}}}
	\hypo{}
	\ellipsis{$\tderiv$}{\tyjp{(1,0)}{(\la{\var}{\var}) \tm}{}{\zero}}
	\infer2[$\app$]{\tyjp{(2,0)}{(\la{\vartwo}I)((\la{\var}{\var}) \tm)}{}{\mult{\normal}}}
\end{prooftree*}
The term normalises in just 1 $\mcbneed$-step to $I \esub{\vartwo}{(\la{\var}{\var}) \tm}$ but the multiplicative index of the derivation is 2. The mismatch is 
due to a bad derivation of $\zero$ used as right premise of an $\app$ rule. Similarly, the induced typing of $I \esub{\vartwo}{(\la{\var}{\var}) \tm}$ 
is an example of a bad derivation used as right premise of a rule $\esrule$:
	\begin{prooftree*}
		\infer0[$\normal$]{\tyjp{(0,0)}{I}{}{\normal}}
		\infer1[$\many_{>0}$]{\tyjp{(0,0)}{I}{}{\mult{\normal}}}
		\hypo{}
		\ellipsis{$\tderiv$}{\tyjp{(1,0)}{(\la{\var}{\var}) \tm}{}{\zero}}
		\infer2[$\esrule$]{\tyjp{(1,0)}{I \esub{\vartwo}{(\la{\var}{\var}) \tm}}{}{\mult{\normal}}}
	\end{prooftree*}

\paragraph{The Actual Type System.} Our solution to such an issue is to modify the system as to avoid as much as possible derivations of $\zero$. 
The idea is that deriving $\zero$ is only needed for the right premise of rules $\app$ and $\esrule$, when $\mtypetwo = \zero$, and so we add two dedicated rules $\appgc$ and $\esgc$, and instead remove rule $\many_0$ and forbid axioms to introduce $\zero$---the system is in \reffig{need-type-system} and it is based on the same grammar of types of the na\"ive system. Note that rules $\app$ and $\esrule$ now also require $\mtypetwo$ to be different from $\zero$, to avoid overlaps with $\appgc$ and $\esgc$. 

Note that the indices $\msteps$ and $\esteps$ are incremented and summed exactly as in the \cbn and \cbv type systems.

%
%
%
%
%
%

\begin{figure}[t]
\centering

\ovalbox{
$\begin{array}{cc}

\infer
	[\!\ax]
	{\tyjp{(0,1)} {\var} {\var \hastype \mtype} {\mtype}}
	{}

&

\infer
	[\!\normal]
	{\tyjp{(0,0)} { \la\var\tm } {} {\normal}}
	{}	

\\\\

\infer
	[\!\fun]
	{\tyjp{(\msteps, \esteps)}{\la{\var}{\tm}}{\typctx}{\ty{\mtype}{\mtypetwo}}}
	{\tyjp{(\msteps,\esteps)}{\tm}{\typctx, \var \hastype \mtype}{\mtypetwo}}
&
\infer
	[\!\many]
	{\tyjp{(\sum_{i \in J\!} \msteps_{i}, \sum_{i \in J\!} \esteps_{i})}{\la{\var}{\tm}}{\bigmplus_{i \in J} \typctx_{i}}{\mult{\type_{i}}_{i \in J}}}
	{(\tyjp{(\msteps_{i}, \esteps_{i})}{\la{\var}{\tm}}{\typctx_{i}}{\type_{i}})_{i \in J}
	\quad
	J \neq \emptyset}

\\\\

\infer
	[\!\appgc]
	{\tyjp{(\msteps + 1,\esteps)}{\tm \tmtwo}{\typctx} {\mtype}}
	{\tyjp{(\msteps,\esteps)} {\tm}{\typctx} {\mult {\ty{\emptytype}{\mtype}}}}
&
\infer
	[\!\app]
	{\tyjp{(\msteps + \msteps' + 1, \esteps + \esteps')}{\tm \tmtwo} {\typctx \mplus \typctxtwo} {\mtype}}
	{\tyjp{(\msteps,\esteps)} {\tm}{\typctx} {\mult {\ty{\mtypetwo}{\mtype}}} 
	\quad
	\tyjp{(\msteps',\esteps')} {\tmtwo} {\typctxtwo} {\mtypetwo}
	\quad
	\mtypetwo \neq \emptytype}

\\\\

\infer
	[\!\esgc]
	{\tyjp{(\msteps,\esteps)}{\tm \esub\var\tmtwo}{\typctx} {\mtype}}
	{\tyjp{(\msteps,\esteps)} {\tm} {\typctx} {\mtype}
	\quad
	\typctx(\var) = \emptytype}
&
\infer
	[\!\esrule]
	{\tyjp{(\msteps + \msteps', \esteps + \esteps')}{\tm \esub{\var}{\tmtwo}}{\typctx \mplus \typctxtwo}{\mtype}}
	{\tyjp{(\msteps,\esteps)}{\tm}{\typctx, \var \hastype \mtypetwo}{\mtype}
	\quad
	\tyjp{(\mstepstwo,\estepstwo)}{\tmtwo}{\typctxtwo}{\mtypetwo}
	\quad
	\mtypetwo \neq \emptytype}


\end{array}$

}

\caption{Type system for \cbneed evaluation.}
\label{fig:need-type-system}
\end{figure}

\begin{definition}[Tight derivations for \cbneed]
\label{def:need-tight-deriv}
A derivation $\tderiv \exder[\cbneedup] \allowbreak \tyjp{(\msteps, \esteps)}{\tm}{\typctx}{\mtype}$ is \emph{tight} if $\mtype = [\normal]$ and $\typctx$ is empty.
\end{definition}


\begin{example} [The needed one]
\label{ex:need-type-deriv}
We return 
to $\tm \defeq ((\la{\var}{\la{\vartwo}{\var \var}})(II))(II)$ used in \refex{evaluation} and we give it a tight derivation in the $\cbneed$ type system.

Again, we shorten $\normal$ to $\shortnormal$. Then, we define $\tderivtwo$ as follows
$$
\scalebox{0.95}{
\infer
			[\many]
			{\tyjp{(1,2)}{\la{\var}{\la{\vartwo}{\var \var}}}{}{\mult{\ty{\mult{\shortnormal, \ty{\mult{\shortnormal}}{\mult{\shortnormal}}}}{\mult{\ty{\zero}{\mult{\shortnormal}}}}}}}
			{\infer
				[\fun]
				{\tyjp{(1,2)}{\la{\var}{\la{\vartwo}{\var \var}}}{}{\ty{\mult{\shortnormal, \ty{\mult{\shortnormal}}{\mult{\shortnormal}}}}{\mult{\ty{\zero}{\mult{\shortnormal}}}}}}
				{\infer
					[\many]
					{\tyjp{(1,2)}{\la{\vartwo}{\var \var}}{\var : \mult{\shortnormal, \ty{\mult{\shortnormal}}{\mult{\shortnormal}}}}{\mult{\ty{\zero}{\mult{\shortnormal}}}}}
					{\infer
						[\fun]
						{\tyjp{(1,2)}{\la{\vartwo}{\var \var}}{\var : \mult{\shortnormal, \ty{\mult{\shortnormal}}{\mult{\shortnormal}}}}{\ty{\zero}{\mult{\shortnormal}}}}
						{\infer
							[\app]
							{\tyjp{(1,2)}{\var \var}{\var : \mult{\shortnormal, \ty{\mult{\shortnormal}}{\mult{\shortnormal}}}}{\mult{\shortnormal}}}
							{\infer
								[\ax]
								{\tyjp{(0,1)}{\var}{\var : \mult{\ty{\mult{\shortnormal}}{\mult{\shortnormal}}}}{\mult{\ty{\mult{\shortnormal}}{\mult{\shortnormal}}}}}
								{}
							\quad
							\infer
								[\ax]
								{\tyjp{(0,1)}{\var}{\var : \mult{\shortnormal}}{\mult{\shortnormal}}}
								{}}
						}
					}
				}
			}
}
$$
and, shortening $\ty{\mult{\shortnormal}}{\mult{\shortnormal}}$ to $\shortnormal^{\shortnormal}$, we define $\tderivthree$ as follows
$$
\scalebox{0.93}{
\infer
	[\!app]
	{\tyjp{(1,2)}{II}{}{\mult{\shortnormal, \shortnormal^{\shortnormal}}}}
	{\infer
		[\!\many]
		{\tyjp{(0,1)}{\la{\varthree}{\varthree}}{}{\mult{\ty{\mult{\shortnormal, \shortnormal^{\shortnormal}}}{\mult{\shortnormal, \shortnormal^{\shortnormal}}}}}}
		{\infer
			[\!\fun]
			{\tyjp{(0,1)}{\la{\varthree}{\varthree}}{}{\ty{\mult{\shortnormal, \shortnormal^{\shortnormal}}}{\mult{\shortnormal, \shortnormal^{\shortnormal}}}}}
			{\infer
				[\!\ax]
				{\tyjp{(0,1)}{\varthree}{\varthree : \mult{\shortnormal, \shortnormal^{\shortnormal}}}{\mult{\shortnormal, \shortnormal^{\shortnormal}}}}
				{}
			}
		}
	\quad
	\infer
		[\!\many]
		{\tyjp{(0,1)}{\la{\varfour}{\varfour}}{}{\mult{\shortnormal, \shortnormal^{\shortnormal}}}}
		{\infer
			[\!\normal]
			{\tyjp{(0,0)}{\la{\varfour}{\varfour}}{}{\shortnormal}}
			{}
		\quad
		\infer
			[\!\fun]
			{\tyjp{(0,1)}{\la{\varfour}{\varfour}}{}{\shortnormal^{\shortnormal}}}
			{\infer
				[\!\ax]
				{\tyjp{(0,1)}{\varfour}{\varfour : \mult{\shortnormal}}{\mult{\shortnormal}}}
				{}
			}
		}
	}
}
$$
Finally, we put $\tderivtwo$ and $\tderivthree$ together 
in the following derivation $\tderiv$ for $\tm$

\begin{prooftree*}
\scalebox{0.95}{
	\hypo{}
	\ellipsis{$\tderivtwo$}{\tyjp{(1,2)}{\la{\var}{\la{\vartwo}{\var \var}}}{}{\mult{\ty{\mult{\shortnormal, \ty{\mult{\shortnormal}}{\mult{\shortnormal}}}}{\mult{\ty{\zero}{\mult{\shortnormal}}}}}}}
	\hypo{}
	\ellipsis{$\tderivthree$}{\tyjp{(1,2)}{II}{}{\mult{\shortnormal, \shortnormal^{\shortnormal}}}}
	\infer2[$\app$]{\tyjp{(3,4)}{(\la{\var}{\la{\vartwo}{\var \var}})(II)}{}{\mult{\ty{\zero}{\mult{\shortnormal}}}}}
	\infer1[$\appgc$]{\tyjp{(4,4)}{((\la{\var}{\la{\vartwo}{\var \var}})(II))(II)}{}{\mult{\shortnormal}}}
}
\end{prooftree*}

\noindent Note that the indices $(4,4)$ correspond exactly to the number of 
$\mcbneed$-steps and $\ecbneed$-steps respectively, from $\tm$ to its $\cbneedsym$-normal form---as shown in \refex{evaluation}---, and that $\tderiv$ is a \emph{tight} derivation. 
Forthcoming \refthm{need-correctness} shows once again that this is not by chance: $\cbneed$ tight derivations are minimal and provides exact bounds to evaluation lengths in $\tocbneed$.
\end{example}

Remarkably, the technical development to prove \emph{correctness} and \emph{completeness} of the \cbneed type system with respect to \cbneed evaluation follows smoothly along the same lines of the two other systems, \emph{mutatis mutandis}.



\subsection{\cbneed Correctness} 

\begin{toappendix}
\begin{lemma}[\cbneed linear substitution]
\label{l:need-linear-substitution}
Let $\tderiv_{\cbneedctx \cwc{\var}} \exder[\cbneedup] \tyjp{(\msteps, \esteps)}{\cbneedctx\cwc{\var}}{\var \hastype \mtype;\typctx}{\mtypethree}$ be a derivation and $\val$ a value such that $\mtypethree \neq \emptytype $ and $\cbneedctx$ does not capture the free variables of $\val$. Then there exists a splitting $ \mtype = \mtype_{1} \mplus \mtype_{2}$, with $\mtype_{1} \neq \emptytype$, such that for every derivation $\tderivtwo \exder[\cbneedup] \tyjp{(\mstepstwo, \estepstwo)}{\val}{\typctxtwo}{\mtype_{1}}$ there exists a derivation $\tderiv_{\cbneedctx\cwc{\val}} \exder[\cbneedup] \tyjp{(\msteps + \mstepstwo, \esteps + \estepstwo - 1)}{\cbneedctx\cwc{\val}}{\var \hastype \mtype_{2}; \typctx \mplus \typctxtwo}{\mtypethree}$.
\end{lemma}
\end{toappendix}

\begin{toappendix}
\begin{proposition}[Quantitative subject reduction for $\cbneed$]
\label{prop:need-subject-reduction}
Let $\tderiv \exder[\cbneedup] \tyjp{(\msteps, \esteps)}{\tm}{\typctx}{\mtype}$ be a derivation such that $\mtype \neq \emptytype$.
\begin{itemize}
\item \emph{Multiplicative}: if $\tm \tomcbneed \tmtwo$ then $\msteps \geq 1$ and there is a derivation $\tderiv' \exder[\cbneedup] \tyjp{(\msteps - 1, \esteps)}{\tm}{\typctx}{\mtype}$.
\item \emph{Exponential}: if $\tm \toecbneed \tmtwo$ then $\esteps \geq 1$ and there exists a derivation $\tderiv' \exder[\cbneedup] \tyjp{(\msteps, \esteps - 1)}{\tm}{\typctx}{\mtype}$.
\end{itemize}
\end{proposition}
\end{toappendix}

\begin{toappendix}
\begin{proposition}[$\mult\normal$ typings for normal forms for \cbneed]
\label{prop:need-normal-forms-forall}
Let $\tderiv \exder[\cbneedup] \tyjp{(\msteps, \esteps)}{\tm}{\typctx}{\mult\normal}$ be a derivation, with $\normalpr{\tm}$. 	Then $\typctx$ is empty, and so $\tderiv$ is tight, and $\msteps = \esteps = 0$.
\end{proposition}
\end{toappendix}

\begin{toappendix}
\begin{theorem} [\cbneed tight correctness]
\label{thm:need-correctness}
Let $\tm$ be a closed term. 
If $\tderiv \exder[\cbneedup] \tyjp{(\msteps, \esteps)}{\tm}{}{\mtype}$ then there is $\tmtwo$ such that $\deriv \colon \tm \tocbneedn \tmtwo$, $\normalpr{\tmtwo}$, $\sizem{\deriv} \leq \msteps$, $\sizee{\deriv} \leq \esteps$. 
Moreover, if $\tderiv$ is tight then $\sizem{\deriv} = \msteps$ and $\sizee{\deriv} = \esteps$.
\end{theorem}
\end{toappendix}

\subsection{\cbneed Completeness} 

\begin{toappendix}
\begin{proposition}[Normal forms are tightly typable for $\cbneed$]
\label{prop:need-normal-forms-exist}
Let $\tm$ be such that $\normalpr{\tm}$. Then there is a tight derivation $\tderiv \exder[\cbneedup] \tyjp{(0,0)}{\tm}{}{\mult{\normal}}$.
\end{proposition}
\end{toappendix}

\begin{toappendix}
\begin{lemma}[Linear removal for $\cbneed$]
\label{l:need-linear-anti-substitution}
Let  $\tderiv \exder[\cbneedup] \tyjp{(\msteps, \esteps)}{\cbneedctx\cwc{\val}}{\typctx}{\mtypethree}$ be a derivation, with $\mtypethree \neq \emptytype$ and $\var \notin \fv\val$. Then  there exist 
\begin{itemize}
\item a multi type $\mtype$,
\item a derivation $\tderiv_{\val}  \exder[\cbneedup] \tyjp{(\msteps_{\val}, \esteps_{\val})}{\val}{\typctx_{\val}}{\mtype}$, and
\item a derivation $\tderiv_{\cbneedctx \cwc{\var}} \exder[\cbneedup] 
 \tyjp{(\msteps', \esteps')}{\cbneedctx\cwc{\var}}{\typctx' \mplus \{\var \colon \mtype\} }{\mtypethree}$
\end{itemize}
such that 
\begin{itemize}
	\item \emph{Type contexts}: $\typctx = \typctx' \mplus \typctx_{\val}$.

	\item \emph{Indices}: $(\msteps, \esteps) = (\msteps' +  \msteps_{\val}, \esteps' + \esteps_{\val} - 1)$.
\end{itemize}
\end{lemma}
\end{toappendix}

\begin{toappendix}
\begin{proposition}[Quantitative subject expansion for $\cbneed$]
\label{prop:need-subject-expansion}
Let $\tderiv \exder[\cbneedup] \tyjp{(\msteps,\esteps)}{\tmtwo}{\typctx}{\mtype}$ be a derivation such that $\mtype \neq \emptytype$. Then,

\begin{itemize}
\item \emph{Multiplicative:} if $\tm \tomcbneed \tmtwo$ then there is a  derivation $\tderiv' \exder[\cbneedup] \tyjp{(\msteps + 1, \esteps)}{\tm}{\typctx}{\mtype}$,
\item \emph{Exponential}: if $\tm \toecbneed \tmtwo$ then there is a  derivation $\tderiv' \exder[\cbneedup] \tyjp{(\msteps, \esteps + 1)}{\tm}{\typctx}{\mtype} $.
\end{itemize}

\end{proposition}
\end{toappendix}

\begin{toappendix}
\begin{theorem} [\cbneed tight completeness]
\label{thm:need-completeness}
Let $\tm$ be a closed term.
If $\deriv \colon \tm \tocbneedn \tmtwo$ and $\normalpr{\tmtwo}$
then there exists a tight  derivation $\tderiv \exder[\cbneedup] \tyjp{(\sizem{\deriv},\sizee{\deriv})}{\tm}{}{\mult{\normal}}$.
\end{theorem}
\end{toappendix}

\paragraph{\cbneed Model.} The interpretation $\sem{\tm}_{\vec{\var}}^\cbneed$ with respect to the \cbneed system is defined as the set (where $\vec{\var} = (\var_1, \dots, \var_n)$ is a list of variables suitable for $\tm$):
$$\begin{array}{llll}
    \{((\mtype_1,\dots, \mtype_n),\mtypetwo) \mid \exists 
    \, 
    \tderiv \exder[\cbneedup] \Deri[(\msteps, \esteps)]{\var_1 \colon\! \mtype_1, \dots, \var_n \colon\! \mtype_n}\tm\mtypetwo \mbox{ and $\mtypetwo \neq \zero$}
    \} \,.
\end{array}$$
Note that the right multi type is required to be non-empty. 
The \emph{invariance} 
and the \emph{adequacy} of $\sem{\tm}_{\vec{\var}}^\cbneed$ with respect to \cbneed evaluation are obtained exactly as for the \cbn and \cbv cases. 
 
\section{A New Fundamental Theorem for Call-by-Need}
\label{sect:new-theorem-need}

\paragraph{\cbneed Erases Wisely.} In the literature, \emph{the} theorem about \cbneed is the fact that it is operationally equivalent to \cbn. This result was first proven independently by two groups, Maraist, Odersky, and Wadler \cite{DBLP:journals/jfp/MaraistOW98}, and Ariola and Felleisen \cite{DBLP:journals/jfp/AriolaF97}, in the nineties, using heavy rewriting techniques. 

Recently, Kesner gave a much simpler proof via \cbn multi types \cite{DBLP:conf/fossacs/Kesner16}. 
She uses multi types to first show termination equivalence of \cbn and \cbneed, from which she then infers operational equivalence. 
Termination equivalence 
means that a given term terminates in \cbn if and only if terminates in \cbneed, and it is a consequence of our slogan that \emph{\cbn and \cbneed both erase wisely}.

With our terminology and notations, Kesner's result takes the following form.

\begin{theorem}[Kesner \cite{DBLP:conf/fossacs/Kesner16}]
Let $\tm$ be a closed term.
  \begin{enumerate}
    \item \emph{Correctness}: if $\tderiv \exder[\cbnup]
    \Deri[(\msteps, \esteps)] {}{\tm}{\type}$ then there exists $\tmtwo$ such that $\deriv \colon \tm
  \tocbneedn \tmtwo$, $\normalpr\tmtwo$, $\sizem\deriv \leq \msteps$ and $\sizee\deriv \leq \esteps$.  
  
    \item \emph{Completeness}: if                                                                                                                                                                                                                                                                                                                                                                                                                                                                         $\deriv \colon\! \tm \!\tocbneedn \tmtwo$ and $\normalpr\tmtwo$ then there is 
    $\tderiv \exder[\cbnup\!] \Deri[(\msteps, \esteps)] {}{\tm\!} {\!\normal}$.
  \end{enumerate}
\end{theorem}

Note that, with respect to the other similar theorems in this paper, the result does not cover tight derivations and it does not provide exact bounds. In fact, the \cbn system \emph{cannot} provide exact bounds for \cbneed, because it does provide them for \cbn evaluation, that in general is slower than \cbneed. 
Consider for instance the term $\tm$ in \refex{evaluation} and its \cbn tight derivation in \refex{cbn-type-deriv}: the derivation provides indices $(5,5)$ for $\tm$ (and so $\tm$ evaluates in 10 \cbn steps), but $\tm$ evaluates in 8 \cbneed steps. 
Closing such a gap is the main motivation behind this paper, achieved by the \cbneed multi type system in \refsect{cbneed}.

\paragraph{\cbneed Duplicates Wisely.} Curiously, in the literature there are no dual results showing that \cbneed duplicates as wisely as \cbv. One of the reasons is that it is a theorem that does not admit a simple formulation such as operational or termination equivalence, because \cbneed and \cbv are not in such relationships. Morally, this is subsumed by the logical interpretation according to which \cbneed corresponds to an affine variant of the linear logic representation of \cbv. Yet, it would be nice to have a precise, formal statement establishing that \emph{\cbneed duplicates as wisely as \cbv}---we provide it here.

Our result is that the \cbv multi type system is correct with respect to \cbneed evaluation. In particular, the indices $(\msteps, \esteps)$ provided by a \cbv type derivation provide bounds for \cbneed evaluation lengths. Two important remarks before we proceed with the formal statement:
\begin{itemize}
  \item \emph{Bounds are not exact}: the indices of a \cbv derivation do not generally provide exacts bounds for \cbneed, not even in the case of tight derivations. The reason is that \cbneed does not evaluate unneeded subterms (\ie those typed with $\zero$), while \cbv does. Consider again the term $\tm$ of \refex{evaluation}, for instance, whose \cbv tight derivation has indices $(5,5)$ (and so $\tm$ evaluates in 10 \cbv steps) but it \cbneed evaluates in 8 steps.
  
  \item \emph{Completeness cannot hold}: we prove correctness but not completeness simply because the \cbv system is not complete with respect to \cbneed evaluation. Consider for instance $(\la\var I) \Omega$: it is \cbv untypable by \refthm{value-completeness}, because it is \cbv divergent, and yet it is \cbneed normalisable.
\end{itemize}

\paragraph{\cbv Correctness with Respect to \cbneed.} Pleasantly, our presentations of \cbv and \cbneed make the proof of the result straightforward. It is enough to observe that, since we do not consider garbage collection and we adopt a non-deterministic formulation of \cbv, \cbneed is a subsystem of \cbv. Formally, if $\tm \tocbneed \tmtwo$ then $\tm \tocbv \tmtwo$, as it is easily seen from the definitions (\cbneed reduces only \emph{some} subterms of applications and ES, while \cbv reduces \emph{all} such subterms). The result is then a corollary of the correctness theorem for \cbv.

\begin{toappendix}
\begin{corollary}[\cbv correctness wrt \cbneed]
	\label{coro:value-need}
	Let $\tm$ be a closed term and $\tderiv \exder[\cbvsym] \Deri[(\msteps,\esteps)]{}{\tm}{\mtype}$ be a  derivation. Then there exists $\tmtwo$ such that $\deriv \colon \tm \tocbneedn \tmtwo$ and $\normalpr{\tmtwo}$, with $\sizem{\deriv} \leq \msteps$ and $\sizee{\deriv} \leq \esteps$.
\end{corollary}  
\end{toappendix}

Since the \cbneed system provides exact bounds (\refthm{need-correctness}), we obtain that \cbneed duplicates as wisely as \cbv, when the comparison makes sense, that is, on \cbv normalisable terms.

\begin{toappendix}
\begin{corollary}[\cbneed duplicates as wisely as \cbv]
	\label{cor:value-longer-than-need}
	Let $\deriv\colon \tm \tocbvn \tmthree$ with $\normalcbvpr{\tmthree}$. 
	Then there is $\derivtwo \colon \tm \tocbneedn \tmtwo$ with $\normalpr{\tmtwo}$ and $\sizem{\derivtwo} \leq \sizem{\deriv}$ and $\sizee{\derivtwo} \leq \sizee{\deriv}$.
\end{corollary}
\end{toappendix}

\section{Conclusions}
\paragraph{Contributions.} This paper introduces a multi type system for \cbneed evaluation, carefully blending ingredients from multi type systems for \cbn and \cbv evaluation in the literature. Notably, it is the first type system whose minimal derivations---explicitly characterised---provide exact bounds for evaluation lengths. 
It also characterises \cbneed termination, and thus its judgements provide an adequate relational semantics, which is the first one precisely reflecting \cbneed evaluation.

The technical development is simple, and uniform with respect to those of \cbn and \cbv multi type systems. The typing rules count evaluation steps following \emph{exactly} the same schema of the \cbn and \cbv rules. The proofs of correctness and completeness also follow \emph{exactly} the same structure.

A further side contribution of the paper is a new fundamental result of \cbneed, formally stating that it duplicates as wisely as \cbv. More precisely, the \cbv multi type system is (quantitatively) correct with respect to \cbneed evaluation. Pleasantly, our presentations of \cbv and \cbneed provide the result for free. This result dualises the other fundamental theorem stating that \cbneed erases as wisely as \cbn, usually formulated as termination equivalence, and recently re-proved by Kesner  using \cbn multi types \cite{DBLP:conf/fossacs/Kesner16}. 

\paragraph{Future Work.} Recently, Barenbaum et al. extended \cbneed to strong evaluation \cite{DBLP:journals/pacmpl/BalabonskiBBK17}, and it is natural to try to extend  our type system as well. The definition of the system, in particular the extension of \emph{tight} derivations to that setting, seems however far from being evident. Barembaum, Bonelli, and Mohamed also apply \cbn multi types to a \cbneed calculus extended with pattern matching and fixpoints \cite{DBLP:conf/ppdp/BarenbaumBM18}, that might be interesting to refine along the lines of our work.

An orthogonal direction is the study of the denotational models of \cbneed. It would be interesting to have a categorical semantics of \cbneed, as well as a categorical way of discriminating our quantitative precise model from the quantitatively lax one given by \cbn multi types. It would also be interesting to obtain game semantics of \cbneed, hopefully satisfying a strong correspondence with our multi types in the style of what happens in \cbn \cite{DBLP:journals/tcs/GianantonioHL08,DBLP:conf/csl/GianantonioL13,DBLP:conf/lics/TsukadaO16,DBLP:conf/lics/Ong17}.

A further, unconventional direction is to dualise the inception of the \cbneed type system trying to mix silly duplication from \cbn and silly erasure from \cbv, obtaining---presumably---a multi types system measuring a perpetual strategy.

\paragraph{Acknowledgements.} This work has been partially funded by the ANR JCJC grant COCA HOLA (ANR-16-CE40-004-01) and by the EPSRC grant EP/R029121/1 ``Typed Lambda-Calculi with Sharing and Unsharing''.

%

 \bibliographystyle{splncs04}
 \bibliography{\macrospath/biblio}

 \newpage
 \appendix

\chapter*{Proof Appendix}


\section{Closed $\lambda$-Calculi (\refsect{closed})}

\begin{remark}
	\label{rem:normal}
	Let $\tm$ be a term. 
	According to definition of the predicate $\normal$, $\normalpr{\tm}$ if and only if $\tm \defeq \sctxp{\la{\var}\tmtwo}$ for some term $\tmtwo$ and substitution context $\sctx$.
\end{remark}

\begin{remark}
	\label{rem:normal-cbv}
	Let $\tm$ be a term. 
	According to definition of the predicate $\normalcbv$, $\normalcbvpr{\tm}$ if and only if $\tm \defeq (\la{\var}\tmtwo) \esub{\vartwo_1}{\tmtwo_1} \dots \esub{\vartwo_n}{\tmtwo_n}$ for some $n \in \nat$ and terms $\tmtwo, \tmtwo_1, \dots, \tmtwo_n$ such that $\normalcbvpr{\tmtwo_i}$ for all $1 \leq i \leq n$.
	In particular, if $\normalcbvpr{\tm}$ then $\normalpr{\tm}$.
\end{remark}

\begin{remark}
	\label{rem:writing}
	Every term can be written in a unique way as $\cbnctxp{\tm}$ for some \cbn context $\cbnctx$, where $\tm$ is either a variable or an abstraction.
\end{remark}

\begin{remark}
	\label{rem:cbn-cbneed-context}
	Every \cbn context is a \cbneed context and a \cbv context.
\end{remark}

\gettoappendix{prop:syntactic-characterization-closed-normal}

\begin{proof}\hfill
	\begin{enumerate}
	\item First, we prove that a closed term $\tm$ is \cbn normal if and only if $\normalpr{\tm}$. 
 	\begin{description}
		\item[$\Rightarrow$:] Let $\tm$ be a closed and \cbn normal term. 
		We prove by induction on $\tm$ 
		that $\normalpr{\tm}$. 
		Cases:
		\begin{itemize}
			\item \emph{Variable}, \ie $\tm \defeq \var$: it is impossible because $\tm$ is closed by hypothesis.
			
			\item \emph{Abstraction}, \ie $\tm \defeq \la{\vartwo}\tmtwo$: then, $\normalpr{\tm}$ according to the definition of the predicate $\normal$.
			
			\item \emph{Application}, \ie $\tm \defeq \tmtwo\tmthree$: this case is impossible, we prove it by contradiction. 
			Suppose $\tm \defeq \tmtwo\tmthree$, then $\tmtwo$ would be closed and \cbn normal (as $\tm$ is so), and hence $\normalpr{\tmtwo}$ by \ih
			According to \refrem{normal}, $\tmtwo = \sctxp{\la{\vartwo}\tmfour}$ and so $\tm = \sctxp{\la{\vartwo}\tmfour}\tmthree$, which is impossible because $\tm$ would be a $\msym$-redex.
			
			\item \emph{Explicit substitution}, \ie $\tm \defeq \tmtwo \esub{\var}{\tmthree}$: then, $\tmtwo$ is \cbn normal.
			There are four subcases:
			\begin{itemize}
				\item $\tmtwo$ is closed, then $\normalpr{\tmtwo}$ by \ih, and hence $\normalpr{\tm}$;
				\item $\tmtwo \defeq \sctxp{\la{\vartwo}\tmfour}$, then $\normalpr{\tmtwo}$ by \refrem{normal}, and hence $\normalpr{\tm}$;
				\item $\tmtwo \defeq \cbnctxp{\vartwo}$: this is impossible because otherwise (since $\tm$ is closed) $\tm = \cbnctx'\hole{\cbnctxtwo\cwc{\vartwo}\esub\vartwo{\tmfour}}$ which is not $\ecbn$-normal; 
				\item $\tmtwo$ is none of the above, then by \refrem{writing}, $\tmtwo \defeq \cbnctxp{\la{\vartwo}\tmfour}$ for some \cbn context $\cbnctx$ that is not a substitution context: this case is impossible because otherwise, by \refrem{writing}, $\tmtwo \defeq \cbnctxtwop{\sctxp{\la{\vartwo}\tmfour}\tmfive}$  for some \cbn context $\cbnctxtwo$, which is not $\mcbn$-normal.
			\end{itemize}
		\end{itemize}
		
		\item[$\Leftarrow$:] We prove by induction on the definition of $\normalpr{\tm}$ a stronger statement: for any term $\tm$, if $\normalpr{\tm}$ then $\tm$ is \cbn normal (we dropped the hypothesis that $\tm$ is closed).
		There are only two cases:
		\begin{itemize}
			\item \emph{Abstraction}, \ie $\tm \defeq \la{\var}\tmtwo$, then $\tm$ is \cbn normal because $\tocbn$ does not reduce under abstractions.
			
			\item \emph{Explicit substitution}, \ie $\tm \defeq \tmtwo \esub{\var}{\tmthree}$ with $\normalpr{\tmtwo}$, then $\tmtwo = \sctxp{\la{\vartwo}\tmfour}$ for some substitution context $\sctx$ according to \refrem{normal}, thus $\tm = \sctxtwop{\la{\vartwo}\tmfour}$ for the substitution context $\sctxtwo = \sctx \esub\var\tmthree$.
			By \ih, $\tmtwo$ is \cbn normal.
			Hence, $\tm$ is $\mcbn$-normal.
			Moreover, $\tm \neq \cbnctxp{\var}$ for any \cbn context $\cbnctx$, so $\tm$ is also $\ecbn$-normal and hence \cbn normal.
		\end{itemize}
	\end{description}
	
	\item We prove that a closed term $\tm$ is \cbneed normal if and only if $\normalpr{\tm}$. 
	\begin{description}
		\item[$\Rightarrow$:] Let $\tm$ be a closed and \cbneed normal term. 
		We prove by induction on $\tm$ 
		that $\normalpr{\tm}$. 
		Cases:
		\begin{itemize}
			\item \emph{Variable}, \ie $\tm \defeq \var$: it is impossible because $\tm$ is closed by hypothesis.
			
			\item \emph{Abstraction}, \ie $\tm \defeq \la{\vartwo}\tmtwo$: then, $\normalpr{\tm}$ according to the definition of the predicate $\normal$.
			
			\item \emph{Application}, \ie $\tm \defeq \tmtwo\tmthree$: this case is impossible, we prove it by contradiction. 
			Suppose $\tm \defeq \tmtwo\tmthree$, then $\tmtwo$ would be closed and \cbneed normal (as $\tm$ is so), and hence $\normalpr{\tmtwo}$ by \ih
			According to \refrem{normal}, $\tmtwo = \sctxp{\la{\vartwo}\tmfour}$ and so $\tm = \sctxp{\la{\vartwo}\tmfour}\tmthree$, which is impossible because $\tm$ would be a $\msym$-redex.
			
			\item \emph{Explicit substitution}, \ie $\tm \defeq \tmtwo \esub{\var}{\tmthree}$: then, $\tmtwo$ is \cbneed normal.
			There are four subcases:
			\begin{itemize}
				\item $\tmtwo$ is closed, then $\normalpr{\tmtwo}$ by \ih, and hence $\normalpr{\tm}$;
				\item $\tmtwo \defeq \sctxp{\la{\vartwo}\tmfour}$, then $\normalpr{\tmtwo}$ by \refrem{normal}, and hence $\normalpr{\tm}$;
				\item $\tmtwo \defeq \cbneedctxp{\vartwo}$: this is impossible because otherwise (since $\tm$ is closed) $\tm = \cbneedctx'\hole{\cbneedctxtwo\cwc{\vartwo}\esub\vartwo{\tmfour}}$ for some $\tmfour \defeq \sctxp{\val}$ (by \ih and \refrem
				{normal}, since $\tmfour$ is \cbneed normal), and so $\tm$ would not be $\ecbneed$-normal; 
				\item $\tmtwo$ is none of the above, then by \refrem{writing}, $\tmtwo \defeq \cbnctxp{\la{\vartwo}\tmfour}$ for some \cbn (and hence \cbneed, by \refrem{cbn-cbneed-context}) context $\cbnctx$ that is not a substitution context: this case is impossible because otherwise, according to \refrem{writing}, $\tmtwo \defeq \cbnctxtwop{\sctxp{\la{\vartwo}\tmfour}\tmfive}$ for some \cbn context $\cbnctxtwo$, which is not $\mcbneed$-normal (since $\cbnctxtwo$ is a \cbneed context by \refrem{cbn-cbneed-context}).
			\end{itemize}
		\end{itemize}
		
		\item[$\Leftarrow$:] We prove by induction on the definition of $\normalpr{\tm}$ a stronger statement: for any term $\tm$, if $\normalpr{\tm}$ then $\tm$ is \cbneed normal (we dropped the hypothesis that $\tm$ is closed).
		There are only two cases:
		\begin{itemize}
			\item \emph{Abstraction}, \ie $\tm \defeq \la{\var}\tmtwo$, then $\tm$ is \cbneed normal because $\tocbneed$ does not reduce under abstractions.
			
			\item \emph{Explicit substitution}, \ie $\tm \defeq \tmtwo \esub{\var}{\tmthree}$ with $\normalpr{\tmtwo}$, then $\tmtwo = \sctxp{\la{\vartwo}\tmfour}$ for some substitution context $\sctx$ according to \refrem{normal}, thus $\tm = \sctxtwop{\la{\vartwo}\tmfour}$ for the substitution context $\sctxtwo = \sctx \esub\var\tmthree$.
			By \ih, $\tmtwo$ is \cbneed normal.
			Hence, $\tm$ is $\mcbneed$-normal.
			Moreover, $\tm \neq \cbneedctxp{\var}$ for any \cbneed context $\cbneedctx$, so $\tm$ is also $\ecbneed$-normal and hence \cbneed normal.
		\end{itemize}
	\end{description}

	\item We prove that a closed term $\tm$ is \cbv normal if and only if $\normalcbvpr{\tm}$. 
	\begin{description}
		\item[$\Rightarrow$:] Let $\tm$ be a closed and \cbv normal term. 
		We prove by induction on $\tm$ 
		that $\normalcbvpr{\tm}$. 
		Cases:
		\begin{itemize}
			\item \emph{Variable}, \ie $\tm \defeq \var$: it is impossible because $\tm$ is closed by hypothesis.
			
			\item \emph{Abstraction}, \ie $\tm \defeq \la{\vartwo}\tmtwo$: then, $\normalcbvpr{\tm}$ according to the definition of the predicate $\normalcbv{}$.
			
			\item \emph{Application}, \ie $\tm \defeq \tmtwo\tmthree$: this case is impossible, we prove it by contradiction. 
			Suppose $\tm \defeq \tmtwo\tmthree$, then $\tmtwo$ would be closed and \cbv normal, and hence $\normalcbvpr{\tmtwo}$ by \ih
			According to \refrem{normal-cbv}, $\tmtwo = \sctxp{\la{\vartwo}\tmfour}$ and so $\tm = \sctxp{\la{\vartwo}\tmfour}\tmthree$, which is impossible because $\tm$ would be a $\msym$-redex.
			
			\item \emph{Explicit substitution}, \ie $\tm \defeq \tmtwo \esub{\var}{\tmthree}$: then, $\tmtwo$ and $\tmthree$ are \cbv normal, and $\tmthree$ is closed (as $\tm$ is so), thus $\normalcbvpr{\tmthree}$ by \ih
			There are four subcases:
			\begin{itemize}
				\item $\tmtwo$ is closed, then $\normalcbvpr{\tmtwo}$ by \ih, and hence $\normalcbvpr{\tm}$;
				\item $\tmtwo \defeq (\la{\vartwo}\tmfour)\esub{\vartwo_1}{\tmtwo_1}\dots \esub{\vartwo_n}{\tmtwo_n}$, then all $\tmfour_i$'s are \cbv normal (as $\tmtwo$ is so);
				by \ih, $\normalcbvpr{\tmfour_i}$ for all $1 \leq i \leq n$, thus $\normalcbvpr{\tmtwo}$ (since $\normalcbvpr{\la{\vartwo}\tmfour}$)
				 and hence $\normalcbvpr{\tm}$;
				\item $\tmtwo \defeq \cbvctxp{\vartwo}$: this is impossible because otherwise (since $\tm$ is closed) $\tm = \cbvctxtwo\hole{\cbvctxthree\cwc{\vartwo}\esub\vartwo{\tmfour}}$ for some $\tmfour \defeq \sctxp{\val}$ (by \ih and \refrem{normal-cbv}, since $\tmfour$ is \cbv normal), and so $\tm$ would not be $\ecbv$-normal; 
				\item $\tmtwo$ is none of the above, then by \refrem{writing}, $\tmtwo \defeq \cbnctxp{\la{\vartwo}\tmfour}$ for some \cbn (and hence \cbv, by \refrem{cbn-cbneed-context}) context $\cbnctx$ that is not a substitution context: this case is impossible because otherwise, by \refrem{writing}, $\tmtwo \defeq \cbnctxtwop{\sctxp{\la{\vartwo}\tmfour}\tmfive}$ for some \cbn context $\cbnctxtwo$, which is not $\mcbv$-normal (since $\cbnctxtwo$ is a \cbv context by \refrem{cbn-cbneed-context}).
			\end{itemize}
		\end{itemize}
		
		\item[$\Leftarrow$:] We prove by induction on the definition of $\normalcbvpr{\tm}$ a stronger statement: for any term $\tm$, if $\normalcbvpr{\tm}$ then $\tm$ is \cbv normal (we dropped the hypothesis that $\tm$ is closed).
		There are only two cases:
		\begin{itemize}
			\item \emph{Abstraction}, \ie $\tm \defeq \la{\var}\tmtwo$, then $\tm$ is \cbv normal because $\tocbv$ does not reduce under abstractions.
			
			\item \emph{Explicit substitution}, \ie $\tm \defeq \tmtwo \esub{\var}{\tmthree}$ with $\normalcbvpr{\tmtwo}$ and $\normalcbvpr{\tmthree}$, then $\tm = (\la{\vartwo}\tmfour)\esub{\vartwo_1}{\tmtwo_1}\dots\esub{\vartwo_n}{\tmtwo_n}\esub{\var}{\tmthree}$ where $\normalcbvpr{\tmtwo_i}$ for all $1 \leq i \leq n$, by \refrem{normal-cbv} applied to $\tmtwo$.
			By \ih, $\tmthree, \tmtwo_1, \dots, \tmtwo_n$ and $\tmtwo$ are \cbv normal.
			Hence, $\tm$ is $\mcbv$-normal.
			Moreover, $\tm \neq \cbvctxp{\var}$ for any \cbv context $\cbvctx$, so $\tm$ is also $\ecbv$-normal and hence \cbv normal.
			\qed
		\end{itemize}
	\end{description}
	
	\end{enumerate}
\end{proof}

\section{Types by Name (\refsect{cbn})}

\subsection{\cbn Correctness}

In order to prove subject reduction, we have to first show that typability is preserved by linear substitutions, via a dedicated lemma. We also need the following splitting property of multi-sets, whose proof is omitted because straightforward.

\begin{lemma}[Splitting multi-sets with respect to derivations]
\label{l:name-splitting-multisets}
Let $\tm$ be a term, $\tderiv \exder[\cbn] \tyjp{(\msteps,\esteps)}{\tm}{\typctx}{\mtype}$ a derivation, and $\mtype = \mtypetwo \mplus \mtypethree$  a splitting. Then there exist two derivations 
\begin{itemize}
  \item $\tderiv_{\mtypetwo} \exder[\cbn] \tyjp{(\msteps_{\mtypetwo},\esteps_{\mtypetwo})}{\tm}{\typctx_{\mtypetwo}}{\mtypetwo}$, and
  \item $\tderiv_{\mtypethree} \exder[\cbn] \tyjp{(\msteps_{\mtypethree},\esteps_{\mtypethree})}{\tm}{\typctx_{\mtypethree}}{\mtypethree}$ 
\end{itemize}
 such that 
	\begin{itemize}
	\item \emph{Type contexts}: $\typctx = \typctx_{\mtypetwo} \mplus \typctx_{\mtypethree}$,
	\item \emph{Indices}: $\msteps = \msteps_{\mtypetwo} + \msteps_{\mtypethree}$ and $\esteps = \esteps_{\mtypetwo} + \esteps_{\mtypethree}$.
	\end{itemize}
\end{lemma}

\gettoappendix {l:name-linear-substitution}
\begin{proof} 
By induction on $\cbnctx$. Cases:
\begin{itemize}
\item \emph{Empty context}, \ie $\cbnctx = \ctxhole$. The typing derivation $\tderiv$ is simply 
$$ \infer[\ax]{\Deri[(0, 1)] {\var \col \single \type} \var \type}{} $$
and $\typctx$ is empty. Then $\mtype = \mult\type$ and so $\mtypetwo$ is empty.  The statement then holds with respect to $\tderiv_{\cbnctx\cwc{\tm}} \defeq \tderivtwo$, because  $\msteps = 0$ and $\esteps = 1$.

\item \emph{Left on an application}, \ie $\cbnctx = \cbnctxtwo \tmtwo$.
The last rule of $\tderiv$ can only be $\app$, and so $\tderiv$ has the form:
			$$
			\infer
				[\app]
				{\tyjp{(\msteps_{\typctxthree} + \msteps_{\typctxfour} + 1, \esteps_{\typctxthree} + \esteps_{\typctxfour})}{\cbnctxtwo \cwc{\var} \tmtwo}{\var \colon (\mtype_{\typctxthree} \bigmplus \mtype_{\typctxfour}) ; (\typctxthree \bigmplus \typctxfour)}{\type}}
				{\tyjp{(\msteps_{\typctxthree},\esteps_{\typctxthree})}{\cbnctxtwo \cwc{\var}}{\var \colon \mtype_{\typctxthree} ; \typctxthree}{\ty{\mtypetwo}{\type}}
				\quad
				\tyjp{(\msteps_{\typctxfour}, \esteps_{\typctxfour})}{\tmtwo}{\var \colon \mtype_{\typctxfour} ; \typctxfour}{\mtypetwo}}
			$$
where $\typctx = \typctxthree \bigmplus \typctxfour$, $\typctxthree(\var) = \typctxfour(\var) = \emptytype$, $\mtype_{\typctxthree} \mplus \mtype_{\typctxfour} = \mtype $, $\msteps = \msteps_{\typctxthree} + \msteps_{\typctxfour} + 1$, and $\esteps = \esteps_{\typctxthree} + \esteps_{\typctxfour}$.

	By \ih, there exists a splitting $\mtype_{\typctxthree} = \mult\typetwo \mplus \mtypethree$ such that for every derivation $\tderivtwo \exder[\cbn]\Deri[(\mstepstwo, \estepstwo)]{\typctxtwo}{\tm}{\typetwo}$ there exists a derivation 
	$$\tderiv_{\cbnctxtwo\cwc{\tm}} \exder[\cbn] \Deri[(\msteps_{\typctxthree}+\mstepstwo, \esteps_{\typctxthree} + \estepstwo-1)]{\var:\mtypethree;
  \typctxthree \bigmplus \typctxtwo}{\cbnctxtwo\cwc{\tm}}{\ty{\mtypetwo}{\type}}$$
By applying an \app rule we obtain:
		$$
			\infer
				[\app]
				{\tyjp{(\msteps_{\typctxthree} + \mstepstwo + \msteps_{\typctxfour} + 1, \esteps_{\typctxthree} + \estepstwo + \esteps_{\typctxfour}-1)}{\cbnctxtwo \cwc{\var} \tmtwo}{\var \colon (\mtypethree \mplus \mtype_{\typctxfour}) ; (\typctxthree \mplus \typctxtwo \mplus  \typctxfour)}{\type}}
				{\Deri[(\msteps_{\typctxthree}+\mstepstwo, \esteps_{\typctxthree} + \estepstwo-1)]{\var:\mtypethree;
  \typctxthree \mplus \typctxtwo}{\cbnctxtwo\cwc{\tm}}{\ty{\mtypetwo}{\type}}
				\quad
				\tyjp{(\msteps_{\typctxfour}, \esteps_{\typctxfour})}{\tmtwo}{\var \colon \mtype_{\typctxfour} ; \typctxfour}{\mtypetwo}}
			$$
Now, by defining $\mtypetwo \defeq  \mtypethree \mplus \mtype_{\typctxfour}$, we obtain $\mtype = \mtype_{\typctxthree} \mplus \mtype_{\typctxfour}  = \mult\typetwo \mplus \mtypethree \mplus \mtype_{\typctxfour}  = \mult\typetwo \mplus \mtypetwo$. Therefore by applying the equalities on the type context the last obtained judgement is in fact:
$$\tyjp{(\msteps_{\typctxthree} + \mstepstwo + \msteps_{\typctxfour} + 1, \esteps_{\typctxthree} + \estepstwo + \esteps_{\typctxfour}-1)}{\cbnctxtwo \cwc{\var} \tmtwo}{\var \colon \mtypetwo ; (\typctx \mplus \typctxtwo)}{\type}$$
and by applying those on the indices we obtain:
$$\tyjp{(\msteps + \mstepstwo, \esteps + \estepstwo-1)}{\cbnctxtwo \cwc{\var} \tmtwo}{\var \colon \mtypetwo ; (\typctx \mplus \typctxtwo)}{\type}$$
as required.

\item \emph{Left of a substitution}, \ie $\cbnctx = \cbnctxtwo \esub\vartwo\tmtwo$. Note that $\var \neq \vartwo$, because the hypothesis $\cbnctx \cwc{\var}$ implies that $\cbnctx$ does not capture $\var$. 

The last rule of $\tderiv$ can only be $\ES$, and so $\tderiv$ has the form:
			$$
			\infer
				[\esrule]
				{\tyjp{(\msteps_{\typctxthree} + \msteps_{\typctxfour} + 1, \esteps_{\typctxthree} + \esteps_{\typctxfour})}{\cbnctxtwo \cwc{\var} \esub\vartwo\tmtwo}{\var \colon (\mtype_{\typctxthree} \mplus \mtype_{\typctxfour}) ;  (\typctxthree \mplus \typctxfour)}{\type}}
				{\tyjp{(\msteps_{\typctxthree},\esteps_{\typctxthree})}{\cbnctxtwo \cwc{\var}}{\var \colon \mtype_{\typctxthree} ; \vartwo \colon \mtype' ; \typctxthree}{\type}
				\quad
				\tyjp{(\msteps_{\typctxfour}, \esteps_{\typctxfour})}{\tmtwo}{\var \colon \mtype_{\typctxfour} ; \typctxfour}{\mtype'}}
			$$
where $\typctx = \typctxthree \mplus \typctxfour$, $\typctxthree(\var) = \typctxfour(\var) = \emptytype$, $\mtype_{\typctxthree} \mplus \mtype_{\typctxfour} = \mtype $, $\msteps = \msteps_{\typctxthree} + \msteps_{\typctxfour} + 1$, and $\esteps = \esteps_{\typctxthree} + \esteps_{\typctxfour}$.

	By \ih, there exists a splitting $\mtype_{\typctxthree} = \mult\typetwo \mplus \mtypethree$ such that for every derivation $\tderivtwo \exder[\cbn]\Deri[(\mstepstwo, \estepstwo)]{\typctxtwo}{\tm}{\typetwo}$ there exists a derivation 
	$$\tderiv_{\cbnctxtwo\cwc{\tm}} \exder[\cbn] \Deri[(\msteps_{\typctxthree}+\mstepstwo, \esteps_{\typctxthree} + \estepstwo-1)]{\var:\mtypethree; \vartwo \colon \mtype' ;
  \typctxthree \mplus \typctxtwo}{\cbnctxtwo\cwc{\tm}}{\type}$$
  
Note that by \reflemma{name-typctx-varocc} and the fact that we are working up to $\alpha$-equivalence, we can prove that $\vartwo \notin \dom{\typctxtwo} $.
By applying a rule \esrule we obtain 
		$$
			\infer
				[\esrule]
				{\tyjp{(\msteps_{\typctxthree} + \mstepstwo + \msteps_{\typctxfour} + 1, \esteps_{\typctxthree} + \estepstwo + \esteps_{\typctxfour}-1)}{\cbnctxtwo \cwc{\var} \esub\vartwo\tmtwo}{\var \hastype \mtypethree \mplus \mtype_{\typctxfour} ; \typctxthree \mplus \typctxtwo \mplus \typctxfour}{\type}}
				{\Deri[(\msteps_{\typctxthree}+\mstepstwo, \esteps_{\typctxthree} + \estepstwo-1)]{\var \hastype \mtypethree; \vartwo \colon \mtype' ;
  \typctxthree \mplus \typctxtwo}{\cbnctxtwo\cwc{\tm}}{\type}
				\quad
				\tyjp{(\msteps_{\typctxfour}, \esteps_{\typctxfour})}{\tmtwo}{\var \colon \mtype_{\typctxfour} ; \typctxfour}{\mtype'}}
			$$
			Now, by defining $\mtypetwo \defeq  \mtypethree \mplus \mtype_{\typctxfour}$, we obtain $\mtype = \mtype_{\typctxthree} \mplus \mtype_{\typctxfour}  = \mult\typetwo \mplus \mtypethree \mplus \mtype_{\typctxfour}  = \mult\typetwo \mplus \mtypetwo$. Therefore by applying the equalities on the type context the last obtained judgement is in fact:
$$\tyjp{(\msteps_{\typctxthree} + \mstepstwo + \msteps_{\typctxfour} + 1, \esteps_{\typctxthree} + \estepstwo + \esteps_{\typctxfour}-1)}{\cbnctxtwo \cwc{\var} \esub\vartwo\tmtwo}{\var \colon \mtypetwo ; (\typctx \mplus \typctxtwo)}{\type}$$
and by applying those on the indices we obtain:
$$\tyjp{(\msteps + \mstepstwo, \esteps + \estepstwo-1)}{\cbnctxtwo \cwc{\var} \esub\vartwo\tmtwo}{\var \colon \mtypetwo ; (\typctx \mplus \typctxtwo)}{\type}$$
as required.\qed
\end{itemize}
\end{proof}

\gettoappendix {prop:name-subject-reduction}
\begin{proof} \hfill
\begin{enumerate}
\item By induction on $\tm \tom \tmtwo$. Cases:
\begin{itemize}
\item \emph{Step at top level}, \ie  $\tm =  \sctxp{\la\var \tmthree} \tmfour \tom \sctxp{ \tmthree \esub\var\tmfour } = \tmtwo$. This case is itself by induction on $\sctx$. Two sub-cases:
\begin{itemize}
\item \emph{Empty substitution context}, \ie $\sctx = \ctxhole$. By 
construction the derivation $\tderiv$ is of the form: 
$$
\infer[\app]{
	\Deri[(\msteps_\tmthree +  \msteps_\tmfour +1, \esteps_\tmthree + \esteps_\tmfour)]{\typctx_\tmthree  \mplus  \typctx_\tmfour}{(\la \var \tmthree)\tmfour}{\type}
}{
	\infer[\fun]{
		\Deri[(\msteps_\tmthree, \esteps_\tmthree)]{\typctx_\tmthree}{\la \var \tmthree}{ \mtype \rightarrow \type}
		}{
		\Deri[(\msteps_\tmthree, \esteps_\tmthree)]{\var:\mtype;\typctx_\tmthree}{\tmthree}{\type}
		} 
	\quad 
    \Deri[(\msteps_\tmfour, \esteps_\tmfour)]{\typctx_\tmfour}{\tmfour}{\mtype}  
}$$
With $\typctx = \typctx_\tmthree  \mplus  \typctx_\tmfour$, $\msteps = \msteps_\tmthree + \msteps_\tmfour$, and $\esteps = \esteps_\tmthree + \esteps_\tmfour$. Note that $\msteps \geq 1$ as required. We construct the 
following derivation $\tderivtwo$, verifying the statement:
$$ 
\infer[\esrule]{
	\Deri[(\msteps_\tmthree +  \msteps_\tmfour, \esteps_\tmthree + \esteps_\tmfour)]{\typctx_\tmthree  \mplus  \typctx_\tmfour}{\tmthree \esub\var\tmfour}{\type}
}
{
	\Deri[(\msteps_\tmthree, \esteps_\tmthree)]{\var:\mtype;\typctx_\tmthree}{\tmthree}{\type}
    \quad 
	\Deri[(\msteps_\tmfour, \esteps_\tmfour)]{\typctx_\tmfour}{\tmfour}{\mtype}    
}$$
 
\item \emph{Non-empty substitution context}, \ie $\sctx = \sctxtwo \esub\vartwo\tmfive$. Then $\tderiv$ has the following structure:
$$
\infer[\app]{
	\Deri[(\msteps_\tmthree + \msteps_\tmfive +  \msteps_\tmfour +1, \esteps_\tmthree + \esteps_\tmfive + \esteps_\tmfour)]{\typctx_\tmthree  \mplus  \typctx_\tmfive \mplus  \typctx_\tmfour}{\sctxtwop{\la \var \tmthree}\esub\vartwo\tmfive\tmfour}{\type}
}{
	\infer[\esrule]{
		\Deri[(\msteps_\tmthree + \msteps_\tmfive, \esteps_\tmthree + \esteps_\tmfive)]{\typctx_\tmthree\mplus  \typctx_\tmfive }{\sctxtwop{\la \var \tmthree}\esub\vartwo\tmfive}{ \mtype \rightarrow \type}
		}{
		\Deri[(\msteps_\tmthree, \esteps_\tmthree)]{\vartwo:\mtypetwo; \typctx_\tmthree}{\sctxtwop{\la \var \tmthree}}{ \mtype \rightarrow \type}
		\quad
		    \Deri[(\msteps_\tmfive, \esteps_\tmfive)]{\typctx_\tmfive}{\tmfive}{\mtypetwo}  
		} 
	\quad 
    \Deri[(\msteps_\tmfour, \esteps_\tmfour)]{\typctx_\tmfour}{\tmfour}{\mtype}  
}$$
With $\typctx = \typctx_\tmthree  \mplus \typctx_\tmfive \mplus  \typctx_\tmfour$, $\msteps = \msteps_\tmthree + \msteps_\tmfive + \msteps_\tmfour + 1$, and $\esteps = \esteps_\tmthree + \esteps_\tmfive + \esteps_\tmfour$. Note that $\msteps \geq 1$ as required.

Consider the following derivation, obtained by removing the rule \esrule:
$$
\infer[\app]{
	\Deri[(\msteps_\tmthree +  \msteps_\tmfour +1, \esteps_\tmthree +  \esteps_\tmfour)]{\vartwo:\mtypetwo;\typctx_\tmthree  \mplus  \typctx_\tmfour}{\sctxtwop{\la \var \tmthree}\tmfour}{\type}
}{
	\Deri[(\msteps_\tmthree, \esteps_\tmthree)]{\vartwo:\mtypetwo; \typctx_\tmthree}{\sctxtwop{\la \var \tmthree}}{ \mtype \rightarrow \type} 
	\quad 
    \Deri[(\msteps_\tmfour, \esteps_\tmfour)]{\typctx_\tmfour}{\tmfour}{\mtype}  
}$$
By \ih, we obtain a derivation 
$$\tderivthree \exder \Deri[(\msteps_\tmthree +  \msteps_\tmfour, \esteps_\tmthree +  \esteps_\tmfour)]{\vartwo:\mtypetwo;\typctx_\tmthree  \mplus  \typctx_\tmfour}{\sctxtwop{ \tmthree \esub\var\tmfour}}{\type}$$ 
Now, we apply a rule \esrule with respect to $\vartwo$ and $\tmfive$, obtaining the following derivation $\tderivtwo$, satisfying the statement:
$$
\infer[\esrule]{
	\Deri[(\msteps_\tmthree + \msteps_\tmfive +  \msteps_\tmfour, \esteps_\tmthree + \esteps_\tmfive + \esteps_\tmfour)]{\typctx_\tmthree  \mplus  \typctx_\tmfive \mplus  \typctx_\tmfour}{\sctxtwop{ \tmthree \esub\var\tmfour}\esub\vartwo\tmfive}{\type}
	}{
	\Deri[(\msteps_\tmthree +  \msteps_\tmfour, \esteps_\tmthree +  \esteps_\tmfour)]{\vartwo:\mtypetwo;\typctx_\tmthree  \mplus  \typctx_\tmfour}{\sctxtwop{ \tmthree \esub\var\tmfour}}{\type}
	\quad 
    \Deri[(\msteps_\tmfive, \esteps_\tmfive)]{\typctx_\tmfive}{\tmfive}{\mtypetwo}  
}$$

\end{itemize}

\item \emph{Contextual closure.} We have $\tm = \cbnctxp\tmthree \tom \cbnctxp\tmfour = \tmtwo$. Cases of $\cbnctx$:
	\begin{itemize}
		\item \emph{Left on an application}, \ie $\cbnctx = \cbnctxtwo \tmfive$. The last typing rule in $\tderiv$ is necessarily $\app$ and $\tderiv$ is of the form
		$$
		\infer
			[\app]
			{\tyjp{(\msteps_\tmthree + \msteps_\tmfive + 1, \esteps_\tmthree + \esteps_\tmfive)}{\cbnctxtwo \hole{\tmthree} \tmfive}{\typctx_\tmthree \mplus \typctx_\tmfive}{\type}}
			{\tyjp{(\msteps_\tmthree,\esteps_\tmthree)}{\cbnctxtwop{\tmthree}}{\typctx_\tmthree}{\ty{\mtype}{\type}}
			\quad
			\tyjp{(\msteps_\tmfive,\esteps_\tmfive)}{\tmfive}{\typctx_\tmfive}{\mtype}}
		$$
		With $\typctx = \typctx_\tmthree  \mplus  \typctx_\tmfive$, $\msteps = \msteps_\tmthree + \msteps_\tmfive + 1$, and $\esteps = \esteps_\tmthree + \esteps_\tmfive$. 
		
By \ih, $\msteps_\tmthree \geq 1$ and there exists a derivation $\tyjp{(\msteps_\tmthree-1,\esteps_\tmthree)}{\cbnctxtwop{\tmfour}}{\typctx_\tmthree}{\ty{\mtype}{\type}}$, thus allowing us to construct $\tderivtwo$ as follows:
		$$
		\infer
			[\appsteps]
			{\tyjp{(\msteps_\tmthree + \msteps_\tmfive, \esteps_\tmthree + \esteps_\tmfive)}{\cbnctxtwop{\tmfour} \tmfive}{\typctx_\tmthree \mplus \typctx_\tmfive}{\type}}
			{\tyjp{(\msteps_\tmthree-1,\esteps_\tmthree)}{\cbnctxtwop{\tmfour}}{\typctx_\tmthree}{\ty{\mtype}{\type}}
			\quad
			\tyjp{(\msteps_\tmfive,\esteps_\tmfive)}{\tmfive}{\typctx_\tmfive}{\mtype}}
		$$
Note that $(\msteps_\tmthree + \msteps_\tmfive, \esteps_\tmthree + \esteps_\tmfive) = (\msteps - 1,\esteps) $.

	\item Let $\cbnctx = \cbnctxtwo \esub{\var}{\tmfive}$. The last typing rule in $\tderiv$ is necessarily $\ES$  and $\tderiv$ is of the form
		$$
		\infer
			[\ES]
			{\tyjp{(\msteps_\tmthree + \msteps_\tmfive, \esteps_\tmthree + \esteps_\tmfive)}{\cbnctxtwo \hole{\tmthree} \esub\var\tmfive}{\typctx_\tmthree \mplus \typctx_\tmfive}{\type}}
			{\tyjp{(\msteps_\tmthree,\esteps_\tmthree)}{\cbnctxtwop{\tmthree}}{\typctx_\tmthree, \var:\mtype}{\type}
			\quad
			\tyjp{(\msteps_\tmfive,\esteps_\tmfive)}{\tmfive}{\typctx_\tmfive}{\mtype}}
		$$
		With $\typctx = \typctx_\tmthree  \mplus  \typctx_\tmfive$, $\msteps = \msteps_\tmthree + \msteps_\tmfive$, and $\esteps = \esteps_\tmthree + \esteps_\tmfive$. 
		
By \ih, $\msteps_\tmthree \geq 1$, and so $\msteps \geq 1$, and there exists a derivation $\tyjp{(\msteps_\tmthree-1,\esteps_\tmthree)}{\cbnctxtwop{\tmfour}}{\typctx_\tmthree, \var:\mtype}{\type}$, thus allowing us to construct $\tderivtwo$ as follows:
$$
		\infer
			[\ES]
			{\tyjp{(\msteps_\tmthree + \msteps_\tmfive -1, \esteps_\tmthree + \esteps_\tmfive)}{\cbnctxtwo \hole{\tmthree} \esub\var\tmfive}{\typctx_\tmthree \mplus \typctx_\tmfive}{\type}}
			{\tyjp{(\msteps_\tmthree-1,\esteps_\tmthree)}{\cbnctxtwop{\tmfour}}{\typctx_\tmthree, \var:\mtype}{\type}
			\quad
			\tyjp{(\msteps_\tmfive,\esteps_\tmfive)}{\tmfive}{\typctx_\tmfive}{\mtype}}
		$$
		Note that $(\msteps_\tmthree + \msteps_\tmfive-1, \esteps_\tmthree + \esteps_\tmfive) = (\msteps - 1,\esteps) $.

\end{itemize}
\end{itemize}

\item By induction on $\tm \toe \tmtwo$. 
\begin{itemize}
\item \emph{Step at top level}, \ie  $\tm =  \cbnctx \cwc{\var} \esub\var\tmthree \toe \cbnctx \cwc{\tmthree} \esub\var\tmthree = \tmtwo$. The last typing rule in $\tderiv$ is necessarily $\ES$  and $\tderiv$ is of the form
		$$
		\infer
			[\ES]
			{\tyjp{(\msteps_\cbnctx + \msteps_\tmthree, \esteps_\cbnctx + \esteps_\tmthree)}{\cbnctx \cwc{\var} \esub\var\tmthree}{\typctx_\cbnctx \mplus \typctxtwo}{\type}}
			{\tderiv_{\cbnctx\cwc{\var}} \exder \tyjp{(\msteps_\cbnctx,\esteps_\cbnctx)}{\cbnctx\cwc{\var}}{\typctx_\cbnctx, \var:\mtype}{\type}
			\quad
			\tyjp{(\mstepstwo,\estepstwo)}{\tmthree}{\typctxtwo}{\mtype}}
		$$
		With $\typctx = \typctx_\cbnctx  \mplus  \typctxtwo$, $\msteps = \msteps_\cbnctx + \mstepstwo$, and $\esteps = \esteps_\cbnctx + \estepstwo$. 
		
	Let $\mtype = \mult\typetwo \mplus \mtypetwo$ be the splitting of $\mtype$ given by the linear substitution lemma (\reflemma{name-linear-substitution}) applied to $\tderiv_{\cbnctx\cwc{\var}}$. By the multi-sets splitting lemma (\reflemma{name-splitting-multisets}) there exist two derivations 
\begin{enumerate}
\item  $\tderivtwo_\typetwo \exder[\cbn]\Deri[(\mstepstwo_\typetwo, \estepstwo_\typetwo)]{\typctxtwo_\typetwo}{\tmthree}{\typetwo}$ and

\item  $\tderivtwo_\mtypetwo \exder[\cbn]\Deri[(\mstepstwo_\mtypetwo, \estepstwo_\mtypetwo)]{\typctxtwo_\mtypetwo}{\tmthree}{\mtypetwo}$.
\end{enumerate}
such that $\typctxtwo = \typctxtwo_\typetwo \mplus \typctxtwo_\mtypetwo$, $\mstepstwo = \mstepstwo_\typetwo + \mstepstwo_\mtypetwo$,  and $\estepstwo = \estepstwo_\typetwo + \estepstwo_\mtypetwo$.

Now, by applying again the linear substitution lemma to $\tderiv_{\cbnctx\cwc{\var}}$ with respect to $\tderivtwo_\typetwo$, we obtain a derivation
$$\tderiv_{\cbnctx\cwc{\tmthree}} \exder[\cbn] \Deri[(\msteps_\cbnctx+\mstepstwo_\typetwo, \esteps_\cbnctx + \estepstwo_\typetwo-1)]{\var:\mtypetwo;
  \typctx_\cbnctx \mplus \typctxtwo_\typetwo}{\cbnctx\cwc{\tmthree}}{\type}$$

Then $\tderivtwo$ is built as follows:
$$
		\infer
			[\ES]
			{\tyjp{(\msteps_\cbnctx + \msteps_\typetwo + \msteps_\mtypetwo, \esteps_\cbnctx + \esteps_\typetwo+ \esteps_\mtypetwo -1)}{\cbnctx \cwc{\tmthree} \esub\var\tmthree}{\typctx_\cbnctx \mplus \typctxtwo_\typetwo \mplus \typctxtwo_\mtypetwo}{\type}}
			{\Deri[(\msteps_\cbnctx+\mstepstwo_\typetwo, \esteps_\cbnctx + \estepstwo_\typetwo-1)]{\var:\mtypetwo;
  \typctx_\cbnctx \mplus \typctxtwo_\typetwo}{\cbnctx\cwc{\tmthree}}{\type}
			\quad
			\Deri[(\mstepstwo_\mtypetwo, \estepstwo_\mtypetwo)]{\typctxtwo_\mtypetwo}{\tmthree}{\mtypetwo}}
		$$
	Now, note that the last judgement is in fact
	$$ \tyjp{(\msteps_\cbnctx + \mstepstwo, \esteps_\cbnctx + \estepstwo -1)}{\cbnctx \cwc{\tmthree} \esub\var\tmthree}{\typctx_\cbnctx \mplus \typctxtwo}{\type} $$
	which in turn is 
	$$ \tyjp{(\msteps, \esteps -1)}{\cbnctx \cwc{\tmthree} \esub\var\tmthree}{\typctx}{\type} $$
	as required.
	
	\item \emph{Contextual closure.} As in the $\tom$ case. Note that indeed those cases do not depend on the details of the step itself, but only on the context enclosing it. \qed
\end{itemize}
\end{enumerate}
\end{proof}

\gettoappendix {prop:name-normal-forms-forall}
\begin{proof}
By induction on the derivation of $\normalpr\tm$. Cases:
\begin{itemize}
\item \emph{Base}, \ie $\normalpr{\la{\var}\tm}$. Then $\tderiv$ can only be:
$$      \infer[\normal]{
        \Deri[(0,0)] {} { \la\var\tm } \normal
      }{}
$$
which satisfies the statement.

\item \emph{Inductive}, \ie $\normalpr{ \tm \esub\var\tmtwo }$ because $\normalpr\tm$. Then $\tderiv$ has the following shape.
$$
      \infer[\esrule]{
        \Deri[(\msteps_\tm + \msteps_\tmtwo, \esteps_\tm + \esteps_\tmtwo)] {\typctx_\tm \mplus \typctx_\tmtwo}
             {\tm \esub\var\tmtwo} \normal
      }{
        \tderiv_\tm \exder \Deri[(\msteps_\tm, \esteps_\tm)] {\typctx_\tm; \var \col \M} \tm {\normal}
        \quad
        \tderiv_\tmtwo \exder \Deri[(\msteps_\tmtwo, \esteps_\tmtwo)] {\typctx_\tmtwo} \tmtwo {\M}
      } 
$$
with $\typctx = \typctx_\tm \mplus \typctx_\tmtwo$, $\msteps = \msteps_\tm + \msteps_\tmtwo$, and $\esteps = \esteps_\tm + \esteps_\tmtwo$.
We can apply the \ih to $\tderiv_\tm$, obtaining that $\mtype = \emptytype$, $\typctx_\tm$ is empty, and $\msteps_\tm = \esteps_\tm= 0$. Then $\tderiv_\tmtwo$ is simply a many rule with 0 premises. Therefore, $\typctx_\tmtwo$ is empty and $\msteps_\tmtwo = \esteps_\tmtwo = 0$. Then $\typctx$ is empty and $\msteps = \esteps = 0$.\qed
\end{itemize}
\end{proof}

\gettoappendix {thm:name-correctness}
\begin{proof} 
By induction on $\msteps+\esteps$ and case analysis on whether $\tm$ reduces or not. If $\tm$ is a $\tocbn$ normal form  then we only have to prove the \emph{moreover} part, that states if $\tderiv$ is tight then $\msteps = \esteps = 0$, which follows from \refprop{name-normal-forms-forall}.

Otherwise, two cases:
\begin{enumerate}
\item \emph{Multiplicative steps}: $\tm \tom \tmthree$ and by quantitative subject reduction (\refprop{name-subject-reduction}) there is a derivation $\tderivtwo \exder[\cbn] \Deri[(\msteps-1, \esteps)] {}{\tmthree}{\type}$. By \ih, there exist $\tmtwo$ and ${\deriv'}$ such that $\normalpr \tmtwo$ and ${\deriv'}:\tmthree
  \tocbnn \tmtwo$, $\sizem{\deriv'} \leq \msteps-1$ and $\sizee{\deriv'} \leq \esteps$. Just note that $\tm \tom \tmthree$ and so, if $\deriv:\tm \tocbnn\tmtwo$ is $\deriv'$ preceeded by such a step, we have $\sizem\deriv \leq \msteps$ and $\sizee\deriv \leq \esteps$. 
  
  If $\tderiv$ is tight then $\tderivtwo$ is tight. Then $\sizem{\deriv'} = \msteps-1$ and $\sizee{\deriv'}  = \esteps$ by \ih, that give $\sizem\deriv = \msteps$ and $\sizee\deriv = \esteps$.
  
  \item \emph{Exponential steps}:  $\tm \toe \tmthree$ and by quantitative subject reduction (\refprop{name-subject-reduction}) there is a derivation $\tderivtwo \exder[\cbn] \Deri[(\msteps, \esteps-1)] {}{\tmthree}{\type}$. By \ih, there exist $\tmtwo$ and ${\deriv'}$ such that $\normalpr \tmtwo$ and ${\deriv'}:\tmthree
  \tocbnn \tmtwo$, $\sizem{\deriv'} \leq \msteps$ and $\sizee{\deriv'} \leq \esteps-1$. Just note that $\tm \toe \tmthree$ and so, if $\deriv:\tm \tocbnn\tmtwo$ is $\deriv'$ preceeded by such a step, we have $\sizem\deriv \leq \msteps$ and $\sizee\deriv \leq \esteps$. 
  
  If $\tderiv$ is tight then $\tderivtwo$ is tight. Then $\sizem{\deriv'} = \msteps$ and $\sizee{\deriv'}  = \esteps-1$ by \ih, that give $\sizem\deriv = \msteps$ and $\sizee\deriv = \esteps$.\qed
  \end{enumerate}
\end{proof}

\subsection{\cbn Completeness}
\gettoappendix {prop:name-normal-forms-exist}
\begin{proof} 
By induction on $\normalpr\tm$. Cases:
\begin{itemize}
  \item \emph{Abstraction}: if $\normalpr\tm$ because $\tm = \la\var\tmtwo$ then 
  $$\infer[\normal]{
        \Deri[(0,0)] {} { \la\var\tmtwo } \normal
      }{}$$
  \item \emph{Substitution}: if $\normalpr\tm$ because $\tm = \tmtwo\esub\var\tmthree$ and $\normalpr\tmtwo$ then by \ih there exists a tight derivation $\tderivtwo\exder[\cbn] \Deri[(0, 0)]{}{\tmtwo}\normal$. Then $\tderiv$ is given by:
  $$
  \infer[\esrule]{
        \Deri[(0, 0)] {}
             {\tmtwo \esub\var\tmthree} \normal
      }{
        \tderivtwo\exder\Deri[(0, 0)] {} \tmtwo {\normal}
        \quad
\infer[\many]{
        \Deri[(0, 0)] {} \tmthree {\emptytype}        
        }{}
      }
      $$\qed
\end{itemize}

\end{proof}

In order to prove subject expansion, we have to first show that typability can also be pulled back along substitutions, via a linear removal lemma. We also need the following merging property of multi sets, whose proof is omitted because straightforward.


\begin{lemma}[Merging of multi-sets with respect to derivations]
\label{l:name-merging-multisets}
Let $\tm$ be a term. For any two derivations 
\begin{itemize}
  \item $\tderiv_{\mtypetwo} \exder[\cbn] \tyjp{(\msteps_{\mtypetwo},\esteps_{\mtypetwo})}{\tm}{\typctx_{\mtypetwo}}{\mtypetwo}$, and
  \item $\tderiv_{\mtypethree} \exder[\cbn] \tyjp{(\msteps_{\mtypethree},\esteps_{\mtypethree})}{\tm}{\typctx_{\mtypethree}}{\mtypethree}$ 
\end{itemize}
 there is a derivation $\tderiv_{\mtypetwo \mplus \mtypethree} \exder[\cbn] \tyjp{(\msteps_{\mtypetwo} + \msteps_{\mtypethree},\esteps_{\mtypetwo} + \esteps_{\mtypethree})}{\tm}{\typctx_{\mtypetwo} \mplus \typctx_{\mtypethree}}{\mtypetwo \mplus \mtypethree}$.
\end{lemma}

\gettoappendix {l:name-linear-removal}
\begin{proof}
By induction on $\cbnctx$. Cases:
\begin{itemize}
  \item \emph{Empty context}, \ie $\cbnctx = \ctxhole$. Then $\tderiv \exder[\cbn] \Deri[(\msteps, \esteps)]{\typctx;\var\hastype\mtype}{\tmtwo}{\type}$. By \reflemma{name-typctx-varocc}, $\var \notin \fv\tmtwo$ implies $\mtype = \emptytype$. Then we simply take 
  \begin{itemize}
  \item $\tderiv_\tmtwo \defeq \tderiv$, that implies $\typetwo \defeq \type$, $\typctx_\tmtwo \defeq \typctx$, $\msteps_\tmtwo \defeq \msteps$, and $\esteps_\tmtwo \defeq \esteps$, and 
  \item $\tderiv_\var$ defined as the axiom 
  $$\infer[\ax]{\Deri[(0, 1)] {\var \col \single \type} \var \type}{} $$
  and for which $\typctx'$ is empty, $\msteps' = 0$, and $\esteps' = 1$. 
  \end{itemize}
  Then the statement holds:
  \begin{itemize}
  \item \emph{Type contexts}: $\typctx = \emptyset \ctxplus \typctx = \emptyset \ctxplus  \typctx_\tmtwo = \typctx' \ctxplus  \typctx_\tmtwo$ and 
  \item \emph{Indices}: $(\msteps, \esteps) = (\msteps_{\tmtwo}, \esteps_{\tmtwo}) = (0 + \msteps_{\tmtwo}, 1 + \esteps_{\tmtwo} - 1) = (\msteps' + \msteps_{\tmtwo}, \esteps' + \esteps_{\tmtwo} - 1) $.
  \end{itemize}
  
  \item \emph{Left of an application}, \ie $\cbnctx = \cbnctxtwo \tmthree$. Then $\tderiv$ has the form
    $$\infer[\app]{
        \Deri[(\msteps_{\cbnctxtwop\tmtwo} + \msteps_\tmthree + 1, \esteps_{\cbnctxtwop\tmtwo} + \esteps_\tmthree)] {\typctx_{\cbnctxtwop\tmtwo} \ctxplus  \typctx_\tmthree;\var\hastype\mtype}
             {\cbnctxtwop\tmtwo \tmthree} \type
      }{
        \tderiv_{\cbnctxtwop\tmtwo} \exder\Deri[(\msteps_{\cbnctxtwop\tmtwo}, \esteps_{\cbnctxtwop\tmtwo})] {\typctx_{\cbnctxtwop\tmtwo};\var\hastype\mtype} {\cbnctxtwop\tmtwo} {\ty{\mtypetwo}{\type}}
        \quad
        \Deri[(\msteps_\tmthree, \esteps_\tmthree)] {\typctx_\tmthree} \tmthree {\mtypetwo}
      }
      $$
  where $\var \notin \dom{\typctx_\tmthree}$ (by \reflemma{name-typctx-varocc}, because $\var \notin\fv\tmthree$ by hypothesis), $\typctx = \typctx_{\cbnctxtwop\tmtwo} \ctxplus \typctx_\tmthree$, $\msteps = \msteps_{\cbnctxtwop\tmtwo} + \msteps_\tmthree + 1$, and $\esteps = \esteps_{\cbnctxtwop\tmtwo} + \esteps_\tmthree$.
  
  Applying the \ih to $\tderiv_{\cbnctxtwop\tmtwo}$ provides a type $\typetwo$ and derivations:
  $$\tderiv_\tmtwo  \exder[\cbn] \Deri[(\msteps_\tmtwo, \esteps_\tmtwo)]{\typctx_\tmtwo}{\tmtwo}{\typetwo}$$
  and 
  $$\tderiv_{\cbnctxtwo\cwc{\var}} \exder[\cbn] 
 \Deri[(\mstepsthree, \estepsthree)]{  \typctx'';\var\hastype\mtype \mplus \mult{\typetwo} }{\cbnctxtwo\cwc{\var}}{\ty{\mtypetwo}{\type}}$$
  such that $\typctx_{\cbnctxtwop\tmtwo} = \typctx'' \ctxplus \typctx_\tmtwo$ and $(\msteps_{\cbnctxtwop\tmtwo}, \esteps_{\cbnctxtwop\tmtwo}) = (\mstepsthree +  \msteps_\tmtwo, \estepsthree + \esteps_\tmtwo- 1)$.
 
 Then $\tderiv_{\cbnctx\cwc{\var}}$ is given by: 
 $$\infer[\app]{
        \Deri[(\mstepsthree + \msteps_\tmthree + 1, \estepsthree + \esteps_\tmthree)] {(\typctx''\ctxplus \typctx_\tmthree);\var\hastype\mtype \mplus \mult{\typetwo}}
             {\cbnctxtwop\var \tmthree} \type
      }{
        \tderiv_{\cbnctxtwo\cwc{\var}} \exder[\cbn] 
 \Deri[(\mstepsthree, \estepsthree)]{  \typctx'' ;\var\hastype \mtype \mplus\mult{\typetwo} }{\cbnctxtwo\cwc{\var}}{\ty{\mtypetwo}{\type}}
        \quad
        \Deri[(\msteps_\tmthree, \esteps_\tmthree)] {\typctx_\tmthree} \tmthree {\mtypetwo}
      }
      $$
  that, by taking $\typctx' \defeq \typctx'' \ctxplus \typctx_\tmthree$, $\mstepstwo = \mstepsthree + \msteps_\tmthree + 1$, and $\estepstwo = \estepsthree + \esteps_\tmthree$, verifies the statement because:
  \begin{itemize}
  \item \emph{Type contexts}: $\typctx = \typctx_{\cbnctxtwop\tmtwo} \ctxplus \typctx_\tmthree = \typctx'' \ctxplus \typctx_\tmtwo \ctxplus \typctx_\tmthree = \typctx' \ctxplus \typctx_\tmtwo$, and 
  \item \emph{Indices}: $(\msteps, \esteps) = (\msteps_{\cbnctxtwop\tmtwo} + \msteps_\tmthree + 1, \esteps_{\cbnctxtwop\tmtwo} + \esteps_\tmthree) = (\mstepsthree +  \msteps_\tmtwo + \msteps_\tmthree + 1, \estepsthree + \esteps_\tmtwo- 1+ \esteps_\tmthree) = (\mstepstwo +  \msteps_\tmtwo, \estepstwo + \esteps_\tmtwo- 1)$.
  \end{itemize}
  
  \item \emph{Left of a substitution}, \ie $\cbnctx = \cbnctxtwo \esub\vartwo\tmthree$. Then $\tderiv$ has the form
    $$\infer[\esrule]{
        \Deri[(\msteps_{\cbnctxtwop\tmtwo} + \msteps_\tmthree, \esteps_{\cbnctxtwop\tmtwo} + \esteps_\tmthree)] {\typctx_{\cbnctxtwop\tmtwo} \ctxplus  \typctx_\tmthree; \var\hastype\mtype}
             {\cbnctxtwop\tmtwo \esub\vartwo\tmthree} \type
      }{
        \tderiv_{\cbnctxtwop\tmtwo} \exder\Deri[(\msteps_{\cbnctxtwop\tmtwo}, \esteps_{\cbnctxtwop\tmtwo})] {\typctx_{\cbnctxtwop\tmtwo}; \var\hastype\mtype;\vartwo\hastype\mtypetwo} {\cbnctxtwop\tmtwo} {\type}
        \quad
        \Deri[(\msteps_\tmthree, \esteps_\tmthree)] {\typctx_\tmthree} \tmthree {\mtypetwo}
      }
      $$
  where $\var \notin \dom{\typctx_\tmthree}$ (by \reflemma{name-typctx-varocc}, because $\var \notin\fv\tmthree$ by hypothesis), $\typctx = \typctx_{\cbnctxtwop\tmtwo} \ctxplus  \typctx_\tmthree$, $\msteps = \msteps_{\cbnctxtwop\tmtwo} + \msteps_\tmthree$, and $\esteps = \esteps_{\cbnctxtwop\tmtwo} + \esteps_\tmthree$.
  
  Applying the \ih to $\tderiv_{\cbnctxtwop\tmtwo}$ provides a type $\typetwo$ and derivations:
  $$\tderiv_\tmtwo  \exder[\cbn] \Deri[(\msteps_\tmtwo, \esteps_\tmtwo)]{\typctx_\tmtwo}{\tmtwo}{\typetwo}$$
  and 
  $$\tderiv_{\cbnctxtwo\cwc{\var}} \exder[\cbn] 
 \Deri[(\mstepsthree, \estepsthree)]{ \typctx''; \var\hastype\mtype \mplus \mult{\typetwo}; \vartwo\hastype\mtypetwo  }{\cbnctxtwo\cwc{\var}}{\type}$$
  such that $\typctx_{\cbnctxtwop\tmtwo} = \typctx'' \ctxplus \typctx_\tmtwo$ and $(\msteps_{\cbnctxtwop\tmtwo}, \esteps_{\cbnctxtwop\tmtwo}) = (\mstepsthree +  \msteps_\tmtwo, \estepsthree + \esteps_\tmtwo- 1)$.
 
 Then $\tderiv_{\cbnctx\cwc{\var}}$ is given by: 
 $$\infer[\esrule]{
        \Deri[(\mstepsthree + \msteps_\tmthree , \estepsthree + \esteps_\tmthree)] {(\typctx''\ctxplus  \typctx_\tmthree); \var\hastype\mtype \mplus \mult{\typetwo}}
             {\cbnctxtwop\var \esub\vartwo\tmthree} \type
      }{
        \tderiv_{\cbnctxtwo\cwc{\var}} \exder[\cbn] 
 \Deri[(\mstepsthree, \estepsthree)]{  \typctx''; \var\hastype\mtype \mplus \mult{\typetwo}; \vartwo\hastype\mtypetwo  }{\cbnctxtwo\cwc{\var}}{\type}
        \quad
        \Deri[(\msteps_\tmthree, \esteps_\tmthree)] {\typctx_\tmthree} \tmthree {\mtypetwo}
      }
      $$
  that, by taking $\typctx' \defeq \typctx'' \ctxplus  \typctx_\tmthree$, $\mstepstwo = \mstepsthree + \msteps_\tmthree$, and $\estepstwo = \estepsthree + \esteps_\tmthree$, verifies the statement because:
  \begin{itemize}
  \item \emph{Type contexts}: $\typctx = \typctx_{\cbnctxtwop\tmtwo} \ctxplus  \typctx_\tmthree = \typctx'' \ctxplus  \typctx_\tmtwo \ctxplus  \typctx_\tmthree = \typctx' \ctxplus \typctx_\tmtwo$, and 
  \item \emph{Indices}: $(\msteps, \esteps) = (\msteps_{\cbnctxtwop\tmtwo} + \msteps_\tmthree, \esteps_{\cbnctxtwop\tmtwo} + \esteps_\tmthree) = (\mstepsthree +  \msteps_\tmtwo + \msteps_\tmthree, \estepsthree + \esteps_\tmtwo- 1+ \esteps_\tmthree) = (\mstepstwo +  \msteps_\tmtwo, \estepstwo + \esteps_\tmtwo- 1)$. \qed
  \end{itemize}
\end{itemize}

\end{proof}

\gettoappendix {prop:name-subject-expansion}
\begin{proof} \hfill
\begin{enumerate}
\item By induction on $\tm \tom \tmtwo$. Cases:
\begin{itemize}
\item \emph{Step at top level}, \ie  $\tm =  \sctxp{\la\var \tmthree} \tmfour \tom \sctxp{ \tmthree \esub\var\tmfour } = \tmtwo$. This case is itself by induction on $\sctx$. Two sub-cases:
\begin{itemize}
\item \emph{Empty substitution context}, \ie $\sctx = \ctxhole$. The derivation $\tderiv$ has the form: 
$$ 
\infer[\ES]{
	\Deri[(\msteps_\tmthree +  \msteps_\tmfour, \esteps_\tmthree + \esteps_\tmfour)]{\typctx_\tmthree  \mplus  \typctx_\tmfour}{\tmthree \esub\var\tmfour}{\type}
}
{
	\Deri[(\msteps_\tmthree, \esteps_\tmthree)]{\var:\mtype;\typctx_\tmthree}{\tmthree}{\type}
    \quad 
	\Deri[(\msteps_\tmfour, \esteps_\tmfour)]{\typctx_\tmfour}{\tmfour}{\mtype}    
}$$
With $\typctx = \typctx_\tmthree  \mplus  \typctx_\tmfour$, $\msteps = \msteps_\tmthree + \msteps_\tmfour$, and $\esteps = \esteps_\tmthree + \esteps_\tmfour$. We construct the 
following derivation $\tderivtwo$, verifying the statement:
$$
\infer[\app]{
	\Deri[(\msteps_\tmthree +  \msteps_\tmfour +1, \esteps_\tmthree + \esteps_\tmfour)]{\typctx_\tmthree  \mplus  \typctx_\tmfour}{(\la \var \tmthree)\tmfour}{\type}
}{
	\infer[\fun]{
		\Deri[(\msteps_\tmthree, \esteps_\tmthree)]{\typctx_\tmthree}{\la \var \tmthree}{ \mtype \rightarrow \type}
		}{
		\Deri[(\msteps_\tmthree, \esteps_\tmthree)]{\var:\mtype;\typctx_\tmthree}{\tmthree}{\type}
		} 
	\quad 
    \Deri[(\msteps_\tmfour, \esteps_\tmfour)]{\typctx_\tmfour}{\tmfour}{\mtype}  
}$$
 
\item \emph{Non-empty substitution context}, \ie $\sctx = \sctxtwo \esub\vartwo\tmfive$. Then $\tderiv$ has the following structure:

$$
\infer[\ES]{
	\Deri[(\msteps_\tmthree + \msteps_\tmfive +  \msteps_\tmfour, \esteps_\tmthree + \esteps_\tmfive + \esteps_\tmfour)]{\typctx_\tmthree  \mplus  \typctx_\tmfive \mplus  \typctx_\tmfour}{\sctxtwop{ \tmthree \esub\var\tmfour}\esub\vartwo\tmfive}{\type}
	}{
	\Deri[(\msteps_\tmthree +  \msteps_\tmfour, \esteps_\tmthree +  \esteps_\tmfour)]{\vartwo:\mtypetwo;\typctx_\tmthree  \mplus  \typctx_\tmfour}{\sctxtwop{ \tmthree \esub\var\tmfour}}{\type}
	\quad 
    \Deri[(\msteps_\tmfive, \esteps_\tmfive)]{\typctx_\tmfive}{\tmfive}{\mtypetwo}  
}$$
With $\typctx = \typctx_\tmthree  \mplus \typctx_\tmfive \mplus  \typctx_\tmfour$, $\msteps = \msteps_\tmthree + \msteps_\tmfive + \msteps_\tmfour + 1$, and $\esteps = \esteps_\tmthree + \esteps_\tmfive + \esteps_\tmfour$. 

By \ih applied to the left premise, we obtain a derivation

$$\Deri[(\msteps_\tmthree +  \msteps_\tmfour +1, \esteps_\tmthree +  \esteps_\tmfour)]{\vartwo:\mtypetwo;\typctx_\tmthree  \mplus  \typctx_\tmfour}{\sctxtwop{\la \var \tmthree}\tmfour}{\type}$$
that has the following structure:
$$
\infer[\app]{
	\Deri[(\msteps_\tmthree +  \msteps_\tmfour +1, \esteps_\tmthree +  \esteps_\tmfour)]{\vartwo:\mtypetwo;\typctx_\tmthree  \mplus  \typctx_\tmfour}{\sctxtwop{\la \var \tmthree}\tmfour}{\type}
}{
	\Deri[(\msteps_\tmthree, \esteps_\tmthree)]{\vartwo:\mtypetwo; \typctx_\tmthree}{\sctxtwop{\la \var \tmthree}}{ \mtype \rightarrow \type} 
	\quad 
    \Deri[(\msteps_\tmfour, \esteps_\tmfour)]{\typctx_\tmfour}{\tmfour}{\mtype}  
}$$
We then construct $\tderivtwo$ has follows:
$$
\infer[\app]{
	\Deri[(\msteps_\tmthree + \msteps_\tmfive +  \msteps_\tmfour +1, \esteps_\tmthree + \esteps_\tmfive + \esteps_\tmfour)]{\typctx_\tmthree  \mplus  \typctx_\tmfive \mplus  \typctx_\tmfour}{\sctxtwop{\la \var \tmthree}\esub\vartwo\tmfive\tmfour}{\type}
}{
	\infer[\ES]{
		\Deri[(\msteps_\tmthree + \msteps_\tmfive, \esteps_\tmthree + \esteps_\tmfive)]{\typctx_\tmthree\mplus  \typctx_\tmfive }{\sctxtwop{\la \var \tmthree}\esub\vartwo\tmfive}{ \mtype \rightarrow \type}
		}{
		\Deri[(\msteps_\tmthree, \esteps_\tmthree)]{\vartwo:\mtypetwo; \typctx_\tmthree}{\sctxtwop{\la \var \tmthree}}{ \mtype \rightarrow \type}
		\quad
		    \Deri[(\msteps_\tmfive, \esteps_\tmfive)]{\typctx_\tmfive}{\tmfive}{\mtypetwo}  
		} 
	\quad 
    \Deri[(\msteps_\tmfour, \esteps_\tmfour)]{\typctx_\tmfour}{\tmfour}{\mtype}  
}$$

\end{itemize}

\item \emph{Contextual closure.} We have $\tm = \cbnctxp\tmthree \tom \cbnctxp\tmfour = \tmtwo$. Cases of $\cbnctx$:
	\begin{itemize}
		\item \emph{Left on an application}, \ie $\cbnctx = \cbnctxtwo \tmfive$. The last typing rule in $\tderiv$ is necessarily $\app$ and $\tderiv$ is of the form
		$$
		\infer
			[\appsteps]
			{\tyjp{(\msteps_\tmfour + \msteps_\tmfive +1, \esteps_\tmfour + \esteps_\tmfive)}{\cbnctxtwop{\tmfour} \tmfive}{\typctx_\tmfour \mplus \typctx_\tmfive}{\type}}
			{\tyjp{(\msteps_\tmfour,\esteps_\tmfour)}{\cbnctxtwop{\tmfour}}{\typctx_\tmfour}{\ty{\mtype}{\type}}
			\quad
			\tyjp{(\msteps_\tmfive,\esteps_\tmfive)}{\tmfive}{\typctx_\tmfive}{\mtype}}
		$$
		With $\typctx = \typctx_\tmfour  \mplus  \typctx_\tmfive$, $\msteps = \msteps_\tmfour + \msteps_\tmfive + 1$, and $\esteps = \esteps_\tmfour + \esteps_\tmfive$. 
		
By \ih, there exists a derivation $\tyjp{(\msteps_\tmfour+1,\esteps_\tmfour)}{\cbnctxtwop{\tmthree}}{\typctx_\tmthree}{\ty{\mtype}{\type}}$, thus allowing us to construct $\tderivtwo$ as follows:
		
		$$
		\infer
			[\app]
			{\tyjp{(\msteps_\tmfour + \msteps_\tmfive + 2, \esteps_\tmfour + \esteps_\tmfive)}{\cbnctxtwo \hole{\tmthree} \tmfive}{\typctx_\tmfour \mplus \typctx_\tmfive}{\type}}
			{\tyjp{(\msteps_\tmfour+1,\esteps_\tmfour)}{\cbnctxtwop{\tmthree}}{\typctx_\tmfour}{\ty{\mtype}{\type}}
			\quad
			\tyjp{(\msteps_\tmfive,\esteps_\tmfive)}{\tmfive}{\typctx_\tmfive}{\mtype}}
		$$
Note that $(\msteps_\tmfour + \msteps_\tmfive+2, \esteps_\tmfour + \esteps_\tmfive) = (\msteps + 1,\esteps) $.

	\item Let $\cbnctx = \cbnctxtwo \esub{\var}{\tmfive}$. The last typing rule in $\tderiv$ is necessarily $\ES$  and $\tderiv$ is of the form
		$$
		\infer
			[\ES]
			{\tyjp{(\msteps_\tmfour + \msteps_\tmfive , \esteps_\tmfour + \esteps_\tmfive)}{\cbnctxtwo \hole{\tmthree} \esub\var\tmfive}{\typctx_\tmfour \mplus \typctx_\tmfive}{\type}}
			{\tyjp{(\msteps_\tmfour,\esteps_\tmfour)}{\cbnctxtwop{\tmfour}}{\typctx_\tmfour, \var:\mtype}{\type}
			\quad
			\tyjp{(\msteps_\tmfive,\esteps_\tmfive)}{\tmfive}{\typctx_\tmfive}{\mtype}}
		$$
		With $\typctx = \typctx_\tmfour  \mplus  \typctx_\tmfive$, $\msteps = \msteps_\tmfour + \msteps_\tmfive$, and $\esteps = \esteps_\tmfour + \esteps_\tmfive$. 
		
By \ih, there exists a derivation $\tyjp{(\msteps_\tmfour+1,\esteps_\tmfour)}{\cbnctxtwop{\tmthree}}{\typctx_\tmfour, \var:\mtype}{\type}$, thus allowing us to construct $\tderivtwo$ as follows:
$$
		\infer
			[\ES]
			{\tyjp{(\msteps_\tmfour + \msteps_\tmfive+1, \esteps_\tmfour + \esteps_\tmfive)}{\cbnctxtwo \hole{\tmthree} \esub\var\tmfive}{\typctx_\tmfour \mplus \typctx_\tmfive}{\type}}
			{\tyjp{(\msteps_\tmfour+1,\esteps_\tmfour)}{\cbnctxtwop{\tmthree}}{\typctx_\tmfour, \var:\mtype}{\type}
			\quad
			\tyjp{(\msteps_\tmfive,\esteps_\tmfive)}{\tmfive}{\typctx_\tmfive}{\mtype}}
		$$

		Note that $(\msteps_\tmfour + \msteps_\tmfive+1, \esteps_\tmfour + \esteps_\tmfive) = (\msteps + 1,\esteps) $.

\end{itemize}
\end{itemize}

\item By induction on $\tm \toe \tmtwo$. 
\begin{itemize}
\item \emph{Step at top level}, \ie  $\tm =  \cbnctx \cwc{\var} \esub\var\tmthree \toe \cbnctx \cwc{\tmthree} \esub\var\tmthree = \tmtwo$. The last typing rule of $\tderiv$ is necessarily $\ES$  and $\tderiv$ is of the form
		$$
		\infer
			[\ES]
			{\tyjp{(\msteps_\cbnctx + \msteps_\tmthree, \esteps_\cbnctx + \esteps_\tmthree)}{\cbnctx \cwc{\tmthree} \esub\var\tmthree}{\typctx_\cbnctx \mplus \typctxtwo}{\type}}
			{\tderiv_{\cbnctx\cwc{\var}} \exder \tyjp{(\msteps_\cbnctx,\esteps_\cbnctx)}{\cbnctx\cwc{\tmthree}}{\typctx_\cbnctx, \var:\mtype}{\type}
			\quad
			\tyjp{(\mstepstwo,\estepstwo)}{\tmthree}{\typctxtwo}{\mtype}}
		$$
		With $\typctx = \typctx_\cbnctx  \mplus  \typctxtwo$, $\msteps = \msteps_\cbnctx + \mstepstwo$, and $\esteps = \esteps_\cbnctx + \estepstwo$. 
		
	The linear removal lemma (\reflemma{name-linear-removal}) applied to $\tderiv_{\cbnctx\cwc{\tmthree}}$ gives a type $\typetwo$ and two derivations 
\begin{enumerate}
\item $\tderiv_\tmthree  \exder[\cbn] \Deri[(\msteps_\tmthree, \esteps_\tmthree)]{\typctx_\tmthree}{\tmthree}{\typetwo}$, and
\item $\tderiv_{\cbnctx\cwc{\var}} \exder[\cbn] \Deri[(\mstepsthree, \estepsthree)]{  \typctx';\var{:}\mtype \mplus \mult{\typetwo} }{\cbnctx\cwc{\var}}{\type}$
\end{enumerate}
such that $\typctx_\cbnctx = \typctx_\tmthree \mplus \typctx'$, $\msteps_\cbnctx = \msteps_\tmthree + \mstepsthree$,  and $\esteps_\cbnctx = \esteps_\tmthree + \estepsthree-1$.

Now, by applying the multi-sets merging lemma (\reflemma{name-merging-multisets}) to $\tderiv_\tmthree $ and the right premise of $\tderiv$:
$$\tyjp{(\mstepstwo,\estepstwo)}{\tmthree}{\typctxtwo}{\mtype} $$
we obtain a derivation
$$\tderivtwo_\tmthree \exder[\cbn] \tyjp{(\mstepstwo + \msteps_\tmthree,\estepstwo + \esteps_\tmthree)}{\tmthree}{\typctxtwo \mplus \typctx_\tmthree}{\mtype\mplus \mult{\typetwo} }$$
Then $\tderivtwo$ is built as follows:
$$
		\infer
			[\ES]
			{\tyjp{(\mstepsthree + \mstepstwo + \msteps_\tmthree, \estepsthree + \estepstwo + \esteps_\tmthree)}{\cbnctx \cwc{\var} \esub\var\tmthree}{\typctx' \mplus \typctxtwo \mplus \typctx_\tmthree}{\type}
			}
			{
			\Deri[(\mstepsthree, \estepsthree)]{  \typctx';\var{:}\mtype \mplus \mult{\typetwo} }{\cbnctx\cwc{\var}}{\type}
			\quad
			\tyjp{(\mstepstwo + \msteps_\tmthree,\estepstwo + \esteps_\tmthree)}{\tmthree}{\typctxtwo \mplus \typctx_\tmthree}{\mtype\mplus \mult{\typetwo} }
			}
		$$
	Now, note that the last judgement is in fact
	$$ \tyjp{(\msteps_\cbnctx + \mstepstwo, \esteps_\cbnctx + \estepstwo +1)}{\cbnctx \cwc{\var} \esub\var\tmthree}{\typctx_\cbnctx \mplus \typctxtwo}{\type} $$
	which in turn is 
	$$ \tyjp{(\msteps, \esteps +1)}{\cbnctx \cwc{\var} \esub\var\tmthree}{\typctx}{\type} $$
	as required.
	
	\item \emph{Contextual closure.} As in the $\tom$ case. Note that indeed those cases do not depend on the details of the step itself, but only on the context enclosing it.\qed
\end{itemize}
\end{enumerate}
\end{proof}

\gettoappendix {thm:name-completeness}
\begin{proof} 
By induction on the length $k \defeq \size{\deriv}$ of the evaluation $\deriv \colon \tm \tocbnn \tmtwo$. If $k = 0$ then $\tm =  \tmtwo$ and $\normalpr\tm$. 
\refprop{name-normal-forms-exist}
gives the existence of  a tight 
derivation $\tderiv \exder[\cbnup]  \Deri[(0,0)] {} \tm \normal$, that satisfies the statement because $\sizem\deriv = \sizee\deriv = 0$. 

If $k > 0$ then
$\tm \tocbn \tmthree \rightarrow_\cbnsym^{k-1} \tmtwo$. Let $\deriv'$ be the evaluation $\tmthree \rightarrow_\cbnsym^{k-1} \tmtwo$. By \ih\  there exists a tight derivation $\tderivtwo\exder[\cbnup]  \Deri[(\sizem{\deriv'},\sizee{\deriv'})]{}
\tmthree \normal$. By quantitative subject expansion \refprop{name-subject-expansion}
there exists a derivation $\tderiv$ of $\tm$ with the same
types in the ending judgement of $\tderivtwo$---then $\tderiv$ is
tight---and with indices $(\sizem\deriv,\sizee\deriv)$. 
\qed
\end{proof}

\section{Types by Value (\refsect{cbv})}

\begin{lemma}[Type contexts and variable occurrences for \cbv]
	\label{l:value-typctx-varocc}
	Let  $\tderiv \exder[\cbvsym] \Deri[(\msteps, \esteps)] \typctx \tm \mtype$ be a derivation. 
	If $\var \not \in \fv\tm$ then $\var \notin \dom\typctx$.
\end{lemma}

\begin{proof}
	By straightforward induction on the derivation $\tderiv$.
	\qed
\end{proof}

\subsection{\cbv Correctness}

\begin{lemma}[Typing values]
	\label{l:value-typing} 
	Let $\val$ be a value and $\tderiv \exder[\cbvsym] \tyjp{(\msteps, \esteps)}{\val}{\typctx}{\mtype}$ be a derivation for it. Then,
	\begin{enumerate}
		\item \label{p:value-typing-empty}
		\emph{Empty multi-set implies null size}: 
		if $\mtype = \emptymset$, then $\dom{\typctx} = \emptyset$ with $\msteps = 0 = \esteps$.
		
		\item \label{p:value-typing-dec}
		\emph{Multi-set splitting}: if 
		$\mtype = \mtypetwo \uplus \mtypethree$, then there are two type contexts $\typctxtwo$ and $\typctxthree$ and two derivations $\tderivtwo \exder[\cbvsym] \Deri[(\mstepstwo,\estepstwo)]{\typctxtwo}{\val}{\mtypetwo}$ and
		$\tderivthree \exder[\cbvsym]\Deri[(\mstepsthree,\estepsthree)]{\typctxthree}{\val}{\mtypethree}$ such that $\typctx = \typctxtwo \uplus \typctxthree$, $\msteps = \mstepstwo +  \mstepsthree$ and $\esteps = \estepstwo +  \estepsthree$
	\end{enumerate}
\end{lemma}

\gettoappendix {l:value-linear-substitution-bis}
\begin{proof}
	By induction on $\cbvctx$.
	Cases:
	\begin{itemize}
		\item \emph{Hole}, \ie $\cbvctx = \ctxhole$, then $\cbvctx\cwc{\val} = \val$ and $\cbvctx\cwc{\var} = \var$, hence $\tderiv$ consists only of an $\ax$-rule and so $\mtypetwo = \mtype$ and $\dom{\typctx} = \emptyset$, with $\msteps = 0$ and $\esteps = 1$.
		Let $\mtypethree \defeq \mtype$ and $\mtypefour \defeq \emptymset$.
		Thus, every $\tderivtwo\exder[\cbvsym] \Deri[(\mstepstwo, \estepstwo)]{\typctxtwo}{\val}{\mtypethree}$ coincides with a type derivation $\tderiv' \exder[\cbvsym] \Deri[(\msteps+\mstepstwo, \esteps+\estepstwo - 1)]{\typctx \mplus \typctxtwo, \var \col \mtypefour}{\cbvctx\cwc{\val}}{\mtypetwo}$, since $\typctx \mplus \typctxtwo, \var \col \mtypefour = \typctx$ and $\mtypetwo = \mtypethree$ and $(\msteps+\mstepstwo, \esteps+\estepstwo-1) = (\mstepstwo,\estepstwo)$.

		\item \emph{Left application}, \ie $\cbvctx \defeq \cbvctxtwo\tm$.
		Then, the derivation $\tderiv$ has the form
		\begin{equation*}
				\infer[\app]{
			\Deri[(\msteps, \esteps)] {\var \col \mtype, \typctx}
			{\cbvctxtwo\cwc{\var}\tm} \mtypetwo
		}{
			\Deri[(\msteps_{1}, \esteps_{1})] {\var \col \mtype_{1}, \typctx_{1}} {\cbvctxtwo\cwc{\var}} {\single{\ty{\mtypetwo'}{\mtypetwo}}}
			\qquad
			\Deri[(\msteps_{2}, \esteps_{2})] {\var \col \mtype_{2}, \typctx_{2}} \tm {\mtypetwo'}
		}
		\end{equation*}
		
		where $\mtype = \mtype_{1} \mplus \mtype_{2}$, $\typctx = \typctx_{1} \mplus \typctx_{2}$, $\msteps = \msteps_{1} + \msteps_{2} + 1$ and $\esteps = \esteps_{1} + \esteps_{2}$.
		By \ih, there exists a splitting $\mtype_{1} = \mtypethree \mplus \mtypefour'$ such that, for every derivation $\tderivtwo \exder[\cbvsym] \Deri[(\mstepstwo,\estepstwo)]{\typctxtwo}{\val}{\mtypethree}$, there exists a derivation with conclusion $\Deri[(\msteps_{1}+\mstepstwo, \esteps_{1}+\estepstwo-1)]{\typctx_{1} \mplus \typctxtwo, \var \col \mtypefour'}{\cbvctxtwo\cwc{\val}}{\single{\ty{\mtypetwo'}{\mtypetwo}}}$.
		So, we can construct the following derivation $\tderiv'$
		\begin{equation*}
		\infer[\app]{
			\Deri[(\msteps+\mstepstwo, \esteps + \estepstwo-1)] {\var \col \mtype_{2} \mplus \mtypefour', \typctx \mplus \typctxtwo}
			{\cbvctxtwo\cwc{\val}\tm} \mtypetwo
		}{
			\Deri[(\msteps_{1}+\mstepstwo, \esteps_{1}+\estepstwo-1)]{\typctx_{1} \mplus \typctxtwo, \var \col \mtypefour'}{\cbvctxtwo\cwc{\val}}{\single{\ty{\mtypetwo'}{\mtypetwo}}}
			\qquad
			\Deri[(\msteps_{2}, \esteps_{2})] {\var \col \mtype_{2}, \typctx_2} \tm {\mtypetwo'}
		}
		\end{equation*}
		where $\mtypefour \defeq \mtype_{2} \mplus \mtypefour'$ and $\mtype = \mtype_{1} \mplus \mtype_{2} = \mtypethree \mplus \mtypefour' \mplus \mtype_{2} = \mtypethree \mplus \mtypefour$.

		\item \emph{Right application}, \ie $\cbvctx \defeq \tm\cbvctxtwo$.
		Then, the derivation $\tderiv$ has the form
		\begin{equation*}
		\infer[\app]{
			\Deri[(\msteps, \esteps)] {\var \col \mtype, \typctx}
			{\tm\,\cbvctxtwo\cwc{\var}} \mtypetwo
		}{
			\Deri[(\msteps_1, \esteps_1)] {\var \col \mtype_1, \typctx_1} {\tm} {\single{\ty{\mtypetwo'}{\mtypetwo}}}
			\qquad
			\Deri[(\msteps_2, \esteps_2)] {\var \col \mtype_2, \typctx_2} {\cbvctxtwo\cwc{\var}} {\mtypetwo'}
		}
		\end{equation*}
		where $\mtype = \mtype_1 \mplus \mtype_2$, $\typctx = \typctx_1 \mplus \typctx_2$, $\msteps = \msteps_1 + \msteps_2 + 1$ and $\esteps = \esteps_1 + \esteps_2$.
		By \ih, there exists a splitting $\mtype_{2} = \mtypethree \mplus \mtypefour'$ such that, for every derivation $\tderivtwo \exder[\cbvsym] \Deri[(\mstepstwo,\estepstwo)]{\typctxtwo}{\val}{\mtypethree}$, there exists a derivation with conclusion $\Deri[(\msteps_{2}+\mstepstwo, \esteps_{2}+\estepstwo-1)]{\typctx_2 \mplus \typctxtwo, \var \col \mtypefour'}{\cbvctxtwo\cwc{\val}}{\mtypetwo'}$.
		So, we can construct the following  derivation $\tderiv'$
		\begin{equation*}
		\infer[\app]{
			\Deri[(\msteps+\mstepstwo, \esteps + \estepstwo-1)] {\var \col \mtype_{1} \mplus \mtypefour', \typctx \mplus \typctxtwo}
			{\tm\cbvctxtwo\cwc{\val}} \mtypetwo
		}{
			\Deri[(\msteps_{1}, \esteps_{1})] {\var \col \mtype_{1}, \typctx_{1}} {\tm} {\single{\ty{\mtypetwo'}{\mtypetwo}}}
			\qquad
			\Deri[(\msteps_{2}+\mstepstwo, \esteps_{2}+\estepstwo-1)]{\typctx_{2} \mplus \typctxtwo, \var \col \mtypefour'}{\cbvctxtwo\cwc{\val}}{\mtypetwo'}
		}
		\end{equation*}
		where $\mtypefour \defeq \mtype_{1} \mplus \mtypefour'$ and $\mtype = \mtype_{1} \mplus \mtype_{2} = \mtype_{1} \mplus \mtypethree \mplus \mtypefour' = \mtypethree \mplus \mtypefour$.
	
		\item \emph{Left explicit substitution}, \ie $\cbvctx \defeq \cbvctxtwo\esub{\vartwo}\tm$.
		We can suppose without loss of generality that $\vartwo \notin \fv{\tm} \cup \fv{\val} \cup \{ \var\}$, and hence $\vartwo \notin \dom{\typctxtwo}$ by \reflemma{value-typctx-varocc}.
		So, the derivation $\tderiv$ has the form
		\begin{equation*}
		\infer[\esrule]{
			\Deri[(\msteps, \esteps)] {\var \col \mtype, \typctx}
			{\cbvctxtwo\cwc{\var}\esub{\vartwo}\tm} \mtypetwo
		}{
			\Deri[(\msteps_1, \esteps_1)] {\var \col \mtype_1, \vartwo \col{\mtypetwo'}, \typctx_1} {\cbvctxtwo\cwc{\var}} {\mtypetwo}
			\qquad
			\Deri[(\msteps_2, \esteps_2)] {\var \col \mtype_2, \typctx_2} \tm {\mtypetwo'}
		}
		\end{equation*}
		where $\mtype = \mtype_1 \mplus \mtype_2$, $\typctx = \typctx_1 \mplus \typctx_2$, $\msteps = \msteps_1 + \msteps_2$ and $\esteps = \esteps_1 + \esteps_2$.
		By \ih, there exists a splitting $\mtype_{1} = \mtypethree \mplus \mtypefour'$ such that, for every derivation $\tderivtwo \exder[\cbvsym] \Deri[(\mstepstwo,\estepstwo)]{\typctxtwo}{\val}{\mtypethree}$, there exists a derivation of $\Deri[(\msteps_{1}+\mstepstwo, \esteps_{1}+\estepstwo-1)]{\typctx_{1} \mplus \typctxtwo, \var \col \mtypefour', \vartwo \col \mtypetwo'}{\cbvctxtwo\cwc{\val}}{\mtypetwo}$.
		Therefore, we can construct the following  derivation $\tderiv'$
		\begin{equation*}
		\infer[\esrule]{
			\Deri[(\msteps+\mstepstwo, \esteps + \estepstwo-1)] {\var \col \mtype_{2} \mplus \mtypefour', \typctx \mplus \typctxtwo}
			{\cbvctxtwo\cwc{\val}\esub{\vartwo}\tm} \mtypetwo
		}{
			\Deri[(\msteps_{1}+\mstepstwo, \esteps_{1}+\estepstwo-1)]{\typctx_{1} \mplus \typctxtwo, \var \col \mtypefour', \vartwo \col \mtypetwo'}{\cbvctxtwo\cwc{\val}}{\mtypetwo}
			\qquad
			\Deri[(\msteps_{2}, \esteps_{2})] {\var \col \mtype_{2}, \typctx_2} \tm {\mtypetwo'}
		}
		\end{equation*}
		where $\mtypefour \defeq \mtype_{2} \mplus \mtypefour'$ and $\mtype = \mtype_{1} \mplus \mtype_{2} = \mtypethree \mplus \mtypefour' \mplus \mtype_{2} = \mtypethree \mplus \mtypefour$.
		
		\item \emph{Right explicit substitution}, \ie $\cbvctx \defeq \tm\esub{\vartwo}\cbvctxtwo$.
		We can suppose without loss of generality that $\vartwo \notin \fv{\val} \cup \{ \var\}$, and hence $\vartwo \notin \dom{\typctxtwo}$ by \reflemma{value-typctx-varocc}.
		Then, the derivation $\tderiv$ has the form
		\begin{equation*}
		\infer[\esrule]{
			\Deri[(\msteps, \esteps)] {\var \col \mtype, \typctx}
			{\tm\esub{\vartwo}{\cbvctxtwo\cwc{\var}}} \mtypetwo
		}{
			\Deri[(\msteps_1, \esteps_1)] {\var \col \mtype_1, \vartwo \col \mtypetwo', \typctx_1} {\tm} {\mtypetwo}
			\qquad
			\Deri[(\msteps_2, \esteps_2)] {\var \col \mtype_2, \typctx_2} {\cbvctxtwo\cwc{\var}} {\mtypetwo'}
		}
		\end{equation*}
		where $\mtype = \mtype_1 \mplus \mtype_2$, $\typctx = \typctx_1 \mplus \typctx_2$, $\msteps = \msteps_1 + \msteps_2$ and $\esteps = \esteps_1 + \esteps_2$.
		By \ih, there is a splitting $\mtype_{2} = \mtypethree \mplus \mtypefour'$ such that, for every derivation $\tderivtwo \exder[\cbvsym] \Deri[(\mstepstwo,\estepstwo)]{\typctxtwo}{\val}{\mtypethree}$, there exists a derivation of $\Deri[(\msteps_{2}+\mstepstwo, \esteps_{2}+\estepstwo-1)]{\typctx_2 \mplus \typctxtwo, \var \col \mtypefour'}{\cbvctxtwo\cwc{\val}}{\mtypetwo'}$.
		So, we can construct the following derivation $\tderiv'$
		\begin{equation*}
		\infer[\esrule]{
			\Deri[(\msteps+\mstepstwo, \esteps + \estepstwo-1)] {\var \col \mtype_{1} \mplus \mtypefour', \typctx \mplus \typctxtwo}
			{\tm \esub{\vartwo}{\cbvctxtwo\cwc{\val}}} \mtypetwo
		}{
			\Deri[(\msteps_{1}, \esteps_{1})] {\var \col \mtype_{1}, \vartwo \col \mtypetwo', \typctx_{1}} {\tm} {\mtypetwo}
			\qquad
			\Deri[(\msteps_{2}+\mstepstwo, \esteps_{2}+\estepstwo-1)]{\typctx_{2} \mplus \typctxtwo, \var \col \mtypefour'}{\cbvctxtwo\cwc{\val}}{\mtypetwo'}
		}
		\end{equation*}
		where $\mtypefour \defeq \mtype_{1} \mplus \mtypefour'$ and $\mtype = \mtype_{1} \mplus \mtype_{2} = \mtype_{1} \mplus \mtypethree \mplus \mtypefour' = \mtypethree \mplus \mtypefour$.
		\qed
	\end{itemize}
\end{proof}

\gettoappendix {prop:value-subject-reduction}

\begin{proof}
	By induction on the reduction relation $\tocbv$, with the root rules $\rtom$ and $\rtoecbv$ as the base case, and the closure by $\cbv$ contexts of $\rtocbv \,\defeq\, \rtom \cup \rtoecbv$ as the inductive one.
	\begin{itemize}
		\item \emph{Root step for $\tomcbv$} \ie $\tm = \sctxp{\la{\var}{\tmthree}} \tmfour \rtom \sctxp{\tmthree \esub{\var}{\tmfour}} = \tm'$ where $\sctx \defeq \esub{\vartwo_1}{\tmtwo_1}\dots \esub{\vartwo_n}{\tmtwo_n}$ for some $n \geq 0$.
		We proceed by induction on $n \in \nat$.
		
		If $n = 0$ then $\sctx = \ctxhole$ and so $\tm = \sctxp{\la{\var}\tmthree} \tmfour = (\la{\var}\tmthree) \tmfour$ and $\tm' = \sctxp{\tmthree \esub{\var}{\tmfour}} = \tmthree \esub{\var}{\tmfour}$. 
		Hence, $\tderiv$ has the form
		$$
		\infer
		[\app]
		{\tyjp{(1 + \msteps' + \msteps'',\esteps' + \esteps'')}{(\la{\var}{\tmthree})\tmfour}{\typctxtwo \mplus \typctxthree }{\mtype}}
		{\infer
			[\many]
			{\tyjp{(\msteps',\esteps')}{\la{\var}{\tmthree}}{\typctxtwo}{\mult{\ty{\mtypethree}{\mtype}}}}
			{\infer
				[\fun]
				{ \tyjp{(\mstepstwo,\estepstwo)}{\la{\var}{\tmthree}}{\typctxtwo}{\ty{\mtypethree}{\mtype}}}
				{\tderivtwo \exder[\cbvsym]\tyjp{(\mstepstwo,\estepstwo)}{\tmthree}{\typctxtwo, \var \col \mtypethree}{\mtype}}}
			\quad
			{\tderivthree \exder[\cbvsym]\tyjp{(\mstepsthree,\estepsthree)}{\tmfour}{\typctxthree}{\mtypethree}}
		}
		$$			
		where $\typctx \defeq \typctxtwo \mplus \typctxthree$, $\msteps \defeq 1 + \mstepstwo + \mstepsthree$ and $\esteps \defeq \estepstwo + \estepsthree$.
		Therefore, $\msteps \geq 1$.
		We can construct the following derivation $\tderiv'$:
		$$
		\infer
		[\esrule]
		{\tyjp{(\msteps'+\msteps'',\esteps'+\esteps'')}{\tmthree\esub{\var}{\tmfour}}{\typctx}{\mtype}}
		{\tderivtwo \exder[\cbvsym]
			\tyjp{(\msteps',\esteps')}{\tmthree}{\typctxtwo, \var \col \mtypethree}{\mtype}
			\quad
			\tderivthree \exder[\cbvsym] \tyjp{(\msteps'',\esteps'')}{\tmfour}{\typctxthree}{\mtypethree}
		}
		$$
		where $(\msteps' + \msteps'', \esteps' + \esteps'') = (\msteps - 1, \esteps)$.
		
		Suppose now $n >  0$. 
		Let $\sctxtwo \defeq \esub{\vartwo_1}{\tmtwo_1} \dots \esub{\vartwo_{n-1}}{\tmtwo_{n-1}}$: then, $\tm = \sctxp{\la{\var}\tmthree} \tmfour = \sctxtwop{\la{\var}\tmthree}\esub{\vartwo_n}{\tmtwo_n} \tmfour$ and $\tm' = \sctxp{\tmthree \esub{\var}{\tmfour}} = \sctxptwo{\tmthree \esub{\var}{\tmfour}}\esub{\vartwo_n}{\tmtwo_n}$. 
		Hence, $\tderiv$ has the form 
		{\small
			\begin{equation*}
			\infer[\app]{
				\Deri[(\mstepsthree + \msteps_n + \mstepstwo_0 + 1,\estepsthree + \esteps_n + \estepstwo_0)] {\typctxtwo \mplus \typctx_n \mplus \typctx_0'}{\sctxp{\la{\var}\tmthree}{\tmfour}}{\mtype}
			}
			{
				\infer[\esrule]{
					\Deri[(\mstepsthree + \msteps_n, \estepsthree + \esteps_n)]{\typctxtwo \mplus \typctx_n}{\sctxp{\la{\var}\tmthree}}{\mtype}
				}{
					\tderivtwo'' \exder[\cbvsym] \Deri[(\mstepsthree,\estepsthree)]{\typctxtwo, \vartwo_n \col \mtypetwo_n}{\sctxtwop{\la{\var}\tmthree}}{\mtype}
					\qquad
					\tderivtwo_n \exder[\cbvsym] \Deri[(\msteps_n,\esteps_n)]{\typctx_n}{\tmtwo_n}{\mtypetwo_n}
				}
				\qquad
				\tderivthree \exder[\cbvsym] \Deri[(\mstepstwo_0,\estepstwo_0)]{\typctx_0'}{\tmfour}{\mtypethree}
			}
			\end{equation*}
		}
		\noindent where $\typctx \defeq \typctxtwo \mplus \typctx_n \mplus \typctx_0'$ and $(\msteps, \esteps) \defeq (\mstepsthree + \msteps_n + \mstepstwo_0 + 1,\estepsthree + \esteps_n + \estepstwo_0)$.
		Note that $\msteps \geq 1$ as required.
		Consider the following derivation $\tderivtwo$
		\begin{equation*}
		\infer[\app]{
			\Deri[(\mstepsthree + \mstepstwo_0 + 1,\estepsthree + \estepstwo_0)] {\typctxtwo \mplus \typctx_0', \vartwo_n \col \mtypetwo_n}{\sctxtwop{\la{\var}\tmthree}{\tmfour}}{\mtype}
		}
		{
			\tderivtwo'' \exder[\cbvsym] \Deri[(\mstepsthree,\estepsthree)]{\typctxtwo, \vartwo_n \col \mtypetwo_n}{\sctxtwop{\la{\var}\tmthree}}{\mtype}
			\qquad
			\tderivthree \exder[\cbvsym] \Deri[(\mstepstwo_0,\estepstwo_0)]{\typctx_0'}{\tmfour}{\mtypethree}
		}
		\end{equation*}
		By \ih applied to $\tderivtwo$ (since $\sctxtwop{\la{\var}\tmthree}\tmfour \rtom \sctxtwop{\tmthree\esub{\var}{\tmfour}}$), there is a derivation $$\tderivtwo' \exder[\cbvsym] \Deri[(\mstepstwo, \estepstwo)]{\typctxtwo \mplus \typctx_0', \vartwo_n \col \mtypetwo_n}{\sctxtwop{\tmthree\esub{\var}{\tmfour}}}{\mtype}$$
		where $(\mstepstwo, \estepstwo) = (\mstepsthree+\mstepstwo_0, \estepsthree+\estepstwo_0)$.
		We can then construct the following  derivation $\tderiv'$
		\begin{equation*}
		\infer[\esrule]{
			\Deri[(\mstepstwo+\msteps_n,\estepstwo+\esteps_n)]{\typctx}{\sctxptwo{\tmthree \esub{\var}{\tmfour}}\esub{\vartwo_n}{\tmtwo_n}}{\mtype}
		}
		{ 
			\tderivtwo' \exder[\cbvsym] \Deri[(\mstepstwo, \estepstwo)]{\typctxtwo \mplus \typctx_0', \vartwo_n \col \mtypetwo_n}{\sctxtwop{\tmthree\esub{\var}{\tmfour}}}{\mtype}
			\qquad
			\tderivtwo_n \exder[\cbvsym] \Deri[(\msteps_n,\esteps_n)]{\typctx_n}{\tmtwo_n}{\mtypetwo_n}
		}
		\end{equation*}
		where $(\mstepstwo+\msteps_n, \estepstwo+\esteps_n) = (\mstepsthree + \mstepstwo_0 + \msteps_n, \estepsthree + \estepstwo_0+\esteps_n) = (\msteps-1, \esteps)$.

		\item \emph{Root step for \!\!$\toecbv$\!\!}, \ie $\tm \defeq \cbvctx\cwc{\var} \esub{\var}{\sctxp{\val}} \rtoecbv\! \sctxp{\cbvctx\cwc{\val} \esub{\var}{\val}} \eqdef \tm'$ with $\sctx \defeq \esub{\vartwo_1}{\tmtwo_1}\dots \esub{\vartwo_n}{\tmtwo_n}$ for some $n \geq 0$. 
		We proceed by induction on $n \in \nat$.
		
		If $n = 0$ then $\sctx = \ctxhole$ and so $\tm = \cbvctx\cwc{\var} \esub{\var}{\val}$ and $\tm' = \cbvctx\cwc{\val} \esub{\var}{\val}$. 
		Hence, $\tderiv$ has the form
		$$
		\infer
		[\esrule]
		{\tyjp{(\msteps' + \msteps_0, \esteps' + \esteps_0)}{\cbvctx \cwc{\var} \esub{\var}{\val}}{\typctxtwo \mplus \typctx_0}{\mtype}}
		{\tderivtwo \exder[\cbvsym] \tyjp{(\msteps',\esteps')}{\cbvctx\cwc{\var}}{\typctxtwo, \var \col \mtypethree}{\mtype}
			\qquad
			\tderivthree \exder[\cbvsym] \tyjp{(\msteps_0, \esteps_0)}{\val}{\typctx_0}{\mtypethree}		}
		$$
		where $\typctx \defeq \typctxtwo \mplus \typctx_0$, and $\msteps \defeq  \mstepstwo + \msteps_0$ and $\esteps \defeq \estepstwo + \esteps_0$.
		Let $\mtypethree = \mtypethree' \mplus \mtypethree''$  be the splitting of $\mtypethree$ given by linear substitution for \cbv (\reflemma{value-linear-substitution-bis}).
		According to the multiset splitting property (\reflemmap{value-typing}{dec}), there exist a splitting $\typctx_0 = \typctx_0' \mplus \typctx_0''$ and the derivations $\tderivthree' \exder[\cbvsym] \tyjp{(\msteps_0',\esteps_0')}{\val}{\typctx_{0}'}{\mtypethree'}$ and $\tderivthree'' \exder[\cbvsym] \Deri[(\mstepsthree_0,\estepsthree_0)]{\typctx_{0}''}{\val}{\mtypethree''}$, with $\msteps_0 = \mstepstwo_0 + \mstepsthree_0$ and $\esteps_0 = \estepstwo_0 + \estepsthree_0$.
		By linear substitution for \cbv (\reflemma{value-linear-substitution-bis}), there exists a derivation $\tderivtwo' \exder[\cbvsym] \tyjp{(\mstepstwo + \mstepstwo_0,\estepstwo+\estepstwo_0-1)}{\cbvctx\cwc{\val}}{\typctxtwo \mplus \typctx_0', \var \col \mtypethree''}{\mtype}$.
		We can then construct the following derivation $\tderiv'$
		\begin{equation*}
		\infer[\esrule]{
			\tyjp{(\mstepstwo+\mstepstwo_0 + \mstepsthree_0,\estepstwo+\estepstwo_0 + \estepsthree_0-1)}{\cbvctx\cwc\val\esub{\var}{\val}}{\typctxtwo \mplus \typctx_0' \mplus \typctx_0''}{\mtype}
		}
		{
			\tderivtwo' \exder[\cbvsym] \tyjp{(\mstepstwo + \mstepstwo_0, \estepstwo+\estepstwo_0-1)} {\cbvctx\cwc{\val}}{\typctxtwo \mplus \typctx_0', \var \col \mtypethree''}{\mtype}
			\qquad
			\tderivthree'' \exder[\cbvsym]\tyjp{(\mstepsthree_0,\estepsthree_0)}{\val}{\typctx_{0}''}{\mtypethree''}
		}
		\end{equation*}
		where $\typctxtwo \mplus \typctx_0' \mplus \typctx_0'' = \typctxtwo \mplus \typctx_0 = \typctx$ and $(\mstepstwo+\mstepstwo_0+\mstepsthree_0,\estepstwo+\estepstwo_0+\estepsthree_0-1) = (\mstepstwo + \msteps_0, \estepstwo + \esteps_0-1) = (\msteps,\esteps-1)$.
		
		\indent Suppose now $n > 0$.
		Let $\sctxtwo \defeq \esub{\vartwo_1}{\tmtwo_1} \dots \esub{\vartwo_{n-1}}{\tmtwo_{n-1}}$: then, $\sctxp{\val} = \sctxtwop{\val}\esub{\vartwo_n}{\tmtwo_n}$ and so $\tm = \cbvctx\cwc{\var}\esub{\var}{\sctxtwop{\val}\esub{\vartwo_n}{\tmtwo_n}}$ and $\tm' = \sctxp{\cbvctx\cwc{\val} \esub{\var}{\val}} = \sctxptwo{\cbvctx\cwc{\val} \esub{\var}{\val}}\esub{\vartwo_n}{\tmtwo_n}$. 
		Hence, $\tderiv$ has the form 
		
		{\small
			\begin{equation*}
			\infer[\esrule]{
				\tyjp{(\mstepstwo_0+\mstepsthree+\msteps_n, \estepstwo_0+\estepsthree+\esteps_n)}{\cbvctx\cwc{\var} \esub{\var}{\sctxp{\val}}}{\typctx_0' \mplus \typctx_0'' \mplus \typctx_n}{\mtype}
			}
			{
				\tderivthree \exder[\cbvsym] \Deri[(\mstepstwo_0, \estepstwo_0)]{\typctx_0', \var \col \mtypetwo}{\cbvctx\cwc{\var}}{\mtype}
				\quad
				\infer[\esrule]{
					\Deri[(\mstepsthree_0+\msteps_n, \estepsthree_0+\esteps_n)]{\typctx_0'' \mplus \typctx_n}{\sctxp{\val}}{\mtypetwo}
				}
				{ 
					\tderivthree'' \exder[\cbvsym\!\!] \Deri[(\mstepsthree_0,\estepsthree_0)]{\vartwo_n \col \mtypetwo_n, \typctx_0''}{\sctxtwop{\val}}{\mtypetwo}
					\quad
					\tderivthree_n \exder[\cbvsym\!\!] \Deri[(\msteps_n,\esteps_n)]{\typctx_n}{\tmtwo_n}{\mtypetwo_n}
				}
			}
			\end{equation*}
		}
		where $\typctx \defeq \typctx_0' \mplus \typctx_0'' \mplus \typctx_n$ and $(\msteps, \esteps) \defeq (\mstepstwo_0+\mstepsthree+\msteps_n, \estepstwo_0+\estepsthree+\esteps_n)$.
		Consider the following derivation $\tderivtwo$
		\begin{equation*}
		\infer[\esrule]{
			\tyjp{(\mstepstwo_0+\mstepsthree, \estepstwo_0+\estepsthree)}{\cbvctx\cwc{\var} \esub{\var}{\sctxtwop{\val}}}{\typctx_0' \mplus \typctx_0'', \vartwo_n \col \mtypetwo_n}{\mtype}
		}
		{
			\tderivthree \exder[\cbvsym] \Deri[(\mstepstwo_0, \estepstwo_0)]{\typctx_0', \var \col \mtypetwo}{\cbvctx\cwc{\var}}{\mtype}
			\quad
			\tderivthree'' \exder[\cbvsym\!\!] \Deri[(\mstepsthree_0,\estepsthree_0)]{\vartwo_n \col \mtypetwo_n, \typctx_0''}{\sctxtwop{\val}}{\mtypetwo}
		}
		\end{equation*}
		By \ih applied to $\tderivtwo$ (since $\cbvctx\cwc{\var} \esub{\var}{\sctxtwop{\val}} \rtoecbv \sctxtwop{\cbvctx\cwc{\val}\esub{\var}{\val}}$), one has $\estepstwo_0 + \estepsthree \geq 1$ and there exists a derivation 
		$$\tderivtwo' \exder[\cbvsym] \Deri[(\mstepstwo, \estepstwo)]{\typctx_0' \mplus \typctx_0'', \vartwo_n \col \mtypetwo_n}{\sctxtwop{\cbvctx\cwc{\val}\esub{\var}{\val}}}{\mtype}$$
		where $(\mstepstwo, \estepstwo) \defeq (\mstepstwo_0+\mstepsthree, \estepstwo_0+\estepsthree-1)$.
		We can then construct the following derivation $\tderiv'$:
		\begin{equation*}
		\infer[\esrule]{
			\Deri[(\mstepstwo+\msteps_n, \estepstwo+\esteps_n)]{\typctx}{\sctxtwop{\cbvctx\cwc{\val}\esub{\var}{\val}}\esub{\vartwo_n}{\tmtwo_n}}{\mtype}
		}{
			\tderivtwo' \exder[\cbvsym] \Deri[(\mstepstwo, \estepstwo)]{\typctx_0' \mplus \typctx_0'', \vartwo_n \col \mtypetwo_n}{\sctxtwop{\cbvctx\cwc{\val}\esub{\var}{\val}}}{\mtype}
			\qquad
			\tderivthree_n \exder[\cbvsym\!\!] \Deri[(\msteps_n,\esteps_n)]{\typctx_n}{\tmtwo_n}{\mtypetwo_n}
		}
		\end{equation*}
		where $(\mstepstwo + \msteps_n, \estepstwo + \esteps_n) =  (\mstepstwo_0+\mstepsthree + \msteps_n, \estepstwo_0+\estepsthree-1 + \esteps_n) = (\msteps, \esteps -1)$.
		
		\item \emph{Application left}, \ie $\tm \defeq \tmtwo\tmthree \torule  \tmtwo'\tmthree \eqdef \tm'$ with $\tmtwo \torule \tmtwo'$ and $\Rule \in \{\mcbv, \ecbv\}$.
		So, $\tderiv$ has the form 
		\begin{equation*}
		\infer[\app]{
			\tyjp{(\msteps_1 + \msteps_2 +1,\esteps_1 + \esteps_2)}{\tm}{\typctx_1 \mplus \typctx_2}{\mtype}
		}
		{
			\tderivtwo_1 \exder[\cbvsym] \tyjp{(\msteps_1,\esteps_1)}{\tmtwo}{\typctx_1}{[\ty{\mtypetwo}{\mtype}]}
			\qquad
			\tderivtwo_2 \exder[\cbvsym] \tyjp{(\msteps_2, \esteps_2)}{\tmthree}{\typctx_2}{\mtypetwo}
		}
		\end{equation*}
		where $\typctx \defeq \typctx_1 \mplus \typctx_2$ and $\msteps \defeq \msteps_1 + \msteps_2 + 1$ and $\esteps \defeq \esteps_1 + \esteps_2$.
		By \ih applied to $\tderivtwo_1$, there is a  derivation $\tderivtwo\exder[\cbv] \Deri[(\mstepstwo, \estepstwo)]{\typctx_1}{\tmtwo'}{[\ty{\mtypetwo}{\mtype}]}$ where 
		\begin{itemize}
			\item $\mstepstwo \defeq \msteps_1 - 1$ and $\estepstwo \defeq \esteps_1$ if $\Rule = \mcbv$,
			\item $\mstepstwo \defeq \msteps_1$ and $\estepstwo\defeq \esteps_1 - 1$ if $\Rule = \ecbv$.
		\end{itemize}
		Thus, we can construct the following  derivation $\tderiv'$
		\begin{equation*}
		\infer[\app]{
			\tyjp{(\mstepstwo + \msteps_2 +1, \estepstwo + \esteps_2)}{\tm'}{\typctx_1 \mplus \typctx_2}{\mtype}
		}
		{
			\tderivtwo \exder[\cbvsym] \Deri[(\mstepstwo, \estepstwo)]{\typctx_1}{\tmtwo'}{[\ty{\mtypetwo}{\mtype}]}
			\qquad
			\tderivtwo_2 \exder[\cbvsym] \tyjp{(\msteps_2, \esteps_2)}{\tmthree}{\typctx_2}{\mtypetwo}
		}
		\end{equation*}
		where 
		\begin{itemize}
			\item $\mstepstwo +  \msteps_2+1 = \msteps_1 - 1 + \msteps_2 + 1 = \msteps - 1$ and $\estepstwo + \esteps_2 = \esteps_1 + \esteps_2 = \esteps$ if $\Rule = \mcbv$,
			\item $\mstepstwo + \msteps_2 + 1 = \msteps_1 + \msteps_2 + 1 = \msteps$ and $\estepstwo + \esteps_2 = \esteps_1 -1 + \esteps_2 = \esteps - 1$ if $\Rule = \ecbv$.
		\end{itemize}
		
		\item \emph{Application right}, \ie $\tm \defeq \tmthree\tmtwo \torule  \tmthree\tmtwo' \eqdef \tm'$ with $\tmtwo \torule \tmtwo'$ and $\Rule \in \{\mcbv, \ecbv\}$.
		So, $\tderiv$ has the form 
		\begin{equation*}
		\infer[\app]{
			\tyjp{(\msteps_1 + \msteps_2 +1,\esteps_1 + \esteps_2)}{\tm}{\typctx_1 \mplus \typctx_2}{\mtype}
		}
		{
			\tderivtwo_1 \exder[\cbvsym] \tyjp{(\msteps_1,\esteps_1)}{\tmthree}{\typctx_1}{[\ty{\mtypetwo}{\mtype}]}
			\qquad
			\tderivtwo_2 \exder[\cbvsym] \tyjp{(\msteps_2, \esteps_2)}{\tmtwo}{\typctx_2}{\mtypetwo}
		}
		\end{equation*}
		where $\typctx \defeq \typctx_1 \mplus \typctx_2$ and $\msteps \defeq \msteps_1 + \msteps_2 + 1$ and $\esteps \defeq \esteps_1 + \esteps_2$.
		By \ih applied to $\tderivtwo_2$, there is a derivation $\tderivtwo\exder[\cbv] \tyjp{(\mstepstwo, \estepstwo)}{\tmtwo'}{\typctx_2}{\mtypetwo}$ where 
		\begin{itemize}
			\item $\mstepstwo \defeq \msteps_2 - 1$ and $\estepstwo \defeq \esteps_2$ if $\Rule = \mcbv$,
			\item $\mstepstwo \defeq \msteps_2$ and $\estepstwo\defeq \esteps_2 - 1$ if $\Rule = \ecbv$.
		\end{itemize}
		Thus, we can construct the following derivation $\tderiv'$
		\begin{equation*}
		\infer[\app]{
			\tyjp{(\msteps_1 + \mstepstwo + 1, \esteps_1 + \estepstwo)}{\tm'}{\typctx_1 \mplus \typctx_2}{\mtype}
		}
		{
			\tderivtwo_1 \exder[\cbvsym] \tyjp{(\msteps_1, \esteps_1)}{\tmthree}{\typctx_1}{[\ty{\mtypetwo}{\mtype}]}
			\qquad
			\tderivtwo \exder[\cbvsym] \tyjp{(\mstepstwo, \estepstwo)}{\tmtwo'}{\typctx_2}{\mtypetwo}
		}
		\end{equation*}
		where 
		\begin{itemize}
			\item $\msteps_1 + \mstepstwo + 1 = \msteps_1 + \msteps_2 - 1 + 1 = \msteps - 1$ and $\esteps_1 + \estepstwo = \esteps_1 + \esteps_2 = \esteps$ if $\Rule = \mcbv$,
			\item $\msteps_1 + \mstepstwo + 1 = \msteps_1 + \msteps_2 + 1 = \msteps$ and $\esteps_1 + \estepstwo = \esteps_1 + \esteps_2 -1 = \esteps - 1$ if $\Rule = \ecbv$.
		\end{itemize}
		
		\item \emph{Left explicit substitution}, \ie $\tm \defeq \tmtwo\esub\var\tmthree \torule  \tmtwo'\esub\var\tmthree \eqdef \tm'$ with $\tmtwo \torule \tmtwo'$ and $\Rule \in \{\mcbv, \ecbv\}$.
		So, $\tderiv$ has the form 
		\begin{equation*}
		\infer[\esrule]{
			\tyjp{(\msteps_1 + \msteps_2,\esteps_1 + \esteps_2)}{\tm}{\typctx_1 \mplus \typctx_2}{\mtype}
		}
		{
			\tyjp{(\msteps_1,\esteps_1)}{\tmtwo}{\typctx_1, \var \col \mtypetwo}{\mtype}
			\qquad
			\tyjp{(\msteps_2, \esteps_2)}{\tmthree}{\typctx_2}{\mtypetwo}
		}
		\end{equation*}
		where $\typctx \defeq \typctx_1 \mplus \typctx_2$ and $\msteps \defeq \msteps_1 + \msteps_2$ and $\esteps \defeq \esteps_1 + \esteps_2$.
		By \ih applied to $\tderivtwo_1$, there is a derivation $\tderivtwo\exder[\cbv] \tyjp{(\mstepstwo, \estepstwo)}{\tmtwo'}{\typctx_1, \var \col \mtypetwo}{\mtype}$ where 
		\begin{itemize}
			\item $\mstepstwo \defeq \msteps_1 - 1$ and $\estepstwo \defeq \esteps_1$ if $\Rule = \mcbv$,
			\item $\mstepstwo \defeq \msteps_1$ and $\estepstwo\defeq \esteps_1 - 1$ if $\Rule = \ecbv$.
		\end{itemize}
		Thus, we construct the following derivation $\tderiv'$
		\begin{equation*}
		\infer[\esrule]{
			\tyjp{(\mstepstwo + \msteps_2, \estepstwo + \esteps_2)}{\tm'}{\typctx_1 \mplus \typctx_2}{\mtype}
		}
		{
			\tderivtwo \exder[\cbvsym] \tyjp{(\mstepstwo, \estepstwo)}{\tmtwo'}{\typctx_1, \var \col \mtypetwo}{\mtype}
			\qquad
			\tderivtwo_2 \exder[\cbvsym] \tyjp{(\msteps_2, \esteps_2)}{\tmthree}{\typctx_2}{\mtypetwo}
		}
		\end{equation*}
		where 
		\begin{itemize}
			\item $\mstepstwo +  \msteps_2 = \msteps_1 - 1 + \msteps_2 = \msteps - 1$ and $\estepstwo + \esteps_2 = \esteps_1 + \esteps_2 = \esteps$ if $\Rule = \mcbv$,
			\item $\mstepstwo + \msteps_2 = \msteps_1 + \msteps_2 = \msteps$ and $\estepstwo + \esteps_2 = \esteps_1 -1 + \esteps_2 = \esteps - 1$ if $\Rule = \ecbv$.
		\end{itemize}
		
		\item \emph{Right explicit substitution}, \ie $\tm \defeq \tmtwo\esub\var\tmthree \torule  \tmtwo\esub\var{\tmthree'} \eqdef \tm'$ with $\tmthree \torule \tmthree'$ and $\Rule \in \{\mcbv, \ecbv\}$.
		So, $\tderiv$ has the form 
		\begin{equation*}
		\infer[\esrule]{
			\tyjp{(\msteps_1 + \msteps_2,\esteps_1 + \esteps_2)}{\tm}{\typctx_1 \mplus \typctx_2}{\mtype}
		}
		{
			\tyjp{(\msteps_1,\esteps_1)}{\tmtwo}{\typctx_1, \var \col \mtypetwo}{\mtype}
			\qquad
			\tyjp{(\msteps_2, \esteps_2)}{\tmthree}{\typctx_2}{\mtypetwo}
		}
		\end{equation*}
		where $\typctx \defeq \typctx_1 \mplus \typctx_2$ and $\msteps \defeq \msteps_1 + \msteps_2$ and $\esteps \defeq \esteps_1 + \esteps_2$.
		By \ih applied to $\tderivtwo_2$, there is a derivation $\tderivtwo\exder[\cbv] \tyjp{(\mstepstwo, \estepstwo)}{\tmthree'}{\typctx_2}{\mtypetwo}$ where 
		\begin{itemize}
			\item $\mstepstwo \defeq \msteps_2 - 1$ and $\estepstwo \defeq \esteps_2$ if $\Rule = \mcbv$,
			\item $\mstepstwo \defeq \msteps_2$ and $\estepstwo\defeq \esteps_2 - 1$ if $\Rule = \ecbv$.
		\end{itemize}
		Thus, we can construct the following derivation $\tderiv'$
		\begin{equation*}
		\infer[\esrule]{
			\tyjp{(\msteps_1 + \mstepstwo, \esteps_1 + \estepstwo)}{\tm'}{\typctx_1 \mplus \typctx_2}{\mtype}
		}
		{
			\tderivtwo_1 \exder[\cbvsym] \tyjp{(\msteps_1, \esteps_1)}{\tmtwo}{\typctx_1, \var \col \mtypetwo}{\mtype}
			\qquad
			\tderivtwo \exder[\cbvsym] \tyjp{(\mstepstwo, \estepstwo)}{\tmthree'}{\typctx_2}{\mtypetwo}
		}
		\end{equation*}
		where 
		\begin{itemize}
			\item $\msteps_1 +  \mstepstwo = \msteps_1 + \msteps_2 - 1 = \msteps - 1$ and $\esteps_1 + \estepstwo = \esteps_1 + \esteps_2 = \esteps$ if $\Rule = \mcbv$,
			\item $\msteps_1 + \mstepstwo = \msteps_1 + \msteps_2 = \msteps$ and $\esteps_1 + \estepstwo = \esteps_1 + \esteps_2 - 1 = \esteps - 1$ if $\Rule = \ecbv$.
			\qed
		\end{itemize}
	\end{itemize}
\end{proof}

\gettoappendix {prop:value-normal-forms-forall}
\begin{proof}
	By induction on the derivation of $\normalcbvpr{\tm}$.
	Cases:
	\begin{itemize}
		\item \emph{Base}, \ie $\tm \defeq \la{\var}{\tmtwo}$.
		Since $\tderiv\exder[\cbvsym] \Deri[(\msteps,\esteps)]{\typctx}{\tm}{\emptymset}$, the last rule of $\tderiv$ can only be a 0-ary instance of $\many$, thus $\dom{\typctx} = \emptyset$ and $\msteps = \esteps=0$.
		
		\item \emph{Inductive step}, \ie $\tm \defeq \tmtwo \esub{\var}{\tmthree}$ with $\normalcbvpr{\tmtwo}$ and $\normalcbvpr{\tmthree}$.
		Hence, $\tderiv$ has the form
		\begin{equation*}
			\infer[\esrule]{
				\tyjp{(\msteps,\esteps)}{\tm}{\typctx}{\emptymset}
			}
			{
				\tderivtwo \exder[\cbvsym] \tyjp{(\mstepstwo,\estepstwo)}{\tmtwo}{\typctxtwo, \var \col \mtypetwo}{\emptymset}
				\qquad
				\tderivthree \exder[\cbvsym] \tyjp{(\mstepsthree,\estepsthree)}{\tmthree}{\typctxthree}{\mtypetwo}
			}
		\end{equation*}
		where $\typctx \defeq \typctxtwo \mplus \typctxthree$ and $\msteps \defeq \mstepstwo + \mstepsthree$ and $\esteps \defeq \estepstwo + \estepsthree$.
		By \ih applied to $\tderivtwo$, $\dom{\typctxtwo} = \emptyset$ and $\mtypetwo = \emptymset$ and $\mstepstwo = 0 = \estepstwo$.
		By \ih applied to $\tderivthree$ (as $\mtypetwo = \emptymset$), $\dom{\typctxthree} = \emptyset$ and $\mstepsthree = 0 = \estepstwo$.
		Therefore, $\dom{\typctx} = \emptyset$ and $\msteps = 0 = \esteps$.
\qed
	\end{itemize} 
\end{proof}

\gettoappendix {thm:value-correctness}
\begin{proof} 
By induction on $\msteps+\esteps$ and case analysis on whether $\tm$ reduces or not.
	
	If $\normalcbvpr{\tm}$ then the statement holds with $\tmtwo \defeq \tm$ and $\deriv$ the empty evaluation, so that $\sizem{\deriv} = 0 = \sizee{\deriv}$.
	If moreover $\tderiv$ is tight then  $\sizem{\deriv} = 0 = \msteps$ and $\sizee{\deriv} = 0 = \esteps$ by \refprop{value-normal-forms-forall}.
	
	Otherwise $\lnot \normalcbvpr{\tm}$, then $\tm \tocbv \tmthree$ according to the syntactic characterization of closed $\cbvsym$-normal forms (\refprop{syntactic-characterization-closed-normal}), since $\tm$ is closed.
	As $\tm \tocbv \tmthree$ means either $\tm \tomcbv \tmthree$ or $\tm \toecbv \tmthree$,
	by quantitative subject reduction (\refprop{value-subject-reduction}) there is $\tderivtwo \exder[\cbvsym] \!\tyjp{(\mstepstwo,\estepstwo)}{\tmthree}{\typctx}{\mtype}$
 	with:
 	\begin{itemize}
 		\item $\mstepstwo \defeq \msteps - 1$ and $\estepstwo = \esteps$ if $\tm \tomcbv \tmthree$,
 		\item $\mstepstwo \defeq \msteps$ and $\estepstwo = \esteps - 1$ if $\tm \toecbv \tmthree$.
 	\end{itemize}
	By \ih (since $\mstepstwo + \estepstwo = \msteps + \esteps - 1$), there is a term $\tmtwo$ such that $\deriv' \colon \tmthree \tocbvn \tmtwo$ and $\normalcbvpr{\tmtwo}$, with $\sizem{\deriv'} \leq \mstepstwo$ and $\sizee{\deriv'} \leq \estepstwo$; and if, moreover, $\tderivtwo$ is tight, then $\sizem{\deriv'} = \mstepstwo$ and $\sizee{\deriv'} = \estepstwo$.
	The evaluation $\deriv \colon \tm \tocbvn \tmtwo$ obtained by prefixing $\deriv'$ with the step $\tm \tocbv \tmthree$ verifies $\sizem{\deriv} \leq \msteps$ and $\sizee{\deriv} \leq \esteps$ (and $\sizem{\deriv} = \msteps$ and $\sizee{\deriv} = \esteps$ if moreover $\tderiv$---and hence $\tderivtwo$ 
	since $\dom{\typctx} = \emptyset$ and $\mtype = \emptymset$---is tight) because:
	\begin{itemize}
		\item if $\tm \tomcbv \tmthree$ then $\sizem{\deriv} = \sizem{\deriv'} + 1 \leq \mstepstwo + 1 = \msteps$ and $\sizee{\deriv} = \sizee{\deriv'} \leq \estepstwo = \esteps$ (and $\sizem{\deriv} = \sizem{\deriv'} + 1 = \mstepstwo + 1 = \msteps$ and $\sizee{\deriv} = \sizee{\deriv'} = \estepstwo = \esteps$ if moreover $\tderiv$ is tight),
		\item if $\tm \toecbv \tmthree$ then $\sizem{\deriv} = \sizem{\deriv'} \leq \mstepstwo = \msteps$ and $\sizee{\deriv} = \sizee{\deriv'} + 1 \leq \estepstwo + 1 = \esteps$ (and $\sizem{\deriv} = \sizem{\deriv'} = \mstepstwo = \msteps$ and $\sizee{\deriv} = \sizee{\deriv'} + 1 = \estepstwo + 1 = \esteps$ if moreover $\tderiv$ is tight).
\qed
	\end{itemize}
\end{proof}

\subsection{\cbv Completeness}

\gettoappendix{prop:value-normal-forms-exist}
\begin{proof} 
	By induction on $\normalcbvpr\tm$. Cases:
	\begin{itemize}
		\item \emph{Abstraction}: if $\normalcbvpr\tm$ because $\tm = \la\var\tmtwo$ then  $\tderiv$ is given by
		$$\infer[\many]{
			\Deri[(0,0)] {} { \la\var\tmtwo } {\emptymset}
		}{}$$
		\item \emph{Substitution}: if $\normalcbvpr\tm$ because $\tm = \tmtwo\esub\var\tmthree$ with $\normalcbvpr\tmtwo$ and $\normalcbvpr\tmthree$ then by \ih there exist tight derivations $\tderivtwo\exder[\cbvsym] \Deri[(0, 0)]{}{\tmtwo}\emptymset$ and $\tderivthree\exder[\cbvsym] \Deri[(0, 0)]{}{\tmthree}\emptymset$. Then $\tderiv$ is given by:
		$$
		\infer[\esrule]{
			\Deri[(0, 0)] {}
			{\tmtwo \esub\var\tmthree} \emptymset
		}{
			\tderivtwo\exder\Deri[(0, 0)] {} \tmtwo {\emptymset}
			\qquad
			\tderivthree\exder\Deri[(0, 0)] {} \tmthree {\emptymset}        
		}
		$$
		\qed
	\end{itemize}
	
\end{proof}

\begin{lemma}[Typability of values]
	\label{l:value-typability}
	Let $\val$ be a value.
	\begin{enumerate}
		\item \label{p:value-typability-judg}
		\emph{Empty judgement}: There is a derivation $\tderivtwo \exder[\cbvsym] \Deri[(0,0)]{}{\val}{\emptymset}$.
		
		\item \label{p:value-typability-merg}
		\emph{Multi-set merging}: If there are derivations $\tderiv \exder[\cbvsym] \tyjp{(\msteps, \esteps)}{\val}{\typctx}{\mtype}$ and  $\tderivtwo \exder[\cbvsym]\Deri[(\mstepstwo,\estepstwo)]{\typctxtwo}{\val}{\mtypetwo}$, then there is a  derivation 
		${\tderivthree} \exder[\cbvsym] \Deri[(\mstepsthree,\estepsthree)]{\typctx \uplus \typctxtwo}{\val}{\mtype \uplus \mtypetwo}$ with 
		$\mstepsthree = \msteps + \mstepstwo$ and $\estepsthree = \esteps + \estepstwo$.
	\end{enumerate}
\end{lemma}


\gettoappendix{l:value-linear-anti-substitution}
\begin{proof} 
	By induction on $\cbvctx$.
	Cases:
	\begin{itemize}
		\item \emph{Hole}, \ie $\cbvctx \defeq \ctxhole$:
		then, $\cbvctxp{\val} = \val$ and $\cbvctxp{\var} = \var$.
		Since $\var \notin \fv{\val}$, then $\mtype = \emptymset$ according to \reflemma{value-typctx-varocc}.
		Let $\mtype' \defeq \mtypetwo$ and $\typctx' \defeq \typctx$ and $\typctxtwo$ be such that $\dom{\typctxtwo} = \emptyset$: hence, $\typctx  = \typctx' \mplus \typctxtwo$.
		Let $\tderiv' \defeq \tderiv$ and $\tderivtwo$ be the following derivation
		\begin{equation*}
			\infer[\ax]{\Deri[(0,1)]{\var \col \mtypetwo}{\var}{\mtypetwo}}{}
		\end{equation*}
		
		Thus, $\tderivtwo \exder[\cbvsym] \Deri[(\mstepsthree,\estepsthree)]{\typctxtwo, \var \col \mtype \mplus \mtype'}{\cbvctx\cwc{\var}}{\mtypetwo}$ with $(\mstepsthree, \estepsthree) \defeq (0,1)$, because $\typctxtwo, \var \col \mtype \mplus \mtype' = \var \col \mtype' = \var \col \mtypetwo$;
		and $\tderiv' \exder[\cbvsym] \Deri[(\msteps, \esteps)]{\typctx'}{\val}{\mtype}$ because $\typctx, \var \col \mtype = \typctx'$.
		Moreover, $(\msteps, \esteps) = (\msteps + 0, \esteps + 1 - 1) = (\msteps + \mstepsthree, \esteps + \estepsthree - 1)$.
				
		\item \emph{Left application}, \ie $\cbvctx = \cbvctxtwo\tm$:
		then, $\cbvctx\cwc{\var} = \cbvctxtwo\cwc{\var} \tm$ and $\cbvctx\cwc{\val} = \cbvctxtwo\cwc{\val} \tm$.
		So, $\tderiv$ has the form 
		\begin{equation*}
			\infer[\app]{
				\Deri[(\msteps,\esteps)]{\typctx,\var \col \mtype}{\cbvctxtwo\cwc{\val} \tm}{\mtypetwo}
			}{
				\tderiv_1 \exder[\cbvsym]\Deri[(\msteps_1,\esteps_1)]{\typctx_1,\var \col \mtype_1}{\cbvctxtwo\cwc{\val}}{\single{\ty{\mtypethree}{\mtypetwo}}}
				\qquad
				\tderiv_2 \exder[\cbvsym]\Deri[(\msteps_2,\esteps_2)]{\typctx_2,\var \col \mtype_2}{ \tm}{\mtypethree}
			}
		\end{equation*}
		where $\typctx \defeq \typctx_1 \mplus \typctx_2$ and $\mtype \defeq \mtype_1 \mplus \mtype_2$ and $(\msteps, \esteps) \defeq (\msteps_1 +\msteps_2+1, \esteps_1+\esteps_2)$.
		By \ih, there exist a multi type $\mtype'$, two type contexts $\typctx_1'$ and $\typctxtwo_1$, and two derivations $\tderivtwo_1 \exder[\cbvsym] \Deri[(\mstepsthree_1,\estepsthree_1)]{\typctxtwo_1, \var \col \mtype_1 \mplus \mtype'}{\cbvctxtwo\cwc{\var}}{\single{\ty{\mtypethree}{\mtypetwo}}}$ and $\tderiv' \exder[\cbvsym] \Deri[(\mstepstwo,\estepstwo)]{\typctx'}{\val}{\mtype'}$ such that $\typctx_1 =\typctx' \mplus \typctxtwo_1$ ad $(\msteps_1,\esteps_1) = (\mstepstwo+\mstepsthree_1,\estepstwo+\estepsthree_1-1)$.
		We can then construct the following derivation $\tderivtwo$
		\begin{equation*}
			\infer[\app]{
				\Deri[(\mstepsthree_1 + \msteps_2+1,\estepsthree_1 + \esteps_2)]{\typctxtwo_1 \mplus \typctx_2,\var \col \mtype_1 \mplus \mtype' \mplus \mtype_2}{\cbvctxtwo\cwc{\var} \tm}{\mtypetwo}
			}{
				\tderivtwo_1 \exder[\cbvsym] \Deri[(\mstepsthree_1,\estepsthree_1)]{\typctxtwo_1, \var \col \mtype_1 \mplus \mtype'}{\cbvctxtwo\cwc{\var}}{\single{\ty{\mtypethree}{\mtypetwo}}}
				\qquad
				\tderiv_2 \exder[\cbvsym]\Deri[(\msteps_2,\esteps_2)]{\typctx_2,\var \col \mtype_2}{ \tm}{\mtypethree}
			}			
		\end{equation*}
		where $\mtype_1 \mplus \mtype' \mplus \mtype_2 = \mtype \mplus \mtype'$. If we set $\typctxtwo \defeq \typctxtwo_1 \mplus \typctx_2$ and $(\mstepsthree, \estepsthree) \defeq (\mstepsthree_1 + \msteps_2+1,\estepsthree_1 + \esteps_2)$, then
		we have $\typctx' \mplus \typctxtwo = \typctx' \mplus \typctxtwo_1 \mplus \typctx_2  = \typctx_1 \mplus \typctx_2 = \typctx$ and  $(\mstepstwo+\mstepsthree, \estepstwo+\estepsthree - 1) = (\mstepstwo + \mstepsthree_1 + \msteps_2+1, \estepstwo + \estepsthree_1+\esteps_2 - 1) = (\msteps_1 + \msteps_2 + 1, \esteps_1+\esteps_2) = (\msteps, \esteps)$, as required.

		\item \emph{Right application}, \ie $\cbvctx = \tm\cbvctxtwo$:
		then, $\cbvctx\cwc{\var} = \tm\cbvctxtwo\cwc{\var}$ and $\cbvctx\cwc{\val} = \tm\cbvctxtwo\cwc{\val}$.
		So, $\tderiv$ has the form 
		\begin{equation*}
		\infer[\app]{
			\Deri[(\msteps,\esteps)]{\typctx,\var \col \mtype}{\tm\cbvctxtwo\cwc{\val}}{\mtypetwo}
		}{
			\tderiv_1 \exder[\cbvsym]\Deri[(\msteps_1,\esteps_1)]{\typctx_1,\var \col \mtype_1}{\tm}{\single{\ty{\mtypethree}{\mtypetwo}}}
			\qquad
			\tderiv_2 \exder[\cbvsym]\Deri[(\msteps_2,\esteps_2)]{\typctx_2,\var \col \mtype_2}{ \cbvctxtwo\cwc{\val}}{\mtypethree}
		}
		\end{equation*}
		where $\typctx \defeq \typctx_1 \mplus \typctx_2$ and $\mtype \defeq \mtype_1 \mplus \mtype_2$ and $(\msteps, \esteps) \defeq (\msteps_1 +\msteps_2+1, \esteps_1+\esteps_2)$.
		By \ih, there exist a multi type $\mtype'$, two type contexts $\typctx_2'$ and $\typctxtwo_2$, and two derivations $\tderivtwo_2 \exder[\cbvsym] \Deri[(\mstepsthree_2,\estepsthree_2)]{\typctxtwo_2, \var \col \mtype_2 \mplus \mtype'}{\cbvctxtwo\cwc{\var}}{\mtypethree}$ and $\tderiv' \exder[\cbvsym] \Deri[(\mstepstwo,\estepstwo)]{\typctx'}{\val}{\mtype'}$ such that $\typctx_2 =\typctx' \mplus \typctxtwo_2$ ad $(\msteps_2,\esteps_2) = (\mstepstwo+\mstepsthree_2,\estepstwo+\estepsthree_2-1)$.
		We can then construct the following derivation $\tderivtwo$
		\begin{equation*}
		\infer[\app]{
			\Deri[(\mstepsthree_2 + \msteps_1+1,\estepsthree_2 + \esteps_1)]{\typctxtwo_2 \mplus \typctx_1,\var \col \mtype_2 \mplus \mtype' \mplus \mtype_1}{\cbvctxtwo\cwc{\var} \tm}{\mtypetwo}
		}{
			\tderiv_1 \exder[\cbvsym]\Deri[(\msteps_1,\esteps_1)]{\typctx_1,\var \col \mtype_1}{ \tm}{\single{\ty{\mtypethree}{\mtypetwo}}}
			\qquad
			\tderivtwo_2 \exder[\cbvsym] \Deri[(\mstepsthree_2,\estepsthree_2)]{\typctxtwo_2, \var \col \mtype_2 \mplus \mtype'}{\cbvctxtwo\cwc{\var}}{\mtypethree}
		}			
		\end{equation*}
		where $\mtype_1 \mplus \mtype' \mplus \mtype_2 = \mtype \mplus \mtype'$.
		If we set $\typctxtwo \defeq \typctxtwo_2 \mplus \typctx_1$ and $(\mstepsthree, \estepsthree) \defeq (\mstepsthree_2 + \msteps_1+1,\estepsthree_2 + \esteps_1)$, then
		we have $\typctx' \mplus \typctxtwo = \typctx' \mplus \typctxtwo_2 \mplus \typctx_1  = \typctx_1 \mplus \typctx_2 = \typctx$ and  $(\mstepstwo+\mstepsthree, \estepstwo+\estepsthree - 1) = (\mstepstwo + \mstepsthree_2 + \msteps_1+1, \estepstwo + \estepsthree_2+\esteps_1 - 1) = (\msteps_1 + \msteps_2 + 1, \esteps_1+\esteps_2) = (\msteps, \esteps)$, as required.
		
		\item \emph{Left explicit substitution}, \ie $\cbvctx = \cbvctxtwo \esub{\vartwo}\tm$:
		then, $\cbvctx\cwc{\var} = \cbvctxtwo\cwc{\var} \esub{\vartwo}\tm$ and $\cbvctx\cwc{\val} = \cbvctxtwo\cwc{\val} \esub{\vartwo}\tm$ where  $\vartwo \notin \fv{\val} \cup \{\var\}$.
		We can suppose without loss of generality that $\vartwo \notin \fv{\tm}$.
		So, $\tderiv$ has the form 
		\begin{equation*}
		\infer[\esrule]{
			\Deri[(\msteps,\esteps)]{\typctx,\var \col \mtype}{\cbvctxtwo\cwc{\val} \esub{\vartwo} \tm}{\mtypetwo}
		}{
			\tderiv_1 \exder[\cbvsym]\Deri[(\msteps_1,\esteps_1)]{\typctx_1,\var \col \mtype_1, \vartwo \col \mtypethree}{\cbvctxtwo\cwc{\val}}{\mtypetwo}
			\qquad
			\tderiv_2 \exder[\cbvsym]\Deri[(\msteps_2,\esteps_2)]{\typctx_2,\var \col \mtype_2}{ \tm}{\mtypethree}
		}
		\end{equation*}
		where $\typctx \defeq \typctx_1 \mplus \typctx_2$ and $\mtype \defeq \mtype_1 \mplus \mtype_2$ and $(\msteps, \esteps) \defeq (\msteps_1 +\msteps_2, \esteps_1+\esteps_2)$.
		By \ih, there exist a multi type $\mtype'$, two type contexts $\typctx_1'$ and $\typctxtwo_1$, and two derivations $\tderivtwo_1 \exder[\cbvsym] \Deri[(\mstepsthree_1,\estepsthree_1)]{\typctxtwo_1, \var \col \mtype_1 \mplus \mtype', \vartwo \col \mtypethree}{\cbvctxtwo\cwc{\var}}{\mtypetwo}$ and $\tderiv' \exder[\cbvsym] \Deri[(\mstepstwo,\estepstwo)]{\typctx'}{\val}{\mtype'}$ such that $\typctx_1 =\typctx' \mplus \typctxtwo_1$ ad $(\msteps_1,\esteps_1) = (\mstepstwo+\mstepsthree_1,\estepstwo+\estepsthree_1-1)$ (note that $\vartwo \notin \dom{\typctx'}$ because of \reflemma{value-typctx-varocc}, since $\vartwo \notin \fv{\val}$).
		We can then construct the following derivation $\tderivtwo$
		\begin{equation*}
		\infer[\esrule]{
			\Deri[(\mstepsthree_1 + \msteps_2,\estepsthree_1 + \esteps_2)]{\typctxtwo_1 \mplus \typctx_2,\var \col \mtype_1 \mplus \mtype' \mplus \mtype_2}{\cbvctxtwo\cwc{\var} \tm}{\mtypetwo}
		}{
			\tderivtwo_1 \exder[\cbvsym] \Deri[(\mstepsthree_1,\estepsthree_1)]{\typctxtwo_1, \var \col \mtype_1 \mplus \mtype', \vartwo \col \mtypethree}{\cbvctxtwo\cwc{\var}}{\mtypetwo}
			\qquad
			\tderiv_2 \exder[\cbvsym]\Deri[(\msteps_2,\esteps_2)]{\typctx_2,\var \col \mtype_2}{ \tm}{\mtypethree}
		}			
		\end{equation*}
		where $\mtype_1 \mplus \mtype' \mplus \mtype_2 = \mtype \mplus \mtype'$. If we set $\typctxtwo \defeq \typctxtwo_1 \mplus \typctx_2$ and $(\mstepsthree, \estepsthree) \defeq (\mstepsthree_1 + \msteps_2,\estepsthree_1 + \esteps_2)$, then
		we have $\typctx' \mplus \typctxtwo = \typctx' \mplus \typctxtwo_1 \mplus \typctx_2  = \typctx_1 \mplus \typctx_2 = \typctx$ and  $(\mstepstwo+\mstepsthree, \estepstwo+\estepsthree - 1) = (\mstepstwo + \mstepsthree_1 + \msteps_2, \estepstwo + \estepsthree_1+\esteps_2 - 1) = (\msteps_1 + \msteps_2, \esteps_1+\esteps_2) = (\msteps, \esteps)$, as required.
		
		\item \emph{Right explicit substitution}, \ie $\cbvctx = \tm\esub{\vartwo}\cbvctxtwo$:
		then, $\cbvctx\cwc{\var} = \tm\esub{\vartwo}{\cbvctxtwo\cwc{\var}}$ and $\cbvctx\cwc{\val} = \tm\esub{\vartwo}{\cbvctxtwo\cwc{\val}}$.
		So, $\tderiv$ has the form 
		\begin{equation*}
		\infer[\esrule]{
			\Deri[(\msteps,\esteps)]{\typctx,\var \col \mtype}{\tm \esub{\vartwo}{\cbvctxtwo\cwc{\val}}}{\mtypetwo}
		}{
			\tderiv_1 \exder[\cbvsym]\Deri[(\msteps_1,\esteps_1)]{\typctx_1,\var \col \mtype_1, \vartwo \col \mtypethree}{\cbvctxthree\cwc\vartwo}{\mtypetwo}
			\qquad
			\tderiv_2 \exder[\cbvsym]\Deri[(\msteps_2,\esteps_2)]{\typctx_2,\var \col \mtype_2}{ \cbvctxtwo\cwc{\val}}{\mtypethree}
		}
		\end{equation*}
		where $\typctx \defeq \typctx_1 \mplus \typctx_2$ and $\mtype \defeq \mtype_1 \mplus \mtype_2$ and $(\msteps, \esteps) \defeq (\msteps_1 +\msteps_2, \esteps_1+\esteps_2)$.
		By \ih, there exist a multi type $\mtype'$, two type contexts $\typctx_2'$ and $\typctxtwo_2$, and two derivations $\tderivtwo_2 \exder[\cbvsym] \Deri[(\mstepsthree_2,\estepsthree_2)]{\typctxtwo_2, \var \col \mtype_2 \mplus \mtype', \vartwo \col \mtypethree}{\cbvctxtwo\cwc{\var}}{\mtypethree}$ and $\tderiv' \exder[\cbvsym] \Deri[(\mstepstwo,\estepstwo)]{\typctx'}{\val}{\mtype'}$ such that $\typctx_2 =\typctx' \mplus \typctxtwo_2$ ad $(\msteps_2,\esteps_2) = (\mstepstwo+\mstepsthree_2,\estepstwo+\estepsthree_2-1)$ (note that $\vartwo \notin \dom{\typctx'}$ because of \reflemma{value-typctx-varocc}, since $\vartwo \notin \fv{\val}$).
		We can then construct the following derivation $\tderivtwo$
		\begin{equation*}
		\infer[\esrule]{
			\Deri[(\mstepsthree_2 + \msteps_1+1,\estepsthree_2 + \esteps_1)]{\typctxtwo_2 \mplus \typctx_1,\var \col \mtype_2 \mplus \mtype' \mplus \mtype_1}{\tm \esub{\vartwo}{\cbvctxtwo\cwc{\var}}}{\mtypetwo}
		}{
			\tderiv_1 \exder[\cbvsym]\Deri[(\msteps_1,\esteps_1)]{\typctx_1,\var \col \mtype_1, \vartwo \col \mtypethree}{\tm}{\mtypetwo}
			\qquad
			\tderivtwo_2 \exder[\cbvsym] \Deri[(\mstepsthree_2,\estepsthree_2)]{\typctxtwo_2, \var \col \mtype_2 \mplus \mtype'}{\cbvctxtwo\cwc{\var}}{\mtypethree}
		}			
		\end{equation*}
		where $\mtype_1 \mplus \mtype' \mplus \mtype_2 = \mtype \mplus \mtype'$.
		If we set $\typctxtwo \defeq \typctxtwo_2 \mplus \typctx_1$ and $(\mstepsthree, \estepsthree) \defeq (\mstepsthree_2 + \msteps_1,\estepsthree_2 + \esteps_1)$, then
		we have $\typctx' \mplus \typctxtwo = \typctx' \mplus \typctxtwo_2 \mplus \typctx_1  = \typctx_1 \mplus \typctx_2 = \typctx$ and  $(\mstepstwo+\mstepsthree, \estepstwo+\estepsthree - 1) = (\mstepstwo + \mstepsthree_2 + \msteps_1, \estepstwo + \estepsthree_2+\esteps_1 - 1) = (\msteps_1 + \msteps_2, \esteps_1+\esteps_2) = (\msteps, \esteps)$, as required.
		\qed
		\end{itemize}
\end{proof}

\gettoappendix{prop:value-subject-expansion}
\begin{proof}
	By induction on the reduction relation $\tocbv$, with the root rules $\rtom$ and $\rtoecbv$ as the base case, and the closure by $\cbv$ contexts of $\rtocbv \,\defeq\, \rtom \cup \rtoecbv$ as the inductive one.
	\begin{itemize}
		\item \emph{Root step for $\tomcbv\!$} \ie $\tm \defeq \sctxp{\la{\var}{\tmthree}} \tmfour \rtom \sctxp{\tmthree \esub{\var}{\tmfour}} \eqdef \tm'$ where $\sctx \defeq \esub{\vartwo_1}{\tmtwo_1}\dots \esub{\vartwo_n}{\tmtwo_n}$ for some $n \geq 0$.
		We proceed by induction on $n \in \nat$.
		
		If $n = 0$ then $\sctx = \ctxhole$ and so $\tm = \sctxp{\la{\var}\tmthree} \tmfour = (\la{\var}\tmthree) \tmfour$ and $\tm' = \sctxp{\tmthree \esub{\var}{\tmfour}} = \tmthree \esub{\var}{\tmfour}$. 
		Hence, $\tderiv'$ has the form
		$$
		\infer[\esrule]{
			\tyjp{(\msteps'+\msteps'',\esteps'+\esteps'')}{\tmthree\esub{\var}{\tmfour}}{\typctx}{\mtype}}
			{
				\tderivtwo \exder[\cbvsym] \tyjp{(\msteps',\esteps')}{\tmthree}{\typctxtwo, \var \col \mtypethree}{\mtype}
				\qquad
				\tderivthree \exder[\cbvsym] \tyjp{(\msteps'',\esteps'')}{\tmfour}{\typctxthree}{\mtypethree}
			}
		$$
		where $\typctx \defeq \typctxtwo \mplus \typctxthree$ and $\msteps \defeq \mstepstwo + \mstepsthree $ and $\esteps \defeq \estepstwo + \estepsthree$.
		We can then construct the following typing derivation $\tderiv$:
		$$
		\infer[\app]
		{\tyjp{(1 + \msteps' + \msteps'',\esteps' + \esteps'')}{(\la{\var}{\tmthree})\tmfour}{\typctx }{\mtype}}
			{\infer
				[\many]
				{\tyjp{(\msteps',\esteps')}{\la{\var}{\tmthree}}{\typctxtwo}{\mult{\ty{\mtypethree}{\mtype}}}}
				{\infer
					[\fun]
					{\tyjp{(\mstepstwo,\estepstwo)}{\la{\var}{\tmthree}}{\typctxtwo}{\ty{\mtypethree}{\mtype}}}
					{\tderivtwo \exder[\cbvsym]\tyjp{(\mstepstwo,\estepstwo)}{\tmthree}{\typctxtwo, \var \col \mtypethree}{\mtype}}}
			\qquad
			\tderivthree \exder[\cbvsym]{\tyjp{(\mstepsthree,\estepsthree)}{\tmfour}{\typctxthree}{\mtypethree}
			}
		}
		$$			
		where $(1 + \msteps' + \msteps'', \esteps' + \esteps'') = (\msteps + 1, \esteps)$.
		
		Suppose now $n >  0$. 
		Let $\sctxtwo \defeq \esub{\vartwo_1}{\tmtwo_1} \dots \esub{\vartwo_{n-1}}{\tmtwo_{n-1}}$: then, $\tm = \sctxp{\la{\var}\tmthree} \tmfour = \sctxtwop{\la{\var}\tmthree}\esub{\vartwo_n}{\tmtwo_n} \tmfour$ and $\tm' = \sctxp{\tmthree \esub{\var}{\tmfour}} = \sctxptwo{\tmthree \esub{\var}{\tmfour}}\esub{\vartwo_n}{\tmtwo_n}$. 
		Hence, $\tderiv'$ has the form 
		\begin{equation*}
		\infer[\esrule]{
			\Deri[(\msteps,\esteps)]{\typctx}{\sctxptwo{\tmthree \esub{\var}{\tmfour}}\esub{\vartwo_n}{\tmtwo_n}}{\mtype}
		}
		{ 
				\tderivtwo' \exder[\cbvsym] \Deri[(\mstepstwo, \estepstwo)]{\typctx', \vartwo_n \col \mtypetwo_n}{\sctxtwop{\tmthree\esub{\var}{\tmfour}}}{\mtype}
			\qquad
			\tderivtwo_n \exder[\cbvsym] \Deri[(\msteps_n,\esteps_n)]{\typctx_n}{\tmtwo_n}{\mtypetwo_n}
		}
		\end{equation*}
		where $\typctx \defeq \typctx' \mplus \typctx_n$ and $(\msteps, \esteps) \defeq (\mstepstwo + \msteps_n, \estepstwo+\esteps_n)$.
		By \ih applied to $\tderivtwo'$ (since $\sctxtwop{\la{\var}\tmthree}\tmfour \rtom \sctxtwop{\tmthree\esub{\var}{\tmfour}}$), there exists a derivation with conclusion $\Deri[(\mstepstwo+1,\estepstwo)]{\typctx', \vartwo_n \col \mtypetwo_n}{\sctxtwop{\la{\var}\tmthree}\tmfour}{\mtype}$, which necessarily has the form (as $\vartwo_n \notin \dom{\typctx_0'}$ by \reflemma{value-typctx-varocc}, since $\vartwo_n \notin \fv{\tmfour}$)
		\begin{equation*}
		\infer[\esrule]{
		\Deri[(\mstepstwo, \estepstwo)]{\typctx', \vartwo_n \col \mtypetwo_n}{\sctxtwop{\la{\var}\tmthree}{\tmfour}}{\mtype}
					}{
						\tderivtwo \exder[\cbvsym] \Deri[(\mstepsthree,\estepsthree)]{\typctxtwo, \var \col \mtypethree, \vartwo_n \col \mtypetwo_n}{\sctxtwop{\la{\var}\tmthree}}{\mtype}
						\qquad
						\tderivthree \exder[\cbvsym] \Deri[(\mstepstwo_0,\estepstwo_0)]{\typctx_0'}{\tmfour}{\mtypethree}
					}
		\end{equation*}
		where $\typctx' \defeq \typctxtwo \mplus \typctx_0'$ and $(\mstepstwo, \estepstwo) = (\mstepsthree+\mstepstwo_0, \estepsthree+\estepstwo_0)$.
		Therefore, we can construct the following derivation $\tderiv$:
		{\small
		\begin{equation*}
		\infer[\app]{
			\Deri[(\mstepsthree + \msteps_n + \mstepstwo_0 + 1,\estepsthree + \esteps_n + \estepstwo_0)] {\typctxtwo \mplus \typctx_n \mplus \typctx_0'}{\sctxp{\la{\var}\tmthree}{\tmfour}}{\mtype}
		}
		{
			\infer[\esrule]{
				\Deri[(\mstepsthree + \msteps_n, \estepsthree + \esteps_n)]{\typctxtwo \mplus \typctx_n, \var \col \mtypethree}{\sctxp{\la{\var}\tmthree}}{\mtype}
			}{
				\tderivtwo \exder[\cbvsym] \Deri[(\mstepsthree,\estepsthree)]{\typctxtwo, \var \col \mtypethree, \vartwo_n \col \mtypetwo_n}{\sctxtwop{\la{\var}\tmthree}}{\mtype}
				\qquad
				\tderivtwo_n \exder[\cbvsym] \Deri[(\msteps_n,\esteps_n)]{\typctx_n}{\tmtwo_n}{\mtypetwo_n}
			}
			\qquad
			\tderivthree \exder[\cbvsym] \Deri[(\mstepstwo_0,\estepstwo_0)]{\typctx_0'}{\tmfour}{\mtypethree}
		}
		\end{equation*}
	}
	where $ \typctxtwo \mplus \typctx_n \mplus \typctx_0' = \typctx' \mplus \typctx_n = \typctx$	and $(\mstepsthree + \msteps_n + \mstepstwo_0 + 1, \estepsthree + \esteps_n + \estepstwo_0) = (\mstepstwo+\msteps_n+1, \estepstwo+\esteps_n + 1)= (\msteps + 1, \esteps)$.

		\item \emph{Root step for $\!\!\toecbv\!\!\!$} \ie $\tm \defeq \cbvctx\cwc{\var} \esub{\var}{\sctx \hole{\val}} \rtoecbv \sctxp{\cbvctx\cwc{\val} \esub{\var}{\val}} \eqdef \tm'$ with $\sctx \defeq \esub{\vartwo_1}{\tmtwo_1}\dots \esub{\vartwo_n}{\tmtwo_n}$ for some $n \geq 0$. 
		We proceed by induction on $n \in \nat$.
		
		\indent If $n = 0$ then $\sctx = \ctxhole$ and so $\tm = \sctxp{\val} = \val$ and $\tm' = \sctxp{\cbvctx\cwc{\val} \esub{\var}{\val}} = \cbvctx\cwc{\val} \esub{\var}{\val}$. 
		Hence, $\tderiv'$ has the form
		\begin{equation*}
			\infer[\esrule]{
			\tyjp{(\msteps, \esteps)}{\cbvctx\cwc\val\esub{\var}{\val}}{\typctx}{\mtype}
			}
			{
				\tderivtwo_0 \exder[\cbvsym] \Deri[(\msteps_0, \esteps_0)]{\typctx_0, \var \col \mtypetwo}{\cbvctx\cwc{\val}}{\mtype}
					\qquad
				\tderivthree_1 \exder[\cbvsym] \Deri[(\msteps_1,\esteps_1)]{\typctx_1}{\val}{\mtypetwo}
			}
		\end{equation*}
	
		\noindent where $\typctx \defeq \typctx_0 \mplus \typctx_1$, and $\msteps \defeq  \msteps_0 + \msteps_1 $ and $\esteps \defeq \esteps_0 + \esteps_1$.
		By linear removal (\reflemma{value-linear-anti-substitution}), there are a multi type $\mtypetwo'$, two type contexts $\typctx_0'$ and $\typctxtwo$ and two derivations $\tderivthree' \exder[\cbvsym] \Deri[(\mstepstwo_0,\mstepstwo_0)]{\typctx_0'}{\val}{\mtypetwo'}$ and 
		$\tderivtwo \exder[\cbvsym] \Deri[(\mstepsthree,\estepsthree)]{\typctxtwo, \var \col \mtypetwo \mplus \mtypetwo'}{\cbvctx\cwc{\val}}{\mtype}$ 
		such that $\typctx_0 = \typctx_0' \mplus \typctxtwo$ and $(\msteps_0, \esteps_0) = (\mstepstwo_0 + \mstepsthree,  \estepstwo_0 + \estepsthree -1)$.
		Note that $\var \notin \dom{\typctx_0'}$ by \reflemma{value-typctx-varocc}, since $\var \notin \fv{\val}$. 
		By merging of multitypes (\reflemmap{value-typability}{merg}), there is a derivation $\tderivthree \exder[\cbvsym] \Deri[(\mstepstwo_0+\msteps_1,\estepstwo_0+\esteps_1)]{\typctx_0' \mplus \typctx_1 }{\val}{\mtypetwo \mplus \mtypetwo'}$.
		We can construct the following  derivation $\tderiv$:
		$$
		\infer
		[\esrule]
		{\tyjp{(\mstepsthree + \mstepstwo_0 + \msteps_1, \estepsthree + \estepstwo_0 + \esteps_1)}{\cbvctx \cwc{\var} \esub{\var}{\sctxp{\val}}}{\typctxtwo \mplus \typctx_0' \mplus \typctx_1}{\mtype}}
		{\tderivtwo \exder[\cbvsym] \Deri[(\mstepsthree,\estepsthree)]{\typctxtwo, \var \col \mtypetwo \mplus \mtypetwo'}{\cbvctx\cwc{\var}}{\mtype}
			\qquad
			\tderivthree \exder[\cbvsym] \Deri[(\mstepstwo_0+\msteps_1, \esteps_0'+\esteps_1)]{\typctx_0' \mplus \typctx_1}{\val}{\mtypetwo \mplus \mtypetwo'}
		}
		$$
		where $\typctxtwo \mplus \typctx_0' \mplus \typctx_1 = \typctx_0 \mplus \typctx_1 = \typctx$ and $(\mstepsthree+\mstepstwo_0+\msteps_1, \estepsthree+\estepstwo_0+\esteps_1) = (\msteps_0 + \msteps_1, \esteps_0 + 1 + \esteps_1) = (\msteps,\esteps+1)$. 	
		
		\indent Suppose now $n > 0$.
		Let $\sctxtwo \defeq \esub{\vartwo_1}{\tmtwo_1} \dots \esub{\vartwo_{n-1}}{\tmtwo_{n-1}}$: then, $\tm = \sctxp{\val} = \sctxtwop{\val}\esub{\vartwo_n}{\tmtwo_n}$ and $\tm' = \sctxp{\cbvctx\cwc{\val} \esub{\var}{\val}} = \sctxptwo{\cbvctx\cwc{\val} \esub{\var}{\val}}\esub{\vartwo_n}{\tmtwo_n}$. 
		Hence, $\tderiv'$ has the form 
		\begin{equation*}
		\infer[\esrule]{
			\tyjp{(\msteps, \esteps)}{\sctxp{\cbvctx\cwc{\val} \esub{\var}{\val}}}{\typctx}{\mtype}
		}
		{
			\tderivtwo' \exder[\cbvsym] \Deri[(\msteps_0, \esteps_0)]{\typctx_0, \vartwo_n \col \mtypetwo_n}{\sctxtwop{\cbvctx\cwc{\val}\esub{\var}{\val}}}{\mtype}
			\qquad
			\tderivthree_n \exder[\cbvsym] \Deri[(\msteps_n,\esteps_n)]{\typctx_n}{\tmtwo_n}{\mtypetwo_n}
		}
		\end{equation*}
		
		\noindent where $\typctx \defeq \typctx_0 \mplus \typctx_n$ and $(\msteps,\esteps) = (\msteps_0 + \msteps_n, \esteps_0+\esteps_n)$.
		By \ih applied to $\tderivtwo'$ (since $\cbvctx\cwc{\var} \esub{\var}{\sctxtwop{\val}} \rtoecbv \sctxtwop{\cbvctx\cwc{\val}\esub{\var}{\val}}$), there exists a derivation with conclusion $\Deri[(\msteps_0, \esteps_0+1)]{\typctx_0, \vartwo_n \col \mtypetwo_n}{\cbvctx\cwc{\var}\esub{\var}{\sctxtwop{\val}}}{\mtype}$,
		which necessarily has the form (as $\vartwo_n \notin \dom{\typctx_0'}$ by \reflemma{value-typctx-varocc}, since $\vartwo_n \notin \fv{\cbvctx\cwc{\var}}$)
		\begin{equation*}
		\infer[\esrule]{
			\tyjp{(\msteps_0, \esteps_0+1)}{\cbvctx\cwc{\var} \esub{\var}{\sctxtwop{\val}}}{\typctx_0, \vartwo_n \col \mtypetwo_n}{\mtype}
		}
		{
			\tderivtwo \exder[\cbvsym] \Deri[(\mstepstwo_0, \estepstwo_0)]{\typctx_0', \var \col \mtypetwo}{\cbvctx\cwc{\var}}{\mtype}
			\qquad
			\tderivtwo'' \exder[\cbvsym] \Deri[(\mstepsthree_0,\estepsthree_0)]{\vartwo_n \col \mtypetwo_n, \typctx_0''}{\sctxtwop{\val}}{\mtypetwo}
		}
		\end{equation*}
		where $\typctx_0 = \typctx_0' \mplus \typctx_0''$ and $(\msteps_0, \esteps_0+1) = (\mstepstwo_0+ \mstepsthree_0, \estepstwo_0+ \estepsthree_0)$.
		Therefore, we can construct the following derivation $\tderiv$:
		{\small
		\begin{equation*}
		\infer[\esrule]{
			\tyjp{(\mstepstwo_0+\mstepsthree+\msteps_n, \estepstwo_0+\estepsthree+\esteps_n)}{\cbvctx\cwc{\var} \esub{\var}{\sctxp{\val}}}{\typctx_0' \mplus \typctx_0'' \mplus \typctx_n}{\mtype}
		}
		{
			\tderivtwo \exder[\cbvsym] \Deri[(\mstepstwo_0, \estepstwo_0)]{\typctx_0', \var \col \mtypetwo}{\cbvctx\cwc{\var}}{\mtype}
			\quad
			\infer[\esrule]{
				\Deri[(\mstepsthree_0+\msteps_n, \estepsthree_0+\esteps_n)]{\typctx_0'' \mplus \typctx_n}{\sctxp{\val}}{\mtypetwo}
			}
			{ 
				\tderivtwo'' \exder[\cbvsym\!\!] \Deri[(\mstepsthree_0,\estepsthree_0)]{\vartwo_n \col \mtypetwo_n, \typctx_0''}{\sctxtwop{\val}}{\mtypetwo}
				\quad
				\tderivthree_n \exder[\cbvsym\!\!] \Deri[(\msteps_n,\esteps_n)]{\typctx_n}{\tmtwo_n}{\mtypetwo_n}
			}
		}
		\end{equation*}}
		where $\typctx_0' \mplus \typctx_0'' \mplus \typctx_n = \typctx_0 \mplus \typctx_n = \typctx$ and $(\mstepstwo_0+\mstepsthree_0+\msteps_n, \estepstwo_0+\estepsthree_0+\esteps_n) = (\msteps_0+\msteps_n,\esteps_0+1+\esteps_n) = (\msteps, \esteps+1)$. 

		\item \emph{Application left}, \ie $\tm \defeq \tmtwo\tmthree \torule  \tmtwo'\tmthree \eqdef \tm'$ with $\tmtwo \torule \tmtwo'$ and $\Rule \in \{\mcbv, \ecbv\}$.
		So, $\tderiv'$ has the form 
		\begin{equation*}
		\infer[\app]{
			\tyjp{(\mstepstwo + \msteps_2 + 1, \estepstwo + \esteps_2)}{\tm'}{\typctx_1 \mplus \typctx_2}{\mtype}
		}
		{
			\tderivtwo_1 \exder[\cbvsym] \Deri[(\mstepstwo, \estepstwo)]{\typctx_1}{\tmtwo'}{[\ty{\mtypetwo}{\mtype}]}
			\qquad
			\tderivtwo_2 \exder[\cbvsym] \tyjp{(\msteps_2, \esteps_2)}{\tmthree}{\typctx_2}{\mtypetwo}
		}
		\end{equation*}
		where $\typctx \defeq \typctx_1 \mplus \typctx_2$ and $\msteps \defeq \mstepstwo + \msteps_2 + 1$ and $\esteps \defeq \estepstwo + \esteps_2$.
		By \ih applied to $\tderivtwo_1$, there is a derivation $\tderivtwo\exder[\cbvsym] \Deri[(\msteps_1, \esteps_1)]{\typctx_1}{\tmtwo}{[\ty{\mtypetwo}{\mtype}]}$ where 
		\begin{itemize}
			\item $\msteps_1 \defeq \mstepstwo + 1$ and $\esteps_1 \defeq \estepstwo$ if $\Rule = \mcbv$,
			\item $\msteps_1 \defeq \msteps_1$ and $\esteps_1 \defeq \estepstwo + 1$ if $\Rule = \ecbv$.
		\end{itemize}
		Thus, we have the derivation $\tderiv$
		\begin{equation*}
		\infer[\app]{
			\tyjp{(\msteps_1 + \msteps_2 +1,\esteps_1 + \esteps_2)}{\tm}{\typctx_1 \mplus \typctx_2}{\mtype}
		}
		{
			\tderivtwo \exder[\cbvsym] \tyjp{(\msteps_1,\esteps_1)}{\tmtwo}{\typctx_1}{[\ty{\mtypetwo}{\mtype}]}
			\qquad
			\tderivtwo_2 \exder[\cbvsym] \tyjp{(\msteps_2, \esteps_2)}{\tmthree}{\typctx_2}{\mtypetwo}
		}
		\end{equation*}
		where 
		\begin{itemize}
			\item $\msteps_1 +  \msteps_2+1 = \mstepstwo + 1 + \msteps_2 + 1 = \msteps + 1$ and $\esteps_1 + \esteps_2 = \estepstwo + \esteps_2 = \esteps$ if $\Rule = \mcbv$,
			\item $\msteps_1 + \msteps_2 + 1 = \mstepstwo + \msteps_2 + 1 = \msteps$ and $\esteps_1 + \esteps_2 = \estepstwo + 1 + \esteps_2 = \esteps + 1$ if $\Rule = \ecbv$.
		\end{itemize}
		
		\item \emph{Application right}, \ie $\tm \defeq \tmthree\tmtwo \torule  \tmthree\tmtwo' \eqdef \tm'$ with $\tmtwo \torule \tmtwo'$ and $\Rule \in \{\mcbv, \ecbv\}$.
		So, $\tderiv'$ has the form 
		\begin{equation*}
		\infer[\app]{
			\tyjp{(\msteps_1 + \mstepstwo + 1, \esteps_1 + \estepstwo)}{\tm'}{\typctx_1 \mplus \typctx_2}{\mtype}
		}
		{
			\tderivtwo_1 \exder[\cbvsym] \tyjp{(\msteps_1, \esteps_1)}{\tmthree}{\typctx_1}{[\ty{\mtypetwo}{\mtype}]}
			\qquad
			\tderivtwo_2 \exder[\cbvsym] \tyjp{(\mstepstwo, \estepstwo)}{\tmtwo'}{\typctx_2}{\mtypetwo}
		}
		\end{equation*}
		where $\typctx \defeq \typctx_1 \mplus \typctx_2$ and $\msteps \defeq \msteps_1 + \mstepstwo + 1$ and $\esteps \defeq \esteps_1 + \estepstwo$.
		By \ih applied to $\derivtwo_2$, there is a derivation $\tderivtwo\exder[\cbvsym] \tyjp{(\msteps_2, \esteps_2)}{\tmtwo}{\typctx_2}{\mtypetwo}$ where 
		\begin{itemize}
			\item $\msteps_2 \defeq \mstepstwo + 1$ and $\esteps_2 \defeq \estepstwo$ if $\Rule = \mcbv$,
			\item $\msteps_2 \defeq \mstepstwo$ and $\esteps_2 \defeq \estepstwo + 1$ if $\Rule = \ecbv$.
		\end{itemize}
		Thus, we have the derivation $\tderiv$
		\begin{equation*}
		\infer[\app]{
			\tyjp{(\msteps_1 + \msteps_2 +1,\esteps_1 + \esteps_2)}{\tm}{\typctx_1 \mplus \typctx_2}{\mtype}
		}
		{
			\tderivtwo_1 \exder[\cbvsym] \tyjp{(\msteps_1,\esteps_1)}{\tmthree}{\typctx_1}{[\ty{\mtypetwo}{\mtype}]}
			\qquad
			\tderivtwo \exder[\cbvsym] \tyjp{(\msteps_2, \esteps_2)}{\tmtwo}{\typctx_2}{\mtypetwo}
		}
		\end{equation*}
		where 
		\begin{itemize}
			\item $\msteps_1 + \msteps_2 + 1 = \msteps_1 + \mstepstwo + 1 + 1 = \msteps + 1$ and $\esteps_1 + \esteps_2 = \esteps_1 + \estepstwo = \esteps$ if $\Rule = \mcbv$,
			\item $\msteps_1 + \msteps_2 + 1 = \msteps_1 + \mstepstwo + 1 = \msteps$ and $\esteps_1 + \esteps_2 = \esteps_1 + \estepstwo + 1 = \esteps +1 1$ if $\Rule = \ecbv$.
		\end{itemize}
		
		\item \emph{Left explicit substitution}, \ie $\tm \defeq \tmtwo\esub\var\tmthree \torule  \tmtwo'\esub\var\tmthree \eqdef \tm'$ with $\tmtwo \torule \tmtwo'$ and $\Rule \in \{\mcbv, \ecbv\}$.
		So, $\tderiv'$ has the form 
		\begin{equation*}
		\infer[\esrule]{
			\tyjp{(\mstepstwo + \msteps_2, \estepstwo + \esteps_2)}{\tm'}{\typctx_1 \mplus \typctx_2}{\mtype}
		}
		{
			\tderivtwo_1 \exder[\cbvsym] \tyjp{(\mstepstwo, \estepstwo)}{\tmtwo'}{\typctx_1, \var \col \mtypetwo}{\mtype}
			\qquad
			\tderivtwo_2 \exder[\cbvsym] \tyjp{(\msteps_2, \esteps_2)}{\tmthree}{\typctx_2}{\mtypetwo}
		}
		\end{equation*}
		where $\typctx \defeq \typctx_1 \mplus \typctx_2$ and $\msteps \defeq \mstepstwo + \msteps_2$ and $\esteps \defeq \estepstwo + \esteps_2$.
		By \ih applied to $\tderivtwo_1$, there is a derivation $\tderivtwo \exder[\cbv] \tyjp{(\msteps_1, \esteps_1)}{\tmtwo}{\typctx_1, \var \col \mtypetwo}{\mtype}$ where 
		\begin{itemize}
			\item $\msteps_1 \defeq \mstepstwo + 1$ and $\esteps_1 \defeq \estepstwo$ if $\Rule = \mcbv$,
			\item $\msteps_1 \defeq \mstepstwo$ and $\esteps_1 \defeq \estepstwo + 1$ if $\Rule = \ecbv$.
		\end{itemize}
		Thus, we have the derivation $\tderiv$
		\begin{equation*}
		\infer[\esrule]{
			\tyjp{(\msteps_1 + \msteps_2,\esteps_1 + \esteps_2)}{\tm}{\typctx_1 \mplus \typctx_2}{\mtype}
		}
		{
			\tderivtwo \exder[\cbvsym] \tyjp{(\msteps_1,\esteps_1)}{\tmtwo}{\typctx_1, \var \col \mtypetwo}{\mtype}
			\qquad
			\tderivtwo_2 \exder[\cbvsym] \tyjp{(\msteps_2, \esteps_2)}{\tmthree}{\typctx_2}{\mtypetwo}
		}
		\end{equation*}
		where 
		\begin{itemize}
			\item $\msteps_1 +  \msteps_2 = \mstepstwo + 1 + \msteps_2 = \msteps + 1$ and $\esteps_1 + \esteps_2 = \estepstwo + \esteps_2 = \esteps$ if $\Rule = \mcbv$,
			\item $\msteps_1 + \msteps_2 = \mstepstwo + \msteps_2 = \msteps$ and $\esteps_1 + \esteps_2 = \estepstwo + 1 + \esteps_2 = \esteps + 1$ if $\Rule = \ecbv$.
			
		\end{itemize}

		\item \emph{Right explicit substitution}, \ie $\tm \defeq \tmtwo\esub\var\tmthree \torule  \tmtwo\esub\var{\tmthree'} \eqdef \tm'$ with $\tmthree \torule \tmthree'$ and $\Rule \in \{\mcbv, \ecbv\}$.
		So, $\tderiv'$ has the form 
		\begin{equation*}
		\infer[\esrule]{
			\tyjp{(\msteps_1 + \mstepstwo, \esteps_1 + \estepstwo)}{\tm'}{\typctx_1 \mplus \typctx_2}{\mtype}
		}
		{
			\tderivtwo_1 \exder[\cbvsym] \tyjp{(\msteps_1, \esteps_1)}{\tmtwo}{\typctx_1, \var \col \mtypetwo}{\mtype}
			\qquad
			\tderivtwo_2 \exder[\cbvsym] \tyjp{(\mstepstwo, \estepstwo)}{\tmthree'}{\typctx_2}{\mtypetwo}
		}
		\end{equation*}
		where $\typctx \defeq \typctx_1 \mplus \typctx_2$ and $\msteps \defeq \msteps_1 + \mstepstwo$ and $\esteps \defeq \esteps_1 + \estepstwo$.
		By \ih applied to $\tderivtwo_2$, there is a derivation $\tderivtwo \exder[\cbv] \tyjp{(\msteps_2, \esteps_2)}{\tmthree}{\typctx_2}{\mtype}$ where 
		\begin{itemize}
			\item $\msteps_2 \defeq \mstepstwo + 1$ and $\esteps_2 \defeq \estepstwo$ if $\Rule = \mcbv$,
			\item $\msteps_2 \defeq \mstepstwo$ and $\esteps_2 \defeq \estepstwo + 1$ if $\Rule = \ecbv$.
		\end{itemize}
		Thus, we have the derivation $\tderiv$
		\begin{equation*}
		\infer[\esrule]{
			\tyjp{(\msteps_1 + \msteps_2,\esteps_1 + \esteps_2)}{\tm}{\typctx_1 \mplus \typctx_2}{\mtype}
		}
		{
			\tderivtwo_1 \exder[\cbvsym] \tyjp{(\msteps_1,\esteps_1)}{\tmtwo}{\typctx_1, \var \col \mtypetwo}{\mtype}
			\qquad
			\tderivtwo \exder[\cbvsym] \tyjp{(\msteps_2, \esteps_2)}{\tmthree}{\typctx_2}{\mtypetwo}
		}
		\end{equation*}
		where 
		\begin{itemize}
			\item $\msteps_1 +  \msteps_2 = \msteps_1 + \mstepstwo + 1 = \msteps + 1$ and $\esteps_1 + \esteps_2 = \esteps_1 + \estepstwo = \esteps$ if $\Rule = \mcbv$,
			\item $\msteps_1 + \msteps_2 = \msteps_1 + \mstepstwo = \msteps$ and $\esteps_1 + \esteps_2 = \esteps_1 + \estepstwo + 1 = \esteps + 1$ if $\Rule = \ecbv$.
\qed	
		\end{itemize}
		
	\end{itemize}
\end{proof}

\gettoappendix{thm:value-completeness}
\begin{proof}
	By induction on the length $k \defeq \size{\deriv}$ of the evaluation $\deriv \colon \tm \tocbvn \tmtwo$. 
	
	If $k = 0$ then $\tm =  \tmtwo$ and $\normalcbvpr\tm$. 
	\refprop{value-normal-forms-exist}
	gives the existence of  a tight 
	derivation $\tderiv \exder[\cbvup]  \Deri[(0,0)] {} \tm \emptymset$, that satisfies the statement because $\sizem\deriv = \sizee\deriv = 0$. 
	
	If $k > 0$ then
	$\deriv \colon \tm \tocbv \tmthree \rightarrow_\cbvsym^{k-1} \tmtwo$. 
	Let $\deriv'$ be the evaluation $\tmthree \rightarrow_\cbvsym^{k-1} \tmtwo$.
	Thus, if $\tm \tomcbv \tmthree$ then $\sizem{\deriv} = \sizem{\deriv'} + 1$ and $\sizee{\deriv} = \sizee{\deriv'}$; otherwise $\tm \toecbv \tmthree$ then $\sizem{\deriv} = \sizem{\deriv'}$ and $\sizee{\deriv} = \sizee{\deriv'} + 1$
	By \ih, there exists a tight derivation $\tderivtwo\exder[\cbvsym]  \Deri[(\sizem{\deriv'},\sizee{\deriv'})]{}
	\tmthree {\emptymset}$. By quantitative subject expansion (\refprop{value-subject-expansion}),
	there exists a derivation $\tderiv \exder[\cbvsym] \Deri[(\sizem{\deriv},\sizee{\deriv})]{}{\tmthree}{\emptymset}$, in particular $\tderiv$ is tight and with indices $(\sizem\deriv,\sizee\deriv)$.
	\qed
\end{proof}

\section{Types by Need (\refsect{cbneed})}

\subsection{\cbneed Correctness}

\begin{lemma}[Basic properties of derivations in \cbneed]
\label{l:need-basic-properties-typing-derivations}
Let $\tderiv \exder[\cbneed] \tyjp{(\msteps,\esteps)}{\tm}{\typctx}{\mtype}$ be a derivation. Then,
\begin{enumerate}
\item if $\var \notin \fv{\tm}$ then $\var \notin \dom{\typctx} $,
\item if $\tm = \cbneedctx \cwc{\var}$ and $\mtype \neq \emptytype$ then $\var \in \dom{\typctx}$.
\end{enumerate}
\end{lemma}
\begin{proof}
\begin{enumerate}
	\item We prove that $\dom{\typctx} \subseteq \fv{\tm}$ by induction on $\tderiv$:
		\begin{itemize}
			\item Rules $\ax$ and $\normal$ satisfy the statement, as can be observed by a simple analysis of the typing rules.
			\item Rules $\app$, $\appgc$ and $\many$ satisfy the statement by a trivial application of the \ih.
			\item Rule $\fun$: By \ih, $\var \cup \dom{\Gamma} \subseteq \fv{\tm}$, and so $\dom{\Gamma} = \dom{\var : \mtype ; \typctx} \setminus \{\var\} \subseteq \fv{\tm} \setminus \{\var\} = \fv{\la{\var}{\tm}}$.
			\item \emph{Rule $\ES$}: We first apply \ih on the left-hand side premise to obtain that $\dom{\var : \mtype ; \typctx} \subseteq \fv{\tm}$, which in turn implies that $\dom{\typctx} \subseteq \fv{\tm} \setminus \{\var\}$. We then apply \ih on the right-hand side premise to obtain that $\dom{\typctxtwo} \subseteq \fv{\tmtwo} $. Hence, $\dom{\typctx} \cup \dom{\typctxtwo} \subseteq (\fv{\tm} \setminus \{\var\}) \cup \fv{\tmtwo} = \fv{\tm \esub{\var}{\tmtwo}}$.
			\item \emph{Rule $\ESgc$}: By \ih, $\dom{\typctx} \subseteq \fv{\tm}$. But $\var \notin \dom{\typctx}$, and so $\dom{\typctx} = \dom{\typctx} \setminus \{\var\} \subseteq \fv{\tm} \setminus \{\var\} \subseteq (\fv{\tm} \setminus \{\var\}) \cup \fv{\tmtwo} = \fv{\tm \esub{\var}{\tmtwo}}$.
		\end{itemize}
	\item By induction on the construction of $\cbneedctx$:
		\begin{itemize}
		\item Let $\cbneedctx = \ctxhole$. Then $\tm = \var$ and so $\tderiv$ is of the form 
			$$
			\infer
				[\ax]
				{\tyjp{(\msteps,\esteps)}{\var}{\var : \mtype}{\mtype}}
				{}
			$$
		Clearly, if $\mtype \neq \emptytype$ then $\var \in \dom{\typctx}$.	
		
		\item If $\cbneedctx = \cbneedctx_{1} \tmtwo$, then $\tm = \cbneedctx_{1} \cwc{\var} \tmtwo$. Then $\tderiv$ can only have either $\app$ or $\appgc$ as the last typing rule. 
			\begin{itemize}
			\item Let $\tderiv$ be of the form
				$$
				\infer
					[\app]
					{\tyjp{(\msteps' + \msteps'' + 1, \esteps' + \esteps'')}{\cbneedctx_{1} \cwc{\var} \tmtwo}{\typctxtwo \mplus \typctxthree}{\mtype}}
					{\tderiv_{\cbneedctx_{1} \cwc{\var}} \exder[\cbneed] \tyjp{(\msteps', \esteps')}{\cbneedctx_{1} \cwc{\var}}{\typctxtwo}{\mult{\ty{\mtypetwo}{\mtype}}}
					\quad
					\tyjp{(\msteps'', \esteps'')}{\tmtwo}{\typctxthree}{\mtypetwo}}
				$$
			Then by application of \ih on $\tderiv_{\cbneedctx_{1} \cwc{\var}}$ we obtain that $\var \in \dom{\typctxtwo}$, finally obtaining $\var \in \dom{\typctxtwo \mplus \typctxthree}$.
			\item If $\tderiv$ has $\appgc$ as its last typing rule instead, then $\var \in \dom{\typctx}$ simply by \ih.
		\end{itemize}
		\item If $\cbneedctx = \cbneedctx_{1} \esub{\vartwo}{\tmtwo}$, then $\tm = \cbneedctx \cwc{\var} = \cbneedctx_{1} \cwc{\var} \esub{\vartwo}{\tmtwo}$, with $\var \neq \vartwo$. Then $\tderiv$ can only have either $\ES$ or $\ESgc$ as the last typing rule.
		\begin{itemize}
			\item Let $\tderiv$ be of the form
				$$
				\infer
					[\ES]
					{\tyjp{(\msteps' + \esteps'', \esteps' + \esteps'')}{\cbneedctx_{1} \cwc{\var} \esub{\vartwo}{\tmtwo}}{\typctxtwo \mplus \typctxthree}{\mtype}}
					{\tderiv_{\cbneedctx_{1} \cwc{\var}} \exder[\cbneed] \tyjp{(\msteps', \esteps')}{\cbneedctx_{1} \cwc{\var}}{\vartwo : \mtypetwo ; \typctxtwo}{\mtype}
					\quad
					\tyjp{(\msteps'', \esteps'')}{\tmtwo}{\typctxthree}{\mtypetwo}}
				$$
			Then by application of \ih on $\tderiv_{\cbneedctx_{1} \cwc{\var}}$ we obtain that $\var \in \dom{\vartwo : \mtypetwo ; \typctxtwo}$, which implies that $\var \in \dom{\typctxtwo}$ and so $\var \in \dom{\typctxtwo \mplus \typctxthree}$.
			
			\item If $\tderiv$ has $\ESgc$ as its last typing rule instead, then $\var \in \dom{\typctx}$ simply by \ih.
		\end{itemize}
		\item Let $\cbneedctx = \cbneedctx_{1} \cwc{\vartwo} \esub{\vartwo}{\cbneedctx_{2}}$, and so $\tm = \cbneedctx \cwc{\var} = \cbneedctx_{1} \cwc{\vartwo} \esub{\vartwo}{\cbneedctx_{2} \cwc{\var}}$, with $\var \neq \vartwo$ because $\var$ is a free variable of $\tm$ while $\vartwo$ is a bound variable of $\tm$, and we are working up to $\alpha$-equivalence. Suppose now that $\ESgc$ was the last typing rule of $\tderiv$. This means that $\tderiv$ is of the form
			$$
			\infer
				[\ESgc]
				{\tyjp{(\msteps,\esteps)}{\cbneedctx_{1} \cwc{\vartwo} \esub{\vartwo}{\cbneedctx_{2} \cwc{\var}}}{\typctx}{\mtype}}
				{\tderiv_{\cbneedctx_{1} \cwc{\vartwo}} \exder[\cbneed] \tyjp{(\msteps,\esteps)}{\cbneedctx_{1} \cwc{\vartwo}}{\typctx}{\mtype}
				\quad
				\vartwo \notin \dom{\typctx}}
			$$
		However, by applying \ih on $\tderiv_{\cbneedctx_{1} \cwc{\vartwo}}$ we obtain that $\vartwo \in \dom{\typctx}$, which is in contradiction with the constraint of rule $\ESgc$.
		
		Hence, $\tderiv$ can only have $\ES$ as the last typing rule. Thus, $\tderiv$ is of the form
			$$
			\infer
				[\ES]
				{\tyjp{(\msteps' + \msteps'', \esteps' + \esteps'')}{\cbneedctx_{1} \cwc{\vartwo} \esub{\vartwo}{\cbneedctx_{2} \cwc{\var}}}{\typctxtwo \mplus \typctxthree}{\mtype}}
				{\tyjp{(\msteps',\esteps')}{\cbneedctx_{1} \cwc{\vartwo}}{\vartwo : \mtypetwo ; \typctxtwo}{\mtype}
				\quad
				\tderiv_{\cbneedctx_{2} \cwc{\var}} \exder[\cbneed] \tyjp{(\msteps'',\esteps'')}{\cbneedctx_{2} \cwc{\var}}{\typctxthree}{\mtypetwo}
				\quad
				\mtypetwo \neq \emptytype}
			$$
			
		We can then apply \ih on $\tderiv_{\cbneedctx_{2} \cwc{\var}}$ to obtain that $\var \in \dom{\typctxthree}$ and so $\var \in \dom{\typctxtwo \mplus \typctxthree}$. \qed
		\end{itemize}
\end{enumerate}
\end{proof}

\begin{lemma}[Splitting of multi-sets with respect to derivations]
\label{l:need-splitting-multisets}
Let $\val$ be a value and $\tderiv \exder[\cbneed] \tyjp{(\msteps,\esteps)}{\val}{\typctx}{\mtype}$ be a  derivation such that $\size{\mtype} \geq 2$. 
Then, for every splitting $\mtype = \mtypetwo \mplus \mtypethree$ such that $\size{\mtypetwo}, \size{\mtypethree} \geq 1$ there are type contexts $\typctx_{\mtypetwo}$ and $\typctx_{\mtypethree} $ and  derivations $\tderiv_{\mtypetwo} \exder[\cbneed] \tyjp{(\msteps_{\mtypetwo}, \esteps_{\mtypetwo})}{\val}{\typctx_{\mtypetwo}}{\mtypetwo} $ and $\tderiv_{\mtypethree} \exder[\cbneed] \tyjp{(\msteps_{\mtypethree}, \esteps_{\mtypethree})}{\val}{\typctx_{\mtypethree}}{\mtypethree} $ such that 
	\begin{itemize}
	\item $\typctx = \typctx_{\mtypetwo} \mplus \typctx_{\mtypethree}$,
	\item $\msteps = \msteps_{\mtypetwo} + \msteps_{\mtypethree} $, and
	\item $\esteps = \esteps_{\mtypetwo} + \esteps_{\mtypethree} $.
	\end{itemize}
\end{lemma}
\begin{proof}
By a simple observation of the typing rules with an abstraction as the term of the final type judgement, we note that $\tderiv$ can only be of the form
	$$
	\infer
		[\many]
		{\tyjp{(\sum_{\iI} \msteps_{i}, \sum_{\iI} \esteps_{i})}{\val}{\bigmplus_{\iI} \typctx_{i}}{\mult{\type_{i}}_{\iI}}}
		{(\tyjp{(\msteps_{i}, \esteps_{i})}{\val}{\typctx_{i}}{\type_{i}})_{\iI}}
	$$
We then appropriately define $\J = \{\iI : \type_{i} \in \mtypetwo \} $ and $\K = \{\iI : \type_{i} \in \mtypethree\}$; \ie, $\mult{\type_{j}}_{\jJ} = \mtypetwo$, $\mult{\type_{k}}_{\kK} = \mtypethree $, and making sure that $\J \cap \K = \emptyset$. Note that $\J, \K \neq \emptyset$, since $\mtypetwo, \mtypethree \neq \emptytype$. Thus, we obtain $\tderiv_{\mtypetwo}$ as
	$$
	\infer
		[\many]
		{\tyjp{(\sum_{\jJ} \msteps_{j}, \sum_{\jJ} \esteps_{j})}{\val}{\bigmplus_{\jJ} \typctx_{j}}{\mult{\type_{j}}_{\jJ}}}
		{(\tyjp{(\msteps_{j}, \esteps_{j})}{\val}{\typctx_{j}}{\type_{j}})_{\jJ}}
	$$
and $\tderiv_{\mtypethree}$ as 
	$$
	\infer
		[\many]
		{\tyjp{(\sum_{\kK} \msteps_{k}, \sum_{\kK} \esteps_{k})}{\val}{\bigmplus_{\kK} \typctx_{k}}{\mult{\type_{k}}_{\kK}}}
		{(\tyjp{(\msteps_{k}, \esteps_{k})}{\val}{\typctx_{k}}{\type_{k}})_{\kK}}
	$$
where $\bigmplus_{\jJ} \typctx_{j} = \typctx_{\mtypetwo}$, $\bigmplus_{\kK} \typctx_{k} = \typctx_{\mtypethree}$, $(\sum_{\jJ} \msteps_{j}, \sum_{\jJ} \esteps_{j}) = (\msteps_{\mtypetwo}, \esteps_{\mtypetwo}) $, and $(\sum_{\kK} \msteps_{k}, \sum_{\kK} \esteps_{k}) = (\msteps_{\mtypethree}, \esteps_{\mtypethree}) $.
\qed
\end{proof}

\gettoappendix {l:need-linear-substitution}
\begin{proof}
We prove this by induction on $N$:
	\begin{itemize}
	\item \emph{Empty context, \ie $N = \ctxhole$}. Then $\typctx = \emptyset$, $\mtypethree = \mtype$, and $\tderiv_{\cbneedctx \cwc{\var}}$ is of the form
		$$
		\infer
			[\ax]
			{\tyjp{(0,1)}{\var}{\var \colon \mtype}{\mtype}}
			{}
		$$
	Therefore, by defining $\mtype_{1} \defeq \mtype$ and $\mtype_{2} \defeq \zero$, the statement holds for every $\tderivtwo \exder[\cbneed] \tyjp{(\msteps', \esteps')}{\val}{\typctxtwo}{\mtype_{1}} $ by taking $\tderiv_{\cbneedctx\cwc{\val}} \defeq \tderivtwo $. In particular, note that $(\msteps + \mstepstwo, \esteps + \estepstwo - 1) = (0 + \msteps', 1 + \estepstwo - 1) = (\mstepstwo, \estepstwo)$.

	\item \emph{Left of an application, \ie $N = N_{1} \tmtwo$}. There are two possible last rules in $\tderiv_{\cbneedctx \cwc{\var}}$, namely $\appsteps$ or $\appgc$.
	\begin{itemize}
	\item Let $\tderiv_{\cbneedctx \cwc{\var}}$ be of the form 
			$$
			\infer
				[\appsteps]
				{\tyjp{(\msteps_{\typctxthree} + \msteps_{\typctxfour} + 1, \esteps_{\typctxthree} + \esteps_{\typctxfour})}{\cbneedctx_{1} \cwc{\var} \tmtwo}{\var \colon (\mtype_{\typctxthree} \mplus \mtype_{\typctxfour}) ; (\typctxthree \mplus \typctxfour)}{\mtypethree}}
				{\tyjp{(\msteps_{\typctxthree},\esteps_{\typctxthree})}{N_{1} \cwc{\var}}{\var \colon \mtype_{\typctxthree} ; \typctxthree}{\mult{\ty{\mtypethree'}{\mtypethree}}}
				\quad
				\tyjp{(\msteps_{\typctxfour}, \esteps_{\typctxfour})}{\tmtwo}{\var \colon \mtype_{\typctxfour} ; \typctxfour}{\mtypethree'}}
			$$
where $\typctx = \typctxthree \mplus \typctxfour$, $\typctxthree(\var) = \typctxfour(\var) = \emptytype$ and $\mtype = \mtype_{\typctxthree} \mplus \mtype_{\typctxfour} $.

	By applying the \ih on the left-hand side premise we obtain that there exists a splitting $\mtype_{\typctxthree} = \mtype_{\typctxthree, 1} \mplus \mtype_{\typctxthree, 2}$, with $\mtype_{\typctxthree, 1} \neq \emptytype$, such that for every derivation $\tderivtwo \exder_{\cbneed} \tyjp{(\mstepstwo, \estepstwo)}{\val}{\typctxtwo}{\mtype_{\typctxthree, 1}}$ there exists a  derivation $\tderiv_{\cbneedctx_{1} \cwc{\val}} \exder[\cbneed] \tyjp{(\msteps_{\typctxthree} + \mstepstwo, \esteps_{\typctxthree} + \estepstwo -1)}{\cbneedctx_{1} {\cwc{\val}}}{\var \colon \mtype_{\typctxthree, 2} ; \typctxthree \mplus \typctxtwo}{\mult{\ty{\mtypethree'}{\mtypethree}}} $. We can then construct $\tderiv_{\cbneedctx \cwc{\val}} $ for such a $\tderivtwo$ as follows
		$$
		\infer
			[\appsteps]
			{\tyjp{(\msteps_{\typctxthree} + \mstepstwo + \msteps_{\typctxfour} + 1, \esteps_{\typctxthree} + \estepstwo - 1 + \esteps_{\typctxfour})}{\cbneedctx_{1} \cwc{\val} \tmtwo}{\var \colon \mtype_{\typctxthree, 2} \mplus \mtype_{\typctxfour} ; \typctxthree \mplus \typctxtwo \mplus \typctxfour}{\mtypethree}}
			{\tyjp{(\msteps_{\typctxthree} + \mstepstwo, \esteps_{\typctxthree} + \estepstwo -1)}{\cbneedctx_{1} {\cwc{\val}}}{\var \colon \mtype_{\typctxthree, 2} ; \typctxthree \mplus \typctxtwo}{\mult{\ty{\mtypethree'}{\mtypethree}}}
			\quad
			\tyjp{(\msteps_{\typctxfour}, \esteps_{\typctxfour})}{\tmtwo}{\var \colon \mtype_{\typctxfour} ; \typctxfour}{\mtypethree'}}
		$$
Note that $\tderiv_{\cbneedctx \cwc{\val}}$ is as desired by splitting $\mtype$ into $\mtype_{1} \defeq \mtype_{\typctxthree,1}$ and $\mtype_{2} \defeq \mtype_{\typctxthree, 2} \mplus \mtype_{\typctxfour}$.
	
	\item Let $\tderiv_{\cbneedctx \cwc{\var}}$ be of the form
		$$
		\infer
			[\appgc]
			{\tyjp{(\msteps,\esteps)}{\cbneedctx_{1} \cwc{\var} \tmtwo}{\var \colon \mtype ; \typctx}{\mtypethree}}
			{\tderiv_{\cbneedctx_{1} \cwc{\var}} \exder[\cbneed] \tyjp{(\msteps - 1, \esteps)}{\cbneedctx_{1} \cwc{\var}}{\var \colon \mtype ; \typctx}{\mult{\ty{\emptytype}{\mtypethree}}}}
		$$
		
	By applying the \ih on $\tderiv_{\cbneedctx_{1} \cwc{\var}} $ we obtain that there exists a splitting $\mtype = \mtype_{1} \mplus \mtype_{2}$, with $\mtype_{1} \neq \emptytype$, such that for every  derivation $\tderivtwo \exder[\cbneed] \tyjp{(\mstepstwo, \estepstwo)}{\val}{\typctxtwo}{\mtype_{1}}$ there exists
a  derivation	$\tderiv_{\cbneedctx_{1} \cwc{\val}} \exder[\cbneed] \tyjp{(\msteps - 1 + \mstepstwo, \esteps + \estepstwo - 1)}{\cbneedctx_{1} \cwc{\val}}{\var \colon \mtype_{2}; \typctx \mplus \typctxtwo}{\mult{\ty{\emptytype}{\mtypethree}}}$. We can then construct $\tderiv_{\cbneedctx \cwc{\val}}$ for such a $\tderivtwo$ as follows
	$$
	\infer
		[\appgc]
		{\tyjp{(\msteps + \mstepstwo, \esteps + \estepstwo - 1)}{\cbneedctx_{1} \cwc{\val} \tmtwo}{\var \colon \mtype_{2} ; \typctx \mplus \typctxtwo}{\mtypethree}}
		{\tyjp{(\msteps - 1 + \mstepstwo, \esteps + \estepstwo - 1)}{\cbneedctx_{1} \cwc{\val}}{\var \colon \mtype_{2} ; \typctx \mplus \typctxtwo}{\mult{\ty{\emptytype}{\mtypethree}}}}
	$$
	\end{itemize}
	
	\item \emph{Left of a substitution; \ie $\cbneedctx = \cbneedctx_{1} \esub{\vartwo}{\tmtwo}$}. Note that $\var \neq \vartwo$, because the hypothesis $\cbneedctx \cwc{\var}$ impies that $\cbneedctx$ does not capture $\var$. There are two possible last rules in $\tderiv_{\cbneedctx \cwc{\var}}$, namely $\ES$ and $\ESgc$.
	\begin{itemize}
	\item Let $\tderiv_{\cbneedctx \cwc{\var}}$ be of the form
		$$
		\infer
			[\ES]
			{\tyjp{(\msteps_{\typctxthree} + \msteps_{\typctxfour}, \esteps_{\typctxthree} + \esteps_{\typctxfour})}{\cbneedctx_{1}\cwc{\var} \esub{\vartwo}{\tmtwo}}{\var \colon (\mtype_{\typctxthree} \mplus \mtype_{\typctxfour}) ; (\typctxthree \sm \vartwo) \mplus \typctxfour}{\mtypethree}}
			{\tyjp{(\msteps_{\typctxthree},\esteps_{\typctxthree})}{\cbneedctx_{1}\cwc{\var}}{\var \colon \mtype_{\typctxthree}; \typctxthree}{\mtypethree}
			\quad
			\tyjp{(\msteps_{\typctxfour},\esteps_{\typctxfour})}{\tmtwo}{\var \colon \mtype_{\typctxfour}; \typctxfour}{\typctxthree(\vartwo)}
			\quad
			\typctxthree(\vartwo) \neq \emptytype}
		$$
where $\mtype = \mtype_{\typctxthree} \mplus \mtype_{\typctxfour} $ and $\typctx = (\typctxthree \sm \vartwo) \mplus \typctxfour$.

	By applying the \ih on the leftmost premise we obtain a splitting $\mtype_{\typctxthree} = \mtype_{\typctxthree,1} \mplus \mtype_{\typctxthree,2}$, with $\mtype_{\typctxthree,1} \neq \emptytype $, such that for every  derivation $\tderivtwo \exder[\cbneed] \tyjp{(\mstepstwo, \estepstwo)}{\val}{\typctxtwo}{\mtype_{\typctxthree,1}}$ there exists a  derivation $\tderiv_{\cbneedctx_{1} \cwc{\val}} \exder[\cbneed] \tyjp{(\msteps_{\typctxthree} + \mstepstwo, \esteps_{\typctxthree} + \estepstwo - 1)}{\cbneedctx_{1} \cwc{\val}}{\var \colon \mtype_{\typctxthree,2} ; \typctxthree \mplus \typctxtwo}{\mtypethree} $. Note however that if $\vartwo \in \dom{\typctxtwo}$, then \reflemma{need-basic-properties-typing-derivations} applied on $\tderivtwo$ would imply that $\vartwo \in \fv{\val}$, which contradicts the hypothesis that $\cbneedctx$ does not capture the free variables of $\val$; \ie, $\vartwo \notin \dom{\typctxtwo}$, and so $(\typctxtwo \mplus \typctxthree)(\vartwo) = \typctxthree(\vartwo) $. We can then construct $\tderiv_{\cbneedctx \cwc{\val}}$ for such a $\tderivtwo$ as follows
	{\scriptsize
	$$
	\infer
		[\ES]
		{\tyjp{(\msteps_{\typctxthree} + \mstepstwo + \msteps_{\typctxfour}, \esteps_{\typctxthree} + \estepstwo - 1 + \esteps_{\typctxfour})}{\cbneedctx_{1} \cwc{\val} \esub{\vartwo}{\tmtwo}}{\var \colon \mtype_{\typctxthree,2} \mplus \mtype_{\typctxfour} ; ((\typctxthree \mplus \typctxtwo) \sm \vartwo) \mplus \typctxfour}{\mtypethree}}
		{\tyjp{(\msteps_{\typctxthree} + \mstepstwo, \esteps_{\typctxthree} + \estepstwo - 1)}{\cbneedctx_{1} \cwc{\val}}{\var \colon \mtype_{\typctxthree,2} ; \typctxthree \mplus \typctxtwo}{\mtypethree}
		\quad
		\tyjp{(\msteps_{\typctxfour},\esteps_{\typctxfour})}{\tmtwo}{\var \colon \mtype_{\typctxfour}; \typctxfour}{\typctxthree(\vartwo)}
		\quad
		\typctxthree(\vartwo) \neq \emptytype}
	$$
	}

by splitting $\mtype$ into $\mtype_{1} \defeq \mtype_{\typctxthree,1}$ and $\mtype_{2} \defeq  \mtype_{\typctxthree,2} \mplus \mtype_{\typctxfour}$. Since $\vartwo \notin \dom{\typctxtwo}$, then $((\typctxthree \mplus \typctxtwo) \sm \vartwo) \mplus \typctxfour = (\typctxthree \sm \vartwo) \mplus \typctxtwo \mplus \typctxfour = \typctx \mplus \typctxtwo$. 
		
	\item Let $\tderiv_{\cbneedctx \cwc{\var}}$ be of the form
		$$
		\infer
			[\ESgc]
			{\tyjp{(\msteps,\esteps)}{\cbneedctx_{1} \cwc{\var} \esub{\vartwo}{\tmtwo}}{\var \colon \mtype ; \typctx}{\mtypethree}}
			{\tyjp{(\msteps, \esteps)}{\cbneedctx_{1} \cwc{\var}}{\var \colon \mtype ; \typctx}{\mtypethree}
			\quad
			\typctx(\vartwo) = \emptytype}
		$$
		
	By applying \ih on the premise we obtain a splitting $\mtype = \mtype_{1} \mplus \mtype_{2}$, with $\mtype_{1} \neq \emptytype$, such that for every  derivation $\tderivtwo \exder[\cbneed] \tyjp{(\mstepstwo, \estepstwo)}{\val}{\typctxtwo}{\mtype_{1}}$ there exists a derivation $\tderiv_{\cbneedctx_{1} \cwc{\val}} \exder[\cbneed] \tyjp{(\msteps + \mstepstwo, \esteps + \estepstwo - 1)}{\cbneedctx_{1} \cwc{\val}}{\var \colon \mtype_{2} ; \typctx \mplus \typctxtwo}{\mtypethree}$. Note that $\vartwo \notin \dom{\typctxtwo}$, because applying \reflemma{need-basic-properties-typing-derivations} on $\tderivtwo$ would otherwise imply that $\vartwo \in \fv{\val}$, which contradicts the hypothesis that $\cbneedctx$ does not capture the free variables of $\val$. 
	Hence, we can then construct $\tderiv_{\cbneedctx \cwc{\val}}$ for such a $\tderivtwo$ as follows
		$$
		\infer
			[\ESgc]
			{\tyjp{(\msteps + \mstepstwo, \esteps + \estepstwo - 1)}{\cbneedctx_{1} \cwc{\val} \esub{\vartwo}{\tmtwo}}{\var \colon \mtype_{2} ; \typctx \mplus \typctxtwo}{\mtypethree}}
			{\tyjp{(\msteps + \mstepstwo, \esteps + \estepstwo - 1)}{\cbneedctx \cwc{\val}}{\var \colon \mtype_{2} ; \typctx \mplus \typctxtwo}{\mtypethree}
			\quad
			(\typctx \mplus \typctxtwo)(\vartwo) = \typctx(\vartwo) = \emptytype}
		$$
	\end{itemize}
	\item Let $\cbneedctx = \cbneedctx_{1} \cwc{\vartwo} \esub{\vartwo}{\cbneedctx_{2}}$. We can safely assume that $\var \neq \vartwo$, since we are working up to $\alpha$-equivalence. \reflemma{need-basic-properties-typing-derivations} implies $\vartwo \in \dom{\var \colon \mtype ; \typctx}$, and so $\tderiv_{\cbneedctx \cwc{\var}}$ can only have $\ES$ as the final type judgement and be of the form
		$$
		\infer
			[\ES]
			{\tyjp{(\msteps_{\typctxthree} + \msteps_{\typctxfour},\esteps_{\typctxthree} + \esteps_{\typctxfour})}{\cbneedctx_{1} \cwc{\vartwo} \esub{\vartwo}{\cbneedctx_{2} \cwc{\var}}}{\var \colon (\mtype_{\typctxthree} \mplus \mtype_{\typctxfour}) ; (\typctxthree \sm \vartwo) \mplus \typctxfour}{\mtypethree}}
			{\tyjp{(\msteps_{\typctxthree}, \esteps_{\typctxthree})}{\cbneedctx_{1} \cwc{\vartwo}}{\var \colon \mtype_{\typctxthree} ; \typctxthree}{\mtypethree}
			\quad
			\tyjp{(\msteps_{\typctxfour}, \esteps_{\typctxfour})}{\cbneedctx_{2} \cwc{\var}}{\var \colon \mtype_{\typctxfour} ; \typctxfour}{\typctxthree(\vartwo)}
			\quad
			\typctxthree(\vartwo) \neq \emptytype}
		$$
		
 \noindent where $\mtype = \mtype_{\typctxthree} \mplus \mtype_{\typctxfour} $, $\typctx =  (\typctxthree \sm \vartwo) \mplus \typctxfour$, and $(\msteps,\esteps) = (\msteps_{\typctxthree} + \msteps_{\typctxfour},\esteps_{\typctxthree} + \esteps_{\typctxfour})$. 

	We can then apply the \ih on the premise in the middle to obtain a splitting $\mtype_{\typctxfour} = \mtype_{\typctxfour,1} \mplus \mtype_{\typctxfour,2} $, with $\mtype_{\typctxfour,1} \neq \emptytype$, such that for every  derivation $\tderivtwo \exder[\cbneed] \tyjp{(\mstepstwo, \estepstwo)}{\val}{\typctxtwo}{\mtype_{\typctxfour,1}} $ there exists a  derivation $\tderiv_{\cbneedctx_{2} \cwc {\val}} \exder[\cbneed] \tyjp{(\msteps_{\typctxfour} + \mstepstwo, \esteps_{\typctxfour} + \estepstwo - 1)}{\cbneedctx_{2} \cwc{\val}}{\var \colon \mtype_{\typctxfour,2} ; \typctxfour \mplus \typctxtwo}{\typctxthree(\vartwo)} $. We can then construct $\tderiv_{\cbneedctx \cwc{\val}}$ for such $\tderivtwo$ as follows
		{\scriptsize
		$$
		\infer
			[\ES]
			{\tyjp{(\msteps_{\typctxthree} + \msteps_{\typctxfour} + \mstepstwo, \esteps_{\typctxthree} + \esteps_{\typctxfour} + \estepstwo - 1)}{\cbneedctx_{1} \cwc{\vartwo} \esub{\vartwo}{\cbneedctx_{2} \cwc{\val}}}{\var \colon (\mtype_{\typctxthree} \mplus \mtype_{\typctxfour,2}) ; (\typctxthree \sm \vartwo) \mplus \typctxfour \mplus \typctxtwo}{\mtypethree}}
			{\tyjp{(\msteps_{\typctxthree}, \esteps_{\typctxthree})}{\cbneedctx_{1} \cwc{\vartwo}}{\var \colon \mtype_{\typctxthree} ; \typctxthree}{\mtypethree}
			\quad
			\tyjp{(\msteps_{\typctxfour} + \mstepstwo, \esteps_{\typctxfour} + \estepstwo - 1)}{\cbneedctx_{2} \cwc{\val}}{\var \colon \mtype_{\typctxfour,2} ; \typctxfour \mplus \typctxtwo}{\typctxthree(\vartwo)}
			\quad
			\typctxthree(\vartwo) \neq \emptytype}
		$$
		}
where we take $\mtype_{1} \defeq \mtype_{\typctxfour,1}$ and $\mtype_{2} \defeq \mtype_{\typctxthree} \mplus \mtype_{\typctxfour,2}$.
\qed
	\end{itemize}
\end{proof}

\gettoappendix {prop:need-subject-reduction}
\begin{proof}
By induction on the reduction relation $\tocbneed$, with $\rtomcbneed$ and $r\toecbneed$ as the base cases, and the closure by $\cbneed$ contexts of $\rtomcbneed \cup \rtoecbneed$ as the inductive one.

\begin{itemize}
\item \emph{Root step for $\tomcbneed$.} Let us assume that $\tm = \sctx \hole{\la{\var}{\tmthree}} \tmfour \rtom \sctx \hole{\tmthree \esub{\var}{\tmfour}} = \tmtwo$, and proceed by induction on \sctx:
	\begin{itemize}
	\item Let $\sctx = \chole$. Then $\tm = (\la{\var}{\tmthree})\tmfour$ and so the last rule of $\tderiv$ is either $\appsteps$ or $\appgc$, because they are the only rules whose term in the conclusion type judgement is an application.
		\begin{itemize}
		\item
		If $\appsteps$ is the last rule of $\tderiv$, then the latter is of the form
			$$
			\infer
				[\appsteps]
				{\tyjp{(\msteps' + \msteps'' + 1,\esteps' + \esteps'')}{(\la{\var}{\tmthree})\tmfour}{(\typctxtwo \sm \var) \bigmplus \typctxthree }{\mtype}}
				{\infer
					[\many]
					{\tyjp{(\msteps',\esteps')}{\la{\var}{\tmthree}}{\typctxtwo \sm \var}{\mult{\ty{\typctxtwo(\var)}{\mtype}}}}
					{\infer
						[\fun]
						{\tyjp{(\msteps',\esteps')}{\la{\var}{\tmthree}}{\typctxtwo \sm \var}{\ty{\typctxtwo(\var)}{\mtype}}}
						{\tyjp{(\msteps',\esteps')}{\tmthree}{\typctxtwo}{\mtype}}}
				\quad
				\tyjp{(\msteps'',\esteps'')}{\tmfour}{\typctxthree}{\typctxtwo(\var)}
				\quad
				\typctxtwo(\var) \neq \emptytype}
			$$
	Therefore, $\msteps \geq 1$. Since $\typctxtwo(\var) \neq \emptytype$, then we can construct $\tderiv'$ as follows:
			$$
			\infer
				[\ES]
				{\tyjp{(\msteps' + \msteps'', \esteps' + \esteps'')}{\tmthree \esub{\var}{\tmfour}}{(\typctxtwo \sm \var) \bigmplus \typctxthree}{\mtype}}
				{\tyjp{(\msteps',\esteps')}{\tmthree}{\typctxtwo}{\mtype}
				\quad
				\tyjp{(\msteps'',\esteps'')}{\tmfour}{\typctxthree}{\typctxtwo(\var)}
				\quad
				\typctxtwo(\var) \neq \emptytype}
			$$
	Note that $(\msteps' + \msteps'', \esteps' + \esteps') = (\msteps - 1, \esteps)$.

		\item If $\appgc$ is the las typing rule of $\tderiv$, then the latter is of the form
			$$
			\infer
				[\appgc]
				{\tyjp{(\msteps' + 1, \esteps')}{(\la{\var}{\tmthree})\tmfour}{\typctxtwo}{\mtype}}
				{\infer
					[\many]
					{\tyjp{(\msteps', \esteps')}{\la{\var}{\tmthree}}{\typctxtwo}{\mult{\ty{\zero}{\mtype}}}}
					{\infer
						[\fun]
						{\tyjp{(\msteps', \esteps')}{\la{\var}{\tmthree}}{\typctxtwo}{\ty{\zero}{\mtype}}}
						{\tyjp{(\msteps', \esteps')}{\tmthree}{\typctxtwo}{\mtype}}
					}
				}
			$$
		with $(\msteps, \esteps) = (\msteps' + 1, \esteps') $ and $\typctx = \typctxtwo \sm \var$. Note that $\var \notin \dom{\typctxtwo}$, because $\tmthree$ is typed with $\mtype$ and $\la{\var}{\tmthree}$ is typed with $\ty{\zero}{\mtype}$, so we can construct $\tderiv'$ as follows:

			$$
			\infer
				[\ESgc]
				{\tyjp{(\msteps', \esteps')}{\tmthree \esub{\var}{\tmfour}}{\typctxtwo}{\mtype}}
				{\tyjp{(\msteps', \esteps')}{\tmthree}{\typctxtwo}{\mtype}
				\quad
				\typctxtwo(\var) = \zero}
			$$
			
		\end{itemize}

	\item Let $\sctx = \sctxtwo \esub{\vartwo}{\tmfive}$. Then $\tm = \sctx \hole{\la{\var}{\tmthree}} \tmfour  = ((\sctxtwo \hole{\la{\var}{\tmthree}})\esub{\vartwo}{\tmfive}) \tmfour \rtom  (\sctxtwo \hole{\tmthree \esub{\var}{\tmfour}}) \esub{\vartwo}{\tmfive} = \sctx \hole{\tmthree \esub{\var}{\tmfour}} = \tmtwo$. Since we are working up to $\alpha$-equivalence, it is safe to assume that $\vartwo \notin \fv{\tmfive}$ and $\vartwo \notin \fv{\tmfour} $. There are several possible forms of $\tderiv$, namely:
		\begin{itemize}
		\item If the last rule is $\appgc$ and $\esub{\vartwo}{\tmfive}$ is appended through rule $\ESgc$, then $\tderiv$ is of the form
			$$
			\infer
				[\appgc]
				{\tyjp{(\msteps' + 1, \esteps')}{\sctx \hole{\la{\var}{\tmthree}} \tmfour}{\typctx}{\mtypetwo}}
				{\infer
					[\ESgc]
					{\tyjp{(\msteps', \esteps')}{\sctx \hole{\la{\var}{\tmthree}}}{\typctx}{\mult{\ty{\emptytype}{\mtypetwo}}}}
					{\tyjp{(\msteps',\esteps')}{\sctxtwo \hole{\la{\var}{\tmthree}}}{\typctx}{\mult{\ty{\emptytype}{\mtypetwo}}}
					\quad
					\typctx(\vartwo) = \emptytype}
				}
			$$
We then construct the following derivation 
			$$
			\infer
				[\appgc]
				{\tyjp{(\msteps' + 1,\esteps')}{\sctxtwo \hole{\la{\var}{\tmthree}} \tmfour}{\typctx}{\mtypetwo}}
				{\tyjp{(\msteps',\esteps')}{\sctxtwo \hole{\la{\var}{\tmthree}}}{\typctx}{\mult{\ty{\emptytype}{\mtypetwo}}}}
			$$
and apply \ih on it to obtain $\msteps = \msteps' + 1 \geq 1$. 
Moreover, the \ih also yields a  derivation $\tderiv'' \exder[\cbneed] \tyjp{(\msteps', \esteps')}{\sctxtwo \hole{\tmthree \esub{\var}{\tmfour}}}{\typctx}{\mtypetwo}$, with which we can then construct $\tderiv'$ as follows:
			$$
			\infer
				[\ESgc]
				{\tyjp{(\msteps', \esteps')}{\sctx \hole{\tmthree \esub{\var}{\tmfour}}}{\typctx}{\mtypetwo}}
				{\tderiv'' \exder[\cbneed] \tyjp{(\msteps', \esteps')}{\sctxtwo \hole{\tmthree \esub{\var}{\tmfour}}}{\typctx}{\mtypetwo}
				\quad
				\typctx(\vartwo) \neq \zero}
			$$
Finally, note that $(\msteps', \esteps') = (\msteps - 1, \esteps) $.
		\item If the last rule is $\appgc$ and $\esub{\vartwo}{\tmfive}$ is appended through rule $\ES$, then $\tderiv$ is 
			$$
			\infer
				[\appgc]
				{\tyjp{(\msteps' + \msteps'' + 1, \esteps' + \esteps'')}{\sctx \hole{\la{\var}{\tmthree}} \tmfour}{(\typctxtwo \sm \vartwo) \bigmplus \typctxthree}{\mtype}}
				{\infer
					[\ES]
					{\tyjp{(\msteps' + \msteps'', \esteps' + \esteps'')}{\sctx \hole{\la{\var}{\tmthree}}}{(\typctxtwo \sm \vartwo) \bigmplus \typctxthree}{\mult{\ty{\emptytype}{\mtype}}}}
					{\tyjp{(\msteps',\esteps')}{\sctxtwo \hole{\la{\var}{\tmthree}}}{\typctxtwo}{\mult{\ty{\emptytype}{\mtype}}}
					\quad
					\tyjp{(\msteps'',\esteps'')}{\tmfour}{\typctxthree}{\mtypetwo}
					\quad
					\typctxtwo(\vartwo) = \mtypetwo \neq \emptytype}
				}
			$$
We can then construct the following derivation
			$$
			\infer
				[\appgc]
				{\tyjp{(\msteps' + 1,\esteps')}{\sctxtwo \hole{\la{\var}{\tmthree}} \tmfour}{\typctxtwo}{\mtype}}
				{\tyjp{(\msteps',\esteps')}{\sctxtwo \hole{\la{\var}{\tmthree}}}{\typctxtwo}{\mult{\ty{\emptytype}{\mtype}}}}
			$$
Applying \ih on it yields a  derivation $\tderiv'' \exder[\cbneed] \tyjp{(\msteps',\esteps')}{\sctxtwo \hole{\tmthree \esub{\var}{\tmfour}}}{\typctxtwo}{\mtype}$ and implies the fact that $\msteps = \msteps' + \msteps'' + 1 \geq \msteps' + 1 \geq 1$. Finally, we construct $\tderiv'$ as follows:
			$$
			\infer
				[\ES]
				{\tyjp{(\msteps' + \msteps'', \esteps' + \esteps'')}{\sctx \hole{\tmthree \esub{\var}{\tmfour}}}{(\typctxtwo \sm \vartwo) \bigmplus \typctxthree}{\mtype}}
				{\tderiv'' \exder[\cbneed] \tyjp{(\msteps',\esteps')}{\sctxtwo \hole{\tmthree \esub{\var}{\tmfour}}}{\typctxtwo}{\mtype}
				\quad
				\tyjp{(\msteps'',\esteps'')}{\tmfour}{\typctxthree}{\mtypetwo}
				\quad
				\typctxtwo(\vartwo) = \mtypetwo \neq \emptytype}
			$$
Note that $(\msteps' + \msteps'', \esteps' + \esteps'') = (\msteps - 1, \esteps) $.
		\item If the last rule is $\appsteps$ and $\esub{\vartwo}{\tmfive}$ is appended through rule $\ESgc$, then $\tderiv$ is 
			$$
			\infer
				[\app]
				{\tyjp{(\msteps' + \msteps'' + 1, \esteps' + \esteps'')}{\sctx \hole{\la{\var}{\tmthree}} \tmfour}{\typctxtwo \mplus \typctxthree}{\mtypetwo}}
				{\infer
					[\esgc]
					{\tyjp{(\msteps',\esteps')}{\sctx \hole{\la{\var}{\tmthree}}}{\typctxtwo}{\mult{\ty{\mtype}{\mtypetwo}}}}
					{\tyjp{(\msteps',\esteps')}{\sctxtwo \hole{\la{\var}{\tmthree}}}{\typctxtwo}{\mult{\ty{\mtype}{\mtypetwo}}}
					\quad
					\typctxtwo(\vartwo) = \emptytype
					}
				\quad
				\tyjp{(\msteps'',\esteps'')}{\tmfour}{\typctxthree}{\mtype}
				\quad
				\mtype \neq \emptytype
				}
			$$
We are now able to give the following derivation
			$$
			\infer
				[\app]
				{\tyjp{(\msteps' + \msteps'' + 1, \esteps' + \esteps'')}{\sctxtwo \hole{\la{\var}{\tmthree}} \tmfour}{\typctxtwo \mplus \typctxthree}{\mtypetwo}}
				{\tyjp{(\msteps',\esteps')}{\sctxtwo \hole{\la{\var}{\tmthree}}}{\typctxtwo}{\mult{\ty{\mtype}{\mtypetwo}}}
				\quad
				\tyjp{(\msteps'',\esteps'')}{\tmfour}{\typctxthree}{\mtype}
				}
			$$
on which application of \ih gives that $\msteps = \msteps' + \msteps'' + 1 \geq 1$, and yields a  derivation $\tderiv'' \exder[\cbneed] \tyjp{(\msteps' + \msteps'', \esteps' + \esteps'')}{\sctxtwo \hole{\tmthree \esub{\var}{\tmfour}}}{\typctxtwo \mplus \typctxthree}{\mtypetwo}$, thus allowing us to construct $\tderiv'$ as follows:
			$$
			\infer
				[\esgc]
				{\tderiv'' \exder[\cbneed] \tyjp{(\msteps' + \msteps'', \esteps' + \esteps'')}{\sctx \hole{\tmthree \esub{\var}{\tmfour}}}{\typctxtwo \mplus \typctxthree}{\mtypetwo}}
				{\tderiv'' \exder[\cbneed] \tyjp{(\msteps' + \msteps'', \esteps' + \esteps'')}{\sctxtwo \hole{\tmthree \esub{\var}{\tmfour}}}{\typctxtwo \mplus \typctxthree}{\mtypetwo}
				\quad
				(\typctxtwo \mplus \typctxthree)(\vartwo) = \emptytype}
			$$
Note that $\vartwo \notin \fv{\tmfive}$ and so via \reflemma{need-basic-properties-typing-derivations} we know that $\typctxthree(\vartwo) = \emptytype$; hence, the use of rule $\esgc$ in $\tderiv'$ is correct. Moreover, note that $(\msteps' + \msteps'', \esteps' + \esteps'') = (\msteps, \esteps) $.
			
		\item If the last rule is $\app$ and $\esub{\vartwo}{\tmfive}$ is appended through rule $\esrule$, then $\tderiv$ is 
			$$
			\infer
				[\app]
				{\tyjp{((\msteps' + \msteps'') + \msteps''' + 1, (\esteps' + \esteps'') + \esteps''')}{\sctx \hole{\la{\var}{\tmthree}} \tmfour}{((\typctxtwo \sm \vartwo) \mplus \typctxthree) \mplus \typctxfour}{\mtypetwo}}
				{\infer
					[\esrule]
					{\tyjp{(\msteps' + \msteps'',\esteps' + \esteps'')}{\sctx \hole{\la{\var}{\tmthree}}}{(\typctxtwo \sm \vartwo) \mplus \typctxthree}{\mult{\ty{\mtype}{\mtypetwo}}}}
					{\tyjp{(\msteps',\esteps')}{\sctxtwo \hole{\la{\var}{\tmthree}}}{\typctxtwo}{\mult{\ty{\mtype}{\mtypetwo}}}
					\quad
					\tyjp{(\msteps'',\esteps'')}{\tmfive}{\typctxthree}{\typctxtwo(\vartwo)}
					\quad
					\typctxtwo(\vartwo) \neq \emptytype}
				\quad
				\tyjp{(\msteps''',\esteps''')}{\tmfour}{\typctxfour}{\mtype}
				}
			$$
Since $\vartwo \notin \fv{\tmfour} \cup \fv{\tmfive}$ then we know through \reflemma{need-basic-properties-typing-derivations} that $\vartwo \notin \dom{\typctxthree}$ and $\vartwo \notin \dom{\typctxfour}$. 
Now, applying \ih on the following  derivation
			$$
			\infer
				[\app]
				{\tyjp{(\msteps' + \msteps''' + 1, \esteps' + \esteps''')}{\sctxtwo \hole{\la{\var}{\tmthree}} \tmfour}{\typctxtwo \mplus \typctxfour}{\mtypetwo}}
				{\tyjp{(\msteps',\esteps')}{\sctxtwo \hole{\la{\var}{\tmthree}}}{\typctxtwo}{\mult{\ty{\mtype}{\mtypetwo}}}
				\quad
				\tyjp{(\msteps''',\esteps''')}{\tmfour}{\typctxfour}{\mtype}}
			$$
yields a derivation $\tderiv'' \exder[\cbneed] \tyjp{(\msteps' + \msteps''', \esteps' + \esteps''')}{L' \hole{\tmthree \esub{\var}{\tmfour}}}{\typctxtwo \mplus \typctxfour}{\mtypetwo}$ and implies $\msteps = (\msteps' +  \msteps'') + \msteps''' + 1 \geq \msteps' + \msteps''' + 1 \geq 1$. 
Finally, given that $(\typctxtwo \mplus \typctxfour)(\vartwo) = \typctxtwo(\vartwo) $, we can finally construct $\tderiv'$ as follows:
			$$
			\infer
				[\ES]
				{\tyjp{((\msteps' + \msteps''') + \msteps'', (\esteps' + \esteps''') + \esteps'')}{\sctx \hole{\tmthree \esub{\var}{\tmfour}}}{((\typctxtwo \mplus \typctxfour) \sm \vartwo) \mplus \typctxthree}{\mtypetwo}}
				{\tderiv'' \exder[\cbneed] \tyjp{(\msteps' + \msteps''', \esteps' + \esteps''')}{\sctxtwo \hole{\tmthree \esub{\var}{\tmfour}}}{\typctxtwo \mplus \typctxfour}{\mtypetwo}
				\quad
				\tyjp{(\msteps'',\esteps'')}{\tmfive}{\typctxthree}{\typctxtwo(\vartwo)}
				\quad
				\typctxtwo(\vartwo) \neq \emptytype}
			$$
Note that $((\msteps' + \msteps''') + \msteps'', (\esteps' + \esteps''') + \esteps'') = (\msteps - 1, \esteps) $, and that since $\vartwo \notin \dom{\typctxfour}$ then $((\typctxtwo \mplus \typctxfour) \sm \vartwo) \mplus \typctxthree = ((\typctxtwo \sm \vartwo) \mplus \typctxthree) \mplus \typctxfour$.
		\end{itemize}
All other typing rules are not possible as the final rule in $\tderiv$. In particular, rule $\many$ is not possible because the term in its final judgement has to be an abstraction, not an application term.
	\end{itemize}
\item \emph{Root step for $\toe$.} Let $\tm = \cbneedctx \cwc{\var} \esub{\var}{\sctx \hole{\val}} \rtoe S \hole{\cbneedctx \cwc{\val} \esub{\var}{\val}}$. We can infer from \reflemma{need-basic-properties-typing-derivations} that $\tderiv$ can only have $\ES$ as its last typing rule, and so can only be of the form
	$$
	\infer
		[\ES]
		{\tyjp{(\msteps' + \msteps'', \esteps' + \esteps'')}{\cbneedctx \cwc{\var} \esub{\var}{\sctxp{\val}}}{\typctxtwo \mplus \typctxthree}{\mtype}}
		{\tderiv_{\cbneedctx \cwc{\var}} \exder[\cbneed] \tyjp{(\msteps',\esteps')}{\cbneedctx \cwc{\var}}{\var \colon \mtypethree; \typctxtwo}{\mtype}
		\quad
		\tderiv_{\sctxp{\val}} \exder[\cbneed] \tyjp{(\msteps'',\esteps'')}{\sctxp{\val}}{\typctxthree}{\mtypethree}
		\quad
		\mtypethree \neq \emptytype}
	$$

	Note that $\var \notin \dom{\typctxthree}$, because otherwise \reflemma{need-basic-properties-typing-derivations} would imply $\var \in \fv{\sctxp{\val}}$ and this cannot be the case, given that we are working up to $\alpha$-equivalence.
	
	We now proceed to prove by induction on $\sctx$ that whenever we have $\tderiv_{\cbneedctx \cwc{\var}} $ and $\tderiv_{\sctxp{\val}} $ we can derive $\tderiv' \exder[\cbneed] \tyjp{(\msteps' + \msteps'', \esteps' + \esteps'' - 1)}{\sctxp{\cbneedctx \cwc{\val} \esub{\var}{\val}}}{\typctxtwo \bigmplus \typctxthree}{\mtype} $.
	\begin{itemize}
	\item Let $\sctx \defeq \ctxhole$. First of all, applying \reflemma{need-linear-substitution} on $\tderiv_{\cbneedctx \cwc{\var}} $ yields a splitting $\mtypethree = \mtypethree_{1} \mplus \mtypethree_{2}$ such that for every  derivation $\tderivtwo \exder[\cbneed] \tyjp{(\msteps''',\esteps''')}{\val}{\typctxfour}{\mtypethree_{1}}$ there exists a  derivation $\tyjp{(\msteps' + \msteps''', \esteps' + \esteps''' - 1)}{\cbneedctx \cwc{\val}}{\var \colon \mtypethree_{2} ; \typctxtwo \bigmplus \typctxfour}{\mtype} $. 
	In particular, if $\mtypethree_{2} = \emptytype$, then we can construct the desired  derivation $\tderiv'$ as follows
		$$
		\infer
			[\ESgc]
			{\tyjp{(\msteps' + \msteps'', \esteps' + \esteps'' - 1)}{\cbneedctx \cwc{\val} \esub{\var}{\val}}{\typctxtwo \bigmplus \typctxthree}{\mtype}}
			{
			\infer
				[\reflemma{need-linear-substitution}]
				{\tyjp{(\msteps' + \msteps'', \esteps' + \esteps'' - 1)}{\cbneedctx \cwc{\val}}{\typctxtwo \bigmplus \typctxthree}{\mtype}}
				{\tyjp{(\msteps',\esteps')}{\cbneedctx \cwc{\var}}{\var \colon \mtypethree; \typctxtwo}{\mtype}
				\qquad
				\tyjp{(\msteps'',\esteps'')}{\val}{\typctxthree}{\mtypethree}}}
		$$

	On the other hand, if $\mtypethree_{2} \neq \emptytype$, we can then apply \reflemma{need-splitting-multisets} on $\tderiv_{\sctxp{\val}}$ to yield  derivations $\tderiv_{\mtypethree_{1}} \exder[\cbneed] \tyjp{(\msteps''_{\mtypethree_{1}}, \esteps''_{\mtypethree_{1}})}{\val}{\typctxthree_{\mtypethree_{1}}}{\mtypethree_{1}}$ and $\tderiv_{\mtypethree_{2}} \exder[\cbneed] \tyjp{(\msteps''_{\mtypethree_{2}}, \esteps''_{\mtypethree_{2}})}{\val}{\typctxthree_{\mtypethree_{2}}}{\mtypethree_{2}}$ such that $\typctxthree = \typctxthree_{\mtypethree_{1}} \mplus \typctxthree_{\mtypethree_{2}} $ and $(\msteps'', \esteps'') = (\msteps''_{\mtypethree_{1}} + \msteps''_{\mtypethree_{2}}, \esteps''_{1} + \esteps''_{\mtypethree_{2}}) $. Thus, we are now able to combine all these  derivations to construct $\tderiv'$ as follows
		$$
		\infer
			[\ES]
			{\tyjp{(\msteps' + \msteps''_{\mtypethree_{1}} + \msteps''_{\mtypethree_{2}}, \esteps' + \esteps''_{\mtypethree_{1}} - 1 + \esteps''_{\mtypethree_{2}})}{\cbneedctx \cwc{\val} \esub{\var}{\val}}{\typctxtwo \bigmplus \typctxthree_{\mtypethree_{1}} \bigmplus \typctxthree_{\mtypethree_{2}}}{\mtype}}
			{\infer
				[\reflemmaeq{need-linear-substitution}]
				{\tyjp{(\msteps' + \msteps''_{\mtypethree_{1}}, \esteps' + \esteps''_{\mtypethree_{1}} - 1)}{\cbneedctx \cwc{\val}}{\var \colon \mtypethree_{2} ; \typctxtwo \bigmplus \typctxthree_{\mtypethree_{1}}}{\mtype}}
				{\tyjp{(\msteps',\esteps')}{\cbneedctx \cwc{\var}}{\var \colon \mtypethree; \typctxtwo}{\mtype}
				\quad
				\tyjp{(\msteps''_{\mtypethree_{1}}, \esteps''_{\mtypethree_{1}})}{\val}{\typctxthree_{\mtypethree_{1}}}{\mtypethree_{1}}}
			\quad
			\tyjp{(\msteps''_{\mtypethree_{2}}, \esteps''_{\mtypethree_{2}})}{\sctxp{\val}}{\typctxthree_{\mtypethree_{2}}}{\mtypethree_{2}}}
		$$
	
	\item Let $\sctx \defeq \sctx' \esub{\vartwo}{\tm}$. There are two possible final typing rules in $\tderiv_{\sctxp{\val}}$, namely $\ES$ and $\ESgc$.
		\begin{itemize}
			\item Let $\tderiv_{\sctxp{\val}}$ be of the form
				$$
				\infer
					[\ESgc]
					{\tyjp{(\mstepsthree, \estepsthree)}{\sctxp{\val}}{\typctxthree}{\mtypethree}}
					{\tyjp{(\mstepsthree, \estepsthree)}{\sctxptwo{\val}}{\typctxthree}{\mtypethree}
					\quad
					\typctxthree(\vartwo) = \emptytype}
				$$
			Note that since we are working up to $\alpha$-equivalence we can safely assume that $\vartwo \notin \fv{\cbneedctx \cwc{\var}}$, and so via \reflemma{need-basic-properties-typing-derivations} we have that $\vartwo \notin \dom{\typctxtwo} $. We can then construct $\tderiv'$ by application of the \ih as follows
				$$
				\infer
					[\ESgc]
					{\tyjp{(\mstepstwo + \mstepsthree, \estepstwo + \estepsthree)}{\sctxp{\cbneedctx \cwc{\val} \esub{\var}{\val}}}{\typctxtwo \bigmplus \typctxthree}{\mtype}}
					{\infer
						[\ih]
						{\tyjp{(\mstepstwo + \mstepsthree, \estepstwo + \estepsthree)}{\sctxptwo{\cbneedctx \cwc{\val} \esub{\var}{\val}}}{\typctxtwo \bigmplus \typctxthree}{\mtype}}
						{\tyjp{(\mstepstwo, \estepstwo)}{\cbneedctx \cwc{\var}}{\var \colon \mtypethree ; \typctxtwo}{\mtype}
						\quad
						\tyjp{(\mstepsthree, \estepsthree)}{\sctxptwo{\val}}{\typctxthree}{\mtypethree}}}
				$$
			\item Let $\tderiv_{\sctxp{\val}}$ be of the form
				$$
				\infer
					[\ES]
					{\tyjp{(\msteps''_{1} + \msteps''_{2}, \esteps''_{1} + \esteps''_{2})}{\sctxp{\val}}{\typctxthree_{1} \bigmplus \typctxthree_{2}}{\mtypethree}}
					{\tyjp{(\msteps''_{1}, \esteps''_{1})}{\sctx' \hole{\val}}{\vartwo \colon \mtypefour; \typctxthree_{1}}{\mtypethree}
					\quad
					\tyjp{(\msteps''_{2}, \esteps''_{2})}{\tm}{\typctxthree_{2}}{\mtypefour}
					\quad
					\mtypefour \neq \emptytype}
				$$
			where $\typctxthree = \typctxthree_{1} \bigmplus \typctxthree_{2} $ and  $(\mstepsthree, \estepsthree) = (\mstepsthree_{1} + \mstepsthree_{2}, \estepsthree_{1} + \estepsthree_{2})$. Note that $\vartwo \notin \fv{\cbneedctx \cwc{\var}}$, and so via \reflemma{need-basic-properties-typing-derivations} we have that $\vartwo \notin \dom{\typctxtwo} $.
			
			We can then construct $\tderiv'$ by application of the \ih and a rearranging of $\tderiv$ as follows
				$$
				\infer
					[\ES]
					{\tyjp{(\msteps' + \msteps''_{1} + \msteps''_{2}, \esteps' + \esteps''_{1} - 1 + \esteps''_{2})}{\sctxp{\cbneedctx \cwc{\val} \esub{\var}{\val}}}{\typctxtwo \bigmplus \typctxthree_{1} \bigmplus \typctxthree_{2}}{\mtype}}
					{\infer
						[\ih]
						{\tyjp{(\msteps' + \msteps''_{1}, \esteps' + \esteps''_{1} - 1)}{\sctx' \hole{\cbneedctx \cwc{\val} \esub{\var}{\val}}}{\vartwo \colon \mtypefour ; \typctxtwo \bigmplus \typctxthree_{1}}{\mtype}}
						{\tyjp{(\msteps', \esteps')}{\cbneedctx \cwc{\var}}{\var \colon \mtypethree; \typctxtwo}{\mtype}
						\quad
						\tyjp{(\msteps''_{1}, \esteps''_{1})}{\sctx' \hole{\val}}{\vartwo \colon \mtypefour; \typctxthree_{1}}{\mtypethree}}
					\quad
					\tyjp{(\msteps''_{2}, \esteps''_{2})}{\tm}{\typctxthree_{2}}{\mtypefour}}
				$$
		
		\end{itemize}
	
	\end{itemize}
	
\item \emph{Contextual closure.} We proceed by induction on the derivation of $\tm = \cbneedctx \hole{\tm_{1}} \tond \cbneedctx \hole{\tm_{2}} = \tmtwo$:
	\begin{itemize}
	\item If $\cbneedctx = \chole$, then $\tm \rtom \tmtwo$ or $\tm \rtoe \tmtwo$, and the statement holds as we have just proved.
	\item Let $\cbneedctx = \cbneedctx_{1} \tmthree$. This implies that the last typing rule in $\tderiv$ is either $\appgc$ or $\appsteps$. We will only cover the case where $\cbneedctx \hole{\tm_{1}} \tom \cbneedctx \hole{\tm_{2}}$ and $\tderiv$ ends in rule $\appsteps$, leaving the rest of the (analogous) cases to the reader.
	
	Now, $\tderiv$ is of the form
		$$
		\infer
			[\appsteps]
			{\tyjp{(\msteps' + \msteps'' + 1, \esteps' + \esteps'')}{\cbneedctx_{1} \hole{\tm_{1}} \tmthree}{\typctxtwo \bigmplus \typctxthree}{\mtype}}
			{\tyjp{(\msteps',\esteps')}{\cbneedctx_{1} \hole{\tm_{1}}}{\typctxtwo}{\mult{\ty{\mtypetwo}{\mtype}}}
			\quad
			\tyjp{(\msteps'',\esteps'')}{\tmthree}{\typctxthree}{\mtypetwo}}
		$$
Then we apply \ih on the left premise of the last rule, obtaining a type derivation whose final judgement is $\tyjp{(\msteps' - 1, \esteps')}{\cbneedctx_{1} \hole{\tm_{2}}}{\typctxtwo}{\mult{\ty{\mtypetwo}{\mtype}}} $, thus allowing us to construct $\tderiv'$ as folows:
		$$
		\infer
			[\appsteps]
			{\tyjp{(\msteps' + \msteps'', \esteps' + \esteps'')}{\cbneedctx_{1} \hole{\tm_{2}} \tmthree}{\typctxtwo \bigmplus \typctxthree}{\mtype}}
			{\tyjp{(\msteps' - 1, \esteps')}{\cbneedctx_{1} \hole{\tm_{2}}}{\typctxtwo}{\mult{\ty{\mtypetwo}{\mtype}}}
			\quad
			\tyjp{(\msteps'',\esteps'')}{\tmthree}{\typctxthree}{\mtypetwo}}
		$$
		
Note that $(\msteps' + \msteps'', \esteps' + \esteps'') = (\msteps - 1,\esteps) $.
	\item Let $\cbneedctx = \cbneedctx_{1} \esub{\var}{\tmthree}$. This implies that the last typing rule in $\tderiv$ is either $\ESgc$ or $\ES$. We will only cover the case where $\cbneedctx \hole{\tm_{1}} \tom \cbneedctx \hole{\tm_{2}}$ and $\tderiv$ ends in rule $\ES$, leaving the rest of the (analogous) cases to the reader.
	
	Now, $\tderiv$ is of the form
		$$
		\infer
			[\ES]
			{\tyjp{(\msteps' + \msteps'', \esteps' + \esteps'')}{\cbneedctx_{1} \hole{\tm_{1}} \esub{\var}{\tmthree}}{(\typctxtwo \sm \var) \bigmplus \typctxthree}{\mtype}}
			{\tyjp{(\msteps',\esteps')}{\cbneedctx_{1} \hole{\tm_{1}}}{\typctxtwo}{\mtype}
			\quad
			\tyjp{(\msteps'', \esteps'')}{\tmthree}{\typctxthree}{\typctxtwo(\var)}
			\quad
			\typctxtwo(\var) \neq \emptytype}
		$$
	
	Applying \ih on the left premise of the last rule yields a  derivation whose final judgement is $\tyjp{(\msteps' - 1, \esteps')}{\cbneedctx_{1} \hole{\tm_{2}}}{\typctxtwo}{\mtype} $, thus allowing us to construct $\tderiv'$ as follows:
		$$
		\infer
			[\ES]
			{\tyjp{(\msteps' - 1 + \msteps'', \esteps' + \esteps'')}{\cbneedctx_{1} \hole{\tm_{2}} \esub{\var}{\tmthree}}{(\typctxtwo \sm \var) \mplus \typctxthree}{\mtype}}
			{\tyjp{(\msteps' - 1, \esteps')}{\cbneedctx_{1} \hole{\tm_{2}}}{\typctxtwo}{\mtype} 
			\quad
			\tyjp{(\msteps'', \esteps'')}{\tmthree}{\typctxthree}{\typctxtwo(\var)}
			\quad
			\typctxtwo(\var) \neq \emptytype
			}
		$$
	
Note that $(\msteps' - 1 + \msteps'', \esteps' + \esteps'') = (\msteps - 1, \esteps)$
		
		\item Let $\cbneedctx = \cbneedctx_{1} \cwc{\var} \esub{\var}{\cbneedctx_{2}}$. We will only consider the case where $\[\cbneedctx_{1} \cwc{\var} \esub{\var}{\cbneedctx_{2}\hole{\tm_{1}}} \tom \cbneedctx_{1} \cwc{\var} \esub{\var}{\cbneedctx_{2} \hole{\tm_{2}}} \]$ leaving the other (analogous) case to the reader.
		
		First of all, \reflemma{need-basic-properties-typing-derivations} implies that the last rule in $\tderiv$ is $\esrule$; \ie, $\tderiv$ is of the form
		$$
		\infer
			[\ES]
			{\tyjp{(\msteps' + \msteps'', \esteps' + \esteps'')}{\cbneedctx_{1} \cwc{\var} \esub{\var}{\cbneedctx_{2}\hole{\tm_{1}}}}{(\typctxtwo \sm \var) \mplus \typctxthree}{\mtype}}
			{\tyjp{(\msteps', \esteps')}{\cbneedctx_{1} \hole{\var}}{\typctxtwo}{\mtype}
			\quad
			\tyjp{(\msteps'',\esteps'')}{\cbneedctx_{2} \hole{\tm_{1}}}{\typctxthree}{\typctxtwo(\var)}
			\quad 
			\typctxtwo(\var) \neq \emptytype}
		$$
		
Applying now the \ih on the premise in the middle of the last rule yields a  derivation with conclusion $\tyjp{(\msteps'' - 1,\esteps'')}{\cbneedctx_{2} \hole{\tm_{2}}}{\typctxthree}{\typctxtwo(\var)} $, thus allowing us to construct $\tderiv'$ as follows
		$$
		\infer
			[\ES]
			{\tyjp{(\msteps' + \msteps'' - 1, \esteps' + \esteps'')}{\cbneedctx_{1} \cwc{\var} \esub{\var}{\cbneedctx_{2}\hole{\tm_{2}}}}{(\typctxtwo \sm \var) \mplus \typctxthree}{\mtype}}
			{\tyjp{(\msteps', \esteps')}{\cbneedctx_{1} \cwc{\var}}{\typctxtwo}{\mtype}
			\quad
			\tyjp{(\msteps'' - 1,\esteps'')}{\cbneedctx_{2} \hole{\tm_{2}}}{\typctxthree}{\typctxtwo(\var)}
			\quad 
			\typctxtwo(\var) \neq \emptytype}
		$$
verifying that $(\msteps' + \msteps'' - 1, \esteps' + \esteps'') = (\msteps - 1, \esteps)$.
\qed
	\end{itemize}

\end{itemize}

\end{proof}

\gettoappendix {prop:need-normal-forms-forall}
\begin{proof}
By induction on $\normalpr\tm$.

\begin{itemize}
\item If $\normalpr\tm$ because $\tm = \la{\var}{\tmtwo}$ then $\tderiv$ can only be of the form 
	$$
	\infer
		[\many]
		{\tyjp{(0,0)}{\la{\var}{\tmtwo}}{}{\mult{\normal}}}
		{\infer
			[\normal]
			{\tyjp{(0,0)}{\la{\var}{\tmtwo}}{}{\normal}}
			{}
		}
	$$
\item If $\normalpr\tm$ because $\tm = \tmtwo \esub{\vartwo}{\tmthree}$ and $\normalpr\tmtwo$ then,  in principle, there are two possible last typing rules to $\tderiv$, namely $\ES$ and $\ESgc$. If we assume that it is $\ES$, then $\tderiv$ is of the form
	$$
	\infer
		[\ES]
		{\tyjp{(\msteps'_{1} + \msteps'_{2}, \esteps'_{1} + \esteps'_{2})}{\tmtwo\esub{\vartwo}{\tmthree}}{\typctxtwo_{1} \mplus \typctxtwo_{2}}{\mult{\normal}}}
		{\tyjp{(\msteps'_{1}, \esteps'_{1})}{\tmtwo}{\vartwo : \mtypethree; \typctxtwo_{1}}{\mult{\normal}}
		\quad
		\tyjp{(\msteps'_{2}, \esteps'_{2})}{\tmthree}{\typctxtwo_{2}}{\mtypethree}
		\quad
		\mtypethree \neq \zero}
	$$
	
with $\typctxtwo = \typctxtwo_{1} \mplus \typctxtwo_{2} $ and $(\msteps', \esteps') = (\msteps'_{1} + \msteps'_{2}, \esteps'_{1} + \esteps'_{2})$. However, application of the \ih on the left-hand side premise gives that $\vartwo : \mtypethree ; \typctxtwo_{1}$ is empty, in turn implying that $\mtypethree = \zero$, which is in contradiction with the constraints of the $\ES$ typing rule.

Therefore, $\ES$ could not be the last typing rule of $\tderiv$, and so the latter can only be of the form
	$$
	\infer
		[\ESgc]
		{\tyjp{(\msteps', \esteps')}{\tmtwo\esub{\vartwo}{\tmthree}}{\typctxtwo}{\mult{\normal}}}
		{\tyjp{(\msteps', \esteps')}{\tmtwo}{\typctxtwo}{\mult{\normal}}
		\quad
		\vartwo \notin \dom{\typctxtwo}
		}
	$$
	
Finally, it suffices to apply \ih on the premise to obtain that $\typctxtwo$ is empty and $\msteps' = \esteps' = 0$.\qed
\end{itemize}
\end{proof}

\gettoappendix {thm:need-correctness}
\begin{proof} 
By induction on $\msteps+\esteps$ and case analysis on whether $\tm$ reduces or not. If $\tm$ is in $\tocbneed$-normal form, then we only have to prove the \emph{moreover} part, which states that if $\tderiv$ is $\tight$ then $\msteps = \esteps = 0$, which follows from \refprop{need-normal-forms-forall}.

Otherwise, there are 2 cases:

\begin{enumerate}
\item \emph{Multiplicative steps}: If $\tm \tomcbneed \tmthree$, then by Quantitative Subject Reduction For $\cbneed$ (\refprop{need-subject-reduction}) there exists a typing derivation $\tderivtwo \exder[\cbneed] \tyjp{(\msteps - 1, \esteps)}{\tmthree}{\typctx}{\mtype}$. By \ih there exist $\tmtwo$ and $\deriv'$ such that $\normalpr{\tmtwo}$, $\deriv' : \tmthree \tocbneedn \tmtwo$, $\size{\deriv'}_{\msteps} \leq \msteps - 1$, and $\size{\deriv'}_{\esteps} \leq \esteps$. Just note that $\tm \tomcbneed \tmthree$ and so, since $\deriv'$ is preceeded by such a step, then we have $\size{\deriv}_{\msteps} = \size{\deriv'}_{\msteps} + 1 \leq \msteps$ and $\size{\deriv}_{\esteps} = \size{\deriv'}_{\esteps} \leq \esteps$.

If $\tderiv $ is $\tight$ then so is $\tderivtwo$. Then by \ih $\size{\deriv'}_{\msteps} = \msteps - 1$ and $\size{\deriv'}_{\esteps} = \esteps$, which finally implies that $\size{\deriv}_{\msteps} = \size{\deriv'}_{\msteps} + 1 = \msteps$ and $\size{\deriv}_{\esteps} = \size{\deriv'}_{\esteps} = \esteps$.

\item \emph{Exponential steps}: If $\tm \toecbneed \tmthree$, then by Quantitative Subject Reduction (\refprop{need-subject-reduction}) there exists a typing derivation $\tderivtwo \exder[\cbneed] \tyjp{(\msteps, \esteps - 1)}{\tmthree}{\typctx}{\mtype}$. By \ih there exists $\tmtwo$ and $\deriv'$ such that $\normalpr{\tmtwo}$, $\deriv' : \tmthree \tocbneedn \tmtwo$, $\size{\deriv'}_{\msteps} \leq \msteps$, and $\size{\deriv'}_{\esteps} \leq \esteps - 1$. Just note that $\tm \toecbneed \tmthree$ and so, since $\deriv'$ is preceeded by such a step, we have $\size{\deriv}_{\msteps} = \size{\deriv'}_{\msteps} \leq \msteps$ and $\size{\deriv}_{\esteps} = \size{\deriv'}_{\esteps} + 1 \leq \esteps$.

If $\tderiv $ is $\tight$ then so is $\tderivtwo$. Then by \ih $\size{\deriv'}_{\msteps} = \msteps$ and $\size{\deriv'}_{\esteps} = \esteps - 1 $, which finally implies that $\size{\deriv}_{\msteps} = \size{\deriv'}_{\msteps}  = \msteps$ and $\size{\deriv}_{\esteps} = \size{\deriv'}_{\esteps} + 1 = \esteps$.
\qed
\end{enumerate}
\end{proof}

\subsection{\cbneed Completeness}
\gettoappendix {prop:need-normal-forms-exist}
\begin{proof} 
We can easily prove by induction on $\normalpr{}$ that if $\normalpr{\tm}$ then $\tm = \sctxp{\la{\var}{\tmtwo}}$, for some abstraction $\la{\var}{\tmtwo}$ and substitution context $\sctx = \ctxhole \esub{\var_{1}}{\tm_{1}} ... \esub{\var_{n}}{\tm_{n}}$, with $n \geq 0$.

Therefore, we can derive $\tderiv$ as follows
	$$
	\infer
		[\ESgc]
		{\tyjp{(0,0)}{\la{\var}{\tmtwo} \esub{\var_{1}}{\tm_{1}} \esub{\var_{n}}{\tm_{n}}}{}{\mult{\normal}}}
		{\infer
			[\ESgc]
			{...}
			{\infer
				[\ESgc]
				{\tyjp{(0,0)}{\la{\var}{\tmtwo} \esub{\var_{1}}{\tm_{1}}}{}{\mult{\normal}}}
				{\infer
					[\many]
					{\tyjp{(0,0)}{\la{\var}{\tmtwo}}{}{\mult{\normal}}}
					{\infer
						[\normal]
						{\tyjp{(0,0)}{\la{\var}{\tmtwo}}{}{\normal}}
						{}
					}
				}
			}
		}
	$$
\qed
\end{proof}

\begin{lemma}[Merging of multi-sets with respect to  derivations]
\label{l:need-merging-multisets}
Given a value $\val$, for any two  derivations $\tderiv_{\mtypetwo} \exder[\cbneed] \tyjp{(\msteps_{\mtypetwo}, \esteps_{\mtypetwo})}{\val}{\typctx_{\mtypetwo}}{\mtypetwo} $ and $\tderiv_{\mtypethree} \exder[\cbneed] \tyjp{(\msteps_{\mtypethree}, \esteps_{\mtypethree})}{\val}{\typctx_{\mtypethree}}{\mtypethree} $, there is a derivation $\tderiv_{\mtypetwo \mplus \mtypethree} \exder[\cbneed] \tyjp{(\msteps_{\mtypetwo} + \msteps_{\mtypethree}, \esteps_{\mtypetwo} + \esteps_{\mtypethree})}{\val}{\typctx_{\mtypetwo} \mplus \typctx_{\mtypethree}}{\mtypetwo \mplus \mtypethree} $.
\end{lemma}
\begin{proof}
Among the different rules that type abstractions (namely $\normal $, $\fun $ and $\many $), only rule $\many$ types them with a multi type. Thus, by properly defining $\J$ and $\K$ such that $\mtypetwo = \mult{\type_{j}}_{\jJ}$ and $\mtypethree = \mult{\type_{k}}_{\kK}$, we have that $\tderiv_{\mtypetwo}$ is of the form
	$$
	\infer
		[\many]
		{\tyjp{(\sum_{\jJ} \msteps_{j},\sum_{\jJ} \esteps_{j})}{\val}{\bigmplus_{\jJ} \typctx_{j}}{\mult{\type_{j}}_{\jJ}}}
		{(\tyjp{(\msteps_{j}, \esteps_{j})}{\val}{\typctx_{j}}{\type_{j}})_{\jJ}}
	$$
and $\tderiv_{\mtypethree}$ is of the form
	$$
	\infer
		[\many]
		{\tyjp{(\sum_{\kK} \msteps_{k},\sum_{\kK} \esteps_{k})}{\val}{\bigmplus_{\kK} \typctx_{k}}{\mult{\type_{k}}_{\kK}}}
		{(\tyjp{(\msteps_{k}, \esteps_{k})}{\val}{\typctx_{k}}{\type_{k}})_{\kK}}
	$$
	
Therefore, we define $\I = \J \cup \K$ and finally obtain $\tderiv_{\mtypetwo \mplus \mtypethree}$ as follows
	$$
	\infer
		[\many]
		{\tyjp{(\sum_{\iI} \msteps_{i},\sum_{\iI} \esteps_{i})}{\val}{\bigmplus_{\iI} \typctx_{i}}{\mult{\type_{i}}_{\iI}}}
		{(\tyjp{(\msteps_{i}, \esteps_{i})}{\val}{\typctx_{i}}{\type_{i}})_{\iI}}
	$$
\qed
\end{proof}

\gettoappendix {l:need-linear-anti-substitution}
\begin{proof}
We prove this by induction on the context $\cbneedctx$:

\begin{itemize}
\item Let $\cbneedctx = \chole$. Note that $\mtypethree \neq \emptytype$, a fact that is verifiable simply by checking the typing rules for abstractions. Now, by taking $\typctx_{\val} \defeq \typctx$, $\typctx' \defeq \emptyset$, $\mtype \defeq \mtypethree$, $(\msteps_{\val}, \esteps_{\val}) \defeq (\msteps, \esteps)$, and $(\msteps', \esteps') \defeq (0,1)$, we can then take $\tderiv_{\val} \defeq \tderiv$ and construct $\tderiv_{\cbneedctx\cwc{\var}}$ as follows:
	$$
	\infer
		[\ax]
		{\tyjp{(0,1)}{\var}{\var \colon \mtype}{\mtype}}
		{}
	$$
verifying that $(\msteps, \esteps) = (\msteps_{\val}, \esteps_{\val}) = (0 + \msteps_{\val}, 1 + \esteps_{\val} - 1) = (\msteps' + \msteps_{\val}, \esteps' + \esteps_{\val} - 1) $ and $\typctx = \emptyset \bigmplus \typctx = \typctx' \bigmplus \typctx_{\val}$.

\item Let $\cbneedctx = \cbneedctx_{1} \tm$, and so $\cbneedctx \cwc{\val} = \cbneedctx_{1} \cwc{\val} \tm$. There are two possible last rules in $\tderiv$, namely $\appsteps$ or $\appgc$.

Let us assume $\tderiv$ is of the form
	$$
	\infer
		[\appsteps]
		{\tyjp{(\msteps_{\typctxtwo} + \msteps_{\typctxthree} + 1, \esteps_{\typctxtwo} + \esteps_{\typctxthree})}{\cbneedctx_{1} \cwc{\val} \tm}{\typctxtwo \bigmplus \typctxthree}{\mtypethree}}
		{\tderiv_{\cbneedctx_{1} \cwc{\val}} \exder[\cbneed] \tyjp{(\msteps_{\typctxtwo}, \esteps_{\typctxtwo})}{\cbneedctx_{1} \cwc{\val}}{\typctxtwo}{\mult{\ty{\mtypethree'}{\mtypethree}}}
		\quad
		\tyjp{(\msteps_{\typctxthree}, \esteps_{\typctxthree})}{\tm}{\typctxthree}{\mtypethree'}
		\quad
		\mtypethree' \neq \emptytype}
	$$

Then we can apply the \ih on $\tderiv_{\cbneedctx_{1} \cwc{\val}}$ to obtain a type $\mtype$ and typing derivations
	$$
		\tderiv_{\val} \exder[\cbneed] \tyjp{(\msteps_{\val},\esteps_{\val})}{\val}{\typctxtwo_{\val}}{\mtype}
	$$
and
	$$
		\tderiv_{\cbneedctx_{1} \cwc{\var}} \exder[\cbneed] \tyjp{(\msteps'', \esteps'')}{\cbneedctx_{1} \cwc{\var}}{\typctxtwo' \bigmplus \{\var \colon \mtype\}}{\mult{\ty{\mtypethree'}{\mtypethree}}}
	$$
such that $\typctxtwo = \typctxtwo' \bigmplus \typctxtwo_{\val} $ and $(\msteps_{\typctxtwo}, \esteps_{\typctxtwo}) = (\msteps'' + \msteps_{\val}, \esteps'' + \esteps_{\val} - 1)$.
	
Thus, we are able to construct $\tderiv_{\cbneedctx \cwc{\var}} $ as follows
	$$
	\infer
		[\appsteps]
		{\tyjp{(\msteps'' + \msteps_{\typctxthree} + 1, \esteps'' + \esteps_{\typctxthree})}{\cbneedctx_{1} \cwc{\var} \tm}{\typctxtwo' \bigmplus \{\var \colon \mtype\} \bigmplus \typctxthree}{\mtypethree}}
		{\tderiv_{\cbneedctx_{1} \cwc{\var}} \exder[\cbneed] \tyjp{(\msteps'', \esteps'')}{\cbneedctx_{1} \cwc{\var}}{\typctxtwo' \bigmplus \{\var \colon \mtype\}}{\mult{\ty{\mtypethree'}{\mtypethree}}}
		\quad
		\tyjp{(\msteps_{\typctxthree}, \esteps_{\typctxthree})}{\tm}{\typctxthree}{\mtypethree'}
		\quad
		\mtypethree' \neq \emptytype}
	$$
and, by taking $\typctx' \defeq \typctxtwo' \bigmplus \typctxthree $, $\typctx_{\val} \defeq \typctxtwo_{\val} $, and $(\msteps', \esteps') \defeq (\msteps'' + \msteps_{\typctxthree} + 1, \esteps'' + \esteps_{\typctxthree})$, then we verify that 
	$$
	\typctx = \typctxtwo \bigmplus \typctxthree = \typctxtwo' \bigmplus \typctxtwo_{\val} \bigmplus \typctxthree = \typctx' \bigmplus \typctx_{\val}
	$$
and
	$$
	(\msteps, \esteps) = 
	(\msteps_{\typctxtwo} + \msteps_{\typctxthree} + 1, \esteps_{\typctxtwo} + \esteps_{\typctxthree}) = 
	(\msteps'' + \msteps_{\val} + \msteps_{\typctxthree} + 1, \esteps'' + \esteps_{\val} - 1 + \esteps_{\typctxthree}) = (\msteps' + \msteps_{\val}, \esteps' + \esteps_{\val} - 1)
	$$

Now, let us assume $\tderiv$ is of the form
	$$
	\infer
		[\appgc]
		{\tyjp{(\msteps, \esteps)}{\cbneedctx_{1} \cwc{\val} \tm}{\typctx}{\mtypethree}}
		{\tderiv_{\cbneedctx_{1} \cwc{\val}} \exder[\cbneed] \tyjp{(\msteps - 1, \esteps)}{\cbneedctx_{1} \cwc{\val}}{\typctx}{\mult{\ty{\emptytype}{\mtypethree}}}}
	$$
We then apply the \ih on $\tderiv_{\cbneedctx_{1} \cwc{\val}}$ to obtain type $\mtype$ and typing derivations
	$$
	\tderiv_{\val} \exder[\cbneed] \tyjp{(\msteps_{\val},\esteps_{\val})}{\val}{\typctx_{\val}}{\mtype}
	$$
and
	$$
	\tderiv_{N_{1} \cwc{\var}} \exder[\cbneed] \tyjp{(\msteps'', \esteps'')}{N_{1} \cwc{\var}}{\typctx' \bigmplus \{\var \colon \mtype\}}{\mult{\ty{\mtypethree'}{\mtypethree}}}
	$$
such that $\typctx = \typctx' \bigmplus \typctx_{\val} $ and $(\msteps - 1,\esteps) = (\msteps'' + \msteps_{\val}, \esteps'' + \esteps_{\val} - 1) $.

Thus, we are able to construct $\tderiv_{\cbneedctx \cwc{\var}}$ as follows
	$$
	\infer
		[\appgc]
		{\tyjp{(\msteps'' + 1, \esteps'')}{\cbneedctx_{1} \cwc{\var} \tm}{\typctx' \bigmplus \{\var \colon \mtype\}}{\mtypethree}}
		{\tderiv_{\cbneedctx_{1} \cwc{\var}} \exder[\cbneed] \tyjp{(\msteps'', \esteps'')}{\cbneedctx_{1} \cwc{\var}}{\typctx' \bigmplus \{\var \colon \mtype\}}{\mult{\ty{\mtypethree'}{\mtypethree}}}}
	$$
and, by taking $(\msteps', \esteps') = (\msteps'' + 1, \esteps'')$, then verify that
	$$
	(\msteps, \esteps) =
	(\msteps'' + \msteps_{\val} + 1, \esteps'' + \esteps_{\val} - 1) = 
	(\msteps' + \msteps_{\val}, \esteps' + \esteps_{\val} - 1)
	$$

Finally, if $\tderiv$ is $\tight$

\item Let $\cbneedctx = \cbneedctx_{1} \esub{\vartwo}{\tm}$, and so $\cbneedctx \cwc{\val} = \cbneedctx_{1} \cwc{\val} \esub{\vartwo}{\tm}$. Note that we can safely assume that $\var \neq \vartwo$, since we are working up to $\alpha$-equivalence and $\vartwo$ has a binding occurrence in $\cbneedctx$ while $\var$ represents a free variable. Moreover, note that $\vartwo \notin \fv{\val}$, since otherwise $\cbneedctx\cwc{\val}$ would not be well-defined. 

There are two possible last rules in $\tderiv$, namely $\ES$ or $\ESgc$.

Let $\tderiv$ be of the form
	$$
	\infer
		[\ES]
		{\tyjp{(\msteps_{\typctxtwo} + \msteps_{\typctxthree}, \esteps_{\typctxtwo} + \esteps_{\typctxthree})}{\cbneedctx_{1} \cwc{\val} \esub{\vartwo}{\tm}}{(\typctxtwo \sm \vartwo) \bigmplus \typctxthree}{\mtypethree}}
		{\tderiv_{\cbneedctx_{1} \cwc{\val}} \exder[\cbneed] \tyjp{(\msteps_{\typctxtwo}, \esteps_{\typctxtwo})}{\cbneedctx_{1} \cwc{\val}}{\typctxtwo}{\mtypethree}
		\quad
		\tyjp{(\msteps_{\typctxthree}, \esteps_{\typctxthree})}{\tm}{\typctxthree}{\typctxtwo(\vartwo)}
		\quad
		\typctxtwo(\vartwo) \neq \emptytype}
	$$
where $\typctx = (\typctxtwo \sm \vartwo) \bigmplus \typctxthree $, $(\msteps, \esteps) = (\msteps_{\typctxtwo} + \msteps_{\typctxthree}, \esteps_{\typctxtwo} + \esteps_{\typctxthree}) $.

Then we can apply the \ih on $\tderiv_{\cbneedctx_{1} \cwc{\val}} $ to obtain type $\mtype$ and typing derivations
	$$
		\tderiv_{\val} \exder[\cbneed] \tyjp{(\msteps_{\val},\esteps_{\val})}{\val}{\typctxtwo_{\val}}{\mtype}
	$$
and
	$$
		\tderiv_{\cbneedctx_{1} \cwc{\var}} \exder[\cbneed] \tyjp{(\msteps'', \esteps'')}{\cbneedctx_{1} \cwc{\var}}{\typctxtwo' \bigmplus \{\var \colon \mtype\}}{\mtypethree}
	$$
such that $\typctxtwo = \typctxtwo' \bigmplus \typctxtwo_{\val} $ and $(\msteps_{\typctxtwo}, \esteps_{\typctxtwo}) = (\msteps'' + \msteps_{\val}, \esteps'' + \esteps_{\val} - 1)$. 

Moreover, since $\vartwo \neq \var$ and $\vartwo \notin \dom{\typctxtwo_{\val}}$ (otherwise \reflemma{need-basic-properties-typing-derivations} would imply that $\vartwo \in \fv{\val}$, which we already know not to be the case), then $(\typctxtwo' \bigmplus \{\var \colon \mtype\})(\vartwo) = \typctxtwo(\vartwo) $ and so we are able to construct $\tderiv_{\cbneedctx \cwc{\var}}$ as follows
	$$
	\infer
		[\ES]
		{\tyjp{(\msteps'' + \msteps_{\typctxthree}, \esteps'' + \esteps_{\typctxthree})}{\cbneedctx_{1} \cwc{\var} \esub{\vartwo}{\tm}}{((\typctxtwo' \bigmplus \{\var \colon \mtype\}) \sm \vartwo) \bigmplus \typctxthree}{\mtypethree}}
		{\tderiv_{\cbneedctx_{1} \cwc{\var}} \exder[\cbneed] \tyjp{(\msteps'', \esteps'')}{\cbneedctx_{1} \cwc{\var}}{\typctxtwo' \bigmplus \{\var \colon \mtype\}}{\mtypethree}
		\quad
		\tyjp{(\msteps_{\typctxthree}, \esteps_{\typctxthree})}{\tm}{\typctxthree}{\typctxtwo(\vartwo)}
		\quad
		\typctxtwo(\vartwo) \neq \emptytype}
	$$
Now, by taking $\typctx' \defeq (\typctxtwo' \sm \vartwo) \bigmplus \typctxthree$, $\typctx_{\val} \defeq \typctxtwo_{\val} $, and $(\msteps', \esteps') \defeq (\msteps'' + \msteps_{\typctxthree}, \esteps'' + \esteps_{\typctxthree}) $, we can verify that 	
	$$
	\typctx = 
	(\typctxtwo \sm \vartwo) \bigmplus \typctxthree =
	((\typctxtwo' \bigmplus \typctxtwo_{\val}) \sm \vartwo) \bigmplus  \typctxthree =
	(\typctxtwo' \sm \vartwo) \bigmplus \typctxtwo_{\val} \bigmplus \typctxthree =
	\typctx' \bigmplus \typctx
	$$
and
	$$
	(\msteps, \esteps) =
	(\msteps_{\typctxtwo} + \msteps_{\typctxthree}, \esteps_{\typctxtwo} + \esteps_{\typctxthree}) =
	((\msteps'' + \msteps_{\val}) + \msteps_{\typctxthree}, (\esteps'' + \esteps_{\val} - 1) + \esteps_{\typctxthree}) =
	(\msteps' + \msteps_{\val}, \esteps' + \esteps_{\val})
	$$
	
If $\tderiv$ is instead of the form
	$$
	\infer
		[\ESgc]
		{\tyjp{(\msteps, \esteps)}{\cbneedctx_{1} \cwc{\val} \esub{\vartwo}{\tm}}{\typctx}{\mtypethree}}
		{\tderiv_{\cbneedctx_{1} \cwc{\val}} \exder[\cbneed] \tyjp{(\msteps, \esteps)}{\cbneedctx_{1} \cwc{\val}}{\typctx}{\mtypethree}
		\quad
		\typctx(\vartwo) = \emptytype}
	$$
then we can apply the \ih on $\tderiv_{\cbneedctx_{1} \cwc{\val}}$ to obtain a type $\mtype$ and typing derivations 
	$$
		\tderiv_{\val} \exder[\cbneed] \tyjp{(\msteps_{\val},\esteps_{\val})}{\val}{\typctxtwo_{\val}}{\mtype}
	$$
and
	$$
		\tderiv_{\cbneedctx_{1} \cwc{\var}} \exder[\cbneed] \tyjp{(\msteps', \esteps')}{\cbneedctx_{1} \cwc{\var}}{\typctx' \bigmplus \{\var \colon \mtype\}}{\mtypethree}
	$$
such that $\typctx = \typctx' \bigmplus \typctx_{\val}$ and $(\msteps, \esteps) = (\msteps' + \msteps_{\val}, \esteps' + \esteps_{\val} - 1) $. Note that this type context and these indices are exactly as desired, and so we can finally construct $\tderiv_{\cbneedctx \cwc{\var}}$ as follows:
	$$
	\infer
		[\ESgc]
		{\tyjp{(\msteps', \esteps')}{\cbneedctx_{1} \cwc{\var} \esub{\vartwo}{\tm}}{\typctx' \bigmplus \{\var \colon \mtype\}}{\mtypethree}}
		{\tderiv_{\cbneedctx_{1} \cwc{\var}} \exder[\cbneed] \tyjp{(\msteps', \esteps')}{\cbneedctx_{1} \cwc{\var}}{\typctx' \bigmplus \{\var \colon \mtype\}}{\mtypethree}}
	$$

\item Let $\cbneedctx = \cbneedctx_{1} \cwc{\vartwo} \esub{\vartwo}{\cbneedctx_{2}}$, and so $\cbneedctx \cwc{\val} = \cbneedctx_{1} \cwc{\vartwo} \esub{\vartwo}{\cbneedctx_{2} \cwc{\val}} $. Once again, we will assume $\var \neq \vartwo$. \reflemma{need-basic-properties-typing-derivations} implies there is only one possible form of $\tderiv$, namely:
	$$
	\infer
		[\ES]
		{\tyjp{(\msteps_{\typctxtwo} + \msteps_{\typctxthree}, \esteps_{\typctxtwo} + \esteps_{\typctxthree})}{\cbneedctx_{1} \cwc{\vartwo} \esub{\vartwo}{\cbneedctx_{2} \cwc{\val}}}{(\typctxtwo \sm \vartwo) \bigmplus \typctxthree}{\mtypethree}}
		{ \tyjp{(\msteps_{\typctxtwo}, \esteps_{\typctxtwo})}{\cbneedctx_{1} \cwc{\vartwo}}{\typctxtwo}{\mtypethree}
		\quad
		\tderiv_{\cbneedctx_{2} \cwc{\val}} \exder[\cbneed] \tyjp{(\msteps_{\typctxthree}, \esteps_{\typctxthree})}{\cbneedctx_{2} \cwc{\val}}{\typctxthree}{\typctxtwo(\vartwo)}
		\quad
		\typctxtwo(\vartwo) \neq \emptytype}
	$$
where $\typctx = (\typctxtwo \sm \vartwo) \bigmplus \typctxthree $, $(\msteps, \esteps) = (\msteps_{\typctxtwo} + \msteps_{\typctxthree}, \esteps_{\typctxtwo} + \esteps_{\typctxthree}) $.

Then we can apply the \ih on $\tderiv_{\cbneedctx_{2} \cwc{\val}}$ to obtain type $\mtype$ and typing derivations
	$$
	\tderiv_{\val} \exder[\cbneed] \tyjp{(\msteps_{\val}, \esteps_{\val})}{\val}{\typctxthree_{\val}}{\mtype}
	$$
and
	$$
	\tderiv_{\cbneedctx_{2} \cwc{\var}} \exder[\cbneed] \tyjp{(\msteps'', \esteps'')}{\cbneedctx_{2} \cwc{\var}}{\typctxthree' \bigmplus \{\var \colon \mtype\}}{\typctxtwo(\vartwo)}
	$$
such that $\typctxthree = \typctxthree' \bigmplus \typctxthree_{\val}$ and $(\msteps_{\typctxthree}, \esteps_{\typctxthree}) = (\msteps'' + \msteps_{\val}, \esteps'' + \esteps_{\val} - 1)$.

We can then construct $\tderiv_{\cbneedctx \cwc{\var}}$ as follows
	$$
	\infer
		[\ES]
		{\tyjp{(\msteps_{\typctxtwo} + \msteps'', \esteps_{\typctxtwo} + \esteps'')}{\cbneedctx_{1} \cwc{\vartwo} \esub{\vartwo}{\cbneedctx_{2} \cwc{\var}}}{(\typctxtwo \sm \vartwo) \bigmplus \typctxthree'}{\mtypethree}}
		{\tyjp{(\msteps_{\typctxtwo}, \esteps_{\typctxtwo})}{\cbneedctx_{1} \cwc{\vartwo}}{\typctxtwo}{\mtypethree}
		\quad
		\tderiv_{\cbneedctx_{2} \cwc{\var}} \exder[\cbneed] \tyjp{(\msteps'', \esteps'')}{\cbneedctx_{2} \cwc{\var}}{\typctxthree' \bigmplus \{\var \colon \mtype\}}{\typctxtwo(\vartwo)}
		\quad
		\typctxtwo(\vartwo) \neq \emptytype}
	$$
and, by taking $\typctx' \defeq (\typctxtwo \sm \vartwo) \bigmplus \typctxthree'$, $\typctx_{\val} \defeq  \typctxthree_{\val}$, $(\msteps', \esteps') = (\msteps_{\typctxtwo} + \msteps'', \esteps_{\typctxtwo} + \esteps'')$, then verify that
	$$
	\typctx = 
	(\typctxtwo \sm \vartwo) \bigmplus \typctxthree =
	(\typctxtwo \sm \vartwo) \bigmplus (\typctxthree' \bigmplus \typctxthree_{\val}) =
	\typctx' \bigmplus \typctx_{\val}
	$$
and
	$$
	(\msteps, \esteps) =
	(\msteps_{\typctxtwo} + \msteps_{\typctxthree}, \esteps_{\typctxtwo} + \esteps_{\typctxthree}) =
	(\msteps_{\typctxtwo} + (\msteps'' + \msteps_{\val}), \esteps_{\typctxtwo} + (\esteps'' + \esteps_{\val} - 1)) =
	(\msteps' + \msteps_{\val}, \esteps' + \esteps_{\val} - 1)
	$$
\qed	
\end{itemize}
\end{proof}

\gettoappendix {prop:need-subject-expansion}
\begin{proof}
By induction on the derivation $\tm \tond \tmtwo $, with the root rules $\rtom$ and $\rtoe$ as the base case, and the closure by $\cbneed$ contexts of $\rtond$ as the inductive one.
\begin{itemize}
\item \emph{Root step for $\tom$.} Let $\tm = \sctx \hole{\la{\var}{\tmthree}} \tmfour \rtom S \hole{\tmthree \esub{\var}{\tmfour}} = \tmtwo$, and proceed by induction on $\sctx$:
	\begin{itemize}
	\item Let $\sctx \defeq \ctxhole$. Then $\tm = (\la{\var}{\tmthree}) \tmfour$ and $\tmtwo = \tmthree \esub{\var}{\tmfour}$, and so $\tderiv$ has either $\ES$ or $\ESgc$ as its last typing rule.
		\begin{itemize}
		\item If $\tderiv$ is of the form
			$$
			\infer
				[\ESgc]
				{\tyjp{(\msteps,\esteps)}{\tmthree \esub{\var}{\tmfour}}{\typctx}{\mtype}}
				{\tyjp{(\msteps,\esteps)}{\tmthree}{\typctx}{\mtype}
				\quad
				\typctx(\var) = \emptytype}
			$$
		then we can construct $\tderiv'$ as follows
			$$
			\infer
				[\appgc]
				{\tyjp{(\msteps + 1,\esteps)}{(\la{\var}{\tmthree})\tmfour}{\typctx}{\mtype}}
				{\infer
					[!]
					{\tyjp{(\msteps,\esteps)}{\la{\var}{\tmthree}}{\typctx}{\mult{\ty{\emptytype}{\mtype}}}}
					{\infer
						[\fun]
						{\tyjp{(\msteps,\esteps)}{\la{\var}{\tmthree}}{\typctx}{\ty{\emptytype}{\mtype}}}
						{\tyjp{(\msteps,\esteps)}{\tmthree}{\typctx}{\mtype}}}}
			$$
		\item Let $\tderiv$ be of the form
			$$
			\infer
				[\ES]
				{\tyjp{(\mstepstwo + \mstepsthree, \estepstwo + \estepsthree)}{\tmthree \esub{\var}{\tmfour}}{\typctxtwo \bigmplus \typctxthree}{\mtype}}
				{\tyjp{(\mstepstwo, \estepstwo)}{\tmthree}{\var \colon \mtypethree; \typctxtwo}{\mtype}
				\quad
				\tyjp{(\mstepsthree, \estepsthree)}{\tmfour}{\typctxthree}{\mtypethree}
				\quad
				\mtypethree \neq \emptytype}
			$$
		where $\typctx = \typctxtwo \bigmplus \typctxthree$ and $(\msteps, \esteps) = (\mstepstwo + \mstepsthree, \estepstwo + \estepsthree) $.
		
		We can then construct $\tderiv'$ as follows
			$$
			\infer
				[\appsteps]
				{\tyjp{(\mstepstwo + \mstepsthree + 1, \estepstwo + \estepsthree)}{(\la{\var}{\tmthree}) \tmfour}{\typctxtwo \bigmplus \typctxthree}{\mtype}}
				{\infer
					[!]
					{\tyjp{(\mstepstwo, \estepstwo)}{\la{\var}{\tmthree}}{\typctxtwo}{\mult{\ty{\mtypethree}{\mtype}}}}
					{\infer
						[\fun]
						{\tyjp{(\mstepstwo, \estepstwo)}{\la{\var}{\tmthree}}{\typctxtwo}{\ty{\mtypethree}{\mtype}}}
						{\tyjp{(\mstepstwo, \estepstwo)}{\tmthree}{\var \colon \mtypethree; \typctxtwo}{\mtype}}}
				\quad
				\tyjp{(\mstepsthree, \estepsthree)}{\tmfour}{\typctxthree}{\mtypethree}}
			$$
		\end{itemize}
	\item Let $\sctx \defeq \sctxtwo \esub{\vartwo}{\tmfive}$. Then $\tm = (\sctxptwo{\la{\var}{\tmthree}} \esub{\vartwo}{\tmfive}) \tmfour$ and $\tmtwo = \sctxptwo{\tmthree \esub{\var}{\tmfour}} \esub{\vartwo}{\tmfive}$. Note that since we are working up to $\alpha$-equivalence, $\vartwo \notin \fv{\tmfour}$. There are two possible last typing rules of $\tderiv$, namely $\ES$ and $\ESgc$.
		\begin{itemize}
		\item Let $\tderiv$ be of the form
			$$
			\infer
				[\ES]
				{\tyjp{(\mstepstwo + \mstepsthree, \estepstwo + \estepsthree)}{\sctxptwo{\tmthree \esub{\var}{\tmfour}} \esub{\vartwo}{\tmfive}}{\typctxtwo \bigmplus \typctxthree}{\mtype}}
				{\tyjp{(\mstepstwo, \estepstwo)}{\sctxptwo{\tmthree \esub{\var}{\tmfour}}}{\vartwo \colon \mtypethree ; \typctxtwo}{\mtype}
				\quad
				\tyjp{(\mstepsthree, \estepsthree)}{\tmfive}{\typctxthree}{\mtypethree}
				\quad
				\mtypethree \neq \emptytype}
			$$
		where $\typctx = \typctxtwo \bigmplus \typctxthree$ and $(\msteps, \esteps) = (\mstepstwo + \mstepsthree, \estepstwo + \estepsthree) $. We can then apply the \ih on the leftmost premise to obtain a typing derivation $\tderiv'_{\ih} \exder[\cbneed] \tyjp{(\mstepstwo + 1, \estepstwo)}{\sctxptwo{\la{\var}{\tmthree}} \tmfour}{\vartwo \colon \mtypethree; \typctxtwo}{\mtype}$. We then analyze the two possibilities of the last typing rule in $\tderiv'_{\ih}$, namely $\appsteps$ or $\appgc$
			\begin{itemize}
				\item Let $\tderiv'_{\ih}$ be of the form
					$$
					\infer
						[\appsteps]
						{\tyjp{(\mstepstwo + 1, \estepstwo)}{\sctxptwo{\la{\var}{\tmthree}} \tmfour}{\vartwo \colon \mtypethree; \typctxtwo}{\mtype}}
						{\tyjp{(\mstepstwo_{1}, \estepstwo_{1})}{\sctxptwo{\la{\var}{\tmthree}}}{\vartwo \colon \mtypethree ; \typctxtwo_{1}}{\mult{\ty{\mtypefour}{\mtype}}}
						\quad
						\tyjp{(\mstepstwo_{2}, \estepstwo_{2})}{\tmfour}{\typctxtwo_{2}}{\mtypefour}
						\quad
						\mtypefour \neq \emptytype}
					$$ 
				where $(\mstepstwo, \estepstwo) = (\mstepstwo_{1} + \mstepstwo_{2}, \estepstwo_{1} + \estepstwo_{2})$, $\typctxtwo = \typctxtwo_{1} \bigmplus \typctxtwo_{2} $, and $\vartwo \notin \dom{\typctxtwo_{2}}$, since otherwise \reflemma{need-basic-properties-typing-derivations} would imply $\vartwo \in \fv{\tmfour}$.
				
				We can then construct $\tderiv'$ as follows
					$$
					\infer
						[\appsteps]
						{\tyjp{(\mstepstwo_{1} + \mstepsthree + \mstepstwo_{2} + 1, \estepstwo_{1} + \estepsthree + \estepstwo_{2})}{\sctxp{\la{\var}{\tmthree}} \tmfour}{\typctxthree \bigmplus \typctxtwo_{1} \bigmplus \typctxtwo_{2}}{\mtype}}
						{\infer
							[\ES]
							{\tyjp{(\mstepstwo_{1} + \mstepsthree, \estepstwo_{1} + \estepsthree)}{\sctxp{\la{\var}{\tmthree}}}{\typctxthree \bigmplus \typctxtwo_{1}}{\mult{\ty{\mtypefour}{\mtype}}}}
							{\tyjp{(\mstepstwo_{1}, \estepstwo_{1})}{\sctxptwo{\la{\var}{\tmthree}}}{\vartwo \colon \mtypethree ; \typctxtwo_{1}}{\mult{\ty{\mtypefour}{\mtype}}}
							\quad
							\tyjp{(\mstepsthree, \estepsthree)}{\tmfive}{\typctxthree}{\mtypethree}}
						\quad
						\tyjp{(\mstepstwo_{2}, \estepstwo_{2})}{\tmfour}{\typctxtwo_{2}}{\mtypefour}}
					$$
				\item Let $\tderiv'_{\ih}$ be of the form
					$$
					\infer
						[\appgc]
						{\tyjp{(\mstepstwo + 1, \estepstwo)}{\sctxptwo{\la{\var}{\tmthree}} \tmfour}{\vartwo \colon \mtypethree; \typctxtwo}{\mtype}}
						{\tyjp{(\mstepstwo, \estepstwo)}{\sctxptwo{\la{\var}{\tmthree}}}{\vartwo \colon \mtypethree; \typctxtwo}{\mult{\ty{\emptytype}{\mtype}}}}
					$$
				We can then construct $\tderiv'$ as follows
					$$
					\infer
						[\appgc]
						{\tyjp{(\mstepstwo + \mstepsthree + 1, \estepstwo + \estepsthree)}{\sctxp{\la{\var}{\tmthree}} \tmfour}{\typctxtwo \bigmplus \typctxthree}{\mtype}}
						{\infer
							[\ES]
							{\tyjp{(\mstepstwo + \mstepsthree, \estepstwo + \estepsthree)}{\sctxp{\la{\var}{\tmthree}}}{\typctxtwo \bigmplus \typctxthree}{\mult{\ty{\emptytype}{\mtype}}}}
							{\tyjp{(\mstepstwo, \estepstwo)}{\sctxptwo{\la{\var}{\tmthree}}}{\vartwo \colon \mtypethree; \typctxtwo}{\mult{\ty{\emptytype}{\mtype}}}
							\quad
							\tyjp{(\mstepsthree, \estepsthree)}{\tmfive}{\typctxthree}{\mtypethree}}}
					$$
			\end{itemize}
			
		\item Let $\tderiv$ be of the form
			$$
			\infer
				[\ESgc]
				{\tyjp{(\msteps, \esteps)}{\sctxptwo{\tmthree \esub{\var}{\tmfour}} \esub{\vartwo}{\tmfive}}{\typctx}{\mtype}}
				{\tyjp{(\msteps, \esteps)}{\sctxptwo{\tmthree \esub{\var}{\tmfour}}}{\typctx}{\mtype}
				\quad
				\typctx(\vartwo) = \emptytype}
			$$
		We then apply the \ih on the leftmost premise to obtain a typing derivation $\tderiv'_{\ih} \exder[\cbneed] \tyjp{(\msteps + 1, \esteps)}{\sctxptwo{\la{\var}{\tmthree}} \tmfour}{\typctx}{\mtype}$ for which there are two possible last typing rules, namely $\appsteps$ and $\appgc$.
			\begin{itemize}
				\item Let $\tderiv'_{\ih}$ be of the form
					$$
					\infer
						[\appsteps]
						{\tyjp{(\msteps + 1, \esteps)}{\sctxptwo{\la{\var}{\tmthree}} \tmfour}{\typctx}{\mtype}}
						{\tyjp{(\mstepstwo, \estepstwo)}{\sctxptwo{\la{\var}{\tmthree}}}{\typctxtwo}{\mult{\ty{\mtypethree}{\mtype}}}
						\quad
						\tyjp{(\mstepsthree, \estepsthree)}{\tmfour}{\typctxthree}{\mtypethree}
						\quad
						\mtypethree \neq \emptytype}
					$$
				where $\typctx = \typctxtwo \bigmplus \typctxthree $ and $(\msteps, \esteps) = (\mstepstwo + \mstepsthree, \estepstwo + \estepsthree) $. 
				
				We can then construct $\tderiv'$ as follows
					$$
					\infer
						[\appsteps]
						{\tyjp{(\mstepstwo + \mstepsthree + 1, \estepstwo + \estepsthree)}{\sctxp{\la{\var}{\tmthree}} \tmfour}{\typctxtwo \bigmplus \typctxthree}{\mtype}}
						{\infer
							[\ESgc]
							{\tyjp{(\mstepstwo, \estepstwo)}{\sctxp{\la{\var}{\tmthree}}}{\typctxtwo}{\mult{\ty{\mtypethree}{\mtype}}}}
							{\tyjp{(\mstepstwo, \estepstwo)}{\sctxptwo{\la{\var}{\tmthree}}}{\typctxtwo}{\mult{\ty{\mtypethree}{\mtype}}}}
						\quad
						\tyjp{(\mstepsthree, \estepsthree)}{\tmfour}{\typctxthree}{\mtypethree}}
					$$
				\item Let $\tderiv'_{\ih}$ be of the form
					$$
					\infer
						[\appgc]
						{\tyjp{(\msteps + 1, \esteps)}{\sctxptwo{\la{\var}{\tmthree}} \tmfour}{\typctx}{\mtype}}
						{\tyjp{(\msteps, \esteps)}{\sctxptwo{\la{\var}{\tmthree}}}{\typctx}{\mult{\ty{\emptytype}{\mtype}}}}
					$$
					
				We can then construct $\tderiv'$ as follows
					$$
					\infer
						[\appgc]
						{\tyjp{(\msteps + 1, \esteps)}{\sctxp{\la{\var}{\tmthree}} \tmfour}{\typctx}{\mtype}}
						{\infer
							[\ESgc]
							{\tyjp{(\msteps, \esteps)}{\sctxp{\la{\var}{\tmthree}}}{\typctx}{\mult{\ty{\emptytype}{\mtype}}}}
							{\tyjp{(\msteps, \esteps)}{\sctxptwo{\la{\var}{\tmthree}}}{\typctx}{\mult{\ty{\emptytype}{\mtype}}}}}
					$$
			\end{itemize}
		\end{itemize}
	\end{itemize}

\item \emph{Root step for $\toe$.} Let $\tm = \cbneedctx \cwc{\var} \esub{\var}{\sctxp{\val}} \rtoe \sctxp{\cbneedctx \cwc{\val} \esub{\var}{\val}} = \tmtwo$, and proceed by induction on $\sctx$:
	\begin{itemize}
	\item Let $\sctx \defeq \ctxhole$.Then $\tm = \cbneedctx \cwc{\var} \esub{\var}{\val} \rtoe \cbneedctx \cwc{\val} \esub{\var}{\val} = \tmtwo$, and so $\tderiv$ has either $\ES$ or $\ESgc$ as its last typing rule.
		\begin{itemize}
		\item Let $\tderiv$ be of the form 
			$$
			\infer
				[\ESgc]
				{\tyjp{(\msteps, \esteps)}{\cbneedctx \cwc{\val} \esub{\var}{\val}}{\typctx}{\mtype}}
				{\tderiv_{\cbneedctx \cwc{\val}} \exder[\cbneed] \tyjp{(\msteps,\esteps)}{\cbneedctx \cwc{\val}}{\typctx}{\mtype}
				\quad
				\typctx(\var) = \emptytype}
			$$
		We apply \reflemma{need-linear-anti-substitution} on $\tderiv_{\cbneedctx \cwc{\val}}$ to obtain typing derivations $\tderiv_{\val} \exder[\cbneed] \tyjp{(\msteps_{\val}, \esteps_{\val})}{\val}{\typctx_{\val}}{\mtypethree} $ and $\tderiv_{\cbneedctx \cwc{\var}} \exder[\cbneed] \tyjp{(\msteps', \esteps')}{\cbneedctx \cwc{\var}}{\typctx' \bigmplus \{\var \colon \mtypethree\}}{\mtype} $ such that $\typctx = \typctx' \bigmplus \typctx_{\val} $ and $(\msteps,\esteps) = (\msteps' + \msteps_{\val}, \esteps' + \esteps_{\val} - 1) $. We can then construct $\tderiv'$ with such derivations as follows
			$$
			\infer
				[\ES]
				{\tyjp{(\msteps' + \msteps_{\val}, \esteps' + \esteps_{\val})}{\cbneedctx \cwc{\var} \esub{\var}{\val}}{\typctx' \bigmplus \typctx_{\val}}{\mtype}}
				{\tyjp{(\msteps', \esteps')}{\cbneedctx \cwc{\var}}{\typctx' \bigmplus \{\var \colon \mtypethree\}}{\mtype}
				\quad
				\tyjp{(\msteps_{\val}, \esteps_{\val})}{\val}{\typctx_{\val}}{\mtypethree}}
			$$
		In particular, note that $(\msteps' + \msteps_{\val}, \esteps' + \esteps_{\val}) = (\msteps, \esteps + 1)$.
		\item Let $\tderiv$ be of the form
			$$
			\infer
				[\ES]
				{\tyjp{(\mstepstwo + \mstepsthree, \estepstwo + \estepsthree)}{\cbneedctx \cwc{\val} \esub{\var}{\val}}{\typctxtwo \bigmplus \typctxthree}{\mtype}}
				{\tyjp{(\mstepstwo, \estepstwo)}{\cbneedctx \cwc{\val}}{\var \colon \mtypethree; \typctxtwo}{\mtype}
				\quad
				\tyjp{(\mstepsthree, \estepsthree)}{\val}{\typctxthree}{\mtypethree}
				\quad
				\mtypethree \neq \emptytype}
			$$
		We can then apply \reflemma{need-linear-anti-substitution} on the leftmost premise with respect to $\var$ to obtain a multi type $\mtypefour$ and typing derivations $\tderiv_{\val} \exder[\cbneed] \tyjp{(\msteps_{\val}, \esteps_{\val})}{\val}{\typctxtwo_{\val}}{\mtypefour} $ and $\tderiv_{\cbneedctx \cwc{\var}} \exder[\cbneed] \tyjp{(\msteps_{\cbneedctx \cwc{\var}}, \esteps_{\cbneedctx \cwc{\var}})}{\cbneedctx \cwc{\var}}{\typctxtwo_{\cbneedctx \cwc{\var}} \bigmplus \{\var \colon \mtypefour \}}{\mtype} $ such that $\var \colon \mtypethree; \typctxtwo = \typctxtwo_{\val} \bigmplus \typctxtwo_{\cbneedctx \cwc{\var}} $ and $(\mstepstwo, \estepstwo) = (\msteps_{\cbneedctx \cwc{\var}} + \msteps_{\val}, \esteps_{\cbneedctx \cwc{\var}} + \esteps_{\val} - 1)$. Note how \reflemma{need-basic-properties-typing-derivations} implies that $\var \notin \dom{\typctxtwo_{\val}}$ -given that $\var \notin \fv{\val}$-, and so $\tderiv_{\cbneedctx \cwc{\var}}$ can be rewritten as $\tderiv_{\cbneedctx \cwc{\var}} \exder[\cbneed] \tyjp{(\msteps_{\cbneedctx \cwc{\var}}, \esteps_{\cbneedctx \cwc{\var}})}{\cbneedctx \cwc{\var}}{\var \colon \mtypethree \bigmplus \mtypefour; \typctxtwo'_{\cbneedctx \cwc{\var}}}{\mtype}$, where $\typctxtwo_{\cbneedctx \cwc{\var}} = \var \colon \mtypethree ; \typctxtwo'_{\cbneedctx \cwc{\var}}$ and so $\typctxtwo = \typctxtwo'_{\cbneedctx \cwc{\var}} \bigmplus \typctxtwo_{\val}$.

		Furthermore, we can apply \reflemma{need-merging-multisets} on $\tyjp{(\mstepsthree, \estepsthree)}{\val}{\typctxthree}{\mtypethree} $ and $\tderiv_{\val}$ to obtain a typing derivation $\tderiv_{\mtypethree \bigmplus \mtypefour} \exder[\cbneed] \tyjp{(\mstepsthree + \msteps_{\val}, \estepsthree + \esteps_{\val})}{\val}{\typctxthree \bigmplus \typctxtwo_{\val}}{\mtypethree \bigmplus \mtypefour} $.
		
		Finally, we can construct $\tderiv'$ as follows
			$$
			\infer
				[\ES]
				{\tyjp{(\msteps_{\cbneedctx \cwc{\var}} + \mstepsthree + \msteps_{\val}, \esteps_{\cbneedctx \cwc{\var}} + \estepsthree + \esteps_{\val})}{\cbneedctx \cwc{\var} \esub{\var}{\val}}{\typctxtwo'_{\cbneedctx \cwc{\var}} \bigmplus \typctxthree \bigmplus \typctxtwo_{\val}}{\mtype}}
				{\tyjp{(\msteps_{\cbneedctx \cwc{\var}}, \esteps_{\cbneedctx \cwc{\var}})}{\cbneedctx \cwc{\var}}{\var \colon \mtypethree \bigmplus \mtypefour; \typctxtwo'_{\cbneedctx \cwc{\var}}}{\mtype}
				\quad
				\tyjp{(\mstepsthree + \msteps_{\val}, \estepsthree + \esteps_{\val})}{\val}{\typctxthree \bigmplus \typctxtwo_{\val}}{\mtypethree \bigmplus \mtypefour}}
			$$
		Note that $(\msteps_{\cbneedctx \cwc{\var}} + \mstepsthree + \msteps_{\val}, \esteps_{\cbneedctx \cwc{\var}} + \estepsthree + \esteps_{\val}) =(\mstepstwo + \mstepsthree, \estepstwo + 1 + \estepsthree) = (\msteps, \esteps + 1) $.
		\end{itemize}
	
	\item Let $\sctx \defeq \sctx \esub{\vartwo}{\tmfour}$. Then $\tm = \cbneedctx \cwc{\var} \esub{\var}{\sctxptwo{\val} \esub{\vartwo}{\tmfour}} \rtoe \sctxptwo{\cbneedctx \cwc{\val} \esub{\var}{\val}} \esub{\vartwo}{\tmfour} = \tmtwo$. Note that $\vartwo \notin \fv{\cbneedctx \cwc{\var}}$, since $\vartwo$ is bound in $\sctxp{\val}$ and we are working up to $\alpha$-equivalence. Then, $\tderiv$ has either $\ES$ or $\ESgc$ as its last typing rule.
	\begin{itemize}
	\item Let $\tderiv$ be of the form
		$$
		\infer
			[\ESgc]
			{\tyjp{(\msteps, \esteps)}{\sctxptwo{\cbneedctx \cwc{\val} \esub{\var}{\val}} \esub{\vartwo}{\tmfour}}{\typctx}{\mtype}}
			{\tyjp{(\msteps, \esteps)}{\sctxptwo{\cbneedctx \cwc{\val} \esub{\var}{\val}}}{\typctx}{\mtype}
			\quad
			\typctx(\vartwo) = \emptytype}
		$$
	We can then apply the \ih on the premise to obtain a typing derivation $\tderiv'_{\ih} \exder[\cbneed] \tyjp{(\msteps,\esteps+1)}{\cbneedctx \cwc{\var} \esub{\var}{\sctxptwo{\val}}}{\typctx}{\mtype} $. Moreover, note that $\tderiv'_{\ih}$ can only have $\ES$ as its last typing rule, by application of \reflemma{need-basic-properties-typing-derivations}. $\tderiv'_{\ih}$ is hence of the form
		$$
		\infer
			[\ES]
			{\tyjp{(\msteps,\esteps+1)}{\cbneedctx \cwc{\var} \esub{\var}{\sctxptwo{\val}}}{\typctx}{\mtype}}
			{\tyjp{(\msteps_{\cbneedctx \cwc{\var}},\esteps_{\cbneedctx \cwc{\var}})}{\cbneedctx \cwc{\var}}{\var \colon \mtypethree ; \typctxtwo}{\mtype}
			\quad
			\tderiv_{\sctxptwo{\val}} \exder[\cbneed] \tyjp{(\msteps_{\sctxptwo{\val}}, \esteps_{\sctxptwo{\val}})}{\sctxptwo{\val}}{\typctxthree}{\mtypethree}
			\quad
			\mtypethree \neq \emptytype}
		$$
	where $\typctxtwo \bigmplus \typctxthree = \typctx$ and $(\msteps, \esteps + 1) = (\msteps_{\cbneedctx \cwc{\var}} + \msteps_{\sctxptwo{\val}}, \esteps_{\cbneedctx \cwc{\var}} + \esteps_{\sctxptwo{\val}}) $.
		
		Since $\typctx(\vartwo) = \emptytype$ then $\typctxthree(\vartwo) = \emptytype $ and we can construct $\tderiv'$ as follows
		$$
		\infer
			[\ES]
			{\tyjp{(\msteps_{\cbneedctx \cwc{\var}} + \msteps_{\sctxptwo{\val}}, \esteps_{\cbneedctx \cwc{\var}} + \esteps_{\sctxptwo{\val}})}{\cbneedctx \cwc{\var} \esub{\var}{\sctxptwo{\val} \esub{\vartwo}{\tmfour}}}{\typctxtwo \bigmplus \typctxthree}{\mtype}}
			{\tyjp{(\msteps_{\cbneedctx \cwc{\var}},\esteps_{\cbneedctx \cwc{\var}})}{\cbneedctx \cwc{\var}}{\var \colon \mtypethree ; \typctxtwo}{\mtype}
			\quad
			\infer
				[\ESgc]
				{\tyjp{(\msteps_{\sctxptwo{\val}}, \esteps_{\sctxptwo{\val}})}{\sctxptwo{\val} \esub{\vartwo}{\tmfour}}{\typctxthree}{\mtypethree}}
				{\tyjp{(\msteps_{\sctxptwo{\val}}, \esteps_{\sctxptwo{\val}})}{\sctxptwo{\val}}{\typctxthree}{\mtypethree}
				\quad
				\typctxthree(\vartwo) = \emptytype}
			}
		$$

	\item Let $\tderiv$ be of the form
		$$
		\infer
			[\ES]
			{\tyjp{(\msteps_{1} + \msteps_{2}, \esteps_{1} + \esteps_{2})}{\sctxptwo{\cbneedctx \cwc{\val} \esub{\var}{\val}} \esub{\vartwo}{\tmfour}}{\typctxtwo \bigmplus \typctxthree}{\mtype}}
			{\tyjp{(\msteps_{1}, \esteps_{1})}{\sctxptwo{\cbneedctx \cwc{\val} \esub{\var}{\val}}}{\vartwo \colon \mtypethree ; \typctxtwo}{\mtype}
			\quad
			\tyjp{(\msteps_{2}, \esteps_{2})}{\tmfour}{\typctxthree}{\mtypethree}}
		$$
	where $\typctxtwo \bigmplus \typctxthree = \typctx$ and $(\msteps_{1} + \msteps_{2}, \esteps_{1} + \esteps_{2}) = (\msteps, \esteps)$. We can then apply the \ih on the leftmost premise to obtain a typing derivation $\tderiv'_{\ih} \exder[\cbneed] \tyjp{(\msteps_{1}, \esteps_{1} + 1)}{\cbneedctx \cwc{\var} \esub{\var}{\sctxptwo{\val}}}{\vartwo \colon \mtypethree ; \typctxtwo}{\mtype}$ which has to have $\ES$ as its last typing rule -via \reflemma{need-basic-properties-typing-derivations}-, as follows
		$$
		\infer
			[\ES]
			{\tyjp{(\msteps_{1,1} + \msteps_{1,2}, \esteps_{1,1} + \esteps_{1,2})}{\cbneedctx \cwc{\var} \esub{\var}{\sctxptwo{\val}}}{\vartwo \colon \mtypethree ; \typctxtwo_{1} \bigmplus \typctxtwo_{2}}{\mtype}}
			{\tyjp{(\msteps_{1,1}, \esteps_{1,1})}{\cbneedctx \cwc{\var}}{\var \colon \mtypefour;\typctxtwo_{1}}{\mtype}
			\quad
			\tyjp{(\msteps_{1,2}, \esteps_{1,2})}{\sctxptwo{\val}}{\vartwo \colon \mtypethree ; \typctxtwo_{2}}{\mtypefour}}
		$$
	where $\typctxtwo_{1} \bigmplus \typctxtwo_{2} = \typctxtwo$, $(\msteps_{1,1} + \msteps_{1,2}, \esteps_{1,1} + \esteps_{1,2}) = (\msteps_{1}, \esteps_{1} + 1) $. We then construct $\tderiv'$ as follows
		$$
		\infer
			[\ES]
			{\tyjp{(\msteps_{1,1} + \msteps_{1,2} + \msteps_{2}, \esteps_{1,1} + \esteps_{1,2} + \esteps_{2})}{\cbneedctx \cwc{\var} \esub{\var}{\sctxptwo{\val}\esub{\vartwo}{\tmfour}}}{\typctxtwo_{1} \bigmplus \typctxtwo_{2} \bigmplus \typctxthree}{\mtype}}
			{\tyjp{(\msteps_{1,1}, \esteps_{1,1})}{\cbneedctx \cwc{\var}}{\var \colon \mtypefour;\typctxtwo_{1}}{\mtype}
			\quad
			\infer
				[\ES]
				{\tyjp{(\msteps_{1,2} + \msteps_{2}, \esteps_{1,2} + \esteps_{2})}{\sctxptwo{\val}\esub{\vartwo}{\tmfour}}{\typctxtwo_{2} \bigmplus \typctxthree}{\mtypefour}}
				{\tyjp{(\msteps_{1,2}, \esteps_{1,2})}{\sctxptwo{\val}}{\vartwo \colon \mtypethree ; \typctxtwo_{2}}{\mtypefour}
				\quad
				\tyjp{(\msteps_{2}, \esteps_{2})}{\tmfour}{\typctxthree}{\mtypethree}}
			}
		$$
	Note that $\typctxtwo_{1} \bigmplus \typctxtwo_{2} \bigmplus \typctxthree = \typctxtwo \bigmplus \typctxthree = \typctx $ and $(\msteps_{1,1} + \msteps_{1,2} + \msteps_{2}, \esteps_{1,1} + \esteps_{1,2} + \esteps_{2}) = (\msteps_{1} + \msteps_{2}, \esteps_{1} + \esteps_{2}) = (\msteps, \esteps + 1) $.
	
	\end{itemize}
	
	\end{itemize}

\item \emph{Contextual closure.} We proceed by induction on the derivation of $\tm = \cbneedctx \hole{\tm'} \tond \cbneedctx \hole{\tmtwo'} = \tmtwo$:
	\begin{itemize}
	\item Let $\cbneedctx = \ctxhole$. Then $\tm \rtom \tmtwo$ or $\tm \rtoe \tmtwo$ and in either case the statement holds, as we have just proved.
	\item Let $\cbneedctx = \cbneedctx_{1} \tmthree$. Then $\tderiv \exder[\cbneed] \tyjp{(\msteps,\esteps)}{\cbneedctx_{1} \hole{\tmtwo'} \tmthree}{\typctx}{\mtype}$ and its last typing rule is either $\appsteps$ or $\appgc$. We will only cover the case where $\cbneedctx \hole{\tm'} \tom \cbneedctx \hole{\tmtwo'} $ and $\tderiv$ ends in rule $\appsteps$, leaving the rest of the (analogous) cases to the reader.
	
	Let $\tderiv$ be of the form
		$$
		\infer
			[\appsteps]
			{\tyjp{(\mstepstwo + \mstepsthree + 1, \estepstwo + \estepsthree)}{\cbneedctx_{1} \hole{\tmtwo'} \tmthree}{\typctxtwo \bigmplus \typctxthree}{\mtype}}
			{\tderiv_{\cbneedctx_{1} \hole{\tmtwo}} \exder[\cbneed] \tyjp{(\mstepstwo, \estepstwo)}{\cbneedctx_{1} \hole{\tmtwo'}}{\typctxtwo}{\mult{\ty{\mtypethree}{\mtype}}}
			\quad
			\tyjp{(\mstepsthree, \estepsthree)}{\tmthree}{\typctxthree}{\mtypethree}}
		$$
	where $\typctx = \typctxtwo \bigmplus \typctxthree$, $(\mstepstwo + \mstepsthree + 1, \estepstwo + \estepsthree) = (\msteps, \esteps) $, and $\mtypethree \neq \emptytype$.
	
	We can then apply the \ih on $\tderiv_{\cbneedctx_{1} \hole{\tmtwo}} $ to obtain a typing derivation $\tderiv'_{\ih} \exder[\cbneed] \tyjp{(\msteps + 1, \esteps)}{\cbneedctx_{1} \hole{\tm'}}{\typctxtwo}{\mult{\ty{\mtypethree}{\mtype}}} $ with which we construct $\tderiv'$ as follows
		$$
		\infer
			[\appsteps]
			{\tyjp{(\mstepstwo + 1 + \mstepsthree + 1, \estepstwo + \estepsthree)}{\cbneedctx_{1} \hole{\tm'} \tmthree}{\typctxtwo \bigmplus \typctxthree}{\mtype}}
			{\tyjp{(\msteps + 1, \esteps)}{\cbneedctx_{1} \hole{\tm'}}{\typctxtwo}{\mult{\ty{\mtypethree}{\mtype}}}
			\quad
			\tyjp{(\mstepsthree, \estepsthree)}{\tmthree}{\typctxthree}{\mtypethree}}
		$$
	Note that $(\mstepstwo + 1 + \mstepsthree + 1, \estepstwo + 1 + \estepsthree) = (\msteps + 1, \esteps) $.
		
	\item Let $\cbneedctx = \cbneedctx_{1} \esub{\var}{\tmthree}$. Then $\tderiv \exder[\cbneed] \tyjp{(\msteps,\esteps)}{\cbneedctx_{1} \hole{\tmtwo'} \esub{\var}{\tmthree}}{\typctx}{\mtype}$ and its last typing rule is either $\ES$ or $\ESgc$. We will only cover the case where $\cbneedctx \hole{\tm_{1}} \tom \cbneedctx \hole{\tmtwo_{1}} $ and $\tderiv$ ends in rule $\ES$, leaving the rest of the (analogous) cases to the reader.
	
	Let $\tderiv$ be of the form
		$$
		\infer
			[\ES]
			{\tyjp{(\mstepstwo + \mstepsthree, \estepstwo + \estepsthree)}{\cbneedctx_{1} \hole{\tmtwo'} \esub{\var}{\tmthree}}{\typctxtwo \bigmplus \typctxthree}{\mtype}}
			{\tderiv_{\cbneedctx_{1} \hole{\tmtwo'}} \exder[\cbneed] \tyjp{(\mstepstwo, \estepstwo)}{\cbneedctx_{1} \hole{\tmtwo'}}{\var \colon \mtypethree ; \typctxtwo}{\mtype}
			\quad
			\tyjp{(\mstepsthree, \estepsthree)}{\tmthree}{\typctxthree}{\mtypethree}}
		$$
	where $\typctxtwo \bigmplus \typctxthree = \typctx$, $(\mstepstwo + \mstepsthree, \estepstwo + \estepsthree) = (\msteps, \esteps) $, and $\mtypethree \neq \emptytype$.
	
	We can then apply the \ih on $\tderiv_{\cbneedctx_{1} \hole{\tmtwo'}} $ to obtain a typing derivation $\tderiv'_{\ih} \exder[\cbneed] \tyjp{(\mstepstwo + 1, \estepstwo)}{\cbneedctx_{1} \hole{\tm'}}{\var \colon \mtypethree ; \typctxtwo}{\mtype} $ with which $\tderiv'$ goes as follows
		$$
		\infer
			[\ES]
			{\tyjp{(\mstepstwo + 1 + \mstepsthree, \estepstwo + \estepsthree)}{\cbneedctx_{1} \hole{\tm'} \esub{\var}{\tmthree}}{\typctxtwo \bigmplus \typctxthree}{\mtype}}
			{\tyjp{(\mstepstwo + 1, \estepstwo)}{\cbneedctx_{1} \hole{\tm'}}{\var \colon \mtypethree ; \typctxtwo}{\mtype}
			\quad
			\tyjp{(\mstepsthree, \estepsthree)}{\tmthree}{\typctxthree}{\mtypethree}}
		$$
	Note that $(\mstepstwo + 1 + \mstepsthree, \estepstwo + \estepsthree) = (\msteps + 1, \esteps) $.
		
	\item Let $\cbneedctx = \cbneedctx_{1} \cwc{\var} \esub{\var}{\cbneedctx_{2}}$. We will only consider the case where 
		$$
		\tm = \cbneedctx \hole{\tm'} = \cbneedctx_{1} \cwc{\var} \esub{\var}{\cbneedctx_{2} \hole{\tm'}} \tom \cbneedctx_{1} \cwc{\var} \esub{\var}{\cbneedctx_{2} \hole{\tmtwo'}} = \cbneedctx \hole{\tmtwo'} = \tmtwo
		$$
leaving the (analogous) case when $\tm \toe \tmtwo $ to the reader. 

	Therefore, $\tderiv$ is of the form	
		$$
		\infer
			[\ES]
			{\tyjp{(\mstepstwo + \mstepsthree, \estepstwo + \estepsthree)}{\cbneedctx_{1} \cwc{\var} \esub{\var}{\cbneedctx_{2} \hole{\tmtwo'}}}{\typctxtwo \bigmplus \typctxthree}{\mtype}}
			{\tyjp{(\mstepstwo, \estepstwo)}{\cbneedctx_{1}\cwc{\var}}{\typctxtwo}{\mtype}
			\quad
			\tderiv_{\cbneedctx_{2} \hole{\tmtwo'}} \exder[\cbneed] \tyjp{(\mstepsthree, \estepsthree)}{\cbneedctx_{2} \hole{\tmtwo'}}{\typctxthree}{\mtypethree}}
		$$
	where $\typctxtwo \bigmplus \typctxthree = \typctx$, $(\mstepstwo + \mstepsthree, \estepstwo + \estepsthree) = (\msteps, \esteps) $, and $\mtypethree \neq \emptytype$.
	
	We can then apply the \ih on $\tderiv_{\cbneedctx_{2} \hole{\tmtwo'}}$ to obtain a typing derivation $\tderiv'_{\ih} \exder[\cbneed] \tyjp{(\mstepsthree + 1, \estepsthree)}{\cbneedctx_{2} \hole{\tm'}}{\typctxthree}{\mtypethree}$, with which we can finally construct $\tderiv'$ as follows
		$$
		\infer
			[\ES]
			{\tyjp{(\mstepstwo + \mstepsthree + 1, \estepstwo + \estepsthree)}{\cbneedctx_{1} \cwc{\var} \esub{\var}{\cbneedctx_{2} \hole{\tm'}}}{\typctxtwo \bigmplus \typctxthree}{\mtype}}
			{\tyjp{(\mstepstwo, \estepstwo)}{\cbneedctx_{1}\cwc{\var}}{\typctxtwo}{\mtype}
			\quad
			\tyjp{(\mstepsthree + 1, \estepsthree)}{\cbneedctx_{2} \hole{\tm'}}{\typctxthree}{\mtypethree}}
		$$
	Note that $(\mstepstwo + \mstepsthree + 1, \estepstwo + \estepsthree) = (\msteps + 1, \esteps) $.
\qed
	\end{itemize}

\end{itemize}
\end{proof}

\gettoappendix {thm:need-completeness}
\begin{proof} 
By induction on $\size{\deriv}$; \ie, on $k$ such that $\deriv \colon \tm \rightarrow_\cbneedsym^{k} \tmtwo$.
\begin{itemize}

\item If $k = 0$, then $\tm = \tmtwo$ and \refprop{syntactic-characterization-closed-normal} implies that $\normalpr{\tm}$. We then obtain $\tderiv$ as desired via application of \refprop{need-normal-forms-exist}.

\item If $k > 0$ then $\deriv \colon \tm \tocbneed \tmthree \rightarrow_\cbneedsym^{k-1} \tmtwo$. We can then apply the \ih on $\deriv' \colon \tmthree \rightarrow_\cbneedsym^{k-1} \tmtwo$ to obtain a $\tight$ derivation $\tderiv_{\ih} \exder[\cbneedup] \tyjp{(\size{\deriv'}_{\msteps}, \size{\deriv'}_{\esteps})}{\tmthree}{}{\mult{\normal}}$.

Now, if $\tm \tomcbneed \tmthree$ then \refprop{need-subject-expansion} implies that there exists a derivation as desired, namely $\tderiv \exder[\cbneedup] \tyjp{(\size{\deriv'}_{\msteps} + 1, \size{\deriv'}_{\esteps})}{\tm}{}{\mult{\normal}}$, since $(\size{\deriv'}_{\msteps} + 1, \size{\deriv'}_{\esteps}) = (\size{\deriv}_{\msteps}, \size{\deriv}_{\esteps})$.

If $\tm \toecbneed \tmthree$ then \refprop{need-subject-expansion} implies that there exists a derivation as desired, namely $\tderiv \exder[\cbneedup] \tyjp{(\size{\deriv'}_{\msteps}, \size{\deriv'}_{\esteps} + 1)}{\tm}{}{\mult{\normal}}$, since $(\size{\deriv'}_{\msteps}, \size{\deriv'}_{\esteps} + 1) = (\size{\deriv}_{\msteps}, \size{\deriv}_{\esteps})$.\qed

\end{itemize}
\end{proof}


\section{A New Fundamental Theorem for Call-by-Need (\refsect{new-theorem-need})}

\gettoappendix {coro:value-need} 

\begin{proof}
By induction on $\msteps+\esteps$ and case analysis on whether $\tm$ reduces or not.
	If $\normalpr{\tm}$ then the statement holds with $\tmtwo \defeq \tm$ and $\deriv$ the empty evaluation, so that $\sizem{\deriv} = 0 = \sizee{\deriv}$.
	
	Otherwise $\lnot \normalpr{\tm}$, then $\tm \tocbneed \tmthree$ according to the syntactic characterization of closed $\cbneedsym$-normal forms (\refprop{syntactic-characterization-closed-normal}), since $\tm$ is closed.
	As $\tm \tocbneed \tmthree$ means either $\tm \tomcbneed \tmthree$ or $\tm \toecbneed \tmthree$,
	by quantitative subject reduction for the \cbv multi type system with respect to \cbneed evaluation (\refprop{value-need-subject-reduction}) there is $\tderivtwo \exder[\cbvsym] \!\tyjp{(\mstepstwo,\estepstwo)}{\tmthree}{\typctx}{\mtype}$
	with:
	\begin{itemize}
		\item $\mstepstwo \defeq \msteps - 1$ and $\estepstwo = \esteps$ if $\tm \tomcbneed \tmthree$,
		\item $\mstepstwo \defeq \msteps$ and $\estepstwo = \esteps - 1$ if $\tm \toecbneed \tmthree$.
	\end{itemize}
	By \ih (since $\mstepstwo + \estepstwo = \msteps + \esteps - 1$), there is a term $\tmtwo$ such that $\deriv' \colon \tmthree \tocbneedn \tmtwo$ and $\normalpr{\tmtwo}$, with $\sizem{\deriv'} \leq \mstepstwo$ and $\sizee{\deriv'} \leq \estepstwo$.
	The evaluation $\deriv \colon \tm \tocbneedn \tmtwo$ obtained by prefixing $\deriv'$ with the step $\tm \tocbneed \tmthree$ verifies $\sizem{\deriv} \leq \msteps$ and $\sizee{\deriv} \leq \esteps$ 
	because:
	\begin{itemize}
		\item if $\tm \tomcbneed \tmthree$ then $\sizem{\deriv} = \sizem{\deriv'} + 1 \leq \mstepstwo + 1 = \msteps$ and $\sizee{\deriv} = \sizee{\deriv'} \leq \estepstwo = \esteps$, 
		\item if $\tm \toecbneed \tmthree$ then $\sizem{\deriv} = \sizem{\deriv'} \leq \mstepstwo = \msteps$ and $\sizee{\deriv} = \sizee{\deriv'} + 1 \leq \estepstwo + 1 = \esteps$.
		\qed
	\end{itemize}
\end{proof}

\gettoappendix {cor:value-longer-than-need}

\begin{proof}
	By tight completeness for \cbv (\refthm{value-completeness}), there exists a tight type derivation $\tderiv \exder[\cbvsym] \Deri[(\msteps,\esteps)]{}{\tm}\emptymset$ with $\sizem{\deriv} = \msteps$ and $\sizee{\deriv} = \esteps$, because by hypothesis $\tm$ is \cbv normalisable.
	Correctness of \cbv with respect to \cbneed (\refcoro{value-need}) then gives $\derivtwo \colon \tm \tocbneedn \tmtwo$ with $\normalpr{\tmtwo}$, $\sizem{\derivtwo} \leq \msteps = \sizem{\deriv}$ and $\sizee{\derivtwo} \leq \esteps = \sizee{\deriv}$.
	\qed
\end{proof}

\end{document}